\documentclass[a4paper, twoside,12pt]{memoir} 

\usepackage[american,german]{babel} %
\usepackage[utf8x]{inputenc}

\usepackage{amsmath}
\usepackage{amssymb}
\usepackage{amsthm}
\usepackage{bbm}
\usepackage{bm} 
\usepackage{enumerate}
\usepackage{makeidx}

\usepackage{hyperref} 

\usepackage{url}

\usepackage{lineno}

\usepackage[NoDate]{currvita} 
\usepackage{tikz}

\usetikzlibrary{backgrounds,fit,decorations.pathreplacing}

\usetikzlibrary{decorations.pathmorphing}
\usetikzlibrary{positioning}
\usetikzlibrary{shapes.callouts}

\usepackage{listings}

\theoremstyle{plain}
\newtheorem{lem}{Lemma}[section] 
\newtheorem{cor}[lem]{Corollary} 
\newtheorem{prop}[lem]{Proposition} 
\newtheorem{thm}[lem]{Theorem}   
\newtheorem{con}[lem]{Conjecture}   
\theoremstyle{definition}
\newtheorem{defi}{Definition}[section]

\newcommand{\tr}{\operatorname{tr}}
\newcommand{\ket}[1]{\ensuremath{\vert#1\rangle}}
\newcommand{\bra}[1]{\ensuremath{\langle #1|}}
\newcommand{\braket}[2]{\ensuremath{\langle #1\vert#2\rangle}}
\newcommand{\ketbra}[2]{\ensuremath{\vert#1\rangle\!\langle #2\vert}}
\newcommand{\ZX}{\ensuremath{Z\!X}}

\newcommand{\etal}{\emph{et~al.}}

\usepackage[normalem]{ulem}
\newcommand{\ens}[0]{\ensuremath}
\newcommand{\iE}[0]{\ens{\mathrm{i}}}
\newcommand{\EZ}[0]{\ens{\mathrm{e}}} 
\newcommand{\GF}[1]{\ens{\F_{#1}}}
\newcommand{\GL}[2]{\ens{\mathrm{GL}_{#1}(#2)}}
\newcommand{\Eins}[0]{\ens{\mathbbm{1}}}
\newcommand{\F}[0]{\ens{\mathbb{F}}}
\newcommand{\N}[0]{\ens{\mathbb{N}}}
\newcommand{\Z}[0]{\ens{\mathbb{Z}}}
\newcommand{\R}[0]{\ens{\mathbb{R}}}
\newcommand{\C}[0]{\ens{\mathbb{C}}}

\newcommand{\cB}[0]{\ens{\mathcal{B}}}
\newcommand{\cC}[0]{\ens{\mathcal{C}}}
\newcommand{\cM}[0]{\ens{\mathcal{M}}} 
\newcommand{\fS}[0]{\ens{\mathfrak{S}}}
\newcommand{\fC}[0]{\ens{\mathfrak{C}}}
\newcommand{\cH}[0]{\ens{\mathcal{H}}}

\newcommand{\Mg}[1]{\ens{\left\lbrace #1 \right\rbrace}}
\newcommand{\MgN}[1]{\ens{\Mg{0,\dots,#1}}}
\newcommand{\MgE}[1]{\ens{\Mg{1,\dots,#1}}}
\newcommand{\abs}[1]{\ens{\left|#1\right|}}
\newcommand{\gen}[1]{\ens{\langle #1 \rangle}}
\newcommand{\suc}[0]{\ens{\mathrm{suc}}}
\newcommand{\key}[0]{\ens{\bm{k}}}
\newcommand{\word}[0]{\ens{\bm{w}}}

\newcolumntype{C}[1]{>{\centering\arraybackslash}p{#1}}

\newcommand*\patchAmsMathEnvironmentForLineno[1]{%
  \expandafter\let\csname old#1\expandafter\endcsname\csname #1\endcsname
  \expandafter\let\csname oldend#1\expandafter\endcsname\csname end#1\endcsname
  \renewenvironment{#1}%
     {\linenomath\csname old#1\endcsname}%
     {\csname oldend#1\endcsname\endlinenomath}}%
\newcommand*\patchBothAmsMathEnvironmentsForLineno[1]{%
  \patchAmsMathEnvironmentForLineno{#1}%
  \patchAmsMathEnvironmentForLineno{#1*}}%
\AtBeginDocument{%
\patchBothAmsMathEnvironmentsForLineno{equation}%
\patchBothAmsMathEnvironmentsForLineno{align}%
\patchBothAmsMathEnvironmentsForLineno{flalign}%
\patchBothAmsMathEnvironmentsForLineno{alignat}%
\patchBothAmsMathEnvironmentsForLineno{gather}%
\patchBothAmsMathEnvironmentsForLineno{multline}%
}

\usepackage{color,calc}

\newlength{\linespace}
\makeatletter
\let\ps@plain\ps@empty
\makeatother

\newsavebox{\ChpNumBox}
\definecolor{ChapBlue}{rgb}{0.00,0.65,0.65}
\makeatletter
\newcommand*{\thickhrulefill}{%
\leavevmode\leaders\hrule height 1\p@ \hfill \kern \z@}

\newcommand*\BuildChpNum[2]{%
\thispagestyle{empty}
\begin{tabular}[t]{@{}c@{}}
\makebox[0pt][c]{#1\strut} \\[.5ex]
\colorbox{ChapBlue}{%
\rule[-10em]{0pt}{0pt}%
\rule{1ex}{0pt}\color{black}#2\strut
\rule{1ex}{0pt}}%
\end{tabular}}

\makechapterstyle{BlueBox}{%
\renewcommand{\chapnamefont}{\large\scshape}
\renewcommand{\chapnumfont}{\Huge\bfseries}

\setlength{\beforechapskip}{20pt}
\setlength{\midchapskip}{26pt}
\setlength{\afterchapskip}{40pt}

\renewcommand{\printchapternum}{%
\sbox{\ChpNumBox}{%
\BuildChpNum{\chapnamefont\@chapapp}%
{\chapnumfont\thechapter}}}
\renewcommand{\printchapternonum}{%
\sbox{\ChpNumBox}{%
\BuildChpNum{\chapnamefont\vphantom{\@chapapp}}%
{\chapnumfont\hphantom{\thechapter}}}%
}

}
\chapterstyle{BlueBox}

\makeindex

\begin{document}
 
\pagenumbering{roman}

\selectlanguage{american}


\thispagestyle{empty} 
\begin{center}
  \vspace{0.5cm}
  {\LARGE \textsc{Cyclic Mutually Unbiased Bases}\vspace{0.7cm}}\\ 
  {\large \textsc{and Quantum Public-Key Encryption}}\vspace{1.2cm}\\

   \includegraphics[height=70pt]{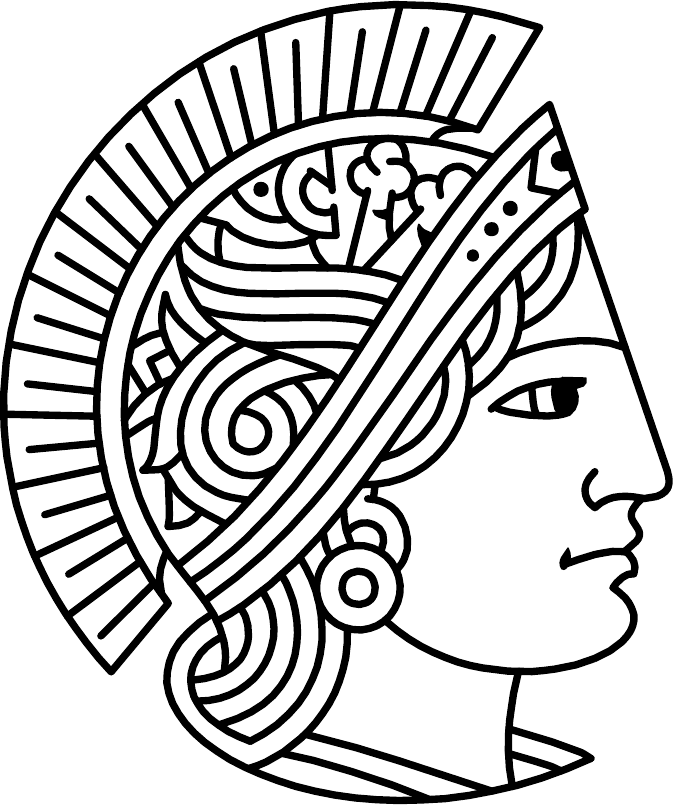}\vspace{1.2cm}\\

  Vom Fachbereich Physik\\
  der Technischen Universität Darmstadt\\
  zur Erlangung der Würde\\
  eines Doktors der Naturwissenschaften\\
  (Dr. rer. nat.)\\
  genehmigte\vspace{1cm}\\

  \textsc{D i s s e r t a t i o n}\vspace{1cm}\\

  von\vspace{1cm}\\

  \textbf{Dipl.-Phys. Ulrich Seyfarth}\vspace{1cm}\\

  aus Hadamar\vspace{1.5cm}\\

  Darmstädter Dissertation\\
  Darmstadt 2013\\
  D17 
  \vspace{0.5cm}\\
\end{center}


\thispagestyle{empty}\enlargethispage{5cm}

\newpage\phantom{Dies hier nur, damit vfill wirkt.}\vfill 
\begin{center}\begin{tabular}{ll}
  Referent:                   & Prof. Dr. rer. nat. Gernot Alber\\
  Korreferent:                & Prof. Dr. rer. nat. Robert Roth\\
  Tag der Einreichung:        & Mittwoch, 28. November 2012\\
  Tag der mündlichen Prüfung: & Montag, 11. Februar 2013\\
\end{tabular}\end{center}

\thispagestyle{empty}
\newpage




%


\chapter*{Cyclic mutually unbiased bases and quantum public-key encryption}
\section*{Abstract}
Based on quantum physical phenomena, quantum information theory has a potential which
goes beyond the classical conditions. Equipped with the resource of complementary information as an intrinsic
property it offers many new perspectives. The field of quantum key distribution, which enables the
ability to implement unconditional security, profits directly from this resource. To measure the state of
quantum systems itself for different purposes in quantum information theory, which may be related to
the construction of a quantum computer, as well as to realize quantum key distribution schemes, a
certain set of bases is necessary. A type of set which is minimal is given by a complete set of mutually unbiased bases.
The construction of these sets is discussed in the first part of this work. We present complete sets of
mutually unbiased bases which are equipped with the additional property to be constructed cyclically, which means,
each basis in the set is the power of a specific 
generating basis of the set. Whereas complete sets of mutually unbiased bases are related to many mathematical problems,
it is shown that a new construction of cyclic sets is related to Fibonacci polynomials. Within this context,
the existence of a symmetric companion matrix over the finite field $\F_2$ is conjectured. For all
Hilbert spaces which have a finite dimension that is a power of two ($d=2^m$), the cyclic sets can be generated explicitely with the discussed
methods. Results for $m=\MgE{600}$ are given. A generalization of this construction is able to generate sets with
different entanglement structures. It is shown that for dimensions $d=2^{2^k}$ with $k$ being a positive
integer, a recursive construction of complete sets exists at least for $k \in \MgN{11}$, where for higher dimensions
a direct connection to an open conjecture in finite field theory by Wiedemann is identified. All discussed sets can
be implemented directly into a quantum circuit by an invented algorithm. The (unitary) equivalence of the considered sets is
discussed in detail.\par
In the second part of this work the security of a quantum public-key encryption protocol is discussed, which was recently
published by Nikolopoulos~\cite{Niko08}, where the information of all published keys is taken into account. Lower bounds on two different
security parameters are given and an attack on single qubits is introduced which is asymptotically equivalent to
the optimal attack. Finally, a generalization of this protocol is given that permits a noisy-preprocessing step
and leads to a higher security against the presented attack for two leaked copies of the public key and to first
results for a non-optimal implementation of the original protocol.


\chapter*{Zyklische komplementäre Basen und Quantenkryptographie mit öffentlichen Schlüsseln}
\section*{Kurzfassung}
\selectlanguage{german}
Quantenmechanische Phänomene verleihen der Quanteninformationstheorie ein Potenzial, welches
über die klassische Informationstheorie hinausgeht. Die hierin verankerte Fähigkeit,
komplementäre Information zu erzeugen, bietet viele neue Möglichkeiten. Die Theorie zur
Quantenschlüsselverteilung nutzt diese Information unmittelbar aus, um beweisbar sichere kryptografische
Verfahren umzusetzen. Zur Realisierung einer solchen Quantenschlüsselverteilung, aber auch zur Bestimmung eines Quantenzustandes,
beispielsweise um einen Quantencomputer zu realisieren, werden gewisse Mengen von Messbasen
benötigt. Eine kleinst\-mö\-gli\-che Menge dieser Messbasen ist eine sogenannte vollständige Menge von komplementären Basen.
Im ersten Teil dieser Arbeit wird die Konstruktion solcher Mengen betrachtet. Diese haben die zusätzliche
Eigenschaft, zyklisch zu sein, d.\,h. jedes Element der Menge lässt sich als Vielfaches
eines bestimmten Generatorelementes der Menge erzeugen. Es wurde bereits gezeigt, dass die
Theorie der komplementären Basen mit einigen anderen mathematischen Gebieten verwandt ist.
Hier wird ein Zusammenhang der Konstruktion von zyklischen komplementären Basen
mit Fibonaccipolynomen beleuchtet. Weiterhin wird die Existenz einer symmetrischen
Begleitmatrix über dem endlichen Körper $\F_2$ vermutet. Die behandelten zyklischen Mengen von
komplementären Basen können für alle endlichen Hilbertraumdimensionen explizit
erzeugt werden, deren Dimension ein Vielfaches von zwei ist ($d=2^m$); Ergebnisse für $m=\MgE{600}$
werden aufgeführt. Eine Verallgemeinerung dieser Konstruktion ist in der Lage, Mengen zu erzeugen, welche
eine alternative Struktur der Verschränkung aufweisen.
Für den Fall dass die Dimension des Hilbertraumes
$d=2^{2^k}$ beträgt, wobei $k$ eine positive ganze Zahl ist, existiert eine rekursive Erzeugungsmethode,
solange $k \in \MgN{11}$ gilt. Für alle höheren Werte von~$k$ wird diese Konstruktion mit einer offenen Vermutung von Wiedemann
aus dem Bereich der Theorie endlicher Körper in Verbindung gebracht.
Alle behandelten Mengen lassen sich mithilfe
eines vorgestellten Algorithmus unmittelbar als
Quantenschaltkreis realisieren. Die (unitäre) Äquivalenz verschiedener behandelter Mengen wird ebenfalls im
Detail betrachtet.\par
Der zweite Teil dieser Arbeit behandelt die Sicherheit eines kürzlich von Nikolopoulos \cite{Niko08}
vorgestellten asymmetrischen Verschlüsselungsprotokolls, welches öffentliche Quantenschlüssel verwendet, wobei der Informationsgewinn
aus allen ver\-öffent\-lichten Schlüsseln für die Betrachtung eines potenziellen Lauschers berücksichtigt wird.
Es werden untere Schranken für zwei verschiedene Sicherheitsparameter angegeben sowie ein Angriff besprochen,
welcher einfach zu realisieren ist, da er nur einzelne Qubits misst. Es wird weiterhin gezeigt, dass dieser
asymptotisch äquivalent zu einem optimalen Angriff ist, welcher physikalisch schwieriger umzusetzen ist.
Abschließend wird eine Verallgemeinerung des Protokolls vorgestellt, welche durch das absichtliche Einbauen
von Störungen zu einer höheren Sicherheit führt. Exemplarisch wird dies für den Fall gezeigt, dass
ein bzw. zwei Exemplare des öffentlichen Schlüssels vom Angreifer abgefangen werden. Diese Verallgemeinerung
kann auch zur Betrachtung einer nicht idealisierten Realisierung des Ausgangsprotokolls genutzt werden.
\selectlanguage{american}

\newpage



\chapter*{Acknowledgments}
\thispagestyle{empty}
This thesis was prepared in the context of my research activities in the group of Prof. Dr. Gernot Alber.
I would like to thank him for his encouragement and support, for his advice and trust in me, and for
inviting great guests.\par
I also like to thank Prof. Dr. Robert Roth for being the second referee of this thesis
and investing his time in studying my work.\par
I am grateful to CASED for supporting me by a scholarship and for the opportunity to work in an interdisciplinary environment.\par
The problem of \emph{constructing complete sets of cyclic mutually unbiased bases}, which is a main goal in the first part
of this work, was brought to me by Christopher Charnes, who visited our group in 2008 for several month
and who deserves many thanks from me, also for introducing me
into mathematical theories like finite field theory or representation theory.\par
Many thanks to Oliver Kern
and Kedar Ranade for having the first ideas on the construction of the sets, for the fruitful collaboration
and many nice and funny discussions.\par
Influenced by discussions with Luis Sánchez-Soto, who I like to thank very much, the focus of the first
part moved a little towards the construction of complete sets of cyclic mutually bases with different
entanglement properties.\par
The second part of this work is based on a manuscript of Georgios Nikolopoulos,
who was also guest in our group and who had the patience to discuss with me my ideas on the security of his
protocol. Let me also thank him for this very nice collaboration.\par
Finally, many thanks for the valuable contributions to both parts of the work which were done in collaboration with
the Bachelor students Niklas Dittmann and Walid Mian.\par
I always enjoyed the environment in the complete group with many discussions and debates on physical topics,
as well as diverse topics which are not related to physics at all. So let me thank all members and former members
I met in this group.\par
I cannot thank Kedar Ranade enough for improving my thesis substantially by proofreading the draft very carefully, for the very fruitful
collaboration, many discussions and his frequent visits.\par
Most importantly, I would like to thank my family, especially for lighten the load for me as much as possible, for tolerating my
addiction to think about research problems almost everywhere, and also for assisting me with nearly everything in 
the last weeks before the submission. Thank you, Andrea! Thank you for your belief in me and your love. Many thanks to Johanna for being such
a nice daughter, who sends laughs the whole day and who had to miss her daddy too much in the last few weeks. It is wonderful to have both of you by my side.\par
Finally, many thanks to my parents, who rendered all this possible and gave me always
the liberty to follow my own path\ldots\thispagestyle{empty}

\cleardoublepage

\tableofcontents
\cleardoublepage

\addcontentsline{toc}{chapter}{Notation}
\chapter*{Notation}
In this document, several mathematical symbols are used. To avoid misinterpretations,
we list the most important symbols in the following, starting with the definition of
different sets, operators, and matrices. The section ends by a description of the notation
of vectors.

\paragraph{Sets}

\begin{center}
\begin{tabular}{ll}
 Symbol & Meaning\\
 $K$ & A general set\\
 $\R$ & Set of real numbers\\
 $\C$ & Set of complex numbers\\
 $\N$ & Set of natural numbers starting with zero\\
 $\N^*$ & Set of natural numbers starting with one\\
 $\Z$ & Set of integers\\
 $\Z_m$ & Set of integers modulo $m$ with $m \in \N^*$\\
 $\F_p$ & Finite field with $p$ elements, $p \in \N^*$ and $p$ prime\\
 $\F_{p^m}$ & Finite field with $p^m$ elements, $p, m \in \N^*$ and $p$ prime\\
 $M_m(K)$ & Set of $m \times m$ matrices with entries from $K$\\
 $\GL{m}{K}$ & Group of invertible $m \times m$ matrices with entries from $K$
\end{tabular}
\end{center}

\paragraph{Operations}\mbox{}\\
\begin{center}
\begin{tabular}{ll}
 Symbol & Meaning\\
 $\oplus_m$ & Addition of two values modulo $m \in \N^*$\\
 $\gcd(a,b)$ & Greatest common divisor of $a,b \in \N^*$\\
 $[x]$ & Largest integer not greater than $x \in \R$ (Gaussian floor)\\
 $\lceil x \rceil$ & Smallest integer not less than $x \in \R$ (Gaussian ceiling)\\
 $\langle a \rangle$ & Group generated by the element $a$\\
 $a \otimes b$ & Tensor product of two linear operations $a,b$\\
  & which act on a finite dimensional Hilbert space
\end{tabular}
\end{center}

\pagebreak

\paragraph{Matrices}\mbox{}\\
\begin{center}
\begin{tabular}{ll}
 Symbol & Meaning\\
 $\Eins_m$ & $m$-dimensional identity matrix\\
 $0_m$ & $m \times m$ zero matrix
\end{tabular}
\end{center}

\paragraph{Vectors}\mbox{}\\
\phantom{a}\\
Usually, vectors are written in the form $\vec a =(a_1,\ldots,a_m)^t$, with numbers
$a_1,\ldots,a_m$ from a certain set, $m \in \N$ and $(\cdot)^t$ denoting the transposition
of a vector or matrix. In few cases, where it is mentioned in addition, the same vector is
given by $\bm{a}$ for a better readability.


\cleardoublepage

\pagenumbering{arabic}

\chapter{Introduction and outline}
As long as humans populate our world, their addiction to analyze the incidents 
and the behavior of this world is gigantic. On the one hand, this attitude may result from
the benefits this knowledge provides. On the other hand, it seems to be based on
the pure exploration urge. Already a long time ago, humans started to systematize their
knowledge; it leads to specifically adapted methods for the different fields and
allows a better overview which is helpful in order to teach younger generations.
A first approach which coined nowadays methods in the occidental culture was given
by the \emph{Platonic Academy}. The idea of science was introduced
and different sciences were defined. Aristoteles was one of the first members of this
Academy who was motivated in discussing natural phenomena; the physical laws
were summarized systematically the first time by Isaac Newton in his famous \emph{Principia Mathematica}
in the $17$th century. Later on, based on the number of results and obviously the different
forms of research, natural sciences were divided amongst others into biology, chemistry,
and physics. Around the $19$th century, different research topics on physics were
of interest, which are nowadays called the areas of \emph{classical physics}, namely
mechanics, electrodynamics, thermodynamics, and optics. With the beginning of the $20$th
century, \emph{modern physics} came up with the theory of \emph{general relativity} and \emph{quantum mechanics}.
The latter became necessary in order to avoid the ultraviolet catastrophe which came up with
the Rayleigh–Jeans law which was derived to describe black-body radiation. Max Planck finally solved this
problem by introducing the so-called Planck constant $h$. The aim of quantum mechanics is,
roughly speaking, the description of physical phenomena at microscopic scales, but many descriptions succeeded this idea
which are contrary to the usual expectation. There are Heisenberg's \emph{uncertainty principle} (\emph{complementary}
variables cannot be measured perfectly), \emph{entanglement} (a system can have more
information than the sum of the information of the subsystems), \emph{no-cloning-theorem}
(\emph{unknown} quantum systems cannot be copied perfectly), and many more. Fundamental
legitimation problems still arise from the Copenhagen interpretation which describes the measurement
process as a wave-function collapse which is not convenient to the normal time evolution in
quantum physics. Bell introduced his famous no-go theorem which describes a test to distinguish
between system states which occur in classical physics and systems states which occur only in quantum
physics. His goal was to solve the famous Einstein-Podolsky-Rosen paradox which challenged
the Copenhagen interpretation in order to ask whether quantum mechanics is complete. Along with these
fundamental questions, the effects which originate from the theory of quantum mechanics (and which
are also observed in experiments) have the potential for several applications.\par

As long as humans populate our world, their creativity in order to avoid undesired duties
is unlimited. An important benefit is to save time for more important concerns, another is possibly
the unattractiveness of tasks which can be schematized easily. A first important innovation were water clocks,
already known in Babylon which were able to measure time automatically. Based on the steam engine,
the industrial revolution started in the $18$th century and helped the workers to implement larger
projects. Around the same time, Charles Babbage invented the first \emph{mechanical computer} and the
concept of a programmable computer in order to automatize tedious and error-prone calculations. 
Ada Lovelace worked theoretically on this computer in the $19$th century and is seen as the first
programmer. The first digital computer was build in the $20$th century and was able to solve simple
mathematical problems; with the miniaturization computers became more and more powerful and helped
to solve complicated mathematical problems. Numerical methods and the idea of simulating systems
with computers took place into the scientific research. Nevertheless, many problems are too hard to
be implementable efficiently into classical computers.\par

As long as humans populate our world, they have secrets which they like to share only with selected
persons. Different \emph{ciphers} were known already in the ancient Greece in order to hide information or
to make information unreadable to third parties. In the $20$th century, more complicated ciphers
were implemented by machines, like the \emph{Enigma} during the second world war, which was able to encrypt
and to decrypt messages in a complicated but logical way. As all algorithms were broken, they became
more and more complex. In 1882, a simple cipher, the \emph{one-time pad} was described by Frank Miller%
\footnote{Steven Bellovin figured out in 2011 that Frank Miller invented the one-time pad $35$ years
before Gilbert Vernam and Joseph Mauborgne, who were in general seen as the inventors \cite{Bellovin11}.}. It
encodes each letter of the message with an individual letter of the key, thus the key has to be as long
as the message. It can be proven by methods of information theory, that this cipher is unbreakable, if
the key is perfectly random, only used once and not leaked by a third party. Those problems seem to
be unsolvable in the classical ways.\par

A combination of these three mentioned research areas, namely \emph{quantum mechanics}, \emph{computer sciences},
and \emph{cryptography}, is given by the field of \emph{quantum information theory}.
Equipped with the effects of quantum mechanics, it seems that a new generation of computers, so-called
\emph{quantum computers}, may have the capability to expand the efficiency of classical computers dramatically.
The first example is the Deutsch–Jozsa algorithm, which scales exponentially faster on a quantum computer, but
is mostly of scientific interest \cite{DJ92}. The first quantum algorithm which demonstrated the practical potential of the
quantum computer was invented by Peter Shor in 1994 \cite{Shor97} and is able to perform prime factorization in polynomial
time--which does not seem to be possible for classical computers. As many common asymmetric cryptographic ciphers such
as RSA are based on the difficulty of this problem, Shor's algorithm would affect instantly the security
of nowadays secret information--which is important for the security of credit cards, for instance.\footnote{In 1996,
another important quantum algorithm was found by Lov Grover \cite{Grover96}, which improves the search in an unsorted list quadratically
whereas the classical algorithm seems to be optimal within a classical setup.} Already in 1984, ideas of Charles Bennett and Gilles Brassard paved the
way to recover a security scheme which may replace the classical cryptography by \emph{quantum cryptography}. After Shor's algorithm
was known, security proofs for the \emph{quantum key distribution} protocol were given.
Finally, as quantum computers seem to be efficient, they are seen as an attractive candidate to simulate
quantum systems for their analysis.\par
To realize a quantum computer, different methods and tools need to be explored. For several purposes, the construction
of \emph{complete sets of cyclic mutually unbiased bases} is relevant as will be seen in Section \ref{sec:intro:mubs}. These bases find
their applications also in the field of quantum cryptography. Regarding the development of the field of quantum cryptography,
the security of many protocols is important and should be analyzed. A discussion on the security of a recently invented quantum
public-key encryption scheme is started in Section~\ref{sec:intro:pqke}.

\section{Cyclic mutually unbiased bases}\label{sec:intro:mubs}
The first part of this work deals with the problem of constructing \emph{complete sets of cyclic mutually unbiased bases} (MUBs).\par
A detailed introduction into the history, the properties and the applications of MUBs is given in \textbf{Chapter \ref{chap:mubsintro}}.
Complete sets of MUBs play an important role for quantum state tomography of finite dimensional quantum systems, as
they define a minimal set of measurement bases. This qualifies them for example to be a candidate which measures states of quantum registers,
a part of a quantum computer. Furthermore, MUBs have a large potential in quantum cryptographic protocols. In order to
increase the efficiency of these protocols it turned out that, by using higher-dimensional information carriers, complete sets of MUBs
need to be constructed for higher dimensions, most suitable with a cyclicity property \cite{Chau02, Chau05, Ranade06}. This cyclicity
is a special property those sets may obey, which means that a whole set is generated by the powers of a single element. As it is not even
known yet if complete sets exist in all complex Hilbert spaces with a finite dimension, MUBs are still a field of current research. \par
Whereas first ideas on complete sets of MUBs were given by Ivanović~\cite{Ivanovic81} and Wootters and Fields \cite{WF89},
applications of cyclic sets and alternative constructions were discussed many years later by Chau \cite{Chau05} and
Gow~\cite{Gow07}. These ideas are retraced in \textbf{Chapter \ref{chap:fundamentals}}, together with mathematical methods
which are used later on.\par
\pagebreak
In \textbf{Chapter \ref{chap:cyclic_mubs}}, a construction of complete sets of cyclic MUBs is introduced.
Following the methods of a systematic scheme of Bandyopadhyay \etal{} \cite{Bandy02}, we published a first work in $2010$ in which we
construct cyclic sets of MUBs in all even prime power dimensions \cite{KRS10}:
\begin{quotation}
 \noindent
 \emph{Complete sets of cyclic mutually unbiased bases in even prime-power dimensions},\\
 by Oliver Kern, Kedar S. Ranade, and Ulrich Seyfarth,\\
 in Journal of Physics A \textbf{43}, 275305 (2010).
\end{quotation}
The introduced methods allow to reduce the problem of explicitely constructing a complete set of MUBs for a Hilbert space of dimension $d=2^m$ in a first step
from $(2^m)^2$ free variables of $\Z_4$ to $m^2$ variables of $\Z_2$.
With the help of a second step, an assumed reduction of the search space, the number of free variables goes below $(m/2)^2$.\par
A second work formalizes and extends these results as it shows the relation of the construction with so-called Fibonacci polynomials \cite{SR12}:
\begin{quotation}
 \noindent
 \emph{Cyclic mutually unbiased bases, Fibonacci polynomials and Wiedemann's conjecture},\\
 by Ulrich Seyfarth and Kedar S. Ranade,\\
 in Journal of Mathematical Physics \textbf{53}, 062201 (2012).
\end{quotation}
Furthermore, we prove the existence of complete sets of cyclic MUBs for the discussed construction. Results of both manuscripts are given in Section
\ref{sec:fibset}; numerical methods which deal with the search of the solutions in the remaining space which is spanned by the free
variables are summarized in Section~\ref{subsec:fibset:numerical}. An analytical approach which may lead to a \emph{symmetric
companion matrix}\footnote{The conjectured construction of a symmetric companion matrix assumes, that for each polynomial with
coefficients in $\F_2$, a symmetric matrix can be constructed which has that polynomial as its characteristic polynomial.} is
given in Section \ref{subsec:fibset:analytical}. In Section \ref{sec:fermatset} results of the second work
are presented, which show that for $m=2^k$ with $k\in \N$, a complete set of cyclic MUBs for dimension $d=2^{2^{k}}$ can be constructed
recursively. Namely, it can be constructed from the complete set of cyclic MUBs for dimension $d=2^{2^{k-1}}$ at least for $k \in \MgE{11}$, which is shown in
Appendix \ref{app:fermat_based:wiedemann}, limited by the largest \emph{Fermat number} for which the prime factorization is known. For
all $k$, it is proven in that section that the problem is related to an open conjecture in finite field theory
by Wiedemann \cite{Wiedemann88} which is still of current interest~\cite{MS96,Voloch10}; an approach for a proof is given in Appendix \ref{app:wiedemann_proofs}. In Section \ref{sec:standardform}
a unique form of representing complete sets of cyclic MUBs in order to be able to compare different sets is given.\par
Discussions with L.~L.~Sánchez-Soto drew the author's attention to the problem of constructing complete sets of cyclic MUBs with different entanglement properties.
An introduction is given in Section \ref{sec:entangleprop}, a first subclass, presented in Section \ref{sec:homogeneous:group}, is
a generalization of the construction which was explored in the two mentioned articles. First approaches on two other classes are given in Sections~\ref{sec:homogeneous:semigroup}
and \ref{sec:inhomogeneous}. Results can be found in Appendices \ref{app:homosets:group}, \ref{app:homosets:semigroup}, and \ref{app:inhomosets}.\par

An important representation of sets of MUBs are sets of unitary operators. Therefore, a transformation of the different sets into a unitary operator
representation is given in Section \ref{sec:Uconstruction}, which is published in the first work. For those sets which can be constructed recursively,
we published another manuscript~\cite{SR11}:
\begin{quotation}
 \noindent
 \emph{Construction of mutually unbiased bases with cyclic symmetry for qubit systems},\\
 by Ulrich Seyfarth and Kedar S. Ranade,\\
 in Physical Review A \textbf{84}, 042327 (2011).
\end{quotation}
The first part of these results is shown in Section \ref{sec:unitop:fermset}, namely that the corresponding unitary operator can also be constructed recursively. Finally,
the generators of these complete sets of MUBs can be implemented by a quantum circuit into an experimental setup. In the context of
the Bachelor thesis of N.~Dittmann, the construction of the circuit was generalized for all discussed sets of cyclic MUBs and even
more general operators, which is presented in Section \ref{sec:circuit}. As the implementation of a large cyclic set may accumulate
errors, a more practical implementation for such sets is drawn in Section \ref{sec:gatedecomp:practical}, which makes nevertheless
use of the cyclic structure.\par
\textbf{Chapter \ref{chap:equivalence}} deals with the equivalence of MUBs. A slight generalization of the results of \cite{SR12} is presented
and it is proved
that the introduced method constructs complete sets of cyclic MUBs which are (unitary) equivalent to others like the
non-cyclic sets constructed by Wootters and Fields \cite{WF89}.\par
Finally, the results are concluded in \textbf{Chapter \ref{chap:conclusionsMUBs}} and an outlook on possible future research topics is givens.\par
To keep the sections short, many results are shown in the appendices as well as basic mathematical properties.\
In \textbf{Appendix \ref{app:algebra}}, tools from algebra and quantum information theory which are important for this
work are summarized.\par
An approach to prove Wiedemann's conjecture is presented in \textbf{Appendix~\ref{app:wiedemann_proofs}}.\par
Most of the computational results are given in a compressed form in\linebreak \textbf{Appendix \ref{app:results}}. Section \ref{app:fractals}
shows the appearance of similar fractal patterns in Fibonacci polynomials and characteristic polynomials of certain matrices
which are both important for the construction of cyclic MUBs. Section \ref{app:fibonacci_based:triangle} lists generators
for complete sets of cyclic MUBs for dimensions $d=2^m$ with $m \in \Mg{2,\ldots,600}$ for the method introduced in \cite{KRS10},
taking advantage of the improvements of \cite{SR12}. The results which may indicate the existence (and maybe a construction) of a symmetric companion
matrix are given shortly in Section \ref{app:fibonacci_based:companion}. Then, generators of sets of cyclic MUBs
with different entanglement properties are listed for four-qubit systems in Sections \ref{app:homosets:group} and \ref{app:homosets:semigroup}
and in Section \ref{app:inhomosets}, respectively. Finally, the \emph{Matlab} code which is used to test Wiedemann's
conjecture for $k \in \MgN{11}$ is shown in Section \ref{app:fermat_based:wiedemann}.

\section{Quantum public-key encryption}\label{sec:intro:pqke}
The second part of this work treats the security of an asymmetric protocol of quantum cryptography which was recently published by
Nikolopoulos \cite{Niko08}.
As quantum cryptography wants to offer \emph{unbounded security}, detailed security analyses are essential. Since no security
proof exists yet for this commonly known \emph{quantum public-key encryption} (QPKE) protocol, a detailed analysis was published \cite{SNA12}:
\begin{quotation}
 \noindent
 \emph{Symmetries and security of a quantum-public-key encryption based on single-qubit rotations},\\
 by Ulrich Seyfarth, Georgios M. Nikolopoulos, and Gernot Alber,\\
 in Physical Review A \textbf{85}, 022342 (2011).
\end{quotation}
In \textbf{Chapter \ref{chap:georgios}}, an introduction into the protocol is given.\par
Within \textbf{Chapter \ref{chap:invest}} all new investigations on the security of the protocol
are depicted. Namely, Section \ref{sec:privkey} deals with the security of the private key
which is essential for the security of the protocol. If though, the security of encrypted messages has to be guaranteed in addition, as they
can be attacked directly and in the worst case by using all copies of the public key. This message security is analyzed in Section~\ref{sec:message},
where the security against a used security parameter is shown for an attack which can be implemented easily. Nevertheless, it is
also shown that this attack has a similar behavior as a class of very general attacks.\par
In the context of the Bachelor thesis of W.~Mian the effect of a noisy-preprocessing step was analyzed. With two simple attacks
in Sections \ref{sec:singletest} and \ref{sec:doubletest} it is shown in Section \ref{sec:noisysecurity} that the security of the
message increases with the help of this method. Nikolopoulos' protocol, which is assumed to work in an \emph{ideal} description of the world,
has to be transformed into a protocol which is used in a \emph{real} description of the world, as errors come into account. But
therefore, this method can also be used in order to model errors in the construction process of the public key. Obviously, the
implementation of error correction protocols needs to be discussed within further work.\par
The second part is concluded and an outlook is given in \textbf{Chapter \ref{chap:sumpk}}.


\part{Cyclic mutually unbiased bases}\label{part:mubs}

\chapter{Introduction}\label{chap:mubsintro}
A fundamental characteristic of quantum mechanical systems is the \emph{uncertainty principle} which was formulated
by Heisenberg in 1927 \cite{Heisenberg27}. It describes the observation that pairs of physical variables
exist which cannot be measured simultaneously with maximal precision. If this mutual influence is maximal,
the variables are called \emph{complementary}.\footnote{A common example is the measurement outcome of the electron
spin; if the spin is known to be in an eigenstate of the Pauli-$\sigma_x$ operator, the outcome of a
measurement in the basis of the Pauli-$\sigma_z$ operator is completely undetermined.} As finite-dimensional quantum mechanical
systems are described by density matrices that are defined in the Hilbert space $\cH = \C^d$, a set of operators
exists that defines a \emph{unitary operator basis} (and is capable to describe properties of complementary variables).
Back in 1960, Schwinger derived initially a \emph{complementary pair of operators} that is able to describe
two complementary variables in an arbitrary finite dimensional Hilbert space \cite{Schwinger60}. It turned out
that the absolute overlap of two vectors from different bases is constant for a given dimension. Ivanović followed
this path, motivated by the complete state estimation of an unknown quantum state, and figured out that the minimal
number of pairwise complementary operators needed to determine the quantum state completely, equals the dimension of
the Hilbert space plus one \cite{Ivanovic81}. His considerations result in a construction method for a \emph{complete set}
in prime Hilbert space dimensions, whereas Wootters coined the notion of \emph{mutually unbiased bases} \index{Mutually unbiased bases} (MUBs) for pairwise
complementary operators \cite{Wootters86}. In collaboration with Fields, he presented a construction method for
complete sets of MUBs in prime power Hilbert space dimensions \cite{WF89}. Pointing back again to Heisenberg's uncertainty
principle, it was conjectured by Kraus and shown by Maassen und Uffink that the optimal solution for a certain
kind of uncertainty relations is a set of MUBs~\cite{Kraus87,MU88}.\par
It still remains an open question that is under current research, how many bases a maximal set of MUBs contains
for \emph{composite} Hilbert space dimensions.%
\footnote{The term ``composite dimension'' refers to dimensions that cannot be represented as the power of a prime number.}
Even for the smallest such dimension, which is $d=6$, this problem is unsolved, where the largest known sets have
only three elements.\footnote{There is strong evidence that these sets are maximal \cite{Zauner99,Grassl04,Bengtsson07,Raynal11}.}
Also many different constructions to that given by Wootters and Fields were figured out (e.\,g.~\cite{Bandy02, Klappenecker04}).\footnote{A recent review article
on MUBs was published by Durt \etal{} \cite{Durt10}.}
In 2005, Chau proved a theorem which predicts the existence of a cyclic group generator that can be used to generate a
complete set of MUBs \cite{Chau05};%
\footnote{The proof by Chau is based on finite field theory, whereas Gow gave a proof based on representation theory \cite{Gow07}.}
in other words, such a complete set of \emph{cyclic} MUBs is given by the powers of
a single unitary operator. It was shown by Gow \cite{Gow07}, that cyclic sets exist only for \emph{even-prime power} dimensions. An example of such a
cyclic set was already used by Gottesman in 1998 in order to transform the three Pauli operators cyclically \cite{Gott98}. The
advantages of a cyclic generation of a set of MUBs were used in proofs of the quantum cryptographic six-state protocol~\cite{Lo01,GL03}
and even an abstract definition of such sets was taken into account to prove generalizations of that
protocol \cite{Chau05}.\par
The aim of the first part of this work is to continue this path by finding explicit constructions of cyclic MUBs in
a straightforward way with suggestions for a direct implementation of these sets into experimental setups. After addressing
the fundamental properties of MUBs, their usage in the fields of \emph{quantum state estimation}, \emph{quantum key distribution},
and further relations in the subsequent sections, different well-known constructions will be discussed in Section \ref{sec:constructions}
of Chapter \ref{chap:fundamentals}. Properties of the so-called \emph{Fibonacci polynomials} are presented and extended in
Section \ref{sec:fibonacci_polynomials}, which are fundamental for the generation of complete sets of cyclic MUBs in
Chapter \ref{chap:cyclic_mubs}. Three different constructions are given in Sections \ref{sec:fermatset}, \ref{sec:homogeneous},
and \ref{sec:inhomogeneous}, that aim on sets with specific entanglement properties which are derived in Section \ref{sec:entangleprop}.
It turns out that another form, the \emph{Fermat-based sets} which are discussed in Section \ref{sec:fibset}, are
related to an open conjecture in finite field theory which was given in 1988 by Wiedemann~\cite{Wiedemann88}. Supposed this conjecture
is true, an unlimited class of complete sets of cyclic MUBs can be created by the given recursive construction. Until then, the results
given in Appendix~\ref{app:fermat_based:wiedemann} can be used for the construction of complete sets of cyclic MUBs for dimensions $d=2^{2^k}$
with $k \in \MgN{11}$.\footnote{The limitation to $k=11$ is limited due to the largest known prime factorization of Fermat numbers.}
By their nice form, these sets can be implemented quite easily into \emph{quantum circuits} as shown in Section \ref{sec:gatedecomp:fibset}.
For more general cyclic MUBs, another method is presented in Section \ref{sec:gatedecomp:homogeneous}. An improved practical
implementation is suggested by a promising method in Section \ref{sec:gatedecomp:practical}. The classification of the generated sets of MUBs
is derived in Chapter \ref{chap:equivalence}, where the \emph{equivalence} of the sets with known constructions is discussed.
The results are summarized and an outlook is provided in Chapter \ref{chap:conclusionsMUBs}.

\section{Relations}
From a mathematical point of view, MUBs are related to many mathematical objects of other research fields. In the context of the existence of
a complete set of MUBs with $d+1$ bases, a well-known connection to \emph{orthogonal Latin squares} was established~\cite{Zauner99,Gibbons04,Paterek09},
but also \emph{finite projective planes} play an important role \cite{Saniga04,Bengtsson04,Bengtsson05}. In the
context of equivalence of MUBs, relations to \emph{symplectic spreads}~\cite{Calderbank97} and \emph{affine planes} \cite{Kantor12}
are of importance.

\section{Quantum state estimation}\label{sec:qsestimation}
As already mentioned, MUBs were originally introduced in the context of quantum state estimation \cite{Ivanovic81,WF89}
and seen as a good candidate for optimal schemes for \emph{quantum state tomography}\index{Quantum state tomography}. For a systematic
approach, we may observe the general state of a quantum system, which is defined in a $d$-dimensional Hilbert space $\cH=\C^d$ by a density operator
\begin{align}
 \rho = \sum_i p_i \ketbra{\psi_i}{\psi_i},
\end{align}
which is diagonal in its orthonormal eigenbasis $\{\ket{\psi_i}\}_i$ with a normalized probability distribution, where $\sum_i p_i = 1$ holds.
More generally, we can describe the state vectors $\ket{\psi_i}$ as the eigenvectors of $\rho$ with their corresponding eigenvalues~$p_i$.
An experiment which measures the state $\rho$ of the system in its orthonormal eigenbasis, will measure the state $\ket{\psi_i}$
with probability $p_i$. Since these probabilities have to be real, quantum mechanics postulates $\rho$ to be Hermitian, thus
$\rho = \rho^\dagger$.\par
If we are aware of the eigenbasis of $\rho$ and have infinitely many copies of the system state $\rho$, we can measure infinitely many
times within the eigenbasis to reconstruct the probability distribution, thus we have to solve a well-known classical problem, namely the approximation
of a probability distribution by sampling.
The number of free parameters of $\rho$ is $d-1$, since $\rho$ can be represented by a diagonal $d\times d$ matrix with a normalized set
of real eigenvalues.\par
Conversely, if we are not aware of the eigenbasis of the state $\rho$, it can still be represented by a normalized Hermitian matrix that
might not be diagonal and has at most $d^2$ nonzero entries. By Hermiticity, those entries can be described by $d^2$ real parameters.
Normalization of the matrix fixes another parameter and leaves $d^2-1$ free parameters\footnote{In contrast to the $d-1$ free parameters in the
case where the system is measured in its eigenbasis, additional free parameters appear which encode the basis information.} that describe an arbitrary quantum state $\rho$
of the $d$-dimensional Hilbert space $\cH$. Measuring on infinitely many copies of the system in some basis $\{\ket{\phi_i}\}_i$, we learn
at most $d-1$ free parameters due to the normalization of the measurement outcome. Since the total number of free parameters of the
quantum state $\rho$ is $d^2-1$ and by a single measurement operator we can figure out $d-1$ of those parameters, it turns out that we
need at least $d+1$ different bases to completely estimate the quantum state.\par
The set of unitary operators of the $d$-dimensional Hilbert space defines exactly the set of possible orthonormal bases, where the row vectors
of the unitary operators are identified as the basis vectors of the corresponding orthonormal basis. Thus, in order to
completely estimate the system state $\rho$, one has to find a set of $d+1$ bases from the set of unitary operators.\footnote{In the following,
we will use this correspondency without mentioning again and declare unitary operators as bases.} Given such a set,
the application of a unitary transformation to all bases will not affect the amount of information that is extracted by the measurements.
Hence, we are free to choose one of the bases to be the standard basis. Measuring a state that is diagonal in that standard basis with 
one of the $d$ remaining bases should lead to equally probable outcomes, even if the system is in a definite state (in the standard
basis). To achieve this goal, the overlap of every vector of this basis with every vector from the standard basis has to be constant. If
we want to find such a set of $d+1$ bases to reconstruct the state $\rho$ completely, the operators have to fulfill those properties
pairwise, which leads to the concept of MUBs.
\begin{defi}[Mutually unbiased bases]\label{defi:limits:mubs}\hfill\\
 A set of orthonormal bases $\fS = \Mg{\cB_0, \ldots, \cB_{r-1}}$, $r \in \N^*$, of the $d$-dimensional Hilbert space $\cH= \C^d$ is called
 a \emph{set of mutually unbiased bases (MUBs)}, if for every pair $(\cB_k, \cB_l)$ with $k\neq l$, the absolute value of the overlap of their basis vectors
 is constant. With $\cB_k = \Mg{\ket{\psi_1^k},\ldots, \ket{\psi_d^k}}$ and $\cB_l = \Mg{\ket{\psi_1^l},\ldots, \ket{\psi_d^l}}$,
 there holds
 \begin{align}
   \abs{\braket{\psi_i^k}{\psi_j^l}} = 1/ \sqrt{d},
 \end{align}
 for $k,l \in \Mg{0,\ldots,r-1}$, $k \neq l$, and $i,j \in \Mg{1,\ldots,d}$.
\end{defi}
Such a set of MUBs is called \emph{complete} if no set in the same Hilbert space exists which has a higher number of elements, it is therefore a
\emph{maximal} set. For a Hilbert space of dimension $d$ this size
is in general still unknown. Nevertheless, it was shown by Wootters and Fields with geometric arguments that each set has at most
$d+1$ elements \cite{WF89};\footnote{Wootters and Fields give also a construction of sets with $d+1$ elements if the dimension of the
Hilbert space is a power of a prime. For none of the remaining dimensions, even for $d=6$, a construction of $d+1$ elements is known.}
another proof based on the connection of sets of MUBs with \emph{pairwise orthogonal matrices} was given
by Bandyopadhyay \etal{} \cite{Bandy02}.

\section{Quantum key distribution}\label{sec:qkd}
Another area of application for sets of MUBs emerged in 1984, when Bennett and Brassard
introduced their ideas for \emph{quantum key distribution}\index{Quantum key distribution} (QKD) \cite{BB84}.
The aim of this approach is to solve the classically unsolved problem of a secure distribution of a common
secret key to two distinct parties, called \emph{Alice} and \emph{Bob}.
The classical cipher which is known as \emph{one-time pad}\index{One-time pad} guarantees a secret transmission
of a message between two parties, provided that both parties share a perfectly random bit-key, which has at least
the size of the message and is used only once. Classically, this key cannot be distributed secretly from Alice to
Bob, but a guaranteed secure key transmission may be implemented with the help of QKD. \par
In a general
formulation, Bennett and Brassard use in their approach a \emph{qubit}\footnote{A \emph{qubit} denotes a two-level
quantum system, which may be in any complex superposition of two possible quantum states $\ket{0}$ and $\ket{1}$, namely
$\alpha \ket{0} + \beta \ket{1}$, with $\alpha,\beta \in \C^2, \vert\alpha\vert^2 + \vert\beta\vert^2 =1$.}
in order to send \emph{quantum information}
from Alice to Bob. At random, Alice prepares the qubit either in an eigenstate of the Pauli-$\sigma_x$ or
Pauli-$\sigma_z$ operator.\footnote{More information about the Pauli operators is given in Appendix \ref{app:pauli}.}
By construction, the eigenbases of these operators are mutually unbiased in the
sense of Definition \ref{defi:limits:mubs}. In the chosen basis, quantum state $\ket{0}$ or $\ket{1}$ is again
taken randomly with equal probability. After the qubit is sent to Bob, he measures randomly in one of the
two distinct bases. If he chooses the same basis as Alice, he will obtain the encoded bit perfectly, if not,
the property of the MUBs leads to a random output.\footnote{Imagine, the eigenstates of the Pauli-$\sigma_z$ operator
are denoted by $\ket{0}$ and $\ket{1}$. If the quantum state $2^{-1/2}(\ket{0}+\ket{1})$ (which is then an
eigenstate of the Pauli-$\sigma_x$ operator) is measured by the operator $\ketbra{1}{1}$, the respective expectation
value is $1/2$, thus no information about the prepared state is measured.} Therefore, Alice and Bob communicate
after the measurement of the qubits over a \emph{classical authenticated channel}\footnote{To guarantee, that the
classical messages Alice and Bob receive are not sent by Eve, they have to be authenticated with the help of a
previously shared key; used schemes take then advantage of \emph{universal hash functions} as introduced by Wegman and Carter \cite{WC81}.}
and discard those pairs of bits, where they have chosen different
bases. If the transmission would be perfect, both parties would share a common key at this point, but as the
transmitted single qubits encounter quantum noise on the channel, a perfect correlation between the qubits of Alice
and Bob cannot be guaranteed. By declaring some of the transmitted qubits as test qubits, they can calculate
the induced \emph{bit error rate}; in principle, this error rate cannot be distinguished from an error a
possible eavesdropper \emph{Eve} induced, so it is seen as information which is leaked to her. Possible attacks
can be considered within a quite general way, but for simplicity reasons one may focus on a rather straightforward,
but powerful attack, which is called \emph{intercept-and-resend attack}. In this case, Eve measures each qubit in
one of the two preparation bases randomly. Consequently, in average, in every second case the basis is correct and in cases
where Eve took a different basis than Alice and Bob, her probability to measure the correct result is only
$1/2$. Therefore, she induces an error of $25\%$ in the case she attacks each qubit. Alice and Bob can
correct this error by using classical \emph{error correction} protocols and try
to rule out the information the eavesdropper got, by so-called \emph{privacy amplification} protocols. The \emph{tolerable bit error rate}
(BER) by using the mentioned \emph{post-processing} protocols is
limited by the corresponding proof of the protocol, which guarantees \emph{unconditional security}. For the
protocol by Bennett and Brassard a first rigorous proof was given by Mayers which allows a tolerable BER of
$7.5\%$ \cite{Mayers96}. Many versions of this BB84 called QKD protocol are known (see e.\,g. \cite{Ekert91, Bennett92}) and
further rigorous proofs also in order to raise this limit were given (e.\,g. \cite{LC99,SP00,GL03}). Finally,
Chau proved that this limit can be raised asymptotically to $20\%$ \cite{Chau02}. A formulation of
this result which is usable in a larger context was given by Ranade and Alber~\cite{Ranade06}.\par
A promising approach which uses the properties of MUBs
and generalizes the BB84 protocol was considered by Bruß in 1998 and is called \emph{six-state} protocol \cite{Bruss98};
its security was again rigorously proven \cite{Inamori00,Lo01}. The six-state protocol makes use of the third
variable $\sigma_y$, that is complementary to $\sigma_x$ and $\sigma_z$. So the protocol uses then a complete set
of three MUBs in a Hilbert space of dimension two. By the same arguments as above, it is clear that only $1/3$ of the
transmitted pairs can be used. It was again proven by Chau as well as by Ranade and Alber, that the tolerable error rate is roughly $27.6\%$,
which is clearly above the limit for the BB84 protocol~\cite{Chau02,Ranade06}. Considerations of possible generalizations of this protocol
show that this rate can be raised asymptotically up to $50\%$ by using complete sets of MUBs and \emph{qudits}\footnote{A \emph{qudit} denotes the
generalization of a qubit, namely a $d$-level quantum system with all possible complex superpositions of $d$ orthogonal states for which their absolute squares sum up to
one.} as information carriers,
which are defined in an $d$-dimensional Hilbert space \cite{Ranade07}. Examinations of these protocols indicate,
that a \emph{cyclic} property of the set of MUBs is advantageous \cite{Chau05}. Those cyclic sets have the property,
that all elements within the set are given by the powers of one element of the set of bases.

\chapter{Fundamentals}\label{chap:fundamentals}
The first natural step when aiming on the construction of complete sets of cyclic MUBs is the
reconstruction of existing complete sets of MUBs with the purpose of finding steps in the construction
that are suitable in order to generate cyclic sets. It may be useful to combine different aspects and
ideas from different approaches in order to achieve this goal. Ultimately, these considerations may
lead to new ideas that require fundamental observations of more distant aspects which become relevant.
The aim of this chapter is to retrace exactly this path. Within Section \ref{sec:constructions}, two
different approaches for the construction of complete sets of MUBs will be discussed in order to get
a notion of MUBs and to have a playground which enhances the potential for the construction of cyclic
sets. It will turn out later in this work (cf. Section \ref{sec:fibset}), that the properties of
so-called \emph{Fibonacci polynomials} are useful for that construction. The basis properties and their
relation to the usual \emph{Fibonacci series} will be discussed in Section \ref{sec:fibonacci_polynomials},
as well as advanced results that appeared in literature and own results.

\section{Constructions}\label{sec:constructions}
Many different constructions of complete sets of MUBs are known in literature. As later discovered
by Klappenecker and Rötteler, Alltop gave a construction for all prime dimensions with
$p\geq 5$ unknowingly in 1980 by solving a different problem \cite{Alltop80}. One year later,
a construction which works for all prime dimensions was given by Ivanović~\cite{Ivanovic81} and
generalized to prime power dimensions by Wootters and Fields~\cite{WF89} many years later. This construction is based
on basic observations of the properties of complete sets of MUBs and known results from number theory
and field theory, respectively. Klappenecker and Rötteler gave a precise formulation of all these
constructions more than a decade later, using finite fields and Galois rings more explicitly~\cite{Klappenecker04}. In the
meantime, a different construction was discussed by Bandyopadhyay \etal{}, which is based on the
partition of the set of Pauli operators~\cite{Bandy02}.\par
As these two different approaches seem to be the most important constructions of complete sets of
MUBs that appeared in literature, they will be discussed in the following two sections. In Section
\ref{sec:wootters}, the ideas of Ivanović and Wootters and Fields are summarized, whereas Section
\ref{sec:bandy} concerns with the construction suggested by Bandyopadhyay \etal{}.

\subsection{Exponential sum analysis}\label{sec:wootters}
The first general construction of complete sets of MUBs was given by Ivanović for all finite dimensional Hilbert spaces with a prime dimension $d=p$ \cite{Ivanovic81}.
To follow his approach, let us assume that a complete set of MUBs exists. According to Definition~\ref{defi:limits:mubs},
it is clear by the usage of the scalar product, that the application of any unitary transformation to all elements of the set
causes again a complete set of MUBs.\footnote{As will be seen later on, this transformation is one of the transformations which
leads to an \emph{equivalent set} of MUBs (cf. Chapter \ref{chap:equivalence}).} Therefore, w.l.o.g., if a complete set of MUBs
exists, there is always another complete set, which includes the standard basis. To fulfill then Definition \ref{defi:limits:mubs},
all remaining bases should have only numbers as entries with an absolute value of $p^{-1/2}$.\par
For odd dimensions, Ivanović uses a property of number theory, namely
\begin{align}\label{eqn:wootters:ivanovictrick}
 \left\vert \sum_{j=0}^{p-1} \EZ^{(2\pi \iE /p) ( s j^2 + t j)} \right\vert = \sqrt{p},
\end{align}
which holds for all $t \in \N$, $s \in \N^*$ and $p$ being an odd prime number. This expression is the absolute value of a \emph{generalized
quadratic Gauss sum}\index{Gauss sum} \cite[p. 13]{BEW98}. If the component $l$ of the vector $k$
within the basis $r$ is denoted as $(v_k^{(r)})_l$, the standard basis is given by
\begin{align}\label{eqn:wootters:stand}
 \left(v_k^{(0)}\right)_l = \delta_{kl},
\end{align}
with $k,l \in \MgN{p-1}$. All remaining bases, i.\,e. with $r\in \MgE{d}$ within a complete set of MUBs read in this construction as
\begin{align}\label{eqn:wootters:nonstand}
 \left(v_k^{(r)}\right)_l = \frac{1}{\sqrt{p}} \EZ^{(2\pi \iE /p) ( r l^2 + kl)}.
\end{align}
It can easily be checked that the bases given by Equation \eqref{eqn:wootters:nonstand} define unitary operators and that all of them
are mutually unbiased with respect to the standard basis.
To test the mutual unbiasedness of all remaining pairs of bases in the fashion of Definition~\ref{defi:limits:mubs}, the expression
given by Equation \eqref{eqn:wootters:ivanovictrick} appears and guarantees the expected result. Wootters and Fields used
the fact, that a generalization of Equation \eqref{eqn:wootters:ivanovictrick} exists in finite field theory and turns out
to be a good candidate in order to generalize this construction of complete sets of MUBs to odd prime-power dimensions \cite{WF89}.
Namely, it is known that
\begin{align}\label{eqn:wootters:wootterstrick}
 \left\vert \sum_{j\in \F_{p^m}} \EZ^{(2\pi \iE /p) \tr( s j^2 + t j)} \right\vert = \sqrt{p^m},
\end{align}
holds for $s\neq 0$ and $s, t$ being elements of the finite field with $2^m$ elements \cite{LN08}.\footnote{It was mentioned
by Klappenecker and Rötteler that Equation \eqref{eqn:wootters:wootterstrick} is related to a stronger version of Weil's theorem \cite[Theorem 5.37]{LN08}.}
The trace denotes the trace defined in field theory, which maps an element of the field $\F_{p^m}$ to the field $\F_p$ according
to Definition \ref{defi:app:algfield:trace} with resulting properties.
The non-standard bases read finally similar to Equation~\eqref{eqn:wootters:nonstand} as
\begin{align}\label{eqn:wootters:nonstandpp}
 \left(v_k^{(r)}\right)_l = \frac{1}{\sqrt{p^m}} \EZ^{(2\pi \iE /p) \tr ( r l^2 + kl)},
\end{align}
with $r,k,l \in \F_{2^m}$. Unitarity can again be checked easily. The standard basis is given analogously to Equation \eqref{eqn:wootters:stand},
but with $k,l$ being elements of the finite field $\F_{p^m}$.\par
As the left-hand side of Equation \eqref{eqn:wootters:wootterstrick} turns to zero in the case that the characteristic of the field
equals $p=2$, Wootters and Fields have reformulated the bases defined by Equation \eqref{eqn:wootters:nonstandpp} in a different
representation in order to generate complete sets of MUBs for even prime-power dimensions as will be discussed in a brief summary in the
following. In principle, every finite field $\F_{p^m}$
can be seen as a vector space over the ground field $\F_p$, thus every element $\beta \in \F_{p^m}$ can be written in a basis as
$\beta = \sum_{i=1}^m \beta_i f_i$ with $f_i$ being the basis vectors and $\beta_i$ the coefficients. The product of two basis vectors
can always be expressed within the bases by a set of coefficients like $f_i f_j =  \sum_{n=1}^m \alpha_{ij}^{(n)} f_n$. With the help
of this argumentation, the expression $l^2$ which appears in the argument of the trace in Equation~\eqref{eqn:wootters:nonstandpp} can
be rewritten as
\begin{align}\label{eqn:wootters:lsquare}
 l^2 = \left(\sum_{i=1}^m l_i f_i\right)^2 = \sum_{n=1}^m \bm{l}^t \alpha^{(n)} \bm{l} f_n,
\end{align}
with $\bm{l} = (l_1,\ldots,l_m)^t$ on the right-hand side being a column vector and using that $\alpha^{(n)}$ is a symmetric
$m \times m$ matrix. Following this notion and that the trace in Equation~\eqref{eqn:wootters:nonstandpp} can be rewritten as
\begin{align}
 \tr (rl^2 + k l) = \sum_{n=1}^m \bm{l}^t \alpha^{(n)} \bm{l} \tr(r f_n) + \tr(kl),
\end{align}
the non-standard bases of Equation \eqref{eqn:wootters:nonstandpp} are given by
\begin{align}\label{eqn:wootters:nonstandnofield}
 \left(v_{\bm{d}}^{(\bm{c})}\right)_{\bm{l}} =  \frac{1}{\sqrt{p^m}} \,\EZ^{(2\pi \iE / p)(\bm{l}^t (\bm{c} \cdot \bm{\alpha}) \bm{l} + \bm{d}^t \bm{l} )},
\end{align}
where $\bm{c}$ is a vector of the coefficients which appear by a similar transformation of that given in Equation \eqref{eqn:wootters:lsquare} from $r$ and $\bm{d}$ a vector
of the coefficients which appear analogously from $k$. Finally, $\bm{\alpha}$ denotes a column vector of the matrices $\alpha^{(1)},\ldots, \alpha^{(n)}$.
The mutual unbiasedness of the bases was proven by Wootters and Fields also in the non-field representation of Equation \eqref{eqn:wootters:nonstandnofield},
but again only for odd prime-power dimensions. For even prime-power dimensions, the construction of the set of bases has to be adapted slightly as
\begin{align}
 \left(v_{\bm{d}}^{(\bm{c})}\right)_{\bm{l}} =&  \frac{1}{\sqrt{2^m}} \,\EZ^{(2\pi \iE / 2)(\bm{l}^t (\bm{c} \cdot \bm{\alpha}) \bm{l} / 2 + \bm{d}^t \bm{l} )}\\
   =& \frac{1}{\sqrt{2^m}} \, \iE^{\bm{l}^t (\bm{c} \cdot \bm{\alpha}) \bm{l}} \, (-1)^{\bm{d}^t \bm{l}},\label{eqn:wootters:nonstandnofieldtwo}
\end{align}
and was again proven to be a complete set of MUBs (cf. also \cite[p. 47]{BEW98}). Identifying the rightmost term as an $m$-folded tensor product of
the Hadamard matrix as given in Equations \eqref{eqn:app:hadamard:mfolded} and \eqref{eqn:app:hadamard:-1rep},
leads to an abbreviated form of Equation \eqref{eqn:wootters:nonstandnofieldtwo}, namely
\begin{align}\label{eqn:wootters:nonstandnofieldtwoH}
 \left(\bm{v}^{(\bm{c})}\right)_{\bm{l}} = \iE^{\bm{l}^t (\bm{c} \cdot \bm{\alpha}) \bm{l}} H^{\otimes m},
\end{align}
where $(\bm{v}^{(\bm{c})})_{\bm{l}}$ denotes the row vector with entries $(\bm{v}_{\bm{d}}^{(\bm{c})})_{\bm{l}}$.\par
A more formal description of the discussed complete sets of MUBs was given by Klappenecker and Rötteler \cite{Klappenecker04},
which is based on \emph{Weil sums} in the case of odd prime-power dimensions and in even prime-power dimensions on \emph{Galois rings},
which form by their roots and the zero element the so-called \emph{Teichmüller set}.

\subsection{Pauli operators partition} \label{sec:bandy}
In $2002$ a paper by Bandyopadhyay \etal{} appeared that follows another approach to construct MUBs in
a way that highlights the structure of these bases~\cite{Bandy02} (cf. also \cite{Lawrence02}). In this construction, in a Hilbert space of dimension $d \in \N^*$,
each basis is seen as the set of common eigenvectors belonging
to a maximal set of commuting Pauli operators within the set of all $d^2$ generalized Pauli operators.
If we exclude the unity operator, which commutes with all operators, we
can find at most $d-1$ pairwise commuting operators within a single set. A partition of $d+1$ such sets results in a
complete set of MUBs. These partitions exist, as the previous construction (cf. Section~\ref{sec:wootters}),
for prime power dimensions, although, the initial approach is limited artificially, which leads to very
specific sets. In this section we will introduce the ideas of the approach Chapter \ref{chap:cyclic_mubs} is based on.
Possible alternatives in the construction process are pointed out.\par

We define the generators of the \emph{generalized Pauli operators}\index{Generalized Pauli operators} $Z$ and $X$, acting
on a state $\ket{i}$ of the $d$-dimensio\-nal Hilbert space $\cH=\C^d$ as
\begin{align}\label{eqn:bandy:ZX}
 Z \ket{i} = \omega^i \ket{i} \quad\mathrm{and}\quad X \ket{i} = \ket{i \oplus_d 1},
\end{align}
with $\omega = \exp{(2 \pi \iE /d)}$ being the first $d$-th root of unity. 
The group which is generated by these operators, i.\,e. $H := \gen{Z,X}$, is known as
\emph{\mbox{(Weyl-)Hei}senberg group}\index{Heisenberg group}\index{Weyl-Heisenberg group} or sometimes
\emph{generalized Pauli group}\index{Generalized Pauli group}. Since phases are not relevant for unitary operators
in quantum theory, we will only refer to elements with a real and positive phase in the following, thus we factorize
the group by its center $\{\pm 1, \pm \iE\} \Eins$ and get the \emph{set of Pauli operators} $\tilde H$. Any element of this set is given by
\begin{align}
 \ZX(k,l) :=
\begin{cases}
  (-\iE)^{k l} Z^k X^l& \mathrm{for}\; d=2,\\
  Z^k X^l&\mathrm{else},
\end{cases}
\end{align}
with $k,l \in \N$, whereas $Z^d = X^d = \Eins_d$ by construction. For prime-power dimensions $d=p^m$ with $p$ being
a prime number and $m \in \N^*$, each Pauli operator of this set can be seen as the representation of a $2m$-dimensional vector
which is an element of the finite field $\F_p^{2m}$. Then, the generalized Pauli operators can be written in a
tensor-product structure as
\begin{align}\label{eqn:bandy:ZXa}
 \ZX(\vec a) = 
\begin{cases}
     (-\iE)^{a^z_1 a^x_1} Z^{a^z_1} X^{a^x_1} \otimes\dots\otimes (-\iE)^{a^z_m a^x_m} Z^{a^z_m} X^{a^x_m}& \mathrm{for}\; p=2,\\
     Z^{a^z_1} X^{a^x_1} \otimes\dots\otimes Z^{a^z_m} X^{a^x_m}&\mathrm{else},
\end{cases}
\end{align}
with $\vec a = (a^z_1,\ldots,a^z_m, a^x_1,\ldots,a^x_m)^t \in \F_p^{2m}$. The commutation relation of two elements of this
set $\vec a, \vec b \in \F_p^m$ is given by
\begin{align}
 \ZX(\vec a) \cdot \ZX(\vec b) = \omega^{(\vec a, \vec b)_{\mathrm{sp}}} \ZX(\vec b) \cdot \ZX(\vec a),
\end{align}
with the \emph{symplectic product}\index{Symplectic product} $(\vec a, \vec b)_{\mathrm{sp}}$ as defined in Definition \ref{def:app:clifford:symplecticproduct}.
Thus, two Pauli operators $\ZX(\vec a)$ and $\ZX(\vec b)$ commute, if and only if the symplectic
product $(\vec a, \vec b)_{\mathrm{sp}}$ equals zero, where the symplectic product is additive (cf. Corollary \ref{cor:app:clifford:addsymp}).

As it is shown by Lemma \ref{lem:pauli:orthogonal}, the set of Pauli operators
is an orthogonal basis of linear operators. Thus, a possible choice is to describe every projective measurement operator by this
set. Since we aim on constructing MUBs which can be represented as unitary operators it is natural to restrict the consideration to projective measurements in the following. 
A discussion on more general measurements is given at the end of this section.
As we can parametrize any unitary measurement basis uniquely with the elements of the set of Pauli operators, we will choose an
exceptional set of MUBs, if each basis is given directly by the common eigenspace of a subset of the generalized Pauli operators,
instead of taking a set of commuting unitary operators in general.
Therefore, we partition the set of Pauli operators into disjoint classes $\cC_j$, such that
\begin{align}
 \tilde H \setminus \{\Eins_d \} = \bigcup\limits_{j=0}^{d} \cC_j. \label{eqn:bandy:partition}
\end{align}
Each class $\cC_j$ is a set of $d+1$ commuting Pauli operators and is created by a $2m \times m$ generator matrix $G_j$
with entries in $\F_p$, as
\begin{align}\label{eqn:bandy:classes}
  \cC^{\prime}_j = \cC_j \cup \{ \Eins_d \} = \{ \ZX(\vec a) : \vec a = G_j \cdot \vec c: \vec c \in \F_p^m \}.
\end{align}
By using the generator matrix in the mentioned way, the resulting set of vectors~$\vec a$ forms an $m$-dimensional subspace
of $\F_p^{2m}$ and the class is therefore called \emph{linear}.\par
It was shown by Bandyopadhyay \etal{}, using Equation \eqref{eqn:bandy:ZX}, that the \emph{Hilbert-Schmidt inner product}\index{Hilbert-Schmidt inner product}
of two elements $\vec a, \vec b \in \tilde H$ of the set of Pauli operators, namely
\begin{align}
 \braket{\ZX(\vec a)}{\ZX(\vec b)}_{\mathrm{HS}} := \tr(\ZX(\vec a)^\dagger \ZX(\vec b)),
\end{align}
vanishes for all $\vec a \neq \vec b$ \cite[Theorem 4.2]{Bandy02}. Pairs of matrices with a vanishing Hilbert-Schmidt inner product are called \emph{orthogonal}.
Since all operators within a single class $\cC_j$ commute, they have a common eigenbasis. But the Hilbert-Schmidt inner product
is invariant under basis transformation and reduces for diagonal matrices to the inner product of their diagonal vectors. As at most $d$ vectors
can be found that are mutually orthogonal in a $d$-dimensional space (and form therefore an orthogonal basis), at most $d$ (mutually orthogonal)
elements of the set of Pauli operators can be found, that commute pairwise \cite[Lemma 3.1]{Bandy02}. These elements will always form a linear class:
\begin{lem}[Linearity of maximal class]\label{lem:bandy:linear}\hfill\\
 A class of $d$ commuting elements of the set of Pauli operators can always be created by a generator matrix and is therefore always linear.
\end{lem}
\begin{proof}
 It was shown above that at most $d$ elements of the set of Pauli operators commute pairwise in a $d$-dimensional Hilbert space. By Corollary
 \ref{cor:app:clifford:addsymp}, also those elements $\ZX(\vec a_k)$ and $\ZX(\vec a_l)$ commute, which can be constructed by the linear combinations of their 
 generating vectors $\vec a_k$ and $\vec a_l$. But there is no class with more than $d$ elements, thus there exists always a basis of $m$
 elements which we call generator, as in Equation \eqref{eqn:bandy:classes}.
\end{proof}
A unitary operator basis, i.\,e. a basis for unitary operators, with $d^2$ elements that can be partitioned into $d+1$ classes with mutual orthogonal
and pairwise commuting elements, is called a \emph{maximal commuting basis}\index{Maximal commuting basis} and can be used to construct
a maximal set of $d+1$ MUBs \cite[Theorem 3.2]{Bandy02}. Conversely, also a complete set of MUBs implies the existence of a maximal commuting
basis \cite[Theorem 3.4]{Bandy02}. It remains an open question to figure out all (or at least all non-equivalent) maximal commuting
bases (see Chapter \ref{chap:equivalence}).\par

Therefore, in order to construct the classes $\cC_j$, we need--according to Equa-\linebreak tion~\eqref{eqn:bandy:classes}--to find $d+1$ generators $G_j$ that partition the set of Pauli operators into disjoint classes
of $d-1$ pairwise commuting elements each\footnote{As the unity element appears obviously in all classes it is excluded in order to construct disjoint classes.}. It will turn out later on
in this work that we are free to fix one of the classes to construct a certain set of MUBs (cf. Chapter~\ref{chap:equivalence}). But still,
specific separability properties of the MUBs are modified by this choice (cf. Section \ref{sec:entangleprop}). A possible choice is to set
the generator of the class $\cC_0$ as
\begin{align}\label{eqn:bandy:G0}
 G_0 = \begin{pmatrix}\Eins_m\\ 0_m\end{pmatrix},
\end{align}
which generates all Pauli-$Z$ operators that obviously commute; the symbols $\Eins_m$ and $0_m$ refer to the identity matrix and a quadratic zero matrix, respectively,
where $\Eins_m, 0_m \in M_m(\F_2)$. In order to
obtain classes $\cC_j$ with $j \in \MgE{d}$ that are disjoint with the class $\cC_0$, the column vectors of their generators have to be
linearly independent, which is exactly true if the block matrices $(G_0, G_j)$ are invertible for $j \in \MgE{d}$.\footnote{A class with a maximal number
of elements can only be created if the column vectors of a single generator are linearly independent.} By Lemma \ref{lem:algfund:detblock} this is true if the determinant
of $G_0^z G_j^x - G_j^z G_0^x$ is not zero for $j\neq 0$, where $G_j=(G_j^z, G_j^x)^t$. Since $G_0^x=0_m$, this equation can only hold if $G_j^x$
is invertible. But if $G_j^x$ is invertible, we can write all generators with $j\neq0$ as 
\begin{align}\label{eqn:bandy:Gj}
 G_j = \begin{pmatrix}G_j^z\\ \Eins_m\end{pmatrix},
\end{align}
which will be proven in 
Corollary \ref{cor:standfor:generatorfree} and called \emph{standard form}\index{Standard form} in Section~\ref{sec:standardform}.
Within this form, the elements in a single class commute, if the symplectic product of all pairs
of vectors of the generating set (thus, all column vectors of the generator) have a vanishing symplectic product. With $G_j^z=(\vec a^z_1,\ldots, \vec a^z_m)$
follows
\begin{align}
 \vec a^z_k \vec a^x_l - \vec a^x_k \vec a^z_l = 0 \quad \mathrm{for}\; k,l \in \MgE{m},
\end{align}
and with $G_j^x=\Eins_m$ finally
\begin{align}
 a^z_{k,l} - a^z_{l,k} = 0 \quad \mathrm{for}\; k,l \in \MgE{m}.
\end{align}
Thus, the matrices $G_j^z$ with $j\in \MgE{d}$ have to be symmetric \cite[Lemma~4.3]{Bandy02}. The last point we have to achieve is that arbitrary pairs of
generators $G_j$ with $j\neq 0$ do not span the same vectors spaces, thus for $k,l \in \MgE{m}$ the determinant of $G_k^z G_l^x - G_l^z G_k^x$ does not
vanish. But with $G_k^x = G_l^x = \Eins_m$ we find
\begin{align}
 \det(G_k^z - G_l^z) \neq 0 \quad \mathrm{for}\; k,l \in \MgE{m}.
\end{align}

In summary, this leads to the following three conditions in order to construct a maximal commuting basis if we set $G_0 = (\Eins_m, 0_m)^t$:
\begin{enumerate}[\quad(1)]
 \item $G_j = (G_j^z, \Eins_m)$ for $j\in \MgE{d}$.
 \item $G_j^z$ is symmetric.
 \item $\det(G_k^z - G_l^z) \neq 0$ for $k, l \in \MgE{d}, k \neq l$.
\end{enumerate}

If the unitary operators within a class $\cC'_j$ are given by
\begin{align}
 \cC'_j = \Mg{U_{j,0}, \ldots, U_{j,d-1}},
\end{align}
with $U_{j,0}$ referring to the unity matrix $\Eins_d$, there is an orthonormal basis in which all of these operators are diagonal. This leads
to a set of eigenvalues $\lambda_{j,k,l}$ where $k \in \MgN{d-1}$ indicates the operator $U_{j,k}$ and $l \in \MgE{d}$ belongs to the eigenvector index. Bandyopadhyay \etal{} 
have shown, that the following construction generates a complete set of MUBs from the maximal commuting basis with $M_j$ being the unitary operator
that is a common eigenbasis of the elements of $\cC'_j$ and serves as an element of the set of MUBs:
\begin{align}\label{eqn:bandy:Mj}
 M_j = 
\begin{pmatrix}
  \lambda_{j,0,1} & \lambda_{j,0,2} & \cdots & \lambda_{j,0,d}\\
  \lambda_{j,1,1} & \lambda_{j,1,2} & \cdots & \lambda_{j,1,d}\\
  \vdots & \vdots & \ddots & \vdots\\
  \lambda_{j,d-1,1} & \lambda_{j,d-1,2} & \cdots & \lambda_{j,d-1,d}
\end{pmatrix}.
\end{align}\par


Finally, we like to mention, that the set of generalized Pauli operators as it was chosen to construct the mutually unbiased bases which may serve as a set
for completely estimating the quantum state of a system, is only a subset of the most general set of measurement operators which contains
all \emph{positive operator valued measurements (POVMs)}.\footnote{A POVM is a set of positive operators, $\{E_n\}$, which are defined by measurement operators $M_n$
as $E_n = M_n^{\dagger} M_n$ with $\sum_n E_n = \Eins$ for $n \in \N^*$. The probability that outcome $n$ occours is given by $p(n) = \bra{\psi} E_n \ket{\psi}$
if the state of the system is given by $\ket{\psi}$.}
Investigations on so-called \emph{symmetric informationally complete POVMs (SIC-POVMs)}\index{SIC-POVM} may expand the state estimation techniques
raised by MUBs. After this concept was introduced in the seventies~\cite{Lemmens73,Prugovecki77}, Zauner conjectured that a
complete state estimation with SIC-POVMs is possible for all Hilbert space dimensions and showed this explicitly for
dimensions $d\leq 5$ \cite{Zauner99}. Whereas Renes \etal{} \cite{Renes04} found numerical evidences for larger dimensions, an explicit construction
of a SIC-POVM for the smallest composite dimension $d=6$ was given by Grassl~\cite{Grassl04}. A formal discussion of the former
results with the notion of an \emph{extended Clifford group}\index{Extended Clifford group} was done by Appleby \cite{Appleby05, Appleby09}.


\section{Fibonacci polynomials}\label{sec:fibonacci_polynomials}
Fibonacci polynomials play an important role in the process of constructing the so-called \emph{cyclic MUBs}, which will be considered in Chapter \ref{chap:cyclic_mubs}.
We will discuss those
properties in detail which are necessary in order to understand the features of most of the sets of cyclic MUBs which are constructed in this work 
(cf. Sections~\ref{sec:fibset} and~\ref{sec:homogeneous:group}).
This section starts by defining the
Fibonacci polynomials and enters directly into the area of important and generally available properties of these polynomials. We will then
examine properties which are limited to the case where these polynomials are defined over the finite field $\GF{2}$.
Most of these lemma with similar proofs were done for \cite{KRS10,SR12}. Some were already given in \cite{Webb69,Bicknell70,Goldwasser97,Goldwasser02},
as well as further properties. The existence of complete sets of cyclic MUBs, using the constructions of Sections~\ref{sec:fibset} and \ref{sec:homogeneous:group},
can be proven with the presented results.

The well-known Fibonacci sequence can be generalized in a way to generate the so-called Fibonacci polynomials $F_n(x)$, which are defined
recursively.
\begin{defi}[Fibonacci polynomials]\hfill\\
 The polynomial $F_n(x)$ is called the \emph{Fibonacci polynomial}\index{Fibonacci!polynomials} of index $n$, and recursively defined as
 \begin{align}\label{eqn:fibpol:normal}
   F_{n+1}(x) := x\cdot F_n(x) + F_{n-1}(x)
 \end{align}
 with $F_0 := 0$ and $F_1 := 1$.
\end{defi}
For $x=1$ we end up with the usual Fibonacci sequence given by the ordered set $\{F_n\}_0^\infty = \Mg{0,1,1,2,3,5,\ldots}$.
For further investigations it is a great advantage to have a generalized recursion relation.
\begin{lem}[General recursion relation]\label{lem:fibpol:general}\hfill\\
  The Fibonacci polynomial $F_{k+l}(x)$ with $k,l \in \N$ can be derived with the help of the Fibonacci polynomials $F_k(x)$ and $F_l(x)$ as
\begin{align}\label{eqn:fibpol:general}
 F_{k+l}(x)=F_{k}(x) F_{l+1}(x) + F_{k-1}(x) F_{l}(x).
\end{align}
\end{lem}
\begin{proof}
 We show this formula by induction. For $l=0$ we find $F_k(x)=F_k(x)$. Assuming that \eqref{eqn:fibpol:general} holds, we get
 $F_{k+(l+1)}(x) = F_k(x) F_{(l+1)+1}(x) + F_{k-1}(x) F_{l+1}(x)=x F_{k+l}(x) + F_{k + (l-1)}(x)$ by using Equation \eqref{eqn:fibpol:normal}.
\end{proof}

We can use this relation in order to prove an important lemma on the divisibility properties of the Fibonacci polynomials. Beforehand,
we need the auxiliary lemma which follows.
\begin{lem}[Coprime Fibonacci polynomials]\label{lem:fibpol:coprime}\hfill\\
 The polynomials $F_n(x)$ and $F_{n+1}(x)$ are coprime for $n \in \N^*$.
\end{lem}
\begin{proof}
 It is given by construction that $\gcd ( F_1(x), F_2(x)) = \gcd(1,x) =1$. If we assume that $\gcd ( F_n(x), F_{n+1}(x))=1$ we can
 step the induction forward by $\gcd ( F_{n+1}(x),$ $F_{n+2}(x)) = \gcd ( F_{n+1}(x), x F_{n+1}(x) + F_n(x)) = 1$, using the assumption.
\end{proof}
This basic divisibility property leads to a more fundamental property of the Fibonacci polynomials.
\begin{lem}[Divisibility of Fibonacci polynomials]\label{lem:fibpol:divisibility}\hfill\\
 The polynomial $F_n(x)$ is divisible by $F_m(x)$ if and only if $m$ divides $n$, with $m, n \in \N^*$.
\end{lem}
\begin{proof}
 To show the implication, let us assume that $n=m m'$ with $m' \in \N^*$. We note that $F_m(x)$ divides $F_n(x)$ trivially for $m'=1$.
 Using relation \eqref{eqn:fibpol:general} with $k:=m$ and $l:=m (m'-1)$, we proceed by induction with the assumption that
 $F_{m(m'-1)}(x)$ is divisible by $F_m(x)$ and see that this implies that
 $F_{m m'}(x) = F_{ m + m (m'-1)}(x) = F_m F_{m (m'-1)+1}(x) + F_{m-1}(x) F_{m (m'-1)}(x)$ is also divisible by $F_m(x)$.
 \par To show the converse we set $n = m m' +r$ with some remainder $r \in \N$ such that $r < m$. The generalized recursion relation
 \eqref{eqn:fibpol:general} gives $F_n(x) = F_{r+mm'}(x) = F_r F_{m m' +1}(x) + F_{r-1}(x) F_{m m'}(x)$. Using the normal recursion relation~\eqref{eqn:fibpol:normal}
 we get $F_n(x) = x F_r(x) \cdot F_{m m'}(x) + F_r(x) F_{m m'-1}(x) + F_{r-1}(x) F_{m m'}(x) $. By the implication, the first and the
 last term are divisible by $F_m(x)$. Since we assume that $F_n(x)$ is divisible by $F_m(x)$, the term $F_r(x) F_{m m' -1}(x)$ should also be divisible by $F_m(x)$ or
 vanish. From Lemma~\ref{lem:fibpol:coprime} we know that $F_{m m'}(x)$ and $F_{m m'-1}(x)$ are coprime, thus $F_{m m'-1}(x)$ is coprime to $F_m(x)$.
 Since we further assume that $r<m$ which implies that the degree of $F_m(x)$ is larger than the degree of $F_r(x)$, the
 polynomial $F_r(x)$ cannot be divisible by $F_m(x)$, thus should vanish by identifying $r=0$.
\end{proof}
Keeping these properties in mind, we can expand our investigation on the Fibonacci polynomials by discussing their coefficients. Therefore, it is useful to
read off the coefficients $a_k^{(n)}$ from Equation \eqref{eqn:fibpol:normal}, meaning the coefficient belonging to $x^k$ in the polynomial $F_n(x)$, and get
\begin{align}\label{eqn:fibpol:coefficents}
 a_k^{(n+1)} = a_{k-1}^{(n)} + a_k^{(n-1)},
\end{align}
with $a_k^{(0)}=0$ for $k \in \N$, $a_0^{(1)}=1$ and $a_k^{(1)}=0$ for $k \in \N^*$. Using this relation, we are able to show the following lemma:
\begin{lem}[Coefficients of Fibonacci polynomials]\hfill\\
 For $F_n(x) = \sum_{k=0}^{n} a_k^{(n)} x^k$ and $n \in \N$, there holds
 \begin{align}
      a_k^{(n)} = \begin{cases}
		  \begin{pmatrix}(n+k+1)/2\\(n-k+1)/2\end{pmatrix},        &\text{ if $n-k \equiv 1 \mod 2$},\\
                  0,                        			           &\text{  otherwise}.
      \end{cases}
 \end{align}
\end{lem}
\begin{proof}
 For $F_0(x)=0$ and $F_1(x)=1$ the statement holds. Using Equation~\eqref{eqn:fibpol:coefficents}, we have to show that 
 \begin{align}
  \binom{\frac{n+k}{2}+1}{\frac{n-k}{2}+1} = \binom{\frac{n+k}{2}}{\frac{n-k}{2}} + \binom{\frac{n+k}{2}}{\frac{n-k}{2}+1}
 \end{align}
holds for the first case. But this is known from Pascal's triangle. Using again Equation \eqref{eqn:fibpol:coefficents}, we see that
the second case always stays zero.
\end{proof}
We can restate the result of this lemma as
\begin{align}
 F_n(x) = \sum\limits_{k=0}^{[n/2]} \binom{n-k}{k+1} x^{n-2k-1}.
\end{align}

\par Since the
Fibonacci polynomials define a \emph{linear recurring sequence}\index{Linear recurring sequence}, namely an \emph{impulse response sequence}\index{Impulse response sequence} \cite[Chapter 8]{LN08},
we can represent the sequence given by Equation \eqref{eqn:fibpol:normal} by the associated \emph{companion matrix}\index{Companion matrix} $A$ of its characteristic polynomial
$\chi_F(Z) = Z^2-x\cdot Z-1$, i.\,e.
\begin{align}
 A=\begin{pmatrix}
     0&1\\1&x
   \end{pmatrix}.
\end{align}
With the \emph{initial state vector}\index{Initial state vector}\footnote{In the theory of linear recurring sequences, an initial state vector is defined
as the vector which describes the initial state of the \emph{feedback shift register} (cf. \cite[Chapter 8]{LN08}).} $(F_0,F_1)^t = (0,1)^t$, we obtain the Fibonacci polynomials by
\begin{align}
 \begin{pmatrix}F_n(x)\\F_{n+1}(x) \end{pmatrix} = \begin{pmatrix}0&1\\1&x\end{pmatrix}^n
 \cdot \begin{pmatrix}F_0\\F_1\end{pmatrix}.
\end{align}
If we identify the entries within the companion matrix\index{Companion matrix} $A$ with their associated Fibonacci polynomials, we find
\begin{align}
 A^n= \begin{pmatrix}
     F_0&F_1\\F_1&F_2(x)
   \end{pmatrix}^n = 
   \begin{pmatrix}
     F_{n-1}(x)&F_n(x)\\F_n(x)&F_{n+1}(x)
   \end{pmatrix}.
\end{align}
In the case that $n$ is decomposable into integers as $n = k + l$, the relation $A^{k + l} = A^k \cdot A^l$, namely
\begin{align}
\begin{pmatrix}
     F_{k+l-1}(x)&F_{k+l}(x)\\F_{k+l}(x)&F_{k+l+1}(x)
\end{pmatrix}= 
\begin{pmatrix}
     F_{k-1}(x)&F_{k}(x)\\F_{k}(x)&F_{k+1}(x)
\end{pmatrix} \cdot
\begin{pmatrix}
     F_{l-1}(x)&F_{l}(x)\\F_{l}(x)&F_{l+1}(x)
\end{pmatrix}
\end{align}
leads directly to the generalized recursion relation \eqref{eqn:fibpol:general}.

\subsection{Fibonacci polynomials over $\F_2$}
If we limit the coefficients of the Fibonacci polynomials to be an element of the finite field which has characteristic two and is prime, namely $\GF{2}$, we
find further properties which are important for the construction of cyclic MUBs. A major feature of fields with characteristic two is the fact
that the minus sign can be replaced by the plus sign. We may then create a second form of Equation \eqref{eqn:fibpol:normal}, i.\,e.
\begin{align}\label{eqn:fibpol:reverse}
 F_{n-1}(x) = x\cdot F_n(x) + F_{n+1}(x).
\end{align}
This leads to Fibonacci polynomials with a negative index, being equal to their positive counterparts, such that $F_{-n}(x) = F_n (x)$. A similar result with
alternating signs can be deduced for different prime fields. Some of the Fibonacci
polynomials have special properties if they are defined over the field $\GF{2}$.
\begin{lem}[Fibonacci polynomials $F_{2n}(x)$]\label{lem:fibpol:evenindex}\hfill\\
  For Fibonacci polynomials that have an even index and are defined over the finite field $\GF{2}$, holds $F_{2n}(x) = x\,F_n(x)^2$ for $n \in \N$. More general, we can state
  $F_{2^m\cdot r}(x) = F_r^{2^m}(x)\, x^{2^m-1}$ with $m,r\in \N$.
\end{lem}
\begin{proof}
 Using the generalized recursion relation \eqref{eqn:fibpol:general}, we find $F_{n+n}(x)=$\linebreak $F_n(x)F_{n+1}(x)+F_{n-1}(x)F_n(x)$. With
 Equation \eqref{eqn:fibpol:normal} this yields the requested result. The second part follows by induction over $m$.
\end{proof}
Despite the general recursion relation \eqref{eqn:fibpol:general}, we can deduce a symmetric generalization of the normal recursion relation~\eqref{eqn:fibpol:normal},
that is only valid for Fibonacci polynomials over $\GF{2}$.
\begin{lem}[Symmetric recursion relation]\label{lem:fibpol:symmetricrecursion}\hfill\\
  For Fibonacci polynomials that are defined over $\GF{2}$, there holds $F_{k+l}(x)+F_{k-l}(x) = x\,F_k(x) \, F_l(x)$ with $k,l \in \N$.
\end{lem}
\begin{proof}
  To show the lemma by induction, we start with the normal recursion relation~\eqref{eqn:fibpol:normal}. We can proceed the induction over $l$ by using
  Equations \eqref{eqn:fibpol:normal} and \eqref{eqn:fibpol:reverse} as $F_{k+(l+1)}(x)+F_{k-(l+1)}(x)=F_{k+l}(x) + F_{k-l}(x) + F_{k+(l-1)}(x) + F_{k-{l-1}}(x)$.
  With the assumption this leads to the expected result.
\end{proof}
Using this lemma we can formulate a subtractive type of a recursion relation.
\begin{lem}[Subtractive recursion relation]\label{lem:fibpol:subtractive}\hfill\\
  For Fibonacci polynomials that are defined over the finite field $\GF{2}$, there holds\linebreak $F_{k+1}(x) F_l (x) + F_k(x) F_{l+1}(x) = F_{\vert k-l \vert}(x)$ with $k,l \in \N$.
\end{lem}
\begin{proof}
  If we multiply the expression $F_{k+1}(x) F_l (x) + F_k(x) F_{l+1}(x)$ by $x$ and apply the implication of Lemma \ref{lem:fibpol:symmetricrecursion}, we have
  $F_{k-1+1}(x) + F_{k-l-1}(x)$. Using the converse of this lemma and ruling out the initially multiplied $x$ we get the expected result.
\end{proof}
With Lemma \ref{lem:fibpol:symmetricrecursion} in mind, we can prove the main theorem of this section.

\begin{thm}[Completeness of Fibonacci polynomials]\label{thm:fibpol:completeness}\hfill\\
 All polynomials of degree $m \in \N^*$ or degree $m' \in \N^*$ such that $m'$ divides $m$, which are irreducible over the field $\GF{2}$, divide either $F_{2^m-1}(x)$ or $F_{2^m+1}(x)$, or equal $x$.
\end{thm}
\begin{proof}
 If we calculate the product of $x$ with $F_{2^m-1}(x)$ and $F_{2^m+1}(x)$ using Lemma~\ref{lem:fibpol:symmetricrecursion}, we get $F_{2^{m+1}}(x)+F_2(x)$, which
 equals $x^{2^{m+1}-1} + x = x (x^{2^m-1} +1)^2$ by Lemma~\ref{lem:fibpol:evenindex}. But the term $x^{2^m} + x$ is, as a well-known result from finite field theory
 (see e.\,g. Lemma~2.13 of \cite{LN08}), equal to the product of all irreducible polynomials over the field $\GF{2}$, whose degree divides $m$.
\end{proof}
This theorem enables us to well-define a property of irreducible polynomials which are defined over $\GF{2}$ that will be important in this work later on.
\begin{defi}[Fibonacci index]\label{lem:fibpol:fibindex}\hfill\\
 Every irreducible polynomial $p \in \F[x]$ of degree $m$ has a property, namely the \emph{Fibonacci index}\index{Fibonacci!index} (sometimes called \emph{depth}), which is
 given by the minimal positive number $d \in \N$, such that $p$ divides $F_d$.
\end{defi}
\begin{cor}[Maximal Fibonacci index]\label{cor:fibpol:fibindex}\hfill\\
 The Fibonacci index of an irreducible polynomial $p \in \F[x]$ of degree $m$ is upper bounded by $2^m+1$.
\end{cor}
\begin{proof}
 The statement follows directly by applying Theorem \ref{thm:fibpol:completeness}.
\end{proof}
To limit the possible Fibonacci index of a specific irreducible polynomial, we state the following lemma which was proven by Sutner (cf. Theorem 3.1 of \cite{Sutner00}).
\begin{lem}[Irreducible factors of $F_{2^m \pm 1}$]\label{lem:fibpol:linear}\hfill\\
 All irreducible polynomials which divide $F_{2^m+1}(x)$, have a linear term, whereas the linear term of those irreducible polynomials which divide $F_{2^m-1}(x)$ vanishes.
\end{lem}

Finally, we like to mention that the structure of the Fibonacci polynomials is related to the Sierpinski triangle (see Appendix \ref{app:fractals}) and to structures that appear
with characteristic polynomials of \emph{reduced stabilizer matrices} with \emph{submatrix $A$} equal to zero, as defined in Equation \eqref{eqn:fibset:A}.


\chapter{Construction of cyclic sets}\label{chap:cyclic_mubs}
In Section \ref{sec:constructions} we have discussed different constructions of MUBs that refer to nice
mathematical structures which lead to direct constructions of the different unitary matrices within a
set of MUBs. Only the construction introduced by Bandyopadhyay \etal{} \cite{Bandy02} breaks this rule by introducing a preceding step,
namely the arrangement of the set of Pauli operators into classes of commuting elements. From a physical point
of view it is useful to equip complete sets of MUBs with further features. One of these features is important
both for theoretical investigations and also for experimental implementations of MUBs. It is a cyclic structure
we might only find for complete sets of MUBs for Hilbert space dimensions which are a power of two~\cite{Gow07},
thus $d=2^m$ with $m \in N^*$. Namely, if one of the bases within the set--which can be represented as a unitary operator--is given by $U$, we
get a complete set of MUBs $\fS = \gen{U} = \Mg{U, U^2, \ldots U^{d+1}}$ with $U^{d+1}=\Eins_d$.
Therefore, we will refer to these sets as \emph{sets of cyclic mutually unbiased bases (cyclic MUBs)} from
now on. It is the main aim of the first part of this work to investigate the properties of cyclic
MUBs, to provide constructions and to expand these ideas. Initially the existence of the chosen construction will be proven. From
a theoretical point of view, cyclic MUBs are easier to handle since they allow elegant proofs, e.\,g. for
QKD protocols~\cite{Gott98,Lo01,Chau05}.
For the experimental aspect cyclic MUBs feature their advantage of being implementable by a single
quantum circuit. As we restrict our quantum systems to be describable in a Hilbert space of dimension
$d=2^m$ with $m \in \N^*$, we can consider our systems to be systems of interacting qubits. Formally,
this provides us a simple tensor product structure. Conceptually, the restriction refers to most
of the valuable applications, which are built on qubits. However, extensions to general prime-power dimensions
appear to be feasible by dropping some dispensable properties, but demand further investigations. In
fact, it seems that the set of MUBs can only be reduced to a \emph{set of generators} \cite{Gow07}.\par
This chapter is organized as follows: In Section \ref{sec:fibset} a method is introduced which is proven to be
able to generate complete sets of cyclic MUBs for all dimensions $2^m$. Numerical methods are discussed in
subsections. Section \ref{sec:fermatset} discusses a special kind of solutions for $m=2^k$ with $k\in\N$
which can be constructed recursively and is related to a mathematical conjecture of finite field theory by Wiedemann.
In Section~\ref{sec:standardform} the \emph{standard form} is introduced which will be used later on in order to
compare sets of cyclic MUBs, followed by a discussion about their entanglement properties (Section~\ref{sec:entangleprop}).
The resulting \emph{homogeneous} and \emph{inhomogeneous} sets are shown in Sections \ref{sec:homogeneous} and \ref{sec:inhomogeneous},
respectively. As all sets are given in a special representation, the construction of the corresponding
unitary operator is given in Section \ref{sec:Uconstruction}. Finally, the implementation of all discussed
sets into a quantum circuit is possible with the methods specified in Section \ref{sec:circuit}.

\section{Fibonacci-based sets}\label{sec:fibset}
The first class of cyclic MUBs we create in this work is based on properties of Fibonacci polynomials which were discussed in detail in
Section \ref{sec:fibonacci_polynomials}. Contrary to other classes, which will be discussed later on in Sections \ref{sec:homogeneous}
and \ref{sec:inhomogeneous}, the connection to Fibonacci polynomials allows a straightforward construction scheme that is based
on finite field theory. The method we will use, finds solutions for the construction of Bandyopadhyay \etal{} \cite{Bandy02}
that was introduced in Section \ref{sec:bandy}, with an additional cyclicity property. Therefore, we will establish a \emph{symplectic matrix} $C \in M_{2m} (\F_2)$ which will be called
\emph{stabilizer matrix}\index{Stabilizer matrix} in the following, that cyclically permutes the $d+1$ classes $\mathcal{C}_j$ of Equation~\eqref{eqn:bandy:classes}.
In other words, that stabilizer matrix is a matrix representation of a generator of a finite group of order $d+1$ and generates a complete
set of MUBs. This section starts by discussing the required properties of the stabilizer matrix and different approaches to achieve them, as
published in \cite{KRS10, SR12}. We then show the existence of the stabilizer matrix for each dimension $d=2^m$ with $m \in \N^*$, followed
by a conjecture that predicts a \emph{symmetric companion matrix}\index{Symmetric companion matrix} of polynomials defined over $\F_2$, which raises the ability of calculating
the stabilizer matrix from a certain polynomial directly. A list of solutions for the symmetric stabilizer matrix is given in
Appendix \ref{app:fibonacci_based:companion} for dimensions with $m=2, \ldots, 36$. Solutions for stabilizer matrices in a certain form are listed for
dimensions with $m=2, \ldots, 600$ in Appendix \ref{app:fibonacci_based:triangle}. Section \ref{sec:unitop:fibset} deals
with a unitary representation of the stabilizer matrix $C$.\par

As the method we choose to generate complete sets of cyclic MUBs is based on the construction of Bandyopadhyay \etal{}, we need to partition
the set of Pauli operators into disjoint classes of commuting elements, as seen by Equation \eqref{eqn:bandy:partition}.
For each of these classes $\mathcal{C}_j$ with $j \in \MgN{d}$, by repeating Equation \eqref{eqn:bandy:classes} with $p=2$, we define
a $2m \times m$ generator matrix $G_j$ that generates the $d-1$ elements of the class $\mathcal{C}_j$
as well as the identity as
\begin{align}\label{eqn:fibset:classes}
  \mathcal{C}^{\prime}_j = \mathcal{C}_j \cup \{ \Eins_d \} = \{ \ZX(\vec a) : \vec a = G_j \cdot \vec c: \vec c \in \F_2^m \}.
\end{align}
To implement a cyclic structure into this partition, we demand that the stabilizer matrix $C \in M_{2m}(\F_2)$ permutes the generators as
$C^k \cdot G_l = G_{k \oplus l}$, which implies that $C^{d+1} = \Eins_{2m}$ and leads to
\begin{align}\label{eqn:fibset:classesC}
  \mathcal{C}^{\prime}_j = \mathcal{C}_j \cup \{ \Eins_d \} = \{ \ZX(\vec a) : \vec a = C^j \cdot G_0 \cdot \vec c: \vec c \in \F_2^m \}.
\end{align}
As we are free to fix one of the generators, we set
$G_0 = ( \Eins_m, 0_m )^t$. To guarantee disjoint classes, the generators need to span non-overlapping vector spaces, i.\,e. the matrix
$(G_k, G_l)$ is invertible for all values $k, l \in \MgN{d}$ with $k \neq l$. Since the set of automorphisms of the Heisenberg group is the
Clifford group and the commutation relations of $\ZX(\vec a)$ indicate a symplectic structure (see Appendix \ref{app:clifford}),
we need $C$ to be symplectic. More precisely, the matrix $C \in M_{2m} (\F_2)$ is a $2m\times 2m$ representation of
a \emph{Clifford unitary}\index{Clifford unitary} matrix. So the problem is boiled down from finding a $2^m \times 2^m$ unitary matrix to the task of finding
this $2m \times 2m$ stabilizer matrix $C$. In summary, we have the following restrictions on the stabilizer matrix $C$:
\begin{enumerate}[(I)]
 \item $C$ is symplectic.
 \item $C^{d+1}=\Eins_{2m}$.
 \item 
  $(G_k, G_l)$ is invertible for $k, l \in \MgN{d}$, $k \neq l$, with $G_k = C^k \cdot (\Eins_m, 0_m)^t$,
\end{enumerate}
where $(G_k, G_l)$ defines a \emph{block matrix} in the form
\begin{align}
 (G_k, G_l) = \begin{pmatrix} G_k^z & G_l^z\\ G_k^x & G_l^x\end{pmatrix},\quad\mathrm{with}\quad G_k = \begin{pmatrix}G_k^z\\ G_k^x\end{pmatrix}.
\end{align}
These restrictions yield still a lot of freedom to create complete sets of cyclic MUBs. We will refer to alternative classes of
possible solutions in Sections \ref{sec:homogeneous} and \ref{sec:inhomogeneous}. Within this section, we limit the stabilizer matrix $C$ to
\begin{align}\label{eqn:fibset:B110}\index{Fibonacci!sets}
 C = \begin{pmatrix}B & \Eins_m \\ \Eins_m & 0_m \end{pmatrix},
\end{align}
with $B \in M_m(\F_2)$. Symplecticity is given (see Definition \ref{def:app:clifford:symplecticmatrix}) when
\begin{align}
 C^t \cdot \begin{pmatrix}0_m & -\Eins_m \\ \Eins_m & 0_m \end{pmatrix} \cdot C \equiv \begin{pmatrix}0_m & -\Eins_m \\ \Eins_m & 0_m \end{pmatrix} \mod 2,
\end{align}
thus, the matrix $B$, which we will refer to as \emph{reduced stabilizer matrix}\index{Reduced stabilizer matrix}, needs to be symmetric to fulfill Condition (I).
As we consider the powers of $C$, namely $C^n$ with $n \in \N$, we get
\begin{align}
 C^n = \begin{pmatrix} F_{n+1} (B) & F_{n} (B) \\ F_{n} (B) & F_{n-1} (B) \end{pmatrix},\label{eqn:fibset:Cn}
\end{align}
with $F_n(x)$ being the Fibonacci polynomials as defined in Section~\ref{sec:fibonacci_polynomials}\index{Fibonacci!polynomials}. In order to achieve Condition (II),
we can read off from Equation \eqref{eqn:fibset:Cn} that for $n=d+1$, $F_{n}(B)$ has to be equal to $0_m$ and $F_{n-1} (B)$, as well as $F_{n+1} (B)$
have to equal $\Eins_m$. But the last property follows from the first two properties by the recursion relation of the Fibonacci polynomials over $\GF{2}$
(see Equation \eqref{eqn:fibpol:reverse}). We conclude this symmetry with the following lemma:
\begin{lem}[Symmetry of Fibonacci polynomials with reduced stabilizer matrix]\label{lem:fibset:symmetry}\hfill\\
 For a reduced stabilizer matrix $B \in M_m(\F_2)$ that creates a complete set of cyclic MUBs, there holds $F_n(B) = F_{2^m+1-n}(B)$ for $n \in \MgN{2^m+1}$;
 in particular, $F_{2^{m-1}}(B) = F_{2^{m-1}+1}(B)$.
\end{lem}
\begin{proof}
 With $F_0(B) = 0_m$ and $F_1(B) = \Eins_m$, the recursion relation of Equation~\eqref{eqn:fibpol:normal} defines all $F_n(B)$ with $n \in \N$. As seen above
 (cf. Equation \eqref{eqn:fibset:Cn}), to create a complete set of cyclic MUBs, it holds $F_{2^m+1}(B) = 0_m$ as well as $F_{2^m}(B) = \Eins_m$. The inverse recursion
 relation of Equation \eqref{eqn:fibpol:reverse} defines again all $F_n(B)$ with $n \in \N$. By simply counting we find that $F_{2^{m-1}}(B) = F_{2^{m-1}+1}(B)$.
\end{proof}
The generator matrix $G_j$ of the class $\cC_j$ with $j \in \MgN{N}$ is given by
\begin{align}\label{eqn:fibset:gen}
 G_j = C^j \cdot G_0 = \begin{pmatrix} F_{j+1}(B) \\ F_j (B) \end{pmatrix},
\end{align}
so to achieve Condition (III), we have to check that the matrices
\begin{align}
  (C^k \cdot G_0, C^l \cdot G_0) = \begin{pmatrix} F_{k+1} (B) & F_{l+1} (B) \\ F_{k} (B) & F_{l} (B) \end{pmatrix},
\end{align}
with $k, l \in \MgN{d}$ and $k \neq l$ are invertible. Thus, as Lemma \ref{lem:algfund:detblock}
indicates, the expression $F_{k+1} (B) \cdot F_{l} (B) + F_{l+1} (B) \cdot F_{k} (B)$ has to be invertible. But this is equal to
$F_{\vert k-l \vert}(B)$ as Lemma \ref{lem:fibpol:subtractive} shows.

Using the results of this discussion and the approach of Equation \eqref{eqn:fibset:B110}, we can reformulate Conditions (I)-(III):
\begin{enumerate}[(i)]
 \item $B$ is symmetric.
 \item $F_k(B)$ equals $\Eins_m$ for $k=d$ and $0_m$ for $k=d+1$.
 \item $F_k(B)$ is invertible for $k \in \MgE{d}$.
\end{enumerate}
By Lemma \ref{lem:fibset:symmetry} we need to check Condition (iii) only for $k \in \MgE{d/2}$, but for the subsequent discussion the chosen
form is preferable.
At this point, it is useful to call the properties of Fibonacci polynomials (cf. Section \ref{sec:fibonacci_polynomials}) back into mind in order
to reduce this list of conditions. In the following, we will discuss how the structure of Fibonacci polynomials implies that Condition (iii) 
follows from Condition (ii). Therefore, we start with a lemma which shows that the characteristic polynomial of the reduced stabilizer matrix
is irreducible, if a complete set of MUBs is created by that matrix.
\begin{lem}[Characteristic polynomial of reduced stabilizer matrix]\label{lem:fibset:charpol}\hfill\\
 The characteristic polynomial $\chi_B$ of a reduced stabilizer matrix $B \in M_m (\F_2)$ with $m \in \N^*$ is irreducible and coincides with the minimal polynomial of $B$,
 if and only if the Fibonacci index of $\chi_B$ equals either $2^m-1$ or $2^m+1$.
\end{lem}
\begin{proof}
 Any irreducible polynomial $p$ which is a factor of the characteristic polynomial of $B$ and has in the set of all factors the smallest Fibonacci index (as defined in Section~\ref{lem:fibpol:fibindex}),
 annihilates $B$. Therefore, the Fibonacci polynomial where $p$ appears first as a factor, annihilates $B$. But all irreducible polynomials
 which have Fibonacci index $2^m \pm 1$, have minimal degree $m$ (see Theorem \ref{thm:fibpol:completeness}), which is also the dimension of $B$.
\end{proof}
In the case that the characteristic polynomial of $B$, namely $\chi_B$, has Fibonacci index $d+1$ with $d=2^m$, the Fibonacci polynomial
$F_{d+1}(B)$ equals~$0_m$. With the help of the following corollary this implies that $F_{d}(B) = \Eins_m$.
\begin{cor}[Multiplicative order of reduced stabilizer matrix]\label{cor:fibset:morder}\hfill\\
 If the characteristic polynomial of $B$ has Fibonacci index $d+1$ with $d \in \N^*$, the order\footnote{
 The order of a non-zero polynomial $f(x) \in \F_{p^m}[x]$ is defined to be the smallest natural number $e$ for which the polynomial divides $x^e-1$ \cite[Definition 3.2]{LN08}.
 If the polynomial has no trivial root, the order of the polynomial equals the order of any of its roots in the multiplicative group $\F^*_{p^m}$ \cite[Theorem 3.3]{LN08}.}
 of $\chi_B$ and therefore the multiplicative order of $B$ divides $d-1$.
\end{cor}
\begin{proof}
 As shown by Lemma \ref{lem:fibset:charpol}, the characteristic polynomial of $B$ is irreducible and has degree $m$. Since it is defined over the
 field $\F_2$, its splitting field is isomorphic to $\F_{2^m}$, thus all roots of $\chi_B$ are elements of the multiplicative group of $\F_{2^m}$
 with $2^m-1$ elements. Since $B$ is a root of $\chi_B$ by the Hamilton-Cayley theorem, it has a multiplicative order that divides $d-1$.
\end{proof}
If we take Lemma \ref{lem:fibpol:evenindex} and set $r=1$, it follows with $F_1(x)=1$ that $F_{2^m}(x) = x^{2^m-1}$. Hence, with Corollary \ref{cor:fibset:morder}
we get $F_d(B)= \Eins_m$ for \mbox{$d=2^m$}. This reduces Condition (ii) to the requirement that the characteristic polynomial of $B$ has Fibonacci
index $d+1$. To combine this new statement with Condition (iii), the following field-theoretic proposition is essential.
\begin{prop}[Representation of fields]\label{prop:fibset:rep}\hfill\\
 Let $L/K$ be a field extension\footnote{An introduction and further information about field extensions can be found in Lidl and Niederreiter \cite[Section 1.4]{LN08}.}
 of order $n \in \N^*$ and let $A \in M_n(K)$ be a matrix that has an irreducible characteristic polynomial $\chi_A$,
 which is naturally an element of the polynomial ring $K[x]$. Then, the set $\Mg{f(A) \vert f \in K[x]}$ is isomorphic to the extension field $L$.
\end{prop}
\begin{proof}
 Since one of the roots of the irreducible polynomial $\chi_A$ is $A$ itself (by Hamilton-Cayley), we can simply adjoin this root $A$ to the
 ground field $K$ in order to obtain the extension $L$ which is isomorphic to the polynomial ring $K[x]$ (cf. \cite[pp. 66]{LN08}).
\end{proof}
Given that Condition (ii) is fulfilled, namely the Fibonacci index of the characteristic polynomial $\chi_B$ of the reduced stabilizer
matrix $B \in M_m(\F_2)$ equals $d+1$, no Fibonacci polynomial with a positive index smaller than $d+1$ annihilates $B$. But since
all of those polynomials are elements of the polynomial ring $\F_2[B]$ which is isomorphic to a finite field (as Proposition \ref{prop:fibset:rep} states),
they are invertible. But this is Condition (iii). Finally, we can rewrite Condition (ii) (which implies Condition (iii)) as:
\begin{enumerate}[(ii')]
 \item The characteristic polynomial of $B$ has Fibonacci index\index{Fibonacci!index} $d+1$.
\end{enumerate}
Hence, if we are interested in the construction of a complete set of mutually unbiased bases with a cyclic generator for a Hilbert space
of dimension $d$--how it is discussed here--we may propose the following algorithm:
\begin{enumerate}
 \item Find an irreducible polynomial $p$ of degree $m$ that divides $F_{d+1}$.
 \item Check if $p$ divides any Fibonacci polynomial with an index that divides $d+1$, if so, go back to $1$.
 \item Find a symmetric matrix $B \in M_m(\F_2)$ that has $p$ as its characteristic polynomial.
\end{enumerate}
To ensure the reliability of this algorithm we need to prove the existence of an appropriate Fibonacci polynomial
as well as the corresponding reduced stabilizer matrix $B$.
\begin{thm}[Existence of reduced stabilizer matrices]\label{thm:fibset:existence}\hfill\\
 For any dimension $d=2^m$ with $m \in \N^*$ there exists a reduced stabilizer matrix $B \in M_m(\F_2)$ that
 is symmetric and has an irreducible characteristic polynomial with Fibonacci index $d+1$, hence fulfills
 Conditions (i) and (ii').
\end{thm}
\begin{proof}
 We will follow the checklist given above. The number of polynomials with Fibonacci index $2^m+1$ is given by
 $\frac{\phi(d+1)}{2m}$, where $\phi$ denotes Euler's totient function (cf. Theorem 8 of \cite{Goldwasser02});
 by Theorem \ref{thm:fibpol:completeness} their degree is $m$.
 Since this expression is non-zero, there exists at least one polynomial for any dimension $d$ that has
 Fibonacci index $d+1$.\par
 For the last step, we take the well-known fact that all polynomials $p \in K[x]$ have a companion matrix $A$,
 i.\,e. a matrix which has a characteristic polynomial that equals $p$.\footnote{Despite the fact that different
 constructions for a companion matrix exist, it is possible to define such a matrix with a simple structure
 uniquely.} Finally, for every monic polynomial that is
 defined over a finite field $K=\F_q$ there exists a symmetric matrix that is similar to the companion matrix and,
 therefore, has the same characteristic polynomial (cf. Lemma 2 of \cite{BT98}).
\end{proof}
Having this theorem in mind, we are free in the way of how to construct the stabilizer matrix. We propose
three different ways:

\begin{itemize}
 \item Testing all symmetric matrices $B \in M_m(\F_2)$.
 \item Testing a reduced set of matrices by defining a more specific form of $B$.
 \item Creating $B$ directly, as a \emph{symmetric companion matrix}\index{Symmetric companion matrix}.
\end{itemize}

The first two suggestions need numerical methods, which will be discussed in the following paragraph, whereas
the last approach aims on an analytic solution that leads to a conjecture given in the subsequent paragraph.

\subsection{Numerical construction of reduced stabilizer matrix}\label{subsec:fibset:numerical}
Within this paragraph we concentrate on the first two suggestions given above to find an appropriate reduced stabilizer
matrix $B$. We will estimate the runtime of the proposed algorithm (Steps 1.--3.) and compare this result with an alternative approach.\par
This idea is not based on the construction of any polynomial, it uses the resulting Condition (ii') indirectly and starts
with testing a chosen set of possible matrices that fulfill Condition (i) by construction. The task of Condition~(ii') is to
ensure, that the characteristic polynomial of $B$ appears the first time as a factor in the Fibonacci polynomial $F_{d+1}$ with $d=2^m$, such that
$F_{d+1}(B)=0_m$. If so, then Corollary \ref{cor:fibset:morder} tells us that $F_{d}(B)=\Eins_m$, thus $C^{d+1}$ equals $\Eins_{2m}$.
Furthermore, it has to be guaranteed that all Fibonacci polynomials $F_j(B)$ with $j \in \MgE{d}$ do not equal $0_m$. But
by Lemma \ref{lem:fibpol:divisibility}, only such Fibonacci polynomials $F_j(B)$ divide $F_{d+1}(B)$, that have an index
$j$ that is a divisor of~$d+1$. Thus, if none of those Fibonacci polynomial $F_j(B)$ is zero, where $j$ is a non-trivial divisor
of $d+1$, the Fibonacci index of $B$ is $d+1$ and Condition~(ii') is fulfilled. This brings us to the following alternative algorithm:
\begin{enumerate}[1'.]
 \item Take a stabilizer matrix $C$ from a previously defined test set.
 \item Check if the power to $d+1$ of this stabilizer matrix equals $\Eins_{2m}$.
 \item Check that for any power of $C$ to a non-trivial divisor of $d+1$ the off-diagonal blocks do not equal $0_m$,
\end{enumerate}
whereas we construct $C$ in the form of Equation \eqref{eqn:fibset:B110}.\footnote{After finishing this thesis, the author
realized that only those powers of $C$ have to be tested, which are given by $d+1$ divided by one of the prime factors.
This holds due to Lemma~\ref{lem:fibpol:divisibility} and gives an exponential advantage over the discussed method. In the
calculation of the computational costs the divisor function has to be replaced by another function with
exponentially reduced values.}\par

To be able to compare the runtime of both methods, we examine this question by calculating the computational costs. For
simplicity reasons, we will take the number of \emph{array accesses}\footnote{The most time consuming operations are reading
and writing elements of matrices, as the complete matrices cannot be stored directly into the processor registers. In comparison
to these operations, arithmetical operations are done instantaneously.} as an appropriate measure.\par

For the first method (Steps 1.--3.) we have to calculate the characteristic polynomial of an $m \times m$ matrix.
Therefore, we may derive the sums of all principal
minors. One of these is the determinant of $B$, which is calculated by the Leibniz formula
\begin{align}
 \det(B) = \sum\limits_{\sigma \in S_m} \mathrm{sgn}(\sigma) \prod\limits_{i=1}^m  b_{i,\sigma(i)},
\end{align}
where $b_{i,j}$ denotes the matrix element of $B$ and $S_m$ is the symmetric group of all possible permutations of
the index set $\MgE{m}$. The calculation of the determinant in this way yields to $\vert S_m \vert \cdot m = m \cdot m!$
array accesses. For the next minor sum, we have to sum up all determinants where any pair $(i,i)$ with $i \in \MgE{m}$
of row and column is deleted. This yields $m$ principal minors with $\vert S_{m-1} \vert \cdot (m-1) = (m-1) \cdot (m-1)!$
array accesses each, thus $(m-1) \cdot m !$ array accesses. If we follow this approach and take the sum of all array
accesses, we end up with
\begin{align}\label{eqn:fibset:badlimit}
 \left( \frac{ (m-1)\cdot m}{2} \right) \cdot m!
\end{align}
array accesses for the calculation of the characteristic polynomial of an $m \times m$ matrix.\par

For the alternative method (Steps 1'. --3'.), we need to calculate the power to $2^m+1$ of a stabilizer matrix $C$ as given in
Equation \eqref{eqn:fibset:B110}. To do this efficiently, we can
multiply the matrix $C$ with itself and write the answer to a new variable, square this, and so on. We would
end with $2^m$, which we multiply by $C$ again and write this final result. This procedure needs
$3 (m+1)$ full read outs or writings of an $2m \times 2m$ matrix, thus $12 m^3 + 12 m^2$ array accesses. At the end, we have
to check if this matrix equals $\Eins_{2m}$ which yields in summary $12 m^3 + 16 m^2$ array accesses.
Since we have to check all possible factors to be not equal to a block-diagonal matrix with $m \times m$ submatrices, we have to take a similar procedure
$d(2^m+1)$ times, where $d(n)$ with $n \in \N^*$ denotes the divisor function that counts the number of divisors of
an integer $n$. Since any divisor of $2^m+1$ can be represented as an $m$-bit number (with leading zeros), we can
use the same number of array accesses as an upper bound for this calculation. We have numerical evidence, that for odd integers in the
form $n=2^m+1$ with $m\geq 2$, the divisor function is upper bounded by $\sqrt{n}$ (including $1$ and $n$), as it is true for at
least $m=200$. This limits the total number of array accesses to 
\begin{align}\label{eqn:fibset:mmlimit}
 (\sqrt{2^m+1}-1) \cdot (12 m^3 + 16 m^2).
\end{align}

It is clear that the second method scales much better since the factorial in Equation~\eqref{eqn:fibset:badlimit} grows faster than
any polynomial and obviously faster than $2^m$.
An explicit calculation shows that for $m \geq 7$ the second method is preferable. We are aware of the fact that we
left out some terms like the generation of all divisors of a number, since they are not that much relevant in
the scaling and cannot easily be taken into account. We did also not mention the fact, that we have to find a polynomial with 
an appropriate Fibonacci index in the first case, which is another reason to use the second method.\par

As a next step we should decide which set of matrices we take as our test set. If we test all possible $m \times m$ matrices which have only ones
and zeros as entries, this yields $m^2$ free parameters, which we denote by $g$. Consequently, we have $2^{m^2}$ matrices to test. 
But we are able to fulfill Condition (i) by construction, by shrinking the number of free parameters, such that the
resulting matrices are symmetric which leaves us $2^{(m^2 + m)/2}$ matrices to test. In this case we have $g=(m^2 + m)/2$.
The number of all $m \times m$ matrices $B$ that are symmetric and invertible is given by\footnote{Refer to 
\emph{Sloane's A086812} and Brent \etal{} \cite{Brent88}.}
\begin{align}
 a(m) = 2^{m (m+1)/2} \frac{\prod_{j=1..2 \cdot \lceil \frac{m}{2} \rceil } (1-(1/2)^j)}{ \prod_{j=1..\lceil \frac{m}{2} \rceil} (1-(1/4)^j)},
\end{align}
but it would be quite challenging to implement a construction that creates only these matrices with higher efficiency
than creating all symmetric matrices.\par

A method that really shrinks the number of free parameters is published in \cite{KRS10}. Therefore, the upper left half of the matrix
is fixed to one, the lower right to zero, but with the lower-right corner being a $r \times r$ matrix $A$ with free parameters ($r \in \MgN{[ m/2 ]}$). For $m=5$
the possible reduced stabilizer matrices would look like
\begin{align}\label{eqn:fibset:A}
 B= \begin{pmatrix}
     1 & 1 & 1 & 1 & 1\\
     1 & 1 & 1 & 1 & 0\\
     1 & 1 & 1 & 0 & 0\\
     1 & 1 & 0 & a_{11} & a_{12}\\
     1 & 0 & 0 & a_{12} & a_{22}
    \end{pmatrix},
\end{align}
leaving $3$ free parameters. For dimensions $m=2$ and $3$ this matrix $A = (a_{i,j})$ is a $1 \times 1$ matrix. For $m=1$
the matrix $B=(1)$ is already a solution ($A$ does not exist by construction). For most of the dimensions with $m\leq30$, the matrix $A$ needs only to
be a $2 \times 2$ matrix. Starting with dimensions $m=12,20,21,25,28$ it turns out, that this matrix needs to be at least a $3 \times 3$ matrix,
thus having six free parameters.
In higher dimensions that size has to be increased further. Comparing with the general $(m^2 + m)/2$ parameters to
test, this approach is quite a good improvement, but we have no rigorous proof that the set of matrices this scheme produces, contains always a solution.
Nevertheless, solutions for $m=2,\ldots,600$ are given in Appendix \ref{app:fibonacci_based:triangle}.\par

An approach that shrinks even more the number of possible matrices in the test set is given in the following paragraph, where we set the number of free
parameters to ${2m-1}$.

\subsection{Reduced stabilizer matrix as symmetric companion matrix}\label{subsec:fibset:analytical}\index{Symmetric companion matrix}
To calculate a certain reduced stabilizer matrix $B$, we like to investigate the existence of a symmetric companion matrix, which means, given
any polynomial $p \in \F_2[x]$ with degree $m$, we can directly construct a matrix that has $p$ as its characteristic polynomial\footnote{Ideally, this
matrix would be unique.}. Thus,
we do not have to test different matrices on their suitability. This construction aims on the first algorithm that was given above (Steps 1.--3.).\\
The idea is the following, based on observations on the different solutions for reduced stabilizer matrices:
The reduced stabilizer matrix $B$ is described by a set of ${2m-1}$ free parameters $\Mg{s_1,\ldots,s_{2m-1}}$ as
\begin{align}\label{eqn:fibset:symmComp}
 B= \begin{pmatrix}
     s_1    & s_2     & \cdots & s_{m-1} & s_m\\
     s_2    & s_3     &        & s_m     & s_{m+1}\\
     \vdots &         & \ddots &         & \vdots\\
     s_{m-1}& s_{m}   &        & s_{2m-3}& s_{2m-2}\\
     s_m    & s_{m+1} & \cdots & s_{2m-2}& s_{2m-1}
    \end{pmatrix}.
\end{align}
In other words, the matrix elements $b_{i,j}$ with $i,j \in \MgE{m}$ of $B$ are given by
\begin{align}\label{eqn:fibset:antidiag}
 b_{i,j} = s_{i+j-1}.
\end{align}
As we were not able to construct this companion matrix uniquely yet, we can only conjecture its existence and guess a possible solution.
\begin{con}[Existence of symmetric companion matrix]\label{con:fibset:exsymmcomp}\hfill\\
 For any polynomial $p \in \F_2^m$ with $m \in \N^*$ there exists a symmetric companion matrix $B \in M_m(F_2)$ in the form of Equation \eqref{eqn:fibset:symmComp} which has
 characteristic polynomial $\chi_B = p$.
\end{con}
A table with string representation of possible solutions for $s$ is given for $m=\MgE{36}$ in Appendix \ref{app:fibonacci_based:companion}. Contrary to the
analytical expectation, the method of the last paragraph that finds a matrix $B$ in the form of Equation \eqref{eqn:fibset:A} succeeds faster than finding the symmetric
companion matrix by testing all possibilities. The reason for this contradiction is that, effectively, we find solutions for matrices $A$ that have
a dimension $r$ which is much smaller than the limit of $[m/2]$. Nevertheless, from a mathematical point of view, it is quite interesting to                                                                                                                                                
have a candidate for a symmetric companion matrix that serves for all polynomials over $\GF{2}$.



\section{Fermat sets}\label{sec:fermatset}
A special subclass of the Fibonacci-based sets which were discussed in Section \ref{sec:fibset}~is given by
sets we denote as \emph{Fermat sets}\index{Fermat sets}. These sets are only defined for dimensions $d=2^{2^k}$ with $k \in \N$. Since the number
of elements of a complete set of MUBs in such a dimension is given by $2^{2^k}+1$, they are called Fermat sets,
where the Fermat numbers are defined as
\begin{align}
 \mathcal{F}_k = 2^{2^k}+1.
\end{align}
The outstanding feature of these sets is their recursive construction that supersedes the search of an
appropriate reduced stabilizer matrix as it is common for general Fibonacci-based sets. The only caveat of these
Fermat sets is that in general their completeness (the fact that a complete set of $d+1$ MUBs is created) can
only be conjectured. Within the subsequent discussion
we will derive a connection of this completeness to an open conjecture on \emph{iterated quadratic extensions of $\F_2$}
that was found by D. Wiedemann in 1988 \cite{Wiedemann88}, is of current interest \cite{MS96,Voloch10},
and was tested for $k=0,\ldots,8$ by Wiedemann himself.
By the methods derived in Section \ref{sec:fibset} we are able to raise this limit to $k=11$, the largest
Fermat number where the complete factorization into prime numbers is currently known. It is clear, that the
problem of testing if possible factors divide a Fermat number is easier than testing if the power of the stabilizer
matrix $C$ to the same factors has vanishing off-diagonal blocks of size $m \times m$. The program code we use is
given in Appendix \ref{app:fermat_based:wiedemann}.\par
The solutions we propose to obtain complete sets of cyclic MUBs for dimensions $d=2^{2^k}$ with $k \in \N$ are based on
Equation \eqref{eqn:fibset:B110}. We thus need to find an appropriate reduced stabilizer matrix $B$. In the following
we add an index to $B$ that refers to the integer $m$ which defines the dimension as $d=2^m$.
In the case of Fermat sets the number $m$ is defined by $2^k$. If we recall some already known solutions from
\cite{KRS10}, we find
\begin{align}
 B_{2^0} = (1),\quad
 B_{2^1} = \begin{pmatrix}1&1\\1&0\end{pmatrix},\quad\mathrm{and}\quad
 B_{2^2} = \begin{pmatrix}1&1&1&0\\1&0&0&1\\1&0&0&0\\0&1&0&0\end{pmatrix}.
\end{align}
If we iterate this construction, we end up with the recursion
\begin{align}\label{eqn:ferset:recursion}
 B_{2^{k+1}} = \begin{pmatrix}B_{2^k}&\Eins_{2^k}\\\Eins_{2^k}&0_{2^k}\end{pmatrix} \in M_{2^{k+1}}(\F_2).
\end{align}
Analogously, we define $C_{2^k}$ as
\begin{align}\label{eqn:ferset:C2k}
 C_{2^k} = \begin{pmatrix}B_{2^k}&\Eins_{2^k}\\\Eins_{2^k}&0_{2^k}\end{pmatrix} \in M_{2^{k+1}}(\F_2),
\end{align}
thus $C_{2^k} \equiv B_{2^{k+1}}$.
At this point, we can already use the program code from Appendix \ref{app:fermat_based:wiedemann} to see that the
recursion produces complete sets of cyclic MUBs for $k\in \MgN{11}$. But this does not help us understanding the
problem and developing ideas that may serve to prove this recursion for all $k$. It is clear by construction that the reduced stabilizer
matrices are symmetric, thus Condition~(i) of Section \ref{sec:fibset} is fulfilled. To be able to check if the
construction is also conform with Condition (ii') we shall start with calculating the characteristic polynomial $\chi_B$
of the reduced stabilizer matrix $B$. It will turn out that the notion of \emph{reciprocal}\index{Reciprocal polynomial}
and \emph{self-reciprocal} polynomials is of interest within this context.

\begin{defi}[Reciprocal polynomial]\hfill\\
  The \emph{reciprocal} of a polynomial $f \in K[x]$ is defined by $f^{*} (x) := x^n f(x^{-1})$, where $n \in \N^*$ is
  the degree of $f$. If a polynomial coincides with its reciprocal, it is called \emph{self-reciprocal}.
\end{defi}
We can create a self-reciprocal polynomial by applying the \emph{reciprocal operator}\index{Reciprocal operator} $Q$ on an
arbitrary polynomial $f$ with degree $n$, namely
\begin{align}
 f^Q(x) := x^n f(x+x^{-1}).
\end{align}
The recursion we propose in Equation \eqref{eqn:ferset:recursion} is a matrix representation of the application of the
reciprocal operator, starting with an appropriate polynomial.
\begin{lem}[Characteristic polynomials]\label{lem:ferset:charpol}\hfill\\
 Let $K$ be a finite field of characteristic $2$ and $\chi_S$ be the characteristic polynomial of a matrix $S \in M_n(K)$ with $n \in \N^*$.
 Then $\chi_{S'} = (\chi_S)^Q$ is the characteristic polynomial of 
\begin{align}
 S' := \begin{pmatrix}S&\Eins_n\\\Eins_n&0_n\end{pmatrix} \in M_{2n}(K).
\end{align}
\end{lem}
\begin{proof}
 The characteristic polynomial $\chi_{S'}$ is given by $\chi_{S^\prime}(x) = \det(x\Eins_{2n} - S^\prime)
 = \det ${\tiny $\begin{pmatrix} x\Eins_n - S & \Eins_n \\ \Eins_n & x\Eins_n \end{pmatrix}$}. If we use Lemma \ref{lem:algfund:detdecomp},
 we get $\chi_{S^\prime}(x) = \det(x\Eins_n) \det(x\Eins_n - S - x^{-1}\Eins_n)$. But the first factor equals $x^n$ and the
 second is given by \mbox{$\chi_S(x - x^{-1})$}. Thus $\chi_{S^\prime}(x) = x^n \cdot \chi_S(x - x^{-1}) = (\chi_S)^Q(x)$.
\end{proof}
We can use Lemma \ref{lem:ferset:charpol} to show that the characteristic polynomials of the~$B_{2^k}$ are factors of the
Fibonacci polynomials $F_{2^m+1}$ with $m=2^k$. With the help of Theorem~\ref{thm:fibpol:completeness} and Lemma \ref{lem:fibpol:linear},
it turns out that these polynomials have to be irreducible and keep their linear term in order to be possible candidates for producing
complete sets of cyclic MUBs.
\begin{lem}[Minimal polynomial of $B_{2^k}$]\label{lem:ferset:minpol}\hfill\\
 The characteristic polynomial of $B_{2^k}$ as defined in Equation \eqref{eqn:ferset:recursion} is given by $f^{Q^k}(x)$ with
 $f(x)=1+x$. It is irreducible and therefore minimal.
\end{lem}
\begin{proof}
 The first part follows by induction from Lemma \ref{lem:ferset:charpol}. It was shown by Varshamov and Garakov \cite{VG69}
 and in a more general form by Meyn \cite{Meyn90} that for a polynomial $g \in \F_2[x]$ and $g$ being irreducible, $g^Q$ is also irreducible,
 if and only if the linear coefficient of $g$ does not vanish. This is true for $f$ by construction. But for an arbitrary
 polynomial $h(x)=\sum_{j=0}^n a_j x^j$ it holds $h^Q(x) = \sum_{j=0}^n \sum_{i=0}^j \binom{j}{i} a_j x^{n+j-2i}$. The only
 contribution to the linear coefficient from $h^Q(x)$ arises from $j=i=n-1$ which is given by $a_{n-1}$. Since the reciprocal
 operator creates a self-reciprocal function this stays true for all~$k$.
\end{proof}
Thus, regarding Condition (ii') of Section \ref{sec:fibset}, we have shown by\linebreak Lemma~\ref{lem:ferset:minpol} in combination with Lemma \ref{lem:fibpol:linear} that the
Fibonacci index of the reduced stabilizer matrices $B_{2^k}$ in the form of Equation \eqref{eqn:ferset:recursion} equals
$d+1$ or is a divisor of $d+1$ with $d=2^{2^k}$. We have tried to prove that the Fibonacci index is $d+1$, but without
success. A promising approach can be found in Appendix \ref{app:wiedemann_proofs}.\par
However, we are able to relate the question of the Fibonacci index to an open conjecture from finite field theory, namely
\emph{Wiedemann's conjecture}\index{Wiedemann's conjecture}, which was given already in 1988 \cite{Wiedemann88}. 
Wiedemann considered \emph{iterated quadratic extensions} of the finite field $\F_2$ using generators $x_j$ which are
recursively defined as
\begin{align}\label{eqn:ferset:wiedemann_rec}
 x_{j+1} + x_{j+1}^{-1} = x_j,
\end{align}
with $j \in \N$ and $x_0 + x_0^{-1} = 1$, where $x_0 \in \F_2$. Rewriting
Equation \eqref{eqn:ferset:wiedemann_rec} into
\begin{align}\label{eqn:ferset:wiedemann_rec2}
 x^2_{j+1} + x_{j+1}x_j +1 = 0
\end{align}
shows that the roots of this equation are not in the field where $x_j$ belongs to, but in its quadratic extension, due to
the quadratic nature of this equation \cite[Theorem 1]{Wiedemann88}. Accordingly, he defined the (extension) fields $E_j$ as
\begin{align}
 E_{j+1} := E_j (x_{j+1}),
\end{align}
with $E_0 := \F_2(x_0)$. The extension field $E_j$ is then isomorphic to the finite field $\F_{2^{2^{j+1}}}$. As $x_j^{-1}$
is the \emph{conjugate} of $x_j$ with respect to the ground field, $x_j^{-1} = x_j^{2^{2^j}}$ holds \cite[Lemma 2.14]{LN08}.
Thus, the multiplicative order of $x_j$ divides the $j$-th Fermat number $\mathcal{F}_j = 2^{2^j} + 1$. Calculating
the product of the roots as $\tilde{x}_n := x_0x_1\dots x_n \in E_n$, the order of $\tilde{x}_n$ is given by the product
of the orders of the $x_j$, since the Fermat numbers are mutually coprime. Wiedemann finally conjectured, that the orders
of $x_j$ are given by $\mathcal{F}_j = 2^{2^j} + 1$, thus $\vert \tilde{x}_n \vert = \F_{2^{2^{n+1}}}-2$ would hold,
meaning $\tilde{x}_n$ would be a \emph{primitive element} of $E_n$, generating its multiplicative subgroup $E^*_n$.
He successfully tested his conjecture computationally for $j=0,\ldots,8$. The matrices $C_{2^k}$ defined in Equation~\eqref{eqn:ferset:C2k}
can be seen as a realization of the $x_k$.

\begin{thm}[Wiedemann analogy]\label{thm:ferset:wiedemann}\hfill\\
  Wiedemann's conjecture is true, if and only if the characteristic polynomial of all reduced stabilizer matrices $B_{2^k}$ which
  are defined in Equation \eqref{eqn:ferset:recursion}, has Fibonacci index $2^{2^k}+1$.
\end{thm}
\begin{proof}
  If the stabilizer matrix $C_{2^k}$ is constructed due to Equation \eqref{eqn:ferset:C2k},\linebreak
  $C_{2^0} + C_{2^0}^{-1} = \Eins_2$ holds obviously. Assuming that the order of $C_{2^k}$ is given by $2^{2^k}+1$, the
  characteristic polynomial of $B_{2^k}=C_{2^{k-1}}$ has Fibonacci index $2^{2^k}+1$ (which is in accordance to Condition (ii') of
  Section \ref{sec:fibset}). The characteristic polynomial of $C_{2^{k}}$ over the field where $B_{2^k}=C_{2^{k-1}}$ belongs to,
  is then by construction given as $x^2 + C_{2^{k-1}} x + 1$. By comparing this result with Equation \eqref{eqn:ferset:wiedemann_rec2},
  the stabilizer matrices $C_{2^{k}}$ can be identified with Wiedemann's $x_k$; in equivalence with Wiedemann's construction
  the $C_{2^k}$ are the roots of their characteristic polynomials.
\end{proof}
With Theorem \ref{thm:ferset:wiedemann} it is shown that the construction of complete sets of cyclic MUBs with a reduced stabilizer matrix
in the form of Equation \eqref{eqn:ferset:C2k} is an instance of Wiedemann's conjecture. A potential approach to proof this conjecture is
given in Appendix \ref{app:wiedemann_proofs}. We followed Wiedemann's idea and tested his conjecture for $j=0,\ldots,11$, limited by
the largest Fermat number with a known prime factorization (cf. Appendix \ref{app:fermat_based:wiedemann}).

\section{Standard form}\label{sec:standardform}
All cyclic sets of MUBs which are based on the ideas of Bandyopadhyay \etal{} \cite{Bandy02}
need to partition the set of Pauli operators into disjoint classes of commuting elements, as discussed
in Section \ref{sec:fibset} that deals with Fibonacci-based sets. The Properties~\text{(I)--(III)} should
always be fulfilled, but the form of the stabilizer matrix $C$ can be different to that given in Equation
\eqref{eqn:fibset:B110}. Additionally, we have the freedom to choose a different generator $G_0$ than
$G_0 = (\Eins_m, 0_m)^t$ in order to prevent the appearance of a set with all Pauli-$Z$ operators (which will be identified as the
standard basis in Section~\ref{sec:Uconstruction}). To be able to observe a complete set of
cyclic MUBs within a common picture--different and more general than by relations to Fibonacci
polynomials--we introduce the so-called \emph{standard form}\index{Standard form}, derived from the
Fibonacci-based sets. Unfortunately, this form does not preserve the cyclic structure of the generators of the
classes, but it can
be seen as an instrument to identify certain properties of the analyzed set of MUBs. It turns out
that this form is also used by Bandyopadhyay \etal{}, but not in this very general form
(cf. Equations \eqref{eqn:bandy:G0} and \eqref{eqn:bandy:Gj} of Section~\ref{sec:bandy}).\par
We start with the generators that follow Equation \eqref{eqn:fibset:Cn} and the fact that $G_0$ is
fixed for Fibonacci-based sets. We find
\begin{align}\label{eqn:standfor:generators}
 G_0 := \begin{pmatrix}\Eins_m\\0_m\end{pmatrix}\quad\mathrm{and}\quad G_j= \begin{pmatrix}F_{j+1}(B)\\F_j(B)\end{pmatrix},
\end{align}
for $j \in \MgE{d}$. As the class $\mathcal{C}_j'$ consists of all Pauli operators with argument
$G_j\cdot \vec c$ and $\vec c \in \F_2^m$ (cf. Equation \eqref{eqn:bandy:classes}),
we have a certain freedom in the choice of $G_j$.
\begin{cor}[Generator matrix]\label{cor:standfor:generatorfree}\hfill\\
 The generator $G_j$ with $j \in \N$ that creates a class in the form 
 $\mathcal{C}^{\prime}_j = \{ \ZX(\vec a) : \vec a = G_j \cdot \vec c: \vec c \in \F_2^m \}$
 can be multiplied from the right by any invertible matrix $P \in \GL{m}{\F_2}$ to produce the same set of Pauli operators, i.\,e.
 \begin{align}
  \mathcal{C}^{\prime}_j \equiv \{ \ZX(\vec a) : \vec a = G_j \cdot P \cdot \vec c: \vec c \in \F_2^m \}.
 \end{align}
\end{cor}
\begin{proof}
 Since the class is created by all vectors $\vec c$ with $\vec c \in \F_2^m$, the invertible matrix $P$ only permutes
 this set by executing the operation $P \cdot \vec c$.
\end{proof}
Using Corollary \ref{cor:standfor:generatorfree}, we can write the generators from above (Equation~\eqref{eqn:standfor:generators})
in a different form by multiplying $G_j$ with $(F_j(B))^{-1}$. Recalling Proposition \ref{prop:fibset:rep}, $F_j(B)$ equals either $0_m$
or its inverse exists. Therefore, we will not change $G_0$ for this form, which yields
\begin{align}\label{eqn:standfor:generatorsstandfor}
 \bar G_0 := \begin{pmatrix}\Eins_m\\0_m\end{pmatrix}\quad\mathrm{and}\quad \bar G_j= \begin{pmatrix}F_{j+1}(B) (F_j(B))^{-1}\\\Eins_m\end{pmatrix},
\end{align}
the generators of a Fibonacci-based set written in \emph{standard form}. We introduce sets which are not
based on Fibonacci polynomials within Sections~\ref{sec:homogeneous} and \ref{sec:inhomogeneous}. It will turn out,
that the following definition is capable to describe them uniquely.
\begin{defi}[Standard form]\label{defi:standfor:standfor}\hfill\\
  Generators $G_j$ with $j \in \MgN{d}, d \in \N^*$, which create classes that partition the set of Pauli operators in order to produce a complete set
  of mutually unbiased bases are in standard form, if there holds
\begin{align}\label{eqn:standfor:standfor}
 \bar G_0 := \begin{pmatrix}\Eins_m\\G_0^x\end{pmatrix}\quad\mathrm{and}\quad \bar G_j= \begin{pmatrix}G_j^z\\\Eins_m\end{pmatrix},
\end{align}
with matrices $G_j^z, G_0^x \in M_m(\F_2)$.
\end{defi}
The values of $G_j^z$ and $G_0^x$ are clear in the case of Fibonacci-based sets. We refer to sets where $G_0^x=0_m$ as 
\emph{homogeneous sets}\index{Homogeneous sets} and sets where $G_0^x\neq 0_m$ as \emph{inhomogeneous sets}\index{Inhomogeneous sets}.
For both sets, we can in principle reuse \mbox{Conditions (I)--(III)} of Section \ref{sec:fibset}.
To fulfill Condition (III) for the homogeneous
sets, all sums $G_k^z + G_l^z$ with $k, l \in \MgE{d}, k \neq l$ have to be invertible due to Lemma \ref{lem:algfund:detblock}.
For the inhomogeneous sets, the expressions $\Eins_m + G_0^x G_j^z$ with $j \in \MgE{d}$ have to be invertible in addition.
In order to generate commuting operators by the $G_j$, the matrices $G_j^z$ and $G_0^x$ have to be symmetric \mbox{\cite[Lemma 4.3]{Bandy02}}.
It is reasonable to discuss the appearance of the different sets in relation to the entanglement properties of the classes of
Pauli operators as will be seen in Section \ref{sec:entangleprop}.

\section{Entanglement properties}\label{sec:entangleprop}
Within the constraints of the approach given by Bandyopadhyay \etal{} (cf. Section \ref{sec:bandy}),
different sets of cyclic MUBs exist. Some of these sets were studied in 2002 by Lawrence \etal{} \cite{Lawrence02},
a complete list of four different sets in $d=2^3$ dimensions was given by Romero \etal{} \cite{Romero05}.
Following the consideration of the latter work we will discuss the properties of the different sets
and formulate conditions on the stabilizer matrix $C$ in order to construct a complete set of cyclic MUBs
with specific properties. Therefore, we first introduce the idea which is behind the classification
of different sets.\par
Recalling Equation \eqref{eqn:bandy:classes},
a single basis within a set of MUBs (that is constructed according to the approach of Bandyopadhyay \etal{})
is given by the eigenspace of the Pauli operators which constitute a class
$\mathcal{C}_j$ with $j \in \MgN{d}$. The elements of such a class are generated by $G_j$ and are by
construction given by a tensor product of Pauli operators of the Hilbert space $\cH=\C^2$ (cf. Equations \eqref{eqn:bandy:classes} and \eqref{eqn:bandy:ZXa}).
Operators which are in the same class commute \cite{Bandy02}. With a single basis we are able to measure those properties
of the complete quantum system, which are describable by the Pauli operators of the corresponding class $\cC_j$.\par 
We can categorize a basis in the following way: If the Pauli operators of each two-dimensional subsystem,
as given by the tensor product decomposition, commute separately, this basis measures properties of a fully
separable system, the \emph{separability count}\index{Separability count} is $s(\cC_j)=m$.
If the Pauli operators of two subsystems within the class do not commute separately, the basis measures properties of
a system that is decomposable into $m-2$ single-qubit subsystems and a two-qubit subsystem. For a system of
$m$ subsystems, the separability count is then $s(\cC_j)=m-1$. The continuation of this classification leads to a couple
of different decomposition structures, which we formalize in the following way:
If a complete set of MUBs is used to determine the quantum state of an $m$-qubit system, we can describe the
decomposition structure of a certain basis by arranging the set of $m$ qubits into different subsets,
where the properties are measured on the completely entangled states of the subsets. For all different possible set structures
we define the separability count, by ordering all possible structures and take the position in this list as the separability count.
This ordering works as follows:
\begin{enumerate}
 \item Order the different structures by their number of subsets, starting with the largest.
 \item Order the different structures which have the same number of subsets by the size of their largest set, starting with the smallest; if they are equal, continue with the second largest set, and so on.
\end{enumerate}
Finally, we describe the decomposition structure of the whole basis set by counting the different decomposition
structures of the classes with a single column vector $\vec n$. The first entry of this vector, namely $n_1$,
describes the number of bases which are fully separable, thus, which measure properties of $m$ completely non-entangled
subsystems. Then follows~$n_2$ which counts the number of bases that describe systems
that are separable into $m-1$ parts. To get $m-2$ parts we have two possibilities. Following the algorithm for the ordering, $n_3$ gives the number
of bases which describe systems that can be decomposed into $m-2$ subsystems, with $m-4$ single-qubit systems and two two-qubit
systems; $n_4$ counts bases that describe $m-3$ single-qubit systems and one triple-qubit system. This goes logically forth until
we end up with a fully-entangled system. This ordering is slightly different from that in \cite{Romero05}, in order to
classify systems with few highly-entangled subsystems as less separable than systems with many sparsely-entangled
subsystems.\par
Since there are only three different Pauli operators for a single-qubit system,
there exist at most three different bases in a complete set of MUBs, that measure properties of a fully separable system.
So the maximum value of $n_1$ is three. The Fibonacci-based sets of Section
\ref{sec:fibset} reach this limit: the generator of the class $\cC_0$ constructs all Pauli-$Z$ operators
that can be totally decomposed into $m$ subsystems. The same holds true for the class $\cC_d$ which is generated by
\begin{align}
 G_d = C^d \cdot G_0 = \begin{pmatrix}0_m&\Eins_m\\\Eins_m& B \end{pmatrix} \cdot \begin{pmatrix}\Eins_m\\0_m\end{pmatrix} = \begin{pmatrix}0_m\\ \Eins_m\end{pmatrix},
\end{align}
and contains all Pauli-$X$ operators. Finally, the class $\cC_{d/2}$ is generated by $G_{d/2} = (F_{d/2}(B), F_{d/2+1}(B))^t$
which equals $(F_{d/2}(B), F_{d/2}(B))^t$ (see Lemma \ref{lem:fibset:symmetry}). Writing this generator in standard form
as allowed by Corollary \ref{cor:standfor:generatorfree}, we find
\begin{align}\label{eqn:entprop:Y}
 \bar G_{d/2} = \begin{pmatrix}\Eins_m\\\Eins_m\end{pmatrix},
\end{align}
that produces all Pauli-$Y$ operators.\par
As an example, we generate a complete Fibonacci-based set of cyclic MUBs for a three-qubit system with the reduced
stabilizer matrix
\begin{align}
 B_3 = \begin{pmatrix}
 1 & 1 & 1 & 1 & 0 & 0\\
 1 & 1 & 0 & 0 & 1 & 0\\
 1 & 0 & 0 & 0 & 0 & 1\\
 1& 0& 0& 0& 0& 0\\
 0& 1& 0& 0& 0& 0\\
 0& 0& 1& 0& 0& 0
 \end{pmatrix},
\end{align}
that generates the classes as follows:
\begin{align}
 \cC_0 =& \Mg{\Eins  \Eins  Z, \Eins  Z  \Eins, \Eins  Z  Z, Z  \Eins  \Eins, Z  \Eins  Z, Z  Z  \Eins, Z  Z  Z},& s(\cC_0)= 3,\nonumber\\
 \cC_1 =& \Mg{\Eins  Y  X, X  X  Z, X  Z  Y, Y  X  Y, Y  Z  Z, Z  \Eins  X, Z  Y  \Eins},& s(\cC_1)= 1,\nonumber\\
 \cC_2 =& \Mg{\Eins  Z  X, X  X  Y, X  Y  Z, Y  \Eins  X, Y  Z  \Eins, Z  X  Z, Z  Y  Y},& s(\cC_2)= 1,\nonumber\\
 \cC_3 =& \Mg{\Eins  Z  Y, X  \Eins  Y, X  Z  \Eins, Y  X  Z, Y  Y  X, Z  X  X, Z  Y  Z},& s(\cC_3)= 1,\nonumber\\
 \cC_4 =& \Mg{\Eins  \Eins  Y, \Eins  Y  \Eins, \Eins  Y  Y, Y  \Eins  \Eins, Y  \Eins  Y, Y  Y  \Eins, Y  Y  Y},& s(\cC_4)= 3,\\
 \cC_5 =& \Mg{\Eins  X  Y, X  Y  X, X  Z  Z, Y  Y  Z, Y  Z  X, Z  \Eins  Y, Z  X  \Eins},& s(\cC_5)= 1,\nonumber\\
 \cC_6 =& \Mg{\Eins  X  Z, X  Y  Y, X  Z  X, Y  \Eins  Z, Y  X  \Eins, Z  Y  X, Z  Z  Y},& s(\cC_6)= 1,\nonumber\\
 \cC_7 =& \Mg{\Eins  Y  Z, X  \Eins  Z, X  Y  \Eins, Y  X  X, Y  Z  Y, Z  X  Y, Z  Z  X},& s(\cC_7)= 1,\nonumber\\
 \cC_8 =& \Mg{\Eins  \Eins  X, \Eins  X  \Eins, \Eins  X  X, X  \Eins  \Eins, X  \Eins  X, X  X  \Eins, X  X  X},& s(\cC_8)= 3,\nonumber
\end{align}
where $\Eins$ abbreviates $\Eins_2$. The numbers
of decomposable systems for the classes are given by the separability counts. So the vector $\vec n$ in the case
of this set of MUBs is given by $\vec n_{1} = (3,0,6)^t$.\footnote{This set is commonly known as the \emph{standard set} \cite{Garcia10}.} Other possible structures are $\vec n_{2} = (2,3,4)^t$, 
$\vec n_{3} = (1,6,2)^t$ and $\vec n_{4} = (0,9,0)^t$ \cite{Romero05}.\par
To generate these structures, which do not have three totally-decomposable classes, we cannot use the already discussed
Fibonacci-based sets of Section \ref{sec:fibset}. Therefore, we have to reduce the number of constraints, which are required
to produce the Fibonacci-based sets, in order to achieve a larger set of possible structures. For Fibonacci-based sets, we find a field structure:
\begin{lem}[Field structure of Fibonacci-based sets]\label{lem:entprop:fibfield}\hfill\\
 The submatrices $G_j^z$ of the generators $\bar G_j = (G_j^z, \Eins_m)^t$ in standard form with $j \in \MgE{d}$ that are used to generate a complete set of cyclic
 Fibonacci-based MUBs, form a representation of the finite field $\F_2^m$.
\end{lem}
\begin{proof}
 Considering the standard form (cf. Section \ref{sec:standardform}) of the Fibonacci-based sets, we find the $2^m$ generators 
 $G_j = (F_{j+1}(B) F_{j}(B)^{-1}, \Eins_m)^t$. As~$B$ has an irreducible characteristic polynomial of degree $m$
 that generates a polynomial ring and $2^m$ different elements $F_{j+1}(B) F_{j}(B)^{-1}$ exist, the set of those elements
 is a representation of a polynomial ring, in particular of the finite field $\F_2^m$.
\end{proof}
To create cyclic sets of MUBs it is more comfortable to use homogeneous sets, as discussed in Section \ref{sec:standardform},
since their first basis is used to create the standard basis (cf. Section \ref{sec:Uconstruction}). In this case the generator
of the first group is given by $G_0 = (\Eins_m, 0_m)^t$. To diminish the first value of the vector $\vec n$ to two, by keeping
the homogeneity property, it is obvious that we loose the field property of Lemma \ref{lem:entprop:fibfield}. What we
get is an additive group of the values $G_j^z$ without the unity element $\Eins_m$ of the multiplication, thus, a class
of all Pauli-$Y$ operators as given by Equation \eqref{eqn:entprop:Y} is not included.
In this case, Condition~(III) of Section \ref{sec:fibset} is fulfilled if the expression $G_l^z + G_k^z$ is invertible
for \mbox{$k, l \in \MgE{d}$} with $k \neq l$, but this element is again an element of the group. We can thus create this group
of matrices by a generating set, say w.\,l\,o.\,g. $H_k \in \GL{m}{\F_2} $ with $k \in \MgE{m}$ and
\begin{align}\label{eqn:standfor:group}
 G_j^z = \sum_{k=1}^m h_k H_k,\quad \mathrm{for}\quad (h_1,\ldots,h_m) \in \F_2^m,
\end{align}
where all $G_j^z$ with $j\neq 0$ have to be invertible. Equation \eqref{eqn:standfor:group} is equivalent to the 
construction
of Bandyopadhyay \etal{} \cite[Section 4.3]{Bandy02}. Details will be discussed in Section \ref{sec:homogeneous}.\par
To obtain $n_1=1$, we need to exclude also the element $0_m$ which forms the class of Pauli-$X$ operators from the set of matrices
$G_j^z$ with $j \in \MgE{d}$. The highest set structure which fits is then a semigroup, which does not allow a generation as
in Equation \eqref{eqn:standfor:group}.\par
Finally, we have to give up the homogeneity property to set $n_1=0$. To be able to prevent all fully-decomposable classes, we have
to change the generator~$G_0$. In Section~\ref{sec:standardform} we discussed why $\bar G_0=(\Eins_m, G_0^x)^t$ is a good choice.
It is clear, that Condition~(III) of Section \ref{sec:fibset} has to be adapted minimally. Details of the inhomogeneous sets
will be discussed in Section \ref{sec:inhomogeneous}.

\section{Homogeneous sets}\label{sec:homogeneous}
The second class of cyclic MUBs we create in this work includes the class of Fibonacci-based MUBs
formally, but does not share completely their nice way of construction. This generalization is based on the
Conditions (I)--(III) of Section \ref{sec:fibset}. As seen in Section \ref{sec:entangleprop}, the entanglement properties of the bases within
a certain set of MUBs are different for this class of homogeneous sets. Experimentally, complete sets with
a reduced number of bases that measure properties of a fully entangled system, are to be preferred.\par
As already stated in Section \ref{sec:entangleprop}, the homogeneous sets can be build either from
generators with a field structure, as done for the Fibonacci-based MUBs, with an additive group structure
or an additive semigroup structure. The number of bases with a fully-decomposable structure is three, two
and one, respectively, in these cases. As the construction of Fibonacci-based MUBs is already solved,
we give a generalized solution in Section \ref{sec:homogeneous:group}, which generates also
group-based sets. In Section \ref{sec:homogeneous:semigroup}, a first notion on the realization
of semigroup-based sets is presented.

\subsection{Group-based sets}\label{sec:homogeneous:group}
Considering the properties of the stabilizer matrix to generate a complete set of cyclic MUBs,
we use Conditions (1)--(3) of Section \ref{sec:bandy}. We propose a generalized form
of the stabilizer matrix as in Section \ref{sec:fibset} in order to implement an additive group structure of the matrices $G_j^z$ in the
standard form, that has solutions which avoid the field structure (cf. Equation \eqref{eqn:standfor:standfor}). The
stabilizer matrix we propose is given by
\begin{align}\label{eqn:homosets:group:C}
 C = \begin{pmatrix}B & R\\ R^{-1} & 0_m\end{pmatrix},
\end{align}
where the reduced stabilizer matrix $B \in M_m(\F_2)$ appears, but also an invertible matrix
$R \in \mathrm{GL}_m(\F_2)$. The powers of the stabilizer matrix,
as given by Equation \eqref{eqn:homosets:group:C}, are
\begin{align}
 C^n =
\begin{pmatrix}
 F_{n+1}(B) & F_n(B) R\\
 R^{-1} F_n(B) & R^{-1} F_{n-1}(B) R
\end{pmatrix},
\end{align}
with $F_n(x)$ being the Fibonacci polynomial with index $n \in \N$, as defined in Equation \eqref{eqn:fibpol:normal}.
In standard form\index{Standard form} the generators $G_j$ of the classes $\cC_j$ (cf. Equation \eqref{eqn:bandy:classes}) look like
\begin{align}\label{eqn:homosets:group:standfor}
 \bar G_0 = \begin{pmatrix}\Eins_m\\0_m\end{pmatrix}\quad\mathrm{and}\quad \bar G_j= \begin{pmatrix} F_{j+1}(B) F_j(B)^{-1} R\\\Eins_m\end{pmatrix},
\end{align}
thus, the field elements $G_j^z$ of Equation \eqref{eqn:standfor:generatorsstandfor} are multiplied by the
invertible matrix $R$. The new elements $G_j^z$ of Equation \eqref{eqn:homosets:group:standfor} form an additive group which may be written by
a generating set as stated by Equation \eqref{eqn:standfor:group}.
We can define conditions which are similar to Conditions (i) and (ii') of Section~\ref{sec:fibset}.
Since those conditions emerge from Conditions (1)--(3) of Section \ref{sec:bandy}, we can start our discussion there. As seen by
Equation \eqref{eqn:homosets:group:standfor}, the generators have the required form, thus, Condition (1) is fulfilled by construction.
For Condition (2), all elements $G_j^z$ with $j \in \MgE{d}$ and $d=2^m$ have to be symmetric. Those elements are given by the polynomials
$p_n(B)$, defined as $p_n(B) = F_{n+1}(B) F_n(B)^{-1}$ with $n \in \MgE{d}$, multiplied by the matrix $R$. By the following lemma the
number of matrices to be tested on their symmetry can be reduced dramatically to two.
\begin{lem}[Symmetry of Fibonacci polynomials multiplied with invertible matrix]\label{lem:homosets:group:symmetry}\hfill\\
 For invertible matrices $A, B \in \GL{m}{K}$, all polynomials $p(A) \in K[x]$ multiplied with $B$ from the right are symmetric,
 if and only if $B$ and $AB$ are symmetric.
\end{lem}
\begin{proof}
 Let $B$ and $AB$ be symmetric matrices. Then for any matrix $A^k B$ with $k \in \N$ there holds
 \begin{align}
   (A^k B)^t = B^t A^t (A^{k-1})^t = A B^t A^t (A^{k-2})^t = \ldots = A^k B.
 \end{align}
 Any polynomial in $A$ multiplied with $B$ is then a sum of symmetric matrices. But sums of symmetric matrices are again symmetric.
 The converse is obvious, if we consider the cases $k=0$ and $k=1$, thus $\Eins_m B$ and $AB$.
\end{proof}
Using Lemma \ref{lem:homosets:group:symmetry}, Condition (2) of Section \ref{sec:bandy} is fulfilled, if and only if $R$ and $BR$ are symmetric.
Condition (3) is also correct, since the difference of two group elements leads again to a group element. If all three conditions
are fulfilled, the stabilizer matrix $C$ has to be symplectic, but as a precaution we can check this easily by using Corollary
\ref{cor:app:clifford:symplecticmatrixprop} of Appendix \ref{app:clifford}. The product $R (R^{-1})^t$ equals $\Eins_m$ as $R$ is
a symmetric matrix. $B^t R^{-1}$ is symmetric, since its transpose $(R^{-1})^t B$ is the product of a polynomial in $R$
(cf. Proposition \ref{prop:fibset:rep}) with the symmetric matrix $B$, therefore, we can apply Lemma \ref{lem:homosets:group:symmetry}.
To get a complete set of MUBs, the similarity of Equation~\eqref{eqn:homosets:group:standfor} with Equation \eqref{eqn:standfor:generatorsstandfor} indicates,
that the characteristic polynomial of $B$ has to have a Fibonacci index of $d+1$. Finally, we find the following conditions:

\begin{enumerate}[(i)]
 \item $R$ and $B  R$ are symmetric.
 \item The characteristic polynomial of $B$ has Fibonacci index $d+1$.
\end{enumerate}
The polynomials $p_n(B)$ have field structure, but multiplied by the matrix $R$ this structure disappears. Only an additive
group structure is left, as wanted. For $R = \Eins_m$ or more generally for $R = p_n(B) \neq 0_m$ with $n \in \MgE{d}$, we find
the Fibonacci-based sets. Thus, for appropriate values of $R \neq p_n(B)$ we find complete sets of cyclic MUBs that have exactly two bases with
a totally-decomposable structure, namely $\cC_0$ and $\cC_d$.\par
According to Equation \eqref{eqn:homosets:group:C}, we find $126$ different pairs $(B,R)$ for a system of three qubits that have
a decomposition structure of the bases as given in~\cite{Romero05}, namely $\vec n = (2,3,4)^t$. One of the solutions is
given by
\begin{align}\label{eqn:homoset:group:C234}
 B_{(2,3,4)} = 
\begin{pmatrix}
 0 & 1 & 1\\
 0 & 0 & 1\\
 1 & 0 & 0
\end{pmatrix}\quad\mathrm{and}\quad
 R_{(2,3,4)} = 
\begin{pmatrix}
 0 & 0 & 1\\
 0 & 1 & 0\\
 1 & 0 & 0
\end{pmatrix}.
\end{align}
It is expected that most of the $126$ solutions can be transformed into each other with transformations that originate from simple
symmetry considerations. One of those transformations corresponds only to sets which are not related to the Fibonacci-based sets and was
indirectly discussed above. If we replace the matrix $R$ in the construction of Equation \eqref{eqn:homosets:group:C} by the product
of a polynomial $p_l$ in $B$, with the old matrix $R$, we get
\begin{align}\label{eqn:homosets:group:Cfactor}
 C = \begin{pmatrix}B & p_l(B) R\\ R^{-1} p_l(B)^{-1}  & 0_m\end{pmatrix},
\end{align}
with generators in standard form as
\begin{align}\label{eqn:homosets:group:standforfactor}
 \bar G_0 = \begin{pmatrix}\Eins_m\\0_m\end{pmatrix}\quad\mathrm{and}\quad \bar G_j= \begin{pmatrix} F_{j+1}(B) F_j(B)^{-1} p_l(B) R\\\Eins_m\end{pmatrix}.
\end{align}
Since all polynomials in $B$ with maximal degree $m-1$ appear in the groups~$\bar G_j$, the polynomial $p_l(B)$ with $l \in \MgE{d}$ permutes
the original generators~$\bar G_j$. Since there are only $2^m-1$ non-zero polynomials with degree smaller than~$m$, we found a symmetry to reduce
the solutions by a factor of $2^m-1$. In the case discussed above, we can divide $126$ by $2^3-1$ and get $18$. We would get cyclic MUBs with the
same entanglement properties by relabeling the different qubits, which lead to $m!$ permutation operators, in the case of three qubits to six
possible permutations. But it is not clear if this symmetry is uncorrelated with the symmetry that appears from multiplying the matrix $R$ with
a polynomial in $B$. A complete explanation of the appearance of $126$ solutions may be an interesting issue but is above the scope of this work.\par
We list seven systems with a different decomposition structure for four-qubit systems in Appendix \ref{app:homosets:group}. Some of the sets
do not appear in the list given in \cite{Romero05}. This is in accordance to a statement in \cite[Section IV D]{Lawrence11}, which claims
an incompleteness of the list of $16$ MUBs in the former article; the complete classification lists 34 sets \cite{Garcia10}.

\subsection{Semigroup-based sets}\label{sec:homogeneous:semigroup}
In order to obtain complete sets of cyclic MUBs which contain only one class of operators that is fully separable into qubit systems, we need to find a
form of the stabilizer matrix $C$ which is even more general than the form of Section~\ref{sec:homogeneous:group} and is given by Equation
\eqref{eqn:homosets:group:C}. A hypothetical idea would be the following: To construct the Fibonacci sets of Section \ref{sec:fibset}, the
stabilizer matrix $C \in M_{2m}(\F_2)$ was completely determined by the reduced stabilizer matrix $B \in M_m(\F_2)$, for the group-based sets of
Section \ref{sec:homogeneous:group}, a second matrix $R \in \GL{m}{\F_2}$ became relevant. To construct semigroup-based sets we can imagine that
a third matrix will be relevant. Nevertheless, we have not been able to find an appropriate generalization of Equation \eqref{eqn:homosets:group:C}
so far to reach this goal. Conversely, from a top-down point of view, we are able to construct at least some solutions within an obtainable computation
time: If we assume that the upper left submatrix of $C$ remains invertible, we find by the properties which are given by Corollary \ref{cor:app:clifford:symplecticmatrixprop},
that for any symplectic matrix $C = \begin{pmatrix}s & t\\ u  & v\end{pmatrix}$, with $s,t,u,v \in M_m(K)$ and $s$ invertible, there holds
\begin{align}
 v = (s^t)^{-1} (1+u^t t).
\end{align}
Thus, we should be able to construct the semigroup-based sets with the three matrices $s,t,u \in M_m(\F_2)$. A similar observation is possible by
assuming that~$u$ is invertible, but does not lead to a logical generalization of the group-based sets, as we are not able to require a certain
value for the characteristic polynomial of $s$. A cyclic version of the corresponding three-qubit system in \cite{Romero05} is generated by
\begin{align}\label{eqn:homoset:semigroup:C162}
C_{(1,6,2)} = 
\begin{pmatrix}
  0 & 1 & 1 & 0 & 1 & 0\\
  0 & 1 & 1 & 0 & 0 & 1\\
  0 & 0 & 0 & 1 & 0 & 0\\
  0 & 0 & 1 & 0 & 0 & 1\\
  0 & 1 & 0 & 0 & 1 & 0\\
  1 & 0 & 0 & 0 & 0 & 0
\end{pmatrix}.
\end{align}
However, as a suitable form is not found yet, computational results based on these approaches and
Conditions (1)--(3) of Section \ref{sec:bandy} are given in Appendix~\ref{app:homosets:semigroup}.

\section{Inhomogeneous sets}\label{sec:inhomogeneous}
The third and last class of cyclic MUBs that are discussed in this work does not contain any class
of operators that is fully separable into qubit systems in its decomposition structure which was introduced
in Section \ref{sec:entangleprop}. Accordingly, these sets need to have a more general standard form
as given in Definition~\ref{defi:standfor:standfor}, thus, the generator $\bar G_0$ does not equal
$(\Eins_m, 0_m)^t$ as for the Fibonacci-based sets and the homogeneous sets in general. The approach we
choose in order to create an inhomogeneous set is the following: At first, a symplectic matrix $C_0 \in M_{2m}(\F_2)$ is
taken which produces a more general generator $\bar G_0$ from the matrix $(\Eins_m, 0_m)^t$. Secondly,
the ordinary symplectic stabilizer matrix $C$ is taken.\par
If we take the generator $\bar G^{(h)}_0(\Eins_m, 0_m)^t$ of the homogeneous
sets, we can produce any generator $\bar G_0$ as
\begin{align}
 \bar G_0 = C_0 \bar G^{(h)}_0 = \begin{pmatrix}\Eins_m & t\\ G_0^x & v\end{pmatrix} \begin{pmatrix} \Eins_m \\0_m \end{pmatrix} = \begin{pmatrix} \Eins_m \\G_0^x \end{pmatrix},
\end{align}
with $t,v \in M_m(\F_2)$, where $t$ and $v$ can be chosen according to fulfill the properties derived by
Corollary \ref{cor:app:clifford:symplecticmatrixprop} to let $C_0$ be symplectic. For the first of these
three properties, $G_0^x$ has to be symmetric, which has also be true in order to create a class of
commuting elements in $\cC_0$, corresponding to the derived Property~(2) of Section \ref{sec:bandy}.\par
The following discussion gives a first idea of the construction of inhomogeneous sets. The results
can be seen as a proof-of-principle that legitimates the approach. In order to generate a class which
is not fully decomposable, we entangle two qubit systems with indices $i, j \in \Mg{1,m}, i\neq j$ by
setting $(G_0^x)_{i,j}=(G_0^x)_{j,i} = 1$, which has the expected effect, regarding the class generation
which is introduced by Equation \eqref{eqn:bandy:classes}. Following then the ideas of Section \ref{sec:homogeneous:semigroup},
we found some of the solutions for four-qubit systems and list them in Appendix \ref{app:inhomosets}.
A cyclic generator for the three-qubit system that equals the system which is generated in \cite{Romero05} is given by
\begin{align}\label{eqn:inhomoset:C090}
C_{(0,9,0)} = 
\begin{pmatrix}
0 & 0 & 0 & 0 & 0 & 1\\
1 & 0 & 0 & 1 & 0 & 0\\
0 & 0 & 0 & 0 & 1 & 0\\
0 & 0 & 1 & 0 & 0 & 0\\
1 & 0 & 0 & 0 & 0 & 0\\
0 & 1 & 0 & 0 & 0 & 0
\end{pmatrix},
\end{align}
where the corresponding generator $G_0$ equals
\begin{align}\label{eqn:inhomoset:G090}
 (G_0^x)_{(0,9,0)} =
\begin{pmatrix}
 0 & 1 & 0\\
 1 & 0 & 0\\
 0 & 0 & 0
\end{pmatrix},
\end{align}
with $G_0=(\Eins_m, G_0^x)^t$.

\section{Unitary operator}\label{sec:Uconstruction}
So far, we have discussed different sets of cyclic MUBs in the view of a partition of the set
of Pauli operators, as introduced in Section \ref{sec:bandy}. The corresponding stabilizer matrices
were introduced in Sections \ref{sec:fibset}, \ref{sec:homogeneous}, and \ref{sec:inhomogeneous} and
serve as generators of the set of classes which partition the set of Pauli operators. But, for measuring
quantum states, we need to transform those matrices into unitary operators, in the fashion of Section
\ref{sec:bandy}. This section starts by deriving that transformation and leads to a comfortable form
of the unitary operators which result from Fibonacci-based sets (Section \ref{sec:unitop:fibset}).
Section \ref{sec:unitop:fermset} recalls
the iterative construction of the Fermat sets of Section \ref{sec:fermatset} in order to derive again
an iterative construction of these MUBs, but this time on the level of unitary operators as published
in \cite{SR11}. The section closes by discussing shortly the generation of the unitary operator in cases of a
more general stabilizer matrix (Section \ref{sec:unitop:general}).\par
For the derivation of the unitary operator we use the theory of the \emph{stabilizer formalism}\index{Stabilizer formalism}\footnote{%
The theory of the stabilizer formalism was developed by Gottesman \cite{Gott96}. See also his Ph.~D.~thesis \cite{Gott97}.}, as could
have already been figured out by the definition of the stabilizer matrix $C \in M_{2m}(\F_p)$.
Within this theory, substitutes of the usual Pauli operators are
defined, namely the \emph{logical Pauli operators}\index{Logical Pauli operator}. The usual Pauli operators
$\ZX(\vec a)$ were defined in Equation \eqref{eqn:bandy:ZXa}. With $\vec z_k \in \F_p^{2m}$ being a vector
that has a one at position $k$ and zeros else and $\vec x_k \in \F_p^{2m}$ a vector with a one at position
$m+k$ and zeros else, we can define the \emph{physical operators}\index{Physical operator}, namely a unitary representation of $\F_p^{2m}$,
as $Z_k = \ZX(\vec z_k)$ and $X_k = \ZX(\vec x_k)$; using a stabilizer matrix $A$ results in the
\emph{logical operators}\index{Logical operator} $\bar Z_k = \ZX(A \vec z_k)$ and $\bar X_k = \ZX(A \vec x_k)$ with $k \in \MgE{m}, m \in \N^*$
and $p$ prime. The physical \mbox{Pauli-$Z$} operators produce a local phase factor, whereas physical \mbox{Pauli-$X$} operators induce
bit (or dit\footnote{A dit denotes a variable defined over $\Z_d$, generalizing a bit.}) flips. Therefore, the logical ground state $\ket{0}_\mathrm{L}$ is naturally defined as the joint
eigenstate of all logical \mbox{Pauli-$Z$} operators with eigenvalue~$+1$. The other logical states can be derived by
applying bit (or dit) flips as
\begin{align}\label{eqn:unitop:jL}
 \ket{j}_\mathrm{L} = \bar X_1^{j_1} \cdots \bar X_m^{j_m} \ket{0}_\mathrm{L},
\end{align}
with $\Mg{j_1, \ldots, j_m}$ being the bit (or dit) decomposition of $j \in \F_p^m$.
At first, we focus on the construction of the unitary operator that generates a Fibonacci-based
set of MUBs.
\subsection{Fibonacci-based sets}\index{Fibonacci!sets}\label{sec:unitop:fibset}
In order to facilitate the following consideration, we use the inverse of the stabilizer matrix for the
Fibonacci-based sets (cf. Equation \eqref{eqn:fibset:B110}), which, in the case of qubits, is obviously given by
\begin{align}\label{eqn:unitop:C-1}
 C^{-1} = \begin{pmatrix}0_m& \Eins_m\\ \Eins_m & B\end{pmatrix}.
\end{align}
This leads to the logical operators
\begin{align}
 \bar Z_k = \ZX(C^{-1} \vec z_k)\quad \mathrm{and} \quad \bar X_k = \ZX(C^{-1} \vec x_k),
\end{align}
with $\bar Z_k = X_k$ for $k \in \MgE{m}$. Thus, the logical ground state is (up to a global phase) given by
\begin{align}\label{eqn:unitop:0L}
 \ket{0}_\mathrm{L} = 2^{-m/2} \sum_{j \in \F_2^m} \ket{j}.
\end{align}
For the construction of a complete set of MUBs for a dimension of the Hilbert space which is a power of two,
the set of Pauli operators is partitioned into $d+1=2^m+1$ disjoint classes
of commuting operators (cf. Equation \eqref{eqn:bandy:partition}).
The common eigenbases of those classes are mutually unbiased, as shown by Bandyopadhyay \etal{} \cite{Bandy02}. In
order to get a cyclic set, the transformation that brings a class $\cC_l$ with $l \in \MgN{d}$ of operators
to the class $\cC_{(l + 1)\;\mathrm{mod}\; (d+1)}$, has to be of multiplicative order $d+1$. By construction,
the class $\cC_0$ of the Fibonacci-based sets appears as
\begin{align}
 \cC_0 = \Mg{ \ZX \begin{pmatrix} \vec a_z\\ \vec 0\end{pmatrix} \middle\vert \vec a_z \in \F_2^m \backslash \Mg{0}}.
\end{align}
For a complete set of cyclic MUBs there exists a unitary operator $U \in M_d(\C)$, such that 
$U \cC_l U^\dagger = \cC_{(l+1) \;\mathrm{mod}\; (d+1)}$ and $U^{d+1} = \Eins_d$. With $\Eins_d$ being the common
eigenbasis of the operators of the class $\cC_0$, the columns of $U^l$ form the vectors of the eigenbases of $\cC_l$. Since $U$ maps
Pauli operators onto Pauli operators, it is called a \emph{Clifford unitary operator}\index{Clifford unitary operator}.
Properties of these operators and their relation to symplectic matrices is given in Appendix \ref{app:clifford}.
As mentioned above, we will calculate the inverse of the unitary operator, thus $U^\dagger$. By definition,
the operator $U^\dagger$ transforms a physical state into the corresponding logical state, namely
$U^\dagger \ket{j} = \ket{j}_\mathrm{L}$ which yields with Equation~\eqref{eqn:unitop:jL}
\begin{align}
 U^\dagger \ket{j} = \EZ^{-\iE \Psi} \prod_{k=1}^m \ZX(C^{-1} \vec x_k)^{j_k} \ket{0}_\mathrm{L},
\end{align}
and with Equation \eqref{eqn:unitop:0L} finally
\begin{align}\label{eqn:unitop:Udaggerstart}
 U^\dagger = 2^{-m/2} \EZ^{-\iE \Psi} \sum_{i,j\in \F_2^m}\prod_{k=1}^m \ZX(C^{-1} \vec x_k)^{j_k} \ketbra{i}{j}.
\end{align}
The global phase factor $\EZ^{-\iE \Psi}$ with $\Psi \in [0,2\pi)$ is fixed by assuming that $(U^\dagger)^{d+1} = \Eins_d$,
but the phases of the elements of the set of Pauli operators within a class $\cC_j$ are arbitrary. Thus, this global phase
factor does not follow from the stabilizer formalism.
To calculate the expression of Equation \eqref{eqn:unitop:Udaggerstart}, we plug in $C^{-1}$ in the form of Equation~\eqref{eqn:unitop:C-1}
with $B=(b_{ij})$. The general $k$-th factor looks like
\begin{align}
 \ZX(C^{-1} \vec x_k) = \left( \bigotimes_{t=1}^{k-1} X^{b_{tk}} \right) \otimes (-\iE)^{b_{kk}} Z X^{b_{kk}} \otimes \left( \bigotimes_{t=k+1}^{m} X^{b_{tk}} \right).
\end{align}
Applying the product, we find for a single tensor factor of the qubit at position~$k$ the term
\begin{align}
 (-\iE)^{b_{kk} j_k} X^{b_{1k} j_1 + \ldots + b_{k-1,k} j_{k-1}} \cdot Z^{j_k} \cdot X^{X_{kk} j_k + \ldots + b_{mk} j_m},
\end{align}
that can, by shifting the operator $Z$ to the left, be rewritten as
\begin{align}
 \iE^{b_{kk} j_k} (-1)^{b_{1k} j_1 j_k + \ldots + b_{kk} j_{k} j_k} Z^{j_k} \cdot X^{X_{1k} j_k + \ldots + b_{mk} j_m}.
\end{align}
To be able to write the tensor product in a short form, we define the abbreviations
\begin{align}\label{eqn:unitop:pj*X_j}
 p^*_j :=& \iE^{\sum_{k=1}^m b_{kk} j_k} (-1)^{\sum_{k=1}^m b_{1k} j_1 j_k + \ldots + b_{kk} j_k j_k},\\
 \tilde X_j :=& \bigotimes_{k=1}^{m} X^{b_{1k} j_1 + \ldots + b_{mk} j_m},
\end{align}
which lead to
\begin{align}
 U^\dagger = 2^{-m/2} \EZ^{-\iE \Psi} \sum_{i,j \in \F_2^m} p^*_j \cdot \left( \bigotimes_{k=1}^m Z^{j_k} \right) \cdot \tilde X_j \ketbra{i}{j}.
\end{align}
Only the factor $\tilde X_j$ acts on $\sum_i \ket{i}$, but it has no argument in $i$ and keeps the sum invariant.
Therefore--summarized over $i$--it equals the unity operator and vanishes. The tensor product $\bigotimes_{k=1}^m Z^{j_k}$ can be
identified with $(-1)^{i\cdot j}$ that equals Sylvester's Hadamard matrix (cf. Equation \ref{eqn:app:hadamard:sylvester}). Thus we end up with the short form of $U^\dagger$
as
\begin{align}
  U^{\dagger} = \EZ^{-\iE \Psi} H^{\otimes m} \cdot P^*,
\end{align}
the $m$-folded tensor product of the \emph{normalized Hadamard matrix}\index{Hadamard matrix} (cf. Equation \eqref{eqn:app:hadamard:H}) which is defined as
\begin{align}
 H=\frac{1}{\sqrt{2}}\begin{pmatrix}1&1\\1&-1\end{pmatrix},
\end{align}
and the diagonal \emph{phase system}\index{Phase system} matrix which is given by $P^*=\mathrm{diag}( (p^*_j)_{j \in \F_2^m} )$.
With~$B$ being symmetric and the quadratic form
\begin{align}
 \bra{j} B \ket{j} = \sum_{k,l=1}^m b_{kl} j_k j_l,
\end{align}
where $\bra{j} B \ket{j} \in \Z_4$ holds, as it is only taken as the power of a root of order four, we finally find
\begin{align}\label{eqn:unitop:pj*}
 p^*_j = \iE^{\bra{j} B \ket{j}} (-1)^{\sum_{k=1}^m b_{kk} j_k}.
\end{align}
For the sake of completeness, the original unitary operator $U$ is accordingly given by
\begin{align}\label{eqn:unitop:U}
  U = \EZ^{\iE \Psi} P \cdot H^{\otimes m},
\end{align}
with $P=\mathrm{diag}( (p_j)_{j \in \F_2^m} )$ and $p_j = (-\iE)^{\bra{j} B \ket{j}} (-1)^{\sum_{k=1}^m b_{kk} j_k}$.\footnote{
The additional minus sign comparing to the results of \cite{KRS10, SR12} is due to the fact that we to construct $U$ here
and not $U^{\dagger}$, which does not matter in principle.} The
consideration of the unitary operators $U$ for the derived stabilizer matrices that follow from the matrices $A$ in Appendix
\ref{app:MUBsolutions} leads to an unproven conjecture.

\begin{con}[Spectrum of generators of MUBs]\label{con:unitor:spectrum}\hfill\\ 
 The spectrum of a unitary operator $U$ that generates a complete set of MUBs is non-degenerate; it consists of all roots of
 unity of order $d+1$ with a single exception.
\end{con}

It is clear, that the eigenvalues of $U$ have to be roots of unity of order $d+1$ in order to fulfill $U^{d+1}=\Eins_d$. But
$U$ is a $d \times d$ matrix, thus if it is non-degenerate, it has $d$ different eigenvalues. Multiplying a valid $U$ by
any root of unity of order $d+1$ will result again in a valid $U$. We are free to exclude the trivial root $1$ from the spectrum which results
in a spectrum that is symmetric to the real axis. In this case, the trace of $U$ is given by $\tr U = -1$. This way, we can
calculate the global phase $\EZ^{\iE \Psi}$ of $U$ by inserting Equation \eqref{eqn:unitop:U} as
\begin{align}
 \EZ^{\iE \Psi} = - \tr ( P^* H^{\otimes m}).
\end{align}
The consideration of derived unitary operators that generate a complete set of cyclic MUBs, where the eigenvalue $-1$ is
excluded from the spectrum, results again in a conjecture.

\begin{con}[Global phase of generators of MUBs]\label{con:unitor:globalphase}\hfill\\ 
 The global phase of the generator $U$ in the form of Equation \eqref{eqn:unitop:U} of a complete set of cyclic
 MUBs, where the eigenvalue $+1$ is chosen to be excluded from the spectrum of $U$, provided Conjecture \ref{con:unitor:spectrum} is true,
 reads
\begin{align}
 \EZ^{\iE \Psi} =
\begin{cases}
 \frac{-1 - \iE}{\sqrt{2}},&\mathrm{for}\;m\;\mathrm{odd},\\
 -\iE,&\mathrm{for}\;m\;\mathrm{even}.
\end{cases}
\end{align}
\end{con}
It is surprising that the global phase $\EZ^{\iE \Psi}$ seems to depend not on the stabilizer matrix $C$ at all, but only on the number
of qubits, thus, on the dimension $d=2^m$. Since by construction of the phase factors $p_j$ (cf. Equation \eqref{eqn:unitop:pj*}),
the product $P \cdot H^{\otimes m}$ leads to a matrix with entries that are roots of unity of order four (and a global real normalization
factor), for even $m$ the unitary matrix $U$ has only roots of unity of order four as entries, whereas for odd $m$ it has only entries
that are principal roots of order eight (so not of order four). The complex conjugation of the phase factor in relation to the results
of \cite{KRS10} comes from the fact that here are operators~$\ZX$ used, in~\cite{KRS10} operators $X\!Z$, but $C$ looks the same.
As an advantage, the approach used here results in a symmetric arrangement of the classes $\cC_j$, for example $\cC_0$ consists only 
of Pauli-$Z$ operators and $\cC_d$ of Pauli-$X$ operators.

\subsection{Fermat sets}\index{Fermat sets}\label{sec:unitop:fermset}

In Section \ref{sec:fermatset} it was shown that for dimensions $d=2^{2^k}$ with $k \in \N$, a complete set of cyclic MUBs
can be build recursively, thus, using the complete set of cyclic MUBs from the next smaller dimension. This procedure in terms of
the stabilizer matrix $C$ was shown in Equation \eqref{eqn:ferset:C2k}. Considering the unitary matrix $U$ which results from
$C$ for Fibonacci-based sets as given by Equation~\eqref{eqn:unitop:U}, we find an analogous recursive construction, but this
time in terms of~$U$. And, comparably to the definition of $C_{2^k}$ of Equation \eqref{eqn:ferset:C2k}, we will refer to the
unitary operator $U$ of dimension $2^m$ with $m=2^k$ in the following as $U_{2^k}$.\par
To simplify matters, we define the matrix $V$ by
\begin{align}
 V = P \cdot \bar H^{\otimes m},
\end{align}
which has, compared to the unitary operator $U$ of Equation \eqref{eqn:unitop:U}, a factor of one for the global phase and with
the \emph{Hadamard matrix}\index{Hadamard matrix} as given in Equation~\ref{eqn:app:hadamard:barH}, namely
\begin{align}
 \bar H = \begin{pmatrix}
      1 & 1\\ 1 & -1
     \end{pmatrix},
\end{align}
no normalization factor. Thus, $V$ is not a unitary matrix, but $U$ can be derived from $V$ simply as $U = V / (-\tr V)$. To calculate
a matrix $V_m$ with a $B_m$ as discussed in Section \ref{sec:fermatset}, we have
\begin{align}\label{eqn:unitop:Vm}
 V_{m} = \mathrm{diag}(((-\iE)^{\bra{i} B_m \ket{i}} (-1)^{\sum_{k=1}^m b_{kk} i_k})_{i \in \F_2^m}) \cdot \bar H^{\otimes m},
\end{align}
using Equation \eqref{eqn:unitop:pj*}. By Equation \eqref{eqn:ferset:recursion}, the only value on the diagonal of $B_{2^k}$ that does not vanish,
is $b_{11}$. We can therefore abbreviate Equation \eqref{eqn:unitop:Vm} as
\begin{align}
 V_{m} = \mathrm{diag}(((-\iE)^{\bra{i} B_{m} \ket{i}} (-1)^{i_1})_{i \in \F_2^m}) \cdot \bar H^{\otimes m}.
\end{align}
The same construction holds for $V_{2m}$; we may consider the new vector\linebreak $i' = (i'_1, \ldots i'_{2m})^t \in \F_2^{2m}$ as
$i' = (i,j)^t = (i_1, \ldots, i_m, j_1, \ldots, j_m)^t \in \F_2^{2m}$, which leads to
\begin{align}\label{eqn:unitop:V2m}
 V_{2m} = \mathrm{diag}(((-\iE)^{\bra{i,j} B_{2m} \ket{i,j}} (-1)^{i_1})_{i,j \in \F_2^m}) \cdot \bar H^{\otimes 2 m}.
\end{align}
The matrix element $\bra{i,j} B_{2m} \ket{i,j}$ with $B_{2m} = \begin{pmatrix} B_m & \Eins_m\\ \Eins_m & 0_m\end{pmatrix}$
equals $\bra{i} B_m \ket{i} + 2 ( i \cdot j )$,\footnote{It has to be taken care, that $2 ( i \cdot j ) \in \Z_4$
holds and does therefore not vanish, as it is part of the exponent of $-\iE$ in Equation \eqref{eqn:unitop:V2m}.} thus, we find for Equation \eqref{eqn:unitop:V2m}:
\begin{align}\label{eqn:unitop:V2mH}
 V_{2m} = \mathrm{diag}(((-\iE)^{\bra{i} B_m \ket{i}} (-1)^{i_1} (-1)^{i\cdot j})_{i,j \in \F_2^m}) \cdot \bar H^{\otimes 2 m}.
\end{align}
The factor $(-1)^{i\cdot j}$ defines the $2^m \times 2^m$ dimensional Hadamard matrix $\bar H^{\otimes m}$, the same, that is used to
describe $V_m$. By definition, $i$ represents the lower bit half and $j$ the higher bit half of $i'$ as $i'=(i,j)^t$. We identify $i$ as the row index
and $j$ as the column index of $H^{\otimes m}$, the factors, namely $(-\iE)^{\bra{i} B_m \ket{i}} (-1)^{i_1}$ are invariant of $j$,
thus repeat for each $j$. The phase vector of $V_{2m}$ arises, by calculating $V_m$ and concatenating its column vectors such that
\begin{align}
 \begin{pmatrix}p_1\\p_2\\\hdots\\ p_{2^{2m}}\end{pmatrix} = \begin{pmatrix}(\vec v_1)_m\\ (\vec v_2)_m\\ \hdots \\ ( \vec v_{2^m})_m\end{pmatrix},
\end{align}
where $(\vec v_k)_m$ is the $k$-th column vector of $V_m$. We will refer to the mapping of $V_m$ to the phase vector $P_{2m}$ of $V_{2m}$ as the
\emph{chop map}\index{Chop map} $\cM$, i.\,e. $\cM(V_m) = P_{2m}$. This defines the recursion relation
\begin{align}\label{eqn:unitop:recursion}
 V_{2m} = \mathrm{diag}( \cM( V_m ) )   \cdot \bar H^{\otimes 2 m}.
\end{align}
For $k=0$, thus $m=2^k=1$, we have with $B_1=(1)$:
\begin{align}
 V_1 =&\;\mathrm{diag} ( (-\iE)^{0\cdot 1\cdot 0} (-1)^0, (-\iE)^{1\cdot 1\cdot 1} (-1)^1)  \cdot \begin{pmatrix}1&1\\1&-1\end{pmatrix}\\
     =& \begin{pmatrix}1&1\\\iE&-\iE\end{pmatrix}\label{eqn:unitop:V1}.
\end{align}
This implies with $\cM(V_1) = ( 1, \iE, 1, -\iE)^t$:
\begin{align}
 V_2 = \begin{pmatrix} 1 & 1& 1& 1\\ \iE & -\iE&\iE&-\iE\\1&1&-1&-1\\-\iE&\iE&\iE&-\iE\end{pmatrix},
\end{align}
and can be continued at least to $V_{2048}$ as checked in Appendix \ref{app:fermat_based:wiedemann}, and if Wiedemann's conjecture is true,
for all $V_{2^k}$ with $k \in \N$.\par
To calculate the unitary operator $U_m$ as $U_m = V_m / (-\tr V_m)$, we need to analyze the matrix $V_m$
in detail. If $i\in \F_2$ represents the rows and $j \in \F_2$ the columns of $V_1$, we find with Equation \eqref{eqn:unitop:V1},
\begin{align}
 (V_1)_{i,j} = 
\begin{cases}
 1 \cdot (-1)^{i\cdot j},&\mathrm{for}\; i\equiv 0\quad \mathrm{mod}\quad 2,\\
 \,\iE \cdot (-1)^{i\cdot j},&\mathrm{for}\; i\equiv 1\quad \mathrm{mod}\quad 2.
\end{cases}
\end{align}
Taking Equation \eqref{eqn:unitop:recursion} into account, $V_2$ reads as
\begin{align}
 (V_2)_{i,j} =
\begin{cases}
 1 \cdot (-1)^{i_1\cdot i_2 + i\cdot j},&\mathrm{for}\; i\equiv 0\quad \mathrm{mod}\quad 2,\\
 \,\iE \cdot (-1)^{i_1\cdot i_2 + i\cdot j},&\mathrm{for}\; i\equiv 1\quad \mathrm{mod}\quad 2,
\end{cases}
\end{align}
with $i = \Mg{i_1, \ldots, i_m}$ in general. As seen above, the lengths of the vectors $i$ and
$j$ are one for $V_1$. For $V_2$ the vectors are two bit long, for $V_4$ four bit and so forth.
Using Equation \eqref{eqn:unitop:V2mH} in order to recognize Hadamard matrices in the construction of
$V_m$ recursively, we find
\begin{align}\label{eqn:unitop:Vmij}
 (V_m)_{i,j} =
\begin{cases}
 1 \cdot (-1)^{x},&\mathrm{for}\; i\equiv 0\quad \mathrm{mod}\quad 2,\\
 \,\iE \cdot (-1)^{x},&\mathrm{for}\; i\equiv 1\quad \mathrm{mod}\quad 2,
\end{cases}
\end{align}
with
\begin{align}
 x= i_1\cdot i_2 + (i_1,i_2)\cdot (i_3,i_4)^t + \ldots + (i_1,\ldots,i_{m/2})\cdot (i_{m/2+1},\ldots,i_{m})^t + i\cdot j.
\end{align}
The trace of $V_m$ is given by the sum of the diagonal values, where $j=i$ holds. We will calculate the
real and the imaginary part separately.\par
For the real part of the trace of $V_m$, the first bit value of $i$ has to be zero, thus we find
\begin{align}
 \mathcal{R}\Mg{\tr V_m} = \sum_{i\in \F_2^m, i_1=0} (V_m)_{i,i} = \sum_{i\in \F_2^m, i_1=0} (-1)^x,
\end{align}
with $x= 0\cdot i_2 + (0,i_2)\cdot (i_3,i_4)^t + \ldots + (0,i_2,\ldots,i_{m/2})\cdot (i_{m/2+1},\ldots,i_{m})^t + i\cdot i$.
We arrange the sum in pairs, where for one element of the pair $i_{m/2+1}=0$ holds and for the other element $i_{m/2+1}=1$. Since
$i_1=0$ holds for all terms in the sum, $x$ differs only in the last term for the two elements of a pair, where $i \cdot i$ gives the
Hamming weight of $i$, that differs by one for the different cases. So for one element of each pair, $x$ is an even number, for the
other element it is odd and as exponents of $-1$ can be taken modulo two, the expression $(-1)^0+(-1)^1=0$ is the sum of both
elements and equals zero. Since all elements can be paired this way, the real part of the trace of $V_m$ vanishes.\par
For the imaginary part of the trace of $V_m$ the value of $i_1$ equals one, we find
\begin{align}
 \mathcal{I}\Mg{\tr V_m} = \sum_{i\in \F_2^m, i_1=1} (V_m)_{i,i} = \sum_{i\in \F_2^m, i_1=1} (-1)^x,
\end{align}
with $x= 1\cdot i_2 + (1,i_2)\cdot (i_3,i_4)^t + \ldots + (1,i_2,\ldots,i_{m/2})\cdot (i_{m/2+1},\ldots,i_{m})^t + i\cdot i$.
We will consider two cases: For those elements $i$, where at least one of the bit values of $\Mg{i_2,\ldots, i_{m/2}}$ equals
zero, we can arrange the elements analogously in pairs as done for the real part. The second case treats all remaining
elements, namely those where all bit values of $\Mg{i_2,\ldots, i_{m/2}}$ equal one, which produces terms with
\begin{align}\label{eqn:unitop:x1111}
 x= 1\cdot 1 + (1,1)\cdot (1,1)^t + \ldots + (1,\ldots,1)\cdot (i_{m/2+1},\ldots,i_{m})^t + i\cdot i.
\end{align}
For $m=1$, we have only one term and this adds $-1$ to the imaginary part. For $m=2$, the imaginary part equals
$\sum_{i_2=0}^1 (-1)^{1\cdot i_2 + (1,i_2)\cdot (1,i_2)^t}$ which leads to $-2$. For all $m>2$ we can use the following
observation: For each of the $2^{m/2}$ elements with $i=(1,\ldots,1,i_{m/2+1},\ldots, i_m)$ the two rightmost terms of
Equation \eqref{eqn:unitop:x1111} give the same result modulo two, since the number of ones in the left part is even.
As $m$ is a power of two, all remaining terms without the first one result in powers of two and the first one equals
one. Therefore, the sum of all those elements provides a contribution of $-2^{m/2}$ to the imaginary part of the trace of $V_m$. Thus, we have
\begin{align}\label{eqn:unitop:trVm}
 \tr V_m = - \iE 2^{m/2}.
\end{align}
Taking all results together, we finally find for the unitary operator $U_m$ which results for the Fermat-based sets,
\begin{align}
 U_{2m} = -\iE 2^{-m} \mathrm{diag}( \cM( V_m ) )  \cdot \bar H^{\otimes 2 m}.
\end{align}
The difference of the operators to those given in \cite{SR11} is due to the replacement of $X\!Z$ operators by $\ZX$ operators
which was motivated in Section \ref{sec:unitop:fibset}. Using the results of this section may lead to a proof of Wiedemann's
conjecture by the analogy that was proven in Theorem \ref{thm:ferset:wiedemann}. An idea is given in Appendix~\ref{app:wiedemann_proofs}.
 
\subsection{More general sets}\label{sec:unitop:general}
The construction of a unitary operator that generates a complete set of cyclic MUBs which is based on a more general stabilizer
matrix, like those discussed in Sections \ref{sec:homogeneous} and \ref{sec:inhomogeneous} is possible in a similar way as was done
for the Fibonacci-based sets in Section \ref{sec:unitop:fibset} by using the same steps. But, presumably, the resulting form will
not be decomposable as simple as in Equation \eqref{eqn:unitop:U}. The search for a nice construction form is beyond the scope of this work;
nevertheless, a physical implementation of general stabilizer matrices will be discussed in Section \ref{sec:circuit}
that is more applicable than the representation as a unitary matrix with a simple construction.

\section{Gate decomposition}\label{sec:circuit}
We have seen so far, that the construction of complete sets of cyclic MUBs for a Hilbert space of
dimension $d=2^m$ with $m \in \N^*$ is based on the construction of a stabilizer matrix (cf. Sections \ref{sec:fibset}, \ref{sec:fermatset}, \ref{sec:homogeneous}, and \ref{sec:inhomogeneous}),
which can be represented as a unitary matrix (cf. Section \ref{sec:Uconstruction}). What still lacks
is the implementation of the individual bases into an experimental setup (e.\,g. a quantum computer).
A common general description of an implementation is the idea of a \emph{quantum circuit}\index{Quantum circuit}.
This is a scheme, similar to a classical circuit, that gives a description of the behavior of a unitary
operation that can be seen as a recipe on how to implement the operation in \emph{quantum gates}.
The challenge in deriving this circuit from the operation is to get a minimal number of those gates
where each gate acts on a minimal number of qubits.\footnote{An advantage of cyclic MUBs is that
a single circuit can be used to implement a complete set of MUBs, where for a non-cyclic set, $d+1$
different circuits have to be implemented, as was done e.\,g. by Klimov \etal{} \cite{Klimov08}.}
It was shown that single-qubit gates and two-qubit gates
produce a \emph{universal set}\index{Universal set}, which can be used to implement any unitary operation \cite{DiVincenzo95,DBE95,Lloyd95,Boykin00}.
It is important to mention that the application of gates in quantum circuits have to be read from left
to right, rather than the application of quantum operations on quantum states.
Within this section, we provide a simple generation of a quantum circuit for the Fibonacci-based sets of
Section \ref{sec:fibset} and show the minimal form of this circuit in the case of the Fermat sets of Section \ref{sec:fermatset}.
Finally, an approach for constructing a quantum circuit from a more general stabilizer matrix is given.
A short discussion on performance and error optimizations closes this section.\par

\subsection{Fibonacci-based sets}\label{sec:gatedecomp:fibset}
In the case of Fibonacci-based sets, it turns out that the construction of an appropriate quantum circuit
that implements the generator of a complete set of cyclic MUBs can be read out from the reduced
stabilizer matrix $B \in M_m(\F_2)$ as it was introduced in Equation \eqref{eqn:fibset:B110}. But, honestly
speaking, this fact results from the direct construction of the unitary operator, as given by Equation~\eqref{eqn:unitop:U}.
Since the global phase cannot be measured and is therefore irrelevant in the physical implementation, we have
to generate a quantum circuit for the operator
\begin{align}
 U' = P \cdot H^{\otimes m}.
\end{align}
Obviously, the application of the Hadamard matrix $H^{\otimes m}$ can be done by applying $H$ to each of the
$m$ qubits individually. The phase system is given by $P = \mathrm{diag}((p_j)_{j \in \F_2^m})$, where the
binary representation of $j$ refers to a specific quantum state. The bit $j_k$ addresses then with $k \in \MgE{m}$
the qubit at position $k$. This phase system depicts directly the implementation into a quantum circuit. More precisely,
we can rewrite the phase factor as it was given in Equation \eqref{eqn:unitop:pj*X_j}, where we find
\begin{align}
 p_j =& (-\iE)^{\sum_{k,l=1}^m b_{kl} j_k j_l} (-1)^{\sum_{k=1}^m b_{kk} j_k}\nonumber\\
     =& \prod_{k>l=1}^m \left((-1)^{j_k j_l}\right)^{b_{kl}} \cdot \prod_{k=1}^m \left(\iE^{ j_k}\right)^{b_{kk}}\label{eqn:gatedecomp:fibset:pj}.
\end{align}
The second term of Equation \eqref{eqn:gatedecomp:fibset:pj} tells us, that a phase of $\iE$ is applied to the qubit with
index $k$, if it is in state $\ket{1}_k\bra{1}$ and $b_{kk}=1$; in the case $\ket{0}_k\bra{0}$ nothing happens. The first
term is only relevant in the case that the qubits with index $k$ and $l$ are in the state $\ket{1}_k\bra{1} \otimes \ket{1}_l\bra{1}$
(which means that $j_k = j_l=1$) and if $b_{kl}=1$. The described actions can be performed by the following elementary
gates:
\begin{itemize}[--]
 \item The single-qubit \texttt{Phase} gate acting on qubit $k$ which is defined as
  \begin{align}
   \texttt{Phase}_k(\EZ^{\iE\phi}) &= \ket{0}_k\bra{0} + \EZ^{\iE \phi} \ket{1}_k\bra{1},
  \end{align}
  and visualized as
  \begin{center}
  \begin{tikzpicture}[scale=1, transform shape]
    \tikzstyle{operator} = [draw,fill=white,minimum size=1.5em] 
    \tikzstyle{phase} = [fill,shape=circle,minimum size=5pt,inner sep=0pt]
    \node at (0,0) (q1) {$\ket{\Psi}$};
    \node[operator] (op1) at (2,0) {$\EZ^{\iE\phi}$} edge [-] (q1);
    \node at (4,0) (q1') {$\ket{\Psi'}$} edge [-] (op1);
  \end{tikzpicture}
  \end{center} 
  with phase $\phi \in [0,2\pi)$.
 \item The two-qubit controlled-\texttt{Phase} gate
  \begin{align}
\texttt{CPhase}_{s \rightarrow t}(\EZ^{\iE\phi}) &= \ket{0}_s\bra{0} \otimes \ket{0}_t\bra{0}
    + \ket{0}_s\bra{0} \otimes \ket{1}_t\bra{1} \nonumber\\
    &\hspace{-1cm}+ \ket{1}_s\bra{1} \otimes \ket{0}_t\bra{0}
    + \EZ^{\iE\phi} \ket{1}_s\bra{1} \otimes \ket{1}_t\bra{1},
  \end{align}
  with control qubit $s$, target qubit $t$ and phase $\phi \in [0,2\pi)$. For a two-qubit state
  $\ket{\Psi_s, \Psi_t}$, the \texttt{CPhase} gate induces in general a transformation
  $\ket{\Psi_s, \Psi_t} \rightarrow \ket{\Psi_s, \Psi'_t}$ and is visualized by the following
  circuit:
  \begin{center}
  \begin{tikzpicture}[scale=1, transform shape]
    \tikzstyle{operator} = [draw,fill=white,minimum size=1.5em] 
    \tikzstyle{phase} = [fill,shape=circle,minimum size=5pt,inner sep=0pt]
    \node at (0,0) (q1) {$\ket{\Psi_s}$};
    \node at (0,-1) (q2) {$\ket{\Psi_t}$};
    \node[phase] (p1) at (2,0) {} edge [-] (q1);
    \node[operator] (op1) at (2,-1) {$\EZ^{\iE\phi}$} edge [-] (q2);
    \node at (4,0) (q1') {$\ket{\Psi_s}$} edge [-] (p1);
    \node at (4,-1) (q2') {$\ket{\Psi'_t}$} edge [-] (op1);
    \draw[-] (p1) -- (op1);
  \end{tikzpicture}
  \end{center} 
\end{itemize}
We can use these two gates to implement the phase factor of Equation \eqref{eqn:gatedecomp:fibset:pj},
namely by identifying the gates from the reduced stabilizer matrix $B$ as:
\begin{itemize}[--]
 \item For $b_{kk} = 1$, apply $\texttt{Phase}_k(\iE)$.
 \item For $b_{kl} = 1$ and $k>l$, apply $\texttt{CPhase}_{k \rightarrow l}(-1)$.
\end{itemize}
As an example, a possible reduced stabilizer matrix that can be used to generate a complete system of cyclic MUBs for a 
quantum system of four qubits, is given by the recursive construction of Equation \eqref{eqn:ferset:recursion} with
\begin{align}
 B_{2^2} = \begin{pmatrix}1&1&1&0\\1&0&0&1\\1&0&0&0\\0&1&0&0\end{pmatrix}.
\end{align}


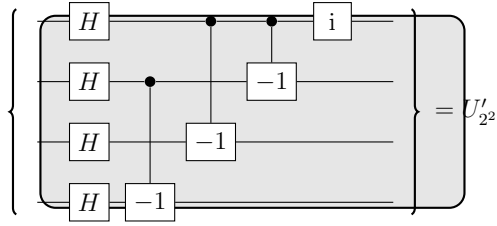
\begin{figure}[t]
  \centerline{
    \begin{tikzpicture}[scale=0.8, transform shape] 
    %
    \tikzstyle{operator} = [draw,fill=white,minimum size=1.5em] 
    \tikzstyle{phase} = [fill,shape=circle,minimum size=5pt,inner sep=0pt]
    \tikzstyle{surround} = [fill=black!10,thick,draw=black,rounded corners=2mm]
    %
    \node at (1,0) (q1) {};
    \node at (1,-1) (q2) {};
    \node at (1,-2) (q3) {};
    \node at (1,-3) (q4) {};
    %
    \node[operator] (op11) at (2,0) {$H$} edge [-] (q1);
    \node[operator] (op21) at (2,-1) {$H$} edge [-] (q2);
    \node[operator] (op31) at (2,-2) {$H$} edge [-] (q3);
    \node[operator] (op41) at (2,-3) {$H$} edge [-] (q4);
    \node[phase] (op22) at (3,-1) {} edge [-] (op21);
    \node[operator] (op42) at (3,-3) {$-1$} edge [-] (op41);
    \draw[-] (op22) -- (op42);
    %
    \node[phase] (op12) at (4,0) {} edge [-] (op11);
    \node[operator] (op32) at (4,-2) {$-1$} edge [-] (op31);
    \draw[-] (op12) -- (op32);
    \node[phase] (op13) at (5,0) {} edge [-] (op12);
    \node[operator] (op23) at (5,-1) {$-1$} edge [-] (op22);
    \draw[-] (op13) -- (op23);
    \node[operator] (op14) at (6,0) {$\iE$} edge [-] (op13);
    \draw[-] (op14) -- (7,0);
    \draw[-] (op23) -- (7,-1);
    \draw[-] (op32) -- (7,-2);
    \draw[-] (op42) -- (7,-3);
    %
    \draw[decorate,decoration={brace},thick] (7.3,0.2) to
	node[midway,right] (bracketR) {$\;\;= U'_{2^2}$}
	(7.3,-3.2);
    \draw[decorate,decoration={brace},thick, rotate=180] (-0.8,3.2) to
	node[midway,left,rotate=180] (bracketL) {}
	(-0.8,-.2);
    %
    \begin{pgfonlayer}{background} 
    \node[surround] (background) [fit = (op11) (op42) (bracketL) (bracketR)] {};
    \end{pgfonlayer}
    \end{tikzpicture}
  }

  \caption{Quantum circuit for the generator $U'_{2^2}$ of a complete set of cyclic MUBs for a four qubit system.}
  \label{fig:gatedecomp:fibset:circuit}
\end{figure} 

Following the results given above, this solution results in one \texttt{Phase} gate, three \texttt{CPhase} gates and
four \texttt{Hadamard} gates, as depicted in
Figure \ref{fig:gatedecomp:fibset:circuit}. Accordingly, in the case of the recursive construction of Fermat sets
given by Equation~\eqref{eqn:ferset:recursion}, the implementation of the generator of a complete set of cyclic MUBs
for $m$ qubits needs one \texttt{Phase} gate, $m-1$ \texttt{CPhase} gates and $m$ \texttt{Hadamard} gates, which
is optimal in the number of two-qubit gates.\footnote{As the bases of a complete set of cyclic MUBs for an $m$ qubit system measure the complete
state of this system, a possible quantum circuit should link all qubits together, which can be done by controlled gates.
If not all qubits were linked together, then only information on subsystems could be measured. Consequently, if for controlled gates
only two qubit gates are used, the minimum number of these gates is $m-1$.} 
In the case of the Fibonacci-based sets in the Form of Equation \eqref{eqn:fibset:A}, an implementation would need
not more than $\lceil m/2 \rceil + 5$ \texttt{Phase} gates, $(\lceil m/2 \rceil^2 - \lceil m/2 \rceil)/2 + 10$ \texttt{CPhase}
gates and $m$ \texttt{Hadamard} gates for $m \in \MgE{600}$ as can be seen in Appendix~\ref{app:fibonacci_based:triangle}.
This does not seem to be optimal at all. Also the proposed method to find a symmetric companion matrix (cf. Conjecture \ref{con:fibset:exsymmcomp})
does not seem to be capable to solve this problem optimally. Therefore, further investigations on how to reduce the number
of non-vanishing entries of the reduced stabilizer matrices are essential in order to reduce the number of gates needed.\par

\subsection{Homogeneous sets}\label{sec:gatedecomp:homogeneous}
To implement homogeneous sets into quantum circuits we cannot use similar methods as for the Fibonacci-based sets as long
as no construction of the unitary matrix is given in an analogous form to that of Equations \eqref{eqn:unitop:U} and~\eqref{eqn:gatedecomp:fibset:pj}.
Therefore, we will introduce a method that is capable to generate the quantum circuit
of nearly any matrix $C \in M_{2m}(\F_2)$ that is symplectic. At first, symplectic representations of different quantum gates will
be discussed; thereafter, it will be shown how a large set of symplectic matrices with entries in $\F_2$ can be decomposed into a
product of symplectic matrices in order to become directly implementable by corresponding quantum gates. Quantum circuits
derived by this method for the homogeneous sets which are discussed, are shown in Appendices \ref{app:homosets:group} and \ref{app:homosets:semigroup}.\par
Let us define an arbitrary symplectic matrix $C \in M_{2m}(\F_2)$ in block-matrix form as
\begin{align}\label{eqn:gatedecomp:homo:C}
  C = \begin{pmatrix} s & t \\u & v \end{pmatrix},
\end{align}
with submatrices $s, t, u, v \in M_m(\F_2)$. As stated in Appendix \ref{app:clifford}, elements of the Clifford group
are automorphisms of the set of Pauli operators. By their group property, these elements are invertible, so also their
representation as symplectic matrices. Therefore, any symplectic matrix has full rank, thus, the row vectors within
the submatrix $(u,v)^t$ are linearly independent. For the group-based sets defined by Equation~\eqref{eqn:homosets:group:C}
this fact can be seen also in the stabilizer matrix directly, since $v:=0_m$ and $u:=R^{-1}$ is invertible. The same holds
for the submatrix $(s,u)$, but is not that obvious. According to this
observation, it is in general possible to obtain by elementary row and column operations a transformed matrix where the submatrix
$u$ is invertible. Within this section, we limit the consideration only to those cases where $u$ is invertible.
To apply these necessary operations and to implement the resulting matrix, different
quantum gates have to be considered.\par
By the construction of the classes $\cC_j$ with $j \in \MgN{d}$ as given in Equation~\eqref{eqn:fibset:classesC}, each
class consists of Pauli operators. The matrix $C$ then realizes by its symplecticity property that the set of a certain
class of Pauli operators is transformed to another set of operators, i.\,e. the next class. Therefore, a circuit of
quantum gates which represents the matrix $C$ has to implement transformations of a set of Pauli operators to another
set. As we can see the whole circuit as the product of basic transformations (i.\,e. the quantum gates), the implication
of each quantum gate can be analyzed by the consideration of its truth table of Pauli operators. To explain the idea, we
analyze the \texttt{Hadamard} gate as an example:\par
The Hadamard operator is defined by Equation \eqref{eqn:app:hadamard:H} as
\begin{align}
 H = \frac{1}{\sqrt{2}} \begin{pmatrix} 1 & 1\\ 1 & -1\end{pmatrix};\quad \mathrm{with} \quad H^{\dagger} = H,
\end{align}
it acts on the Pauli operators $X,Y$ and $Z$, which are listed in Appendix \ref{app:pauli}, as $H X H^{\dagger} = Z$,
$H Y H^{\dagger} = -Y$ and $H Z H^{\dagger} = X$. To construct a complete set
of cyclic MUBs, the eigenbases of the classes $\cC_j$ are relevant, but the phases of the elements are not relevant. A quantum
gate, represented as a symplectic matrix $C_\texttt{H}$, that implements the Hadamard operator, should then act on the
Pauli operators for a single qubit, represented as elements of the finite field $\F_2^2$, as 
\begin{align}
 C_\texttt{H} \begin{pmatrix} 1 \\ 0 \end{pmatrix} = \begin{pmatrix} 0 \\ 1 \end{pmatrix}, 
 \quad C_\texttt{H} \begin{pmatrix} 0 \\ 1 \end{pmatrix} = \begin{pmatrix} 1 \\ 0 \end{pmatrix},\quad\mathrm{and}\quad
 C_\texttt{H} \begin{pmatrix} 1 \\ 1 \end{pmatrix} = \begin{pmatrix} 1 \\ 1 \end{pmatrix}.
\end{align}
Obviously, any symplectic matrix will set the zero vector to the zero vector, which is in accordance to the unitary representation
given above. The representation of the Hadamard operator as a symplectic matrix is finally given by
\begin{align}
 C_\texttt{H} = \begin{pmatrix} 0 & 1\\ 1 & 0\end{pmatrix},
\end{align}
with the corresponding visual representation as the \texttt{Hadamard} gate.
  \begin{center}
  \begin{tikzpicture}[scale=1, transform shape]
    \tikzstyle{operator} = [draw,fill=white,minimum size=1.5em] 
    \tikzstyle{phase} = [fill,shape=circle,minimum size=5pt,inner sep=0pt]
    \node at (1,0) (q1) {};
    \node[operator] (op1) at (2,0) {$\texttt{H}$} edge [-] (q1);
    \node at (3,0) (q1') {} edge [-] (op1);
  \end{tikzpicture}
  \end{center} 
This gate acts only on a single qubit, so the unitary operator on a multi-qubit system with $m$ qubits that applies a
Hadamard gate on the $i$-th qubit, is given by
\begin{align}\label{eqn:gatedecomp:homo:UHi}
 U^{(i)}_{\texttt{H}} = \left( \bigotimes_{l=1}^{i-1} \Eins_2 \right) \otimes H \otimes \left( \bigotimes_{l=i+1}^{m} \Eins_2 \right),
\end{align}
its representation as a symplectic matrix $C^{(i)}_{\texttt{H}} = \left(c_\texttt{H}^{(i)}\right)_{kl}$ has entries
\begin{align}\label{eqn:gatedecomp:homo:cHikl}
\left(c_\texttt{H}^{(i)}\right)_{kl} = 
 \begin{cases}
  \delta_{kl} & \quad\mathrm{for}\quad k,l \notin \Mg{i,m+i},\\
  1-\delta_{kl} & \quad\mathrm{for}\quad k, l \in \Mg{i,m+i}.
 \end{cases}
\end{align}
Another single-qubit gate is the \texttt{Phase} gate, that was already mentioned in Section \ref{sec:gatedecomp:fibset}, but we will
again concentrate on the \texttt{Phase-\iE} gate, with a phase of $\pm \iE$.\footnote{It turns out, that a phase of $+\iE$ leads to
the same result as a phase of $-\iE$, if we do only concentrate on equivalence classes of Pauli operators, thus ignoring additional phases.}
The unitary matrix of the \texttt{Phase-\iE} gate is given by
\begin{align}
 U_{\texttt{P}_{\pm \iE}} = \begin{pmatrix} 1 & 0 \\ 0 & \pm \iE \end{pmatrix},
\end{align}
and can be represented by the symplectic matrix
\begin{align}
 C_{\texttt{P}_{\pm \iE}} = \begin{pmatrix} 1 & 1 \\ 0 & 1 \end{pmatrix},
\end{align}
with the following visual representation.
  \begin{center}
  \begin{tikzpicture}[scale=1, transform shape]
    \tikzstyle{operator} = [draw,fill=white,minimum size=1.5em] 
    \tikzstyle{phase} = [fill,shape=circle,minimum size=5pt,inner sep=0pt]
    \node at (1,0) (q1) {};
    \node[operator] (op1) at (2,0) {$\iE$} edge [-] (q1);
    \node at (3,0) (q1') {} edge [-] (op1);
  \end{tikzpicture}
  \end{center} 
The multi-qubit operators can be built in a similar fashion as for the \texttt{Hadamard} gate which are given by Equations
\eqref{eqn:gatedecomp:homo:UHi} and \eqref{eqn:gatedecomp:homo:cHikl} and yield in the representation as a symplectic matrix,
acting on the $i$-th qubit,
\begin{align}\label{eqn:gatedecomp:homo:cPikl}
 \left(c^{(i)}_{\texttt{P}_{\pm \iE}}\right)_{kl} = \delta_{kl} + \delta_{ki} \delta_{l (m+i)}.
\end{align}
As a next step, we consider controlled two-qubit gates
like the \texttt{CPhase} gate that was already introduced in Section \ref{sec:gatedecomp:fibset}. For these gates, an operation is applied
to the \emph{target} qubit if and only if the state of the \emph{control} qubit equals $\ket{1}$. Here, we concentrate on
the \texttt{CNot} gate and the \texttt{controlled-Z} gate. The unitary operator of the \texttt{CNot} gate, which is equivalent to
a \texttt{controlled-X} gate, is given by
\begin{align}
 U^{(1,2)}_{\texttt{CNot}} = \begin{pmatrix} 1 & 0 & 0 & 0\\ 0 & 1 & 0 & 0\\ 0 & 0 & 0 & 1\\ 0 & 0 & 1 & 0\end{pmatrix}.
\end{align}
The symplectic representation derived from the truth table of Pauli operators of this operation is given by
\begin{align}
 C^{(1,2)}_{\texttt{CNot}} = \begin{pmatrix} 1 & 1 & 0 & 0\\ 0 & 1 & 0 & 0\\ 0 & 0 & 1 & 0\\ 0 & 0 & 1 & 1\end{pmatrix},
\end{align}
and has also a visual representation.
  \begin{center}
  \begin{tikzpicture}[scale=1, transform shape]
    \tikzstyle{not} = [draw,fill=white,shape=circle, minimum size=1.5em] 
    \tikzstyle{phase} = [fill,shape=circle,minimum size=5pt,inner sep=0pt]
    \node at (1,0) (q1) {};
    \node at (1,-1) (q2) {};
    \node[phase] (p1) at (2,0) {} edge [-] (q1);
    \draw (2,-1) node[draw,circle, minimum size=1.1em] {};
    \draw[-] (p1) -- (2,-1cm-.551em);
    \node at (3,0) (q1') {} edge [-] (p1);
    \node at (3,-1) (q2') {} edge [-] (q2);
    \draw[-] (p1) -- (op1);
  \end{tikzpicture}
  \end{center}
The multi-qubit operator for a \texttt{CNot} gate with control qubit $c$ and target qubit~$t$, represented as a symplectic matrix,
is given by
\begin{align}\label{eqn:gatedecomp:homo:cnotmulti}
 \left(c^{(c,t)}_{\texttt{CNot}}\right)_{kl} = \delta_{kl} + \delta_{kc} \delta_{lt} + \delta_{k(t+m)} \delta_{l(c+m)}.
\end{align}

The last gate we want to consider is the \texttt{controlled-Z} gate, which is a special version of the \texttt{controlled-Phase} gate, with a
fixed phase factor of $-1$. We will denote this gate by \texttt{CZ} in the following. Its unitary representation is given by
\begin{align}
 U^{(1,2)}_{\texttt{CZ}} = \begin{pmatrix} 1 & 0 & 0 & 0\\ 0 & 1 & 0 & 0\\ 0 & 0 & 1 & 0\\ 0 & 0 & 0 & -1\end{pmatrix},
\end{align}
and its symplectic representation by
\begin{align}
 C^{(1,2)}_{\texttt{CZ}} = \begin{pmatrix} 1 & 0 & 0 & 1\\ 0 & 1 & 1 & 0\\ 0 & 0 & 1 & 0\\ 0 & 0 & 0 & 1\end{pmatrix}.
\end{align}
The visualization of the \texttt{CZ} gate is the following circuit.
  \begin{center}
  \begin{tikzpicture}[scale=1, transform shape]
    \tikzstyle{operator} = [draw,fill=white,minimum size=1.5em] 
    \tikzstyle{phase} = [fill,shape=circle,minimum size=5pt,inner sep=0pt]
    \node at (1,0) (q1) {};
    \node at (1,-1) (q2) {};
    \node[phase] (p1) at (2,0) {} edge [-] (q1);
    \node[operator] (op1) at (2,-1) {$Z$} edge [-] (q2);
    \node at (3,0) (q1') {} edge [-] (p1);
    \node at (3,-1) (q2') {} edge [-] (op1);
    \draw[-] (p1) -- (op1);
  \end{tikzpicture}
  \end{center} 
Within a multi-qubit environment, the symplectic representation of the \texttt{CZ} gate with control qubit named $c$
and target qubit $t$ is accordingly given by
\begin{align}\label{eqn:gatedecomp:homo:czmulti}
 \left(c^{(c,t)}_{\texttt{CZ}}\right)_{kl} = \delta_{kl} + \delta_{kc} \delta_{l(t+m)} + \delta_{kt} \delta_{l(c+m)}.
\end{align}

Armed with all those gates we can implement any stabilizer matrix $C \in M_{2m}(\F_2)$ in the form of Equation
\eqref{eqn:gatedecomp:homo:C} with $u$ invertible by a quantum circuit. The
process we suggest is realized in two parts. In the first part, the stabilizer matrix is transformed into a special
form by symplectic operations such that~$u$ is mapped to $\Eins_m$, in the second part, this form can be written as a product of symplectic matrices
that can be implemented directly. 

\begin{lem}[Gaussian elimination]\label{lem:gatedecomp:homo:gauss}\hfill\\
 Any symplectic matrix $C \in M_{2m}(\F_2)$ in the form of Equation \eqref{eqn:gatedecomp:homo:C} with $u$ invertible, can be transformed by applying \textup{\texttt{CNot}} operations
 in the form of Equation \eqref{eqn:gatedecomp:homo:cnotmulti} from the left in order to obtain a symplectic matrix in the form
\begin{align}\label{eqn:gatedecomp:homo:C'}
 C' = \begin{pmatrix}s'&t'\\\Eins_m&v'\end{pmatrix}.
\end{align}
\end{lem}
\begin{proof}
 If $u$ is an invertible matrix, it can be diagonalized by the Gaussian elimination, that uses elementary row or column operations. 
 By multiplying a \texttt{CNot} operation in the manner of Equation \eqref{eqn:gatedecomp:homo:cnotmulti} from the left to the
 symplectic matrix $C$, within the block of the submatrix $u$, the row which has the index of the control qubit is added and stored
 to the row which has the index of the target qubit. Two rows $i,j$ of $u$ can be swapped by applying 
 $C^{(i,j)}_{\texttt{CNot}} C^{(j,i)}_{\texttt{CNot}} C^{(i,j)}_{\texttt{CNot}}$, as the elements of $u$ have characteristic two.
 Therefore, all elementary row operations can be realized with \texttt{CNot} operations.\footnote{The third operation, namely the multiplication
 of any row with a non-zero scalar, is given by the identity operation in $\F_2$, as no non-zero element exists besides $1$.}
\end{proof}
The transformation may also be realized by multiplying \texttt{CNot} operations from the right which induces elementary column operations, respectively.\linebreak
Lemma~\ref{lem:gatedecomp:homo:gauss} leads implicitly to the following corollary:
\begin{cor}[Gate commutation]\label{cor:gatedecomp:homo:commute}\hfill\\
 Two quantum gates which have representations as symplectic matrices in the form
 \begin{align}
  C_1= \begin{pmatrix} \Eins_m & a_1\\ 0_m & \Eins_m \end{pmatrix}\quad\mathrm{and}\quad C_2= \begin{pmatrix} \Eins_m & a_2\\ 0_m & \Eins_m \end{pmatrix},
 \end{align}
 with $a_1, a_2 \in M_m(\F_2)$, commute.
\end{cor}
\begin{proof}
 The application of $C_1$ to $C_2$ (or vice versa) gives
 \begin{align}
  C_1 C_2 = \begin{pmatrix} \Eins_m & a_1+a_2\\ 0_m & \Eins_m \end{pmatrix}.
 \end{align}
 Therefore, the \texttt{Phase-\iE} gates and the \texttt{CZ} gates commute.
\end{proof}

The resulting matrix of Lemma \ref{lem:gatedecomp:homo:gauss} can be decomposed into a product of matrices that can be used to read off the required
\texttt{CZ} gates, \texttt{Phase} gates and \texttt{Hadamard} gates directly.
\begin{lem}[Stabilizer matrix decomposition]\label{lem:gatedecomp:homo:decomp}\hfill\\
 Any symplectic matrix $C \in M_{2m}(\F_2)$ in the form of Equation \eqref{eqn:gatedecomp:homo:C'}, can be factorized in matrices that describe the product
 of \texttt{CZ} gates, \texttt{Phase} gates, and an $m$-folded tensor product of the Hadamard matrix as
\begin{align}\label{eqn:gatedecomp:homo:decomp}
 C' = \begin{pmatrix}s'&t'\\ \Eins_m&v'\end{pmatrix} = 
\begin{pmatrix}\Eins_m & s'\\ 0_m & \Eins_m\end{pmatrix} \begin{pmatrix}0_m &\Eins_m\\ \Eins_m & 0_m\end{pmatrix} \begin{pmatrix}\Eins_m & v'\\ 0_m & \Eins_m\end{pmatrix},
\end{align}
which implies $t'=\Eins_m+s'v'$.
\end{lem}
\begin{proof}
 Corollary \ref{cor:app:clifford:symplecticmatrixprop} of Appendix \ref{app:clifford} states the conditions on $C'$ to be a symplectic matrix. From the first condition
 follows, that $s'$ has to be symmetric. From the last condition follows then with $u'=\Eins_m$ that $t'=\Eins_m+s'v'$. Inserting $t'$ into the second condition
 shows finally, that also $v'$ has to be symmetric.\par
 The given equation appears naturally, where the middle term corresponds to Hada\-mard operations given by Equation \eqref{eqn:gatedecomp:homo:cHikl}. The remaining
 two terms can be created using Corollary \ref{cor:gatedecomp:homo:commute} the following way:
 Each non-zero element on the diagonal of $s'$ and $v'$, respectively, refers to a \texttt{Phase-$\iE$} gate, as can be seen from Equation \eqref{eqn:gatedecomp:homo:cPikl}.
 As $s'$ and $v'$ are symmetric matrices, their off-diagonal, non-zero elements refer to \texttt{CZ} gates, according to Equation \eqref{eqn:gatedecomp:homo:czmulti}.
\end{proof}
To implement $s'$ and $v'$, at most $2m$ \texttt{Phase-$\iE$} gates and at most $m^2 - m$ \texttt{CZ} gates are needed. The form derived by Equation
\eqref{lem:gatedecomp:homo:decomp} requires exactly $m$ \texttt{Hadamard} gates. An upper bound to realize the Gaussian elimination as discussed in Lemma
\ref{lem:gatedecomp:homo:gauss} is the implementation of $m^2$ different \texttt{CNot} gates, as any $m$-bit vector can be build by maximally $m$
\texttt{XOR} operations within the set of $m$ linearly independent vectors. For $v'=0_m$ follows $t' = \Eins_m$, thus a Fibonacci-based set. The quantum
circuit which appears then by applying Lemma \ref{lem:gatedecomp:homo:decomp} is the same as the one given by Figure \ref{fig:gatedecomp:fibset:circuit}.\par
Results of the semigroup sets (cf. Appendix \ref{app:homosets:semigroup}) indicate, that the sets of MUBs can be limited to sets which have an invertible $m \times m$ submatrix in the
upper right corner (i.\,e. the matrix $t$ of Equation \eqref{eqn:gatedecomp:homo:C} is invertible). If this is not the case, one may find a set of quantum
gates that realizes an appropriate transformation. Implementations of complete sets of MUBs for three qubit systems are shown by Figure~\ref{fig:gatedecomp:homo:group:circuit}
and Figure \ref{fig:gatedecomp:homo:semigroup:circuit}; quantum circuits for four qubit systems are listed in Appendices \ref{app:homosets:group} and \ref{app:homosets:semigroup}.

\begin{figure}[t]
\centerline{
\begin{tikzpicture}[scale=0.8, transform shape]
\tikzstyle{operator} = [draw,fill=white,minimum size=1.5em]
\tikzstyle{phase} = [fill,shape=circle,minimum size=5pt,inner sep=0pt]
\tikzstyle{surround} = [fill=black!10,thick,draw=black,rounded corners=2mm]
\node at (1,0) (q1) {};
\node at (1,-1) (q2) {};
\node at (1,-2) (q3) {};
\node[phase] (p0) at (2,-2) {} edge [-] (q3);
\draw (2,0) node[draw,circle, minimum size=1.1em] {};
\draw[-] (p0) -- (2,0cm+.551em);
\node[phase] (p1) at (3,0) {} edge [-] (q1);
\draw (3,-2) node[draw,circle, minimum size=1.1em] {};
\draw[-] (p1) -- (3,-2cm-.551em);
\node[phase] (p2) at (4,-2) {} edge [-] (p0);
\draw (4,0) node[draw,circle, minimum size=1.1em] {};
\draw[-] (p2) -- (4,0cm+.551em);
\node[operator] (p3) at (5,0) {$\iE$} edge [-] (p1);
\node[operator] (p4) at (5,-2) {$\iE$} edge [-] (p2);
\node[phase] (p5) at (6,-1) {} edge [-] (q2);
\node[operator] (p6) at (6,-2) {$-1$} edge [-] (p4);
\draw[-] (p5) -- (p6);
\node[operator] (p7) at (7,0) {$H$} edge [-] (p3);
\node[operator] (p8) at (7,-1) {$H$} edge [-] (p5);
\node[operator] (p9) at (7,-2) {$H$} edge [-] (p6);
\draw[-] (p7) -- (8,0);
\draw[-] (p8) -- (8,-1);
\draw[-] (p9) -- (8,-2);
\draw[decorate,decoration={brace},thick] (8.3,0.2) to node[midway,right] (bracketR) {$\;\;=C_{(2,3,4)}$} (8.3,-2.2);
\begin{pgfonlayer}{background}
\node[surround] (background) [fit = (bracketL) (p7) (p9) (bracketR)] {};
\end{pgfonlayer}
\end{tikzpicture}
}
  \caption{Quantum circuit for the generator $C_{(2,3,4)}$ of Equation \eqref{eqn:homoset:group:C234} of a complete, homogeneous set of group-based cyclic MUBs for a three qubit system.}
  \label{fig:gatedecomp:homo:group:circuit}
\end{figure}
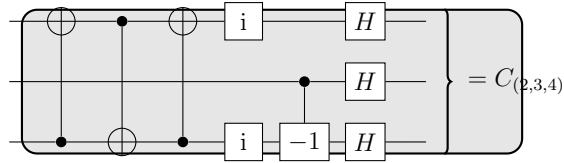 

\begin{figure}[t]
\centerline{
\begin{tikzpicture}[scale=0.8, transform shape]
\tikzstyle{operator} = [draw,fill=white,minimum size=1.5em]
\tikzstyle{phase} = [fill,shape=circle,minimum size=5pt,inner sep=0pt]
\tikzstyle{surround} = [fill=black!10,thick,draw=black,rounded corners=2mm]
\node at (1,0) (q1) {};
\node at (1,-1) (q2) {};
\node at (1,-2) (q3) {};
\node[phase] (p0) at (2,-2) {} edge [-] (q3);
\draw (2,0) node[draw,circle, minimum size=1.1em] {};
\draw[-] (p0) -- (2,0cm+.551em);
\node[phase] (p1) at (3,0) {} edge [-] (q1);
\draw (3,-2) node[draw,circle, minimum size=1.1em] {};
\draw[-] (p1) -- (3,-2cm-.551em);
\node[phase] (p2) at (4,-2) {} edge [-] (p0);
\draw (4,0) node[draw,circle, minimum size=1.1em] {};
\draw[-] (p2) -- (4,0cm+.551em);
\node[operator] (p3) at (5,-1) {$\iE$} edge [-] (q2);
\node[operator] (p4) at (5,-2) {$\iE$} edge [-] (p2);
\node[phase] (p5) at (6,-1) {} edge [-] (p3);
\node[operator] (p6) at (6,-2) {$-1$} edge [-] (p4);
\draw[-] (p5) -- (p6);
\node[operator] (p7) at (7,0) {$H$} edge [-] (p1);
\node[operator] (p8) at (7,-1) {$H$} edge [-] (p5);
\node[operator] (p9) at (7,-2) {$H$} edge [-] (p6);
\node[operator] (p11) at (8,-1) {$\iE$} edge [-] (p8);
\node[operator] (p12) at (8,-2) {$\iE$} edge [-] (p9);
\draw[-] (p7) -- (9,0);
\draw[-] (p11) -- (9,-1);
\draw[-] (p12) -- (9,-2);
\draw[decorate,decoration={brace},thick] (9.3,0.2) to node[midway,right] (bracketR) {$\;\;=C_{(1,6,2)}$} (9.3,-2.2);
\begin{pgfonlayer}{background}
\node[surround] (background) [fit = (bracketL) (p7) (p9) (bracketR)] {};
\end{pgfonlayer}
\end{tikzpicture}
}
  \caption{Quantum circuit for the generator $C_{(1,6,2)}$ of Equation \eqref{eqn:homoset:semigroup:C162} of a complete, homogeneous set of semigroup-based cyclic MUBs for a three qubit system.}
  \label{fig:gatedecomp:homo:semigroup:circuit}
\end{figure}

\subsection{Inhomogeneous sets}\label{sec:gatedecomp:inhomogeneous}

Quantum circuits for inhomogeneous sets can in principle be built in the same way as discussed for the homogeneous sets in the former section.
But these sets do not include the standard basis as shown in Section \ref{sec:inhomogeneous}, i.\,e. the generator $G_0$
of the first class $\cC_0$ is not given by $G'_0=(\Eins_m,0_m)^t$; in standard form it reads $\bar G_0 = (\Eins_m,G_0^x)^t$ (cf. Equation \eqref{eqn:standfor:standfor}).
Any symplectic matrix $C_0$ that transforms the generator $G'_0$ to the generator $\bar G_0$ can be taken to implement the
quantum circuit which produces the class $\cC_0$. In order to generate a class with commuting elements, $G_0^x$ has to be a symmetric matrix (cf. \cite[Lemma 4.3]{Bandy02}).
The symplectic transformation $C_0 G'_0 = \bar G_0$ equals
\begin{align}
 C_0 = \begin{pmatrix} \Eins_m & 0_m\\ G_0^x & \Eins_m \end{pmatrix},
\end{align}
and can be rewritten using \texttt{Hadamard} gates similarly to Lemma \ref{lem:gatedecomp:homo:decomp} as
\begin{align}
 C_0 = \begin{pmatrix}0_m &\Eins_m\\ \Eins_m & 0_m\end{pmatrix} \begin{pmatrix} \Eins_m & G_0^x\\ 0_m & \Eins_m \end{pmatrix} \begin{pmatrix}0_m &\Eins_m\\ \Eins_m & 0_m\end{pmatrix}.
\end{align}
The quantum circuit is then given by a \texttt{Hadamard} gate on each individual qubit (to implement the rightmost matrix),
followed by \texttt{Phase-$\iE$} gates on each diagonal non-zero value of $G_0^x$ and \texttt{CZ} gates on each off-diagonal, upper triangular
value of $G_0^x$, as done for Lemma \ref{lem:gatedecomp:homo:decomp} (to implement the matrix in the middle) and again
\texttt{Hadamard} gates on each individual qubit (to implement the leftmost matrix). Since the application of two consecutively applied
Hadamard operations is the unity operation, we can omit \texttt{Hadamard} gates on those qubits, that are involved neither in any 
\texttt{Phase-$\iE$} nor in any \texttt{CZ} gate.\par
Resulting quantum circuits for complete sets of MUBs on systems with four qubits are given in Appendix \ref{app:inhomosets}, the quantum circuit for
the three-qubit system which is generated by $C_{(0,9,0)}$ of Equation \eqref{eqn:inhomoset:C090} is shown in Figure \ref{fig:gatedecomp:inhomo:circuit}.

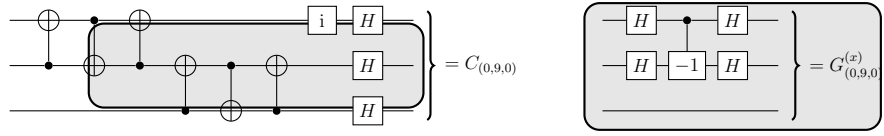
\begin{figure}[t]
  \centerline{
\begin{tikzpicture}[scale=0.6, transform shape]
\tikzstyle{operator} = [draw,fill=white,minimum size=1.5em]
\tikzstyle{phase} = [fill,shape=circle,minimum size=5pt,inner sep=0pt]
\tikzstyle{surround} = [fill=black!10,thick,draw=black,rounded corners=2mm]
\node at (1,0) (q1) {};
\node at (1,-1) (q2) {};
\node at (1,-2) (q3) {};
\node[phase] (p0) at (2,-1) {} edge [-] (q2);
\draw (2,0) node[draw,circle, minimum size=1.1em] {};
\draw[-] (p0) -- (2,0cm+.551em);
\node[phase] (p1) at (3,0) {} edge [-] (q1);
\draw (3,-1) node[draw,circle, minimum size=1.1em] {};
\draw[-] (p1) -- (3,-1cm-.551em);
\node[phase] (p2) at (4,-1) {} edge [-] (p0);
\draw (4,0) node[draw,circle, minimum size=1.1em] {};
\draw[-] (p2) -- (4,0cm+.551em);
\node[phase] (p3) at (5,-2) {} edge [-] (q3);
\draw (5,-1) node[draw,circle, minimum size=1.1em] {};
\draw[-] (p3) -- (5,-1cm+.551em);
\node[phase] (p4) at (6,-1) {} edge [-] (p2);
\draw (6,-2) node[draw,circle, minimum size=1.1em] {};
\draw[-] (p4) -- (6,-2cm-.551em);
\node[phase] (p5) at (7,-2) {} edge [-] (p3);
\draw (7,-1) node[draw,circle, minimum size=1.1em] {};
\draw[-] (p5) -- (7,-1cm+.551em);
\node[operator] (p6) at (8,0) {$\iE$} edge [-] (p1);
\node[operator] (p7) at (9,0) {$H$} edge [-] (p6);
\node[operator] (p8) at (9,-1) {$H$} edge [-] (p4);
\node[operator] (p9) at (9,-2) {$H$} edge [-] (p5);
\draw[-] (p7) -- (10,0);
\draw[-] (p8) -- (10,-1);
\draw[-] (p9) -- (10,-2);
%
%
%
\node at (14,0) (qq1) {};
\node at (14,-1) (qq2) {};
\node at (14,-2) (qq3) {};
\node[operator] (qp0) at (15,0) {$H$} edge [-] (qq1);
\node[operator] (qp1) at (15,-1) {$H$} edge [-] (qq2);
\node[phase] (qp4) at (16,0) {} edge [-] (qp0);
\node[operator] (qp5) at (16,-1) {$-1$} edge [-] (qp1);
\draw[-] (qp4) -- (qp5);
\node[operator] (qp6) at (17,0) {$H$} edge [-] (qp4);
\node[operator] (qp7) at (17,-1) {$H$} edge [-] (qp5);
\draw[-] (qp6) -- (18,0);
\draw[-] (qp7) -- (18,-1);
\draw[-] (qq3) -- (18,-2);
\draw[decorate,decoration={brace},thick] (10.3,0.2) to node[midway,right] (bracketR) {$\;\;=C_{(0,9,0)}$} (10.3,-2.2);
\draw[decorate,decoration={brace},thick] (18.3,0.2) to node[midway,right] (bracketR2) {$\;\;=G^{(x)}_{(0,9,0)}$} (18.3,-2.2);
\begin{pgfonlayer}{background}
\node[surround] (background) [fit = (bracketL) (p7) (p9) (bracketR)] {};
\node at (17,-1.02cm) [surround,minimum width=6.5cm,minimum height=2.8cm] (background) {};
\end{pgfonlayer}
\end{tikzpicture}
}
  \caption{Quantum circuit for the generator $C_{(0,9,0)}$ of Equation \eqref{eqn:inhomoset:C090} and the lower submatrix of the corresponding generator
   of the class $\cC_0$, namely $G_{(0,9,0)}$ of Equation~\eqref{eqn:inhomoset:G090}, of a complete, inhomogeneous set of cyclic MUBs for a three qubit system.}
   \label{fig:gatedecomp:inhomo:circuit}
\end{figure} 

\subsection{Practical implementation}\label{sec:gatedecomp:practical}

For a practical implementation of a complete set of cyclic MUBs, a quantum circuit like the example of Figure \ref{fig:gatedecomp:fibset:circuit}
would be able to transform a quantum state
by an operation that shifts one of the $d+1$ mutually unbiased bases to the next. Thus, after this process, variables that
are complementary to the former ones, could be measured in the computational basis. To be able to measure in a specific
basis of these $d+1$ different bases, the quantum circuit has to be used several times, in order to achieve the desired set of variables.
In average, this would cause $d/2$ applications of the quantum circuit. Consequently, for a large number of bases, namely a high dimension
$d$ of the Hilbert space, we propose to decompose the index of the bases binary and create for each power of two an individual
circuit. This raises the size of the experimental setup, but reduces the number of gates which are used in average to create
a specific basis.\par
To give an example, for a dimension $d=2^{10}$, a complete set of MUBs has $1025$ different bases. Following the suggestion,
there would be ten different quantum circuits, thus any basis can be reached by using five of them in average. This is
a huge improvement over using $512$ times the same circuit in average if only the generator of the cyclic set of MUBs is implemented.\par
Depending on the dimension, it could also be favored to decompose the index to any other number system; fortunately, all those
approaches can be realized by using the powers of the stabilizer matrix.

\chapter{Equivalence of mutually unbiased bases}\label{chap:equivalence}
A set of mutually unbiased bases defines a set of bases, where the absolute value of the overlap
of all pairs of vectors that are taken from two different bases within the set is constant
and equals $d^{-1/2}$, where $d$ is the dimension of the corresponding Hilbert space $\cH = \C^d$,
as was seen in Definition \ref{defi:limits:mubs}. From this definition it does not make any sense
to take the ordering of the bases into account, as well as the ordering of the basis vectors of any
basis. Also the multiplication of any basis vector with a phase factor preserves the mutual
unbiasedness which is finally even invariant under any unitary transformation which is applied to the whole set of bases.
By choosing a unitary transformation which is the inverse of any unitary operator $\cB_j$ of the set, this basis
becomes the standard basis, thus, each set of MUBs we find, can be transformed into a set that includes
the standard basis. We sum up these naturally appearing arguments to define the \emph{equivalence}
of different sets of MUBs.
\begin{defi}[Equivalence of mutually unbiased bases]\label{defi:equiv:equivalence}\hfill\\
 Two sets of mutually unbiased bases, namely $\fS = \Mg{\cB_0, \ldots, \cB_{r-1}}$ and\linebreak
$\fS' = \Mg{\cB'_0, \ldots, \cB'_{r-1}}$ with $r \in \N^*$, are said to be equivalent, if there holds
 \begin{align}
   \cB'_j = U \cB_{\pi(j)} W_j,
 \end{align}
 with a unitary matrix $U \in M_d(\C)$, a permutation $\pi$ on $\MgN{r-1}$ and monomial matrices\footnote{A
 monomial matrix is defined to be a square matrix which has in every row and every column exactly one entry
 that differs from zero.} $W_j$
that are given by the product of a permutation matrix and a diagonal phase matrix, for $j \in \MgN{r-1}$.
\end{defi}
It has to be kept in mind that this definition of equivalence does not distinguish between the different
sets of cyclic MUBs, that were discussed in Sections~\ref{sec:fibset}, \ref{sec:homogeneous} and \ref{sec:inhomogeneous}.
Their equivalence in a Hilbert space of dimension $d=2^3$ was shown already in~\cite{Romero05}. To
distinguish those sets, the consideration of the entanglement properties is essential, as discussed
in Section \ref{sec:entangleprop}.\par
Nevertheless, on this mathematical level of equivalence, different sets were discussed in the literature \cite{Calderbank97, Godsil09, Kantor12}.
It is still an open problem how many inequivalent sets exist and how they are related.
Within the following section, we will discuss the equivalence of sets of MUBs which are based on
the construction of Bandyopadhyay \etal{} \cite{Bandy02} and compare them with other constructions.
The construction of the Fibonacci-based sets of MUBs allows us to construct a set of operators that realize
permutations of the bases which is shown in the subsequent section.

\section{Heisenberg group partition sets}
To analyze the equivalence of sets of MUBs which are based on the construction of Bandyopadhyay \etal{},
we can use their construction principles in order to simplify the comparison of two different sets. Following
the explanation given in Section \ref{sec:bandy}, each element $U_j$ of the set of MUBs $\fS$ which is
generated by the powers of a unitary matrix $U$ as $\fS = \gen{U}$, is given by the set of common eigenvectors
of a corresponding class of Pauli operators $\cC_j$ (cf. Equation~\eqref{eqn:bandy:Mj}) with $j \in \MgN{r-1}$
and $r \in \N^*$. For a complete set of MUBs we have $r=d+1$. By the construction of Equation~\eqref{eqn:bandy:classes},
a class $\cC_j$ is a set of elements $\ZX( \vec a )$ of the set of Pauli operators.
Analogously to the set of unitary matrices, we can define the set $\fC = \Mg{\cC_0, \ldots, \cC_{r-1}}$ of
classes $\cC_j$, that defines a complete set of MUBs, but without determining the order of the basis vectors or
their phase within a single class. To define equivalence on this level, there is no need for matrices $W_j$ as
introduced by Definition \ref{defi:equiv:equivalence}. We can formulate the following lemma:

\begin{lem}[Equivalence of mutually unbiased bases]\label{lem:equiv:bandy:equiv}\hfill\\
  Two sets of MUBs which are characterized by $\fC = \Mg{\cC_0, \ldots, \cC_{r-1}}$ and\linebreak $\fC' = \Mg{\cC'_0, \ldots, \cC'_{r-1}}$
  are equivalent, if there holds
 \begin{align}
   \cC'_j = U \cC_{\pi(j)} U^{\dagger},
 \end{align}
 with a unitary matrix $U \in \C^d$, a permutation $\pi$ on $\MgN{r-1}$, and $r \in \N^*$.
\end{lem}
\begin{proof}
 Starting with the Definition \ref{defi:equiv:equivalence} on the equivalence of MUBs, the matrices $W_j$ do not
 play any role if we observe sets of Pauli operators $\cC_j$ and $\cC'_{\pi(j)}$ that do neither have any ordering
 nor keep any phase. The equivalence has to be formulated the given way, since in the definition, we transform an
 eigenvector $v$ as $Uv$; so, considering $v$ is an eigenvector of a matrix $A$, then $U v$ is an eigenvector
 of $U A U^{\dagger}$.
\end{proof}
As the classes of both sets $\fC$ and $\fC'$ contain only Pauli operators, the unitary transformation $U$ has to be
an element of the Clifford group (cf. Appendix \ref{app:clifford}), that maps by definition Pauli operators onto
Pauli operators. A certain class $\cC_j$ can be written in terms of a $2m \times (d-1)$-matrix, where each column
represents a Pauli operator, namely one of the vectors $\vec a$ given in Equation \eqref{eqn:bandy:classes}. It can be
seen within that formula, that the $d-1$ elements can be generated by the (non-unique) so-called $2m \times m$ generator
matrix $G_j$, thus $\cC_j = \gen{G_j}$.\footnote{Normally, the symbol $\gen{A}$ refers to the set, which is generated by 
the powers of $A$. Adopting the abstract meaning of this symbol, the expression $\gen{G_j}$ refers to a generation
which is specified by Equation \eqref{eqn:bandy:classes}.}
Consequently, equivalence of two sets $\fC$ and $\fC'$ is given, if
and only if
\begin{align}\label{eqn:equiv:bandy:gen}
 \gen{G'_j} = f \gen{G_{\pi(j)}} Q_j,
\end{align}
with a symplectic matrix $f$ (that maps Pauli operators onto Pauli operators), a permutation $Q_j$ and a permutation $\pi$ on
$\MgN{d}$. In the case, we write the
generators in standard form as discussed in Section \ref{sec:standardform}, the generators $G_j$ and $G'_j$ become
unique. Then the permutation matrix $Q_j$ can be set to unity, thus neglected. For the transformation matrix
$f$ we are free to assume the block-matrix form
\begin{align}
 f = \begin{pmatrix}s&t\\u&v\end{pmatrix}.
\end{align}
As it is shown by Corollary \ref{cor:app:clifford:symplecticmatrixprop}, this matrix is symplectic if and only if
$s^t u$ and $t^t v$ are symmetric and $v^t s - t^t u = \Eins_m$. With the methods introduced so far, we are able to
prove the equivalence of all Fibonacci-based complete sets of MUBs for a specific dimension $d=2^m$ with $m \in \N^*$.

\subsection{Fibonacci-based sets}\label{sec:equiv:fibset}
To construct a Fibonacci-based set, we have to follow simply the \mbox{Steps 1.--3.} that were given in Section \ref{sec:fibset}.
It turns out, that neither the choice of the reduced stabilizer matrix $B$ is unique, nor the polynomial with
Fibonacci index $d+1$. Within this section we will show, that all Fibonacci-based sets which are based on the same
polynomial, are equivalent, independent of the choice of~$B$. In a second step it is shown, that also a different
choice of the polynomial will create an equivalent set of cyclic MUBs.\par

Applying the generator of a set to all elements rotates this set, so we can choose $\pi(0) = 0$ for Equation
\eqref{eqn:equiv:bandy:gen}, which implies $\gen{G_0} = \gen{G'_0}$ and in standard form $\bar G_0 = \bar G'_0$.
According to the construction of the Fibonacci-based sets, this generator is given
by
\begin{align}
 G_0 = \begin{pmatrix}\Eins_m \\0_m\end{pmatrix},
\end{align}
which can also be seen in the derivation of the standard form in Equation~\eqref{eqn:standfor:generatorsstandfor}.
Thus, the choice for the permutation $\pi$ maps the class of Pauli-$Z$ operators of the set $\fC$ to the class
of Pauli-$Z$ operators of the set $\fC'$. Using Equation \eqref{eqn:equiv:bandy:gen},
we have to demand that $u=0_m$; to keep $f$ symplectic, we need to set $v=(s^t)^{-1}$. If both sets represent complete
sets of MUBs, we can represent their generators in standard form, where the elements $G^z_j$ as given in Equation~\eqref{eqn:standfor:standfor}
form a matrix representation of the finite field $\F_2^m$; the same holds true for the elements $G'^z_j$. We can
formulate a lemma on equivalent MUBs:
\begin{lem}[Equivalence of certain Fibonacci-based sets]\label{lem:equiv:bandy:equivB}\hfill\\
  The choice of different reduced stabilizer matrices $B \in M_m(\F_2)$ with the same characteristic polynomial that has
  Fibonacci index $2^m+1$ in the construction of Fibonacci-based MUBs leads to equivalent sets.
\end{lem}
\begin{proof}
 Two reduced stabilizer matrices $B, B' \in M_m(\F_2)$, which have to be symmetric according to Condition (i),
 formulated in Section \ref{sec:fibset}, are related as $B' = s B s^t$ with an orthogonal matrix $s \in M_m(\F_2)$, i.\,e.
 $s^{-1} = s^t$. If we
 set the transformation matrix as
 \begin{align}
  f = \begin{pmatrix}s&0_m\\0_m&s\end{pmatrix},
 \end{align}
 this matrix is symplectic. Applied to the generators of the classes $\cC_j$ of $\fC$ in standard form, we get $f \bar G_j = (s p_j(B), s)^t$,
 where $p_j(B)$ is the polynomial $F_{j+1}(B) (F_{j}(B))^{-1}$.
 As seen by Corollary \ref{cor:standfor:generatorfree}, the transformed generators can be multiplied from the right by any invertible, block-diagonal matrix,
 in particular also by $f^{-1}$. This leads to $f \bar G_j f^{-1} = (s p_j(B) s^t, \Eins_m)^t \equiv \bar G'_j$, the generators
 of the classes of the set $\fC'$, in standard form. The case of $j=0$ can be shown the same way.
\end{proof}
We can use this lemma to give a general theorem on the equivalence of Fibonacci-based sets.
\begin{thm}[Equivalence of Fibonacci-based sets]\label{thm:equiv:bandy:equiv}\hfill\\
 All Fibonacci-based complete sets of MUBs are equivalent.
\end{thm}
\begin{proof}
 It was shown by Lemma \ref{lem:equiv:bandy:equivB}, that all reduced stabilizer matrices with the same characteristic polynomial with Fibonacci
 index $d+1$ result in equivalent sets of MUBs. In the case that the two stabilizer matrices $B$ and $B'$ have different characteristic polynomials
 $p$ and $p'$, respectively, with Fibonacci index $d+1$, we can use the following consideration: In general, the two different irreducible polynomials
 have degree $m$ and are defined over the ground field $\F_2$. We pick a pair of their roots, namely $\beta$ and $\beta'$, respectively. The adjunction of $\beta$
 to $\F_2$ results in the field $\F_2(\beta) \cong \F_2[x]/p \F_2[x] \cong \F_{2^m}$, a similar consideration holds for $\beta'$. But as both extensions are \emph{Galois extensions},
 the elements in $\F_2(\beta)$ and in $\F_2(\beta')$ are equivalent up to a permutation \cite[Chapter~4.1]{Bosch06}. The elements of those fields appear
 in the Pauli-$Z$ part of the generators in standard form (cf. Equation \eqref{eqn:standfor:generatorsstandfor}), thus both roots result in the same set of generators.
 If we chose a symmetric matrix $B \in M_m(\F_2)$ to represent the root $\beta$, we have the freedom shown in Lemma \ref{lem:equiv:bandy:equivB},
 where the minimal polynomial of $B$ equals the minimal polynomial of $\beta$.
\end{proof}
In other words, Theorem \ref{thm:equiv:bandy:equiv} shows, that the minimal polynomials of different elements
of the finite field $\F_{2^m}$ are possible minimal polynomials of a polynomial ring $\F_2[x]$ which is congruent
to the field itself; the chosen element only needs to have a minimal polynomial with degree $m$ (which is irreducible by definition).

\subsection{(In-)Homogeneous sets}
To show the equivalence of all sets which are based on the construction of Bandyopadhyay \etal{} in a similar fashion as for the Fibonacci-based sets,
further investigations are essential and exceed the scope of this work. As already stated above, it was shown in \cite{Romero05}
that all discussed sets in Hilbert space dimension $d=2^3$ are equivalent and the equivalence of different sets in higher dimensions
was indicated. Considerations in the fashion of Section~\ref{sec:equiv:fibset} would also lead to explicit transformations between
the different sets. In a first step, one may prove the equivalence of homogeneous sets with an additive group structure in the $Z$ component
of the class generators to Fibonacci-based sets and in a second step their equivalence with sets that possess only an additive semi-group
structure (cf. Sections \ref{sec:homogeneous:group} and \ref{sec:homogeneous:semigroup}). As inhomogeneous sets do not contain
the standard basis, they are by Definition \ref{defi:equiv:equivalence} always equivalent to sets with a standard
basis, where a unitary operation for an arbitrary basis of this set exists, that transforms this set into a homogeneous set.

\subsection{Different constructions}
We can relate the Fibonacci-based sets to different constructions by following the discussion given in \cite{Bandy02}: Theorem 4.4
introduces a construction that is equivalent to the standard form of homogeneous sets (cf. Equation \eqref{eqn:standfor:standfor}).
In Section 4.3 the construction is limited to sets with an additive group structure (cf. Equation \eqref{eqn:standfor:group}). Finally, at
the end of this section, the group structure is fixed in a way to obtain a field structure. Bandyopadhyay \etal{} took this
construction from the Wootters and Fields construction \cite{WF89}, indicating their equivalence.\footnote{In fact, the similarity
of the phase factor in Equation \eqref{eqn:wootters:nonstandnofieldtwoH} and Equation \eqref{eqn:unitop:pj*} is a strong evidence of this equivalence.}
Godsil and Roy \cite{Godsil09} relate these sets also with the sets constructed by Klappenecker and Rötteler
\cite{Klappenecker04} and again to those of Bandyopadhyay \etal{}.


\section{Class-permutation operators}
Starting with the definition of equivalence given by Definition \ref{defi:equiv:equivalence}, the
question arises whether the permutation $\pi$ on the indexing of the different bases can be
realized by a symplectic transformation. In the case of Fibonacci-based sets we discovered a set
of operators that seems to be capable to implement the set of permutations by a set of symplectic
transformations. We call this set the set of \emph{class-permutation operators}%
\index{Class-permutation operators}. The construction is quite intuitive, to follow the ideas, we
start with the definition of the first class-permutation operator $A_1$, that can be deduced from 
the generator of a Fibonacci-based set in the form of Equation \eqref{eqn:fibset:B110} as
\begin{align}
 A_1 = \begin{pmatrix}\Eins_m & B \\ 0_m & \Eins_m\end{pmatrix},
\end{align}
with the same reduced stabilizer matrix $B \in M_m(\F_2)$. Since the generator $G_j$ with $j \in \MgN{d}$
of any class is given by Equation \eqref{eqn:fibset:gen}, the application of $A_1$ to this class generator leads
to
\begin{align}
 A_1 G_j =& \begin{pmatrix}\Eins_m & B \\ 0_m & \Eins_m\end{pmatrix} \begin{pmatrix} F_{j+1} (B)\\ F_{j}(B)\end{pmatrix}
   = \begin{pmatrix} F_{j+1} (B) + B F_{j}(B) \\ F_j(B)\end{pmatrix} = \begin{pmatrix} F_{j-1} (B) \\ F_j(B)\end{pmatrix}.
\end{align}
Using Lemma \ref{lem:fibset:symmetry}, we finally find
\begin{align}
 A_1 G_j = \begin{pmatrix} F_{d+1-j+1} (B) \\ F_{d+1-j}(B)\end{pmatrix} \equiv G_{d+1-j},
\end{align}
with $d=2^m$. The first class-permutation operator $A_1$ inverts therefore the ordering of the classes $\cC_j$,
where $G_0 \rightarrow G_{d+1} \equiv G_0$ keeps its position. The symplecticity of the operator $A_1$
is obvious and can easily be checked with the help of Corollary \ref{cor:app:clifford:symplecticmatrixprop}.
In this fashion, we may construct further class-permutation operators as
\begin{align}
 A'_l = \begin{pmatrix} F_l(B) & F_{l+1}(B)\\ 0_m & F_l(B) \end{pmatrix},
\end{align}
with $l \in \MgE{d}$, where in this case the properties of Corollary \ref{cor:app:clifford:symplecticmatrixprop}
are not fulfilled. To obtain a set of symplectic operators we multiply this set from the right by a
block-diagonal matrix $A^{\times}_l$ and get
\begin{align}
 A_l := A'_l \cdot A^{\times}_l &= \begin{pmatrix} F_l(B) & F_{l+1}(B)\\ 0_m & F_l(B) \end{pmatrix} \begin{pmatrix} (F_l(B))^{-1} & 0_m\\ 0_m & (F_l(B))^{-1} \end{pmatrix}\\
                     &= \begin{pmatrix} \Eins_m & F_{l+1}(B) (F_l(B))^{-1}\\ 0_m & \Eins_m \end{pmatrix}.
\end{align}
By Proposition \ref{prop:fibset:rep} the mentioned polynomials in $B$ exist; we also like to point out the striking similarity of
these operators with the generator matrices in standard form\index{Standard form} given by Equation \eqref{eqn:standfor:generatorsstandfor}.
If we apply these class-permutation operators to an arbitrary generator in standard form, we get for $A_l\cdot \bar G_0 = \bar G_0$ and for $j \neq 0$,
\begin{align}\label{eqn:equiv:classperm:fibperm}
 A_l \cdot \bar G_j = \begin{pmatrix} F_{j+1}(B) (F_j(B))^{-1} + F_{l+1}(B) (F_l(B))^{-1}\\ \Eins_m \end{pmatrix}.
\end{align}
Thus, the operator $A_l$ adds a fixed polynomial of $B$ to the Pauli-$Z$ operator part of the generator $G_j$ in standard
form. As the set of all generator matrices contains all polynomials in $B$, the operator $A_l$ permutes these generators
according to fundamental finite field theory. Since we have $d$ different such class-permutation operators $A_l$, we realize all
basic permutations which are possible this way. Finally, the set of elements of a field is invariant under the addition of any
of its elements, so we can change the position of an arbitrary generator to any new position with this set of transformations.
But it is not clear if all permutations $\pi$ on $\MgE{d}$ can be realized. If so, the permutations can be expanded on the
set $\MgN{d}$ by adding the MUB generator $C$ to the set of permutation operators, since it permutes the classes cyclically, namely
\begin{align}
 C \cdot G_j = G_{(j+1) \; \mathrm{mod} \; (d+1)}.
\end{align}

For general homogeneous sets of MUBs, an appropriate set of class-per\-mu\-ta\-tion operators can be constructed by naturally expanding
Equation \eqref{eqn:equiv:classperm:fibperm} for homogeneous sets, which results in the set of operators
\begin{align}
 A_l^{\mathrm{hom}} := \begin{pmatrix} \Eins_m & F_{l+1}(B) (F_l(B))^{-1} R\\ 0_m & \Eins_m \end{pmatrix},
\end{align}
and leads to the same effective permutations as for the Fibonacci-based sets.

\chapter{Conclusions and further work}\label{chap:conclusionsMUBs}
\emph{Mutually unbiased bases} (MUBs) find their applications in the fields of \emph{quantum state estimation}
and \emph{quantum key distribution}. By their relation to mathematical objects like \emph{orthogonal Latin squares},
\emph{symplectic spreads} and many more, considerations result in immediate effects in these neighboring fields. Many (even fundamental) questions
on MUBs are still unsolved. The most famous are the question of the existence of a set with more than three MUBs in dimension six
(and how large sets may be in all other non-prime power dimensions), as well as the number of inequivalent sets in general.\par
What we treat in this work is the problem of explicitly
constructing \emph{cyclic mutually unbiased bases} in a straightforward way. Those bases have advantages in theoretical tasks as well as in experimental
implementations. They provide also a reduced formulation of a concrete set of MUBs. It turns out that the introduced \emph{Fibonacci-based sets} have
nice properties and are related to well-known algebraic questions and may solve some of them if a construction for the \emph{symmetric companion matrix} is found.
Generalizations indicate the potential of this approach and produce sets with different \emph{entanglement properties} which can be observed in the
\emph{standard form}. A relation to an open conjecture given by Wiedemann 1988 in finite field theory is discovered and serves a realization of this conjecture which may
lead to a proof, as well as a recursive construction of cyclic MUBs in an infinite subset of the dimensions. Most fortunately, such
a construction is not only found in terms of finite field theory, which can be used immediately for a discussed
implementation in \emph{quantum circuits}, but also in terms of unitary matrices. The algorithmical formulation of an implementation strategy that may be suitable for all MUBs based on the
investigations of this work generalizes these observations and leads to feasible implementations. Approaches to prove the conjecture of Wiedemann, a 
positive test of this conjecture for $k \in \MgE{11}$ in dimensions $d=2^{2^k}$ (which is limited by the largest known prime-number factorization of Fermat numbers)
and resulting cyclic MUBs for the Fibonacci-based sets for
dimensions $d=2^m$ with $m=\MgE{600}$ complete the considerations. The examination of the construction schemes relates the discussed MUBs
to known constructions on the level of \emph{equivalence} of the measurement behavior.\par
Obviously, several questions remain open, but lead with possible generalizations of the derived results to promising future investigations. At first,
it seems to be possible to construct a symmetric companion matrix in order to generate the cyclic MUBs directly from the algebraic results. For more
general sets, the derivation of a relation between the entanglement properties and the form of the stabilizer matrix would be valuable. The equivalence
transformations between all sets would show their equivalence and be useful in order to observe the transformations of their properties. Indications
on generalizations of the Fermat-based sets and the use of the doubling scheme for MUBs over the field of real numbers, should be followed.
Considerations for odd-prime power dimensions will probably lead to somehow different constructions as they refer to finite fields with an odd characteristic
(that behave entirely different comparing with an even characteristic), but would again expand the dimensions where sets of MUBs with a reduced
set of generators can be created. Finally, it would be nice to know if the derived methods can be used in a similar way to construct inequivalent sets
of MUBs by replacing the set of Pauli operators with different sets.

\part{Quantum public-key encryption}

\chapter{Introduction}\label{chap:georgios}
The foundation for the field of quantum cryptography was laid in the early seventies of the last century,
when Wiesner had his prospective ideas on \emph{quantum money} that cannot be counterfeited, based on a physical
uncertainty principle (and the fact, that an unknown quantum state cannot be cloned perfectly \cite{NoCloning82}).
Shortly after this work appeared in 1983\footnote{The manuscript by Wiesner on quantum money was
rejected by different journals and finally presented at a conference \cite{Wiesner83}.}, the field emerged with
the \emph{quantum key distribution} (QKD) protocol of Bennett and Brassard \cite{BB84} which is
used for symmetric encryption schemes. Different modifications of this protocol were discussed
\cite{Ekert91,Bennett92} and different security proofs were given later \cite{Mayers96, LC99, SP00, GL03}.%
\footnote{A more detailed discussion on the development of symmetric QKD protocols is given in
Part \ref{part:mubs} in Section \ref{sec:qkd}.} With the new millennium, the potential benefit of the
properties of quantum physics was also discussed for further cryptographic primitives.
Examples are \emph{quantum digital signatures} \cite{GC01}, \emph{quantum fingerprinting} \cite{BCWW01} and
\emph{quantum direct communication} \cite{BF02}. Another primitive is referred to as
\emph{quantum public-key encryption} (QPKE) (see \cite{OTU00}), which will be of interest here.\footnote{
Many propositions for such a scheme were given \cite{KKNY05, HKK08, Kak06, Niko08}.}\par
In contrast to the QKD, the aim of this scheme is to reduce the number of keys needed in a network with a
large number of parties which want to exchange secret information pairwise. To implement a QKD protocol in a network where all parties should be able to communicate
securely, the total number of keys scales quadratically with the number of parties, where in the QPKE scheme this
number scales linearly.\footnote{A different approach, where QKD protocols are used, implements a \emph{key-distribution center}
(KDC), which is a trusted third party, but obviously also an attractive target.}\par
The discussion in this Part of the work is based on a scheme given by Nikolopoulos \cite{Niko08} which achieves its security
from so-called \emph{quantum one-way functions} (qOWF). At first, we will redraw within this chapter the protocol
itself. In Chapter~\ref{chap:invest}, a detailed analysis of the security of this protocol against a powerful \emph{individual
attack} is given as well as an idea to enhance the security against attacks by a \emph{noisy preprocessing} method. The
investigations will be summarized finally in Chapter \ref{chap:sumpk}.

\section{Single-qubit-rotation protocol}\label{sec:georgios}
A proposition for a QPKE protocol was given by Nikolopoulos in 2008 \cite{Niko08}\nocite{Niko08E} which can be
described in a nice way and therefore used to calculate the robustness against different
attacks and easily expanded with preprocessing steps, error performance and so on. The
security of the protocol is based on a \emph{quantum one-way function}\index{Quantum one-way function}.
Alice chooses a classical private key randomly and uses the mentioned function in order to
create as much public keys--which are given as qubit strings--as wanted (and proved to
be secure against certain classes of attacks). A party Bob, which wants to send a message to Alice, can apply for such a key and
perform a determined transformation on each qubit depending on the corresponding bit value
of the message and send it back. As Alice knows by the public key the initial states
of the qubits, she is able to extract the information Bob encrypted. We will sketch in this
section the basic properties of this protocol; a more substantial security analysis will be given in
Sections \ref{sec:privkey} and \ref{sec:message}.\par
The protocol starts by creating a classical private key which is the integer string
$\key = (k_1, \ldots, k_N)$ with $N \in \N^*$ and uses a security parameter $n \in \N^*$, \mbox{$n\gg 1$}.
The key generation part of the protocol is the following:
\begin{enumerate}
 \item Take a (public) random positive integer $n \gg 1$.
 \item Generate a (private) random integer string $\key = (k_1, \ldots, k_N)$ with $N \in \N^*$, where each
     integer is individually chosen randomly and independently from $\Z_{2^n}$.
 \item Prepare $T' \in \N^*$ copies of the (public) $N$-qubit \emph{public-key} state from the \emph{private key} $\textbf k$ as
   \begin{align}\label{eqn:georgios:genPK}
    \ket{\Psi_{\key}(\theta_n)} = \bigotimes_{j=1}^N \ket{\psi_{k_j} (\theta_n)},
   \end{align}
  with
   \begin{align}\label{eqn:georgios:PKsingle}
    \ket{\psi_{k_j} (\theta_n)} = \cos \left(\frac{k_j \theta_n}{2}\right) \ket{0_z} + \sin \left(\frac{k_j \theta_n}{2}\right) \ket{1_z},
   \end{align}
  where $\ket{0_z}$ and $\ket{1_z}$ refer to the eigenstates of the Pauli operator $\sigma_z$.\label{enum:georgios:genPK:public}
\end{enumerate}
An appropriate lower limit on the value of $n$ depends on the number of
provided public keys; it will be derived in Section \ref{sec:privkey}.
A copy of the public key
$\ket{\Psi_{\key}(\theta_n)}$ is simply given by the tensor product of $N$ individual public key qubits
$\ket{\psi_{k_j} (\theta_n)}$. Any such qubit can be represented as a \emph{Bloch vector}\index{Bloch vector}
(cf. Appendix \ref{eqn:pauli:blochvector})
\begin{align}\label{eqn:georgios:bloch}
 \bm R_j(\theta_n) = \cos ( k_j \theta_n) \sigma_z + \sin ( k_j \theta_n) \sigma_x,
\end{align}
with the Pauli operators $\sigma_z$ and $\sigma_x$. Therefore, each chosen bit of the private key $k_j$
results in a qubit that is in the state $\ket{\psi_{k_j} (\theta_n)}$, where the protocol defines the
angles to be equidistant and determined by 
\begin{align}
 \theta_n := \pi/2^{n-1}.
\end{align}
In other words, the public key is a representation of the private key, where each integer $k_j$ is represented by
the angle $\pi/2^{n-1}$ and implemented as a rotation of a qubit within the $z$-$x$-plane of the Bloch sphere as
can be seen in Figure \ref{fig:georgios:xzplane}. This mapping is also called a \emph{quantum one-way function},
where the probability is large to generate the quantum state from the private key and the probability is low 
to reconstruct the private key perfectly from a single copy of the quantum state of the public key, if the initial
transformation is unknown.

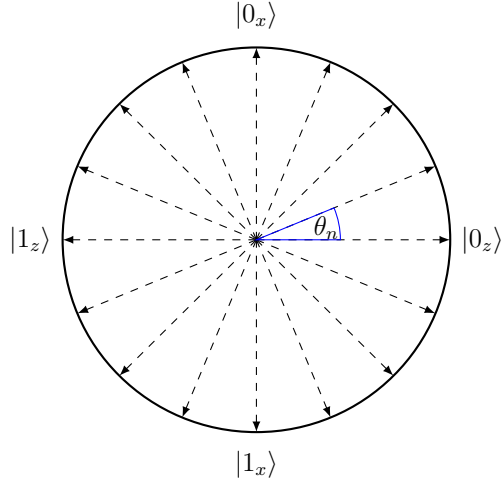
\begin{figure}[t]
  \centerline{
    \begin{tikzpicture}[scale=0.85, transform shape] 
      \draw [thick] (0,0) circle (3cm);
      \node at (0,3.5) {\ket{0_x}};
      \node at (0,-3.5) {\ket{1_x}};
      \node at (3.5,0) {\ket{0_z}};
      \node at (-3.5,0) {\ket{1_z}};
      \foreach \angle in {22.5, 45, ..., 360}
 	\draw[dashed, black, -latex] (0,0) -- (xyz polar cs:angle=\angle, radius=3);
      \draw[blue] (0,0) -- (13mm, 0mm) arc (0:22.5:13mm) -- cycle;
      \draw (10:11mm) node {$\theta_n$};
    %
    \end{tikzpicture}
  }

  \caption{For $n=4$, a single private key bit has one of $16$ different values, which is transformed into one of these $16$ equidistant
   states of the corresponding public-key qubit within the \mbox{$z$-$x$-plane} of the Bloch sphere.}
  \label{fig:georgios:xzplane}
\end{figure}

Alice sends the copies of the public key directly to parties like Bob that want to send encrypted messages back to Alice.
Alice does not delegate the distribution to a so-called
\emph{key distribution center} (KDC).\footnote{Such a KDC which
distributes public keys has to be \emph{unconditionally trusted}.}
Bob may use one of those keys in order to encrypt a message
$\bm{m}$ which has a length of at most $N$ bits. For simplicity and w.\,l.\,o.\,g. we consider a message which
is exactly $N$ bit long, thus $\bm{m} = (m_1, \ldots, m_N)$.\par
The encryption of the protocol is defined as
follows:
For each bit $m_j$ of the message with $j \in \MgE{N}$, apply the operator $\mathcal{E}_{m_j}$ to the qubit
$\ket{\psi_{k_j} (\theta_n)}$ of the public key, where this operator is given by
\begin{align}
 \mathcal{E}_{m_j} = \mathcal{R}_y (m_j \pi),
\end{align}
with $\mathcal{R}_y(\alpha) = \EZ^{-\iE \alpha/2 \sigma_y}$ being a rotation around the $\sigma_y$-axis of the
Bloch sphere. A rotation is therefore applied to the original public key state, if the corresponding bit value of
the message is one and the original state is left untouched if the value equals zero, resulting in the (quantum) cipher
state
\begin{align}\label{eqn:georgios:cipher}
 \ket{\bm X_{\bm k, \bm m} ( \theta_n)} = \bigotimes_{j=1}^N \mathcal{E}_{m_j} \ket{\psi_{k_j}(\theta_n)} = \bigotimes_{j=1}^N \ket{\chi_{k_j,m_j} (\theta_n)}.
\end{align}
In the last step, this cipher state is sent by Bob back to Alice who is able to undo the initial rotations and measure in
the Pauli-$\sigma_z$ basis. This part of the protocol reads as:
\begin{enumerate}
 \item Undo the initial rotations on the cipher state of Equation \eqref{eqn:georgios:cipher}:
 \begin{align}
  \ket{\bm M_{\bm m}} = \bigotimes_{j=1}^N \mathcal{R}_y^{-1} (k_j \theta_n) \ket{X_{k_j}(\theta_n)} = \bigotimes_{j=1}^N \ket{M_{m_j}}.
 \end{align}
 \item Measure each qubit of the resulting state in the eigenbasis of the Pauli-$\sigma_z$ operator to
       extract the sent message.
\end{enumerate}
Finally, Alice has received the classical message from Bob.\par
But conversely, if the cipher state is leaked to the
eavesdropper Eve, she may measure each qubit in a random basis and obtain already by this simple attack the correct
bit value of each message bit with a probability of $3/4$. To avoid similar attacks, Nikolopoulos introduced a second security
parameter
\begin{align}\label{eqn:georgios:s}
 s \in \N^*.
\end{align}
In this generalization, each message bit $m_j$ is encoded first into an $s$-bit codeword~$\bm w_j$ and subsequently
encrypted into $s$ qubits of the public key. Therefore, the public key needs to have at least $s$ times the length of
the message, again we concentrate for simplicity on the case that this ratio is given exactly. The codeword is chosen as follows:
If the message bit $m_j$ is zero, take a codeword $\bm w_j$ of length $s$ randomly from the set of codewords with even parity.
Appropriately, if the message bit $m_j$ equals zero, the codeword is randomly chosen from the set of codewords with odd parity.
In principle, the message is only extended by this procedure, so we can use the already defined steps of the protocol by applying the
replacement
\begin{align}\label{eqn:georgios:word}
 m_j \rightarrow \bm{w}_j.
\end{align}
Nikolopoulos derived an upper bound \cite{Niko08}, that limits the number of public keys Alice can distribute by keeping the
security of the private key against a powerful individual attack which is of interest in Section \ref{sec:message}.
Discussions on a \emph{chosen-plaintext} attack, a \emph{forward-search} attack (cf. also \cite{NI09}) and a
\emph{chosen-ciphertext} were given by Nikolopoulos \cite{Niko08}. A very general attack on the private key as well as a powerful and more
direct attack on the encoded message by using the information of all distributed keys will be discussed in Chapter \ref{chap:invest},
and leads to a limit for the number of public keys which may be distributed in order to achieve a certain security.

\chapter{Security analysis} \label{chap:invest}
For the QPKE protocol presented in Section \ref{sec:georgios}, different attacks were already
discussed by Nikolopolous and Ioannou \cite{Niko08, NI09}, but do not always benefit from the complete set
of available public keys. The aim of this chapter is to take advantage of the complete set and to
derive the strength of practical attacks. In Section \ref{sec:prelim}, preliminary considerations will
be discussed.  For an asymmetric protocol, it is clear that the
security of the private key is the most important goal in order to guarantee the secrecy of the
communicated messages, as if it is broken, all messages which use that key will be leaked. In
Section \ref{sec:privkey}, this security will be discussed. Direct attacks on the message will be investigated
subsequently in Section \ref{sec:message}, where it will be shown that a certain individual attack\footnote{Individual
attacks denote attacks where the transmitted quantum systems are measured individually. Collective attacks denote attacks where
information about the quantum systems is collected, e.\,g. by entangling them to ancilla systems and measuring them after the protocol
is finished and information of post-processing protocols can be taken into account.} possesses
an almost equal strength as collective attacks.\footnote{The ideas discussed in those sections were published
in 2012 \cite{SNA12}.}
Finally, in Section \ref{sec:noisyprepro}, a potential extension of the protocol will be considered
which uses a noisy-preprocessing step in order to reduce the information an eavesdropper may leak.
Obviously, this step influences considerably the setup of the scheme, but it is in principle a
generalization which takes partly the effect of noise into account.

\section{Preliminary considerations} \label{sec:prelim}
All values which refer to security issues of the protocol will be considered in the following in the case of a \emph{single run}
of the protocol, meaning the information of a single bit which is transmitted. For reasons of simplicity, we will
omit the index $j$ which refers to the specific bit, in the following. In some cases, averages of quantities
will be taken into account, but in principle a message has a finite length and ensuring the security of a single
bit is a stronger constraint.\par
To visualize this difference, let us compare the conditional probability of success of an eavesdropping strategy in both cases.
Within a single run of the protocol, values for the key $\key$ and the codeword $\word$ as defined in
Equations \eqref{eqn:georgios:genPK} and \eqref{eqn:georgios:word} are fixed, which lead to a mean probability of success%
\footnote{Probabilities which are denoted in the following in general by $P(\suc \vert X)$, refer to the probability for an attacker
to eavesdrop successfully, given the event $X$. The abbreviation ``\suc``, meaning ``success``, is chosen for a better readability.}
by averaging over all possible keys, namely
\begin{align}
 \bar P (\suc \vert \word) = \sum_{\key} P(\key) P(\suc \vert \key, \word) = \frac{1}{2^{nN}} \sum_{\key} P ( \suc \vert \key, \word),
\end{align}
where the keys are uniformly distributed over the set $\Mg{0,1}^{nN}$. Additionally, an optimal eavesdropping strategy has to be symmetric
with respect to the codewords, as the value for a one-bit message is given by $m \in \Mg{0,1}$ with equal probability
and the conditional probability of obtaining the codeword $\word$ is $P(\word \vert m) = 2^{-(s-1)}$ (cf. Equation \eqref{eqn:georgios:s}).
Thus the probability of any codeword to occur is given by
\begin{align}
 P(\word) = \sum_m P(\word \vert m) / 2 = 2^{-s}.
\end{align}

\section{Security of the private key}\label{sec:privkey}
As given by the third step of the key generation part of the protocol which is introduced in Section \ref{sec:georgios},
$T'$ copies of the public key exist. According to Eve's chosen strategy, she may use for practical reasons $\tau < T'$ copies of the
public key.\footnote{If Eve wants to attack a message, already one key is consumed for the encryption; if her aim is to attack the
key, she may in principle use all $T'$ copies of the public key, but this knowledge does not help her to gain any profitably information.}
Obviously, she can obtain those keys directly from Alice, as they are publicly available. Without measuring
the state of these copies, she holds \emph{a priori} a mixed quantum state, given by
\begin{align}
 \rho_{\mathrm{prior}}^{(\tau)} =& \frac{1}{2^n} \sum_{k'=0}^{2^n-1} \left( \ketbra{\psi_{k'} (\theta_n)}{\psi_{k'} (\theta_n)} \right)^{\otimes \tau}\\
   =& \frac{1}{2^n} \sum_{k'=0}^{2^n-1} \left( \ketbra{\bm{\phi}_{k'}^{(\tau)} (\theta_n)}{\bm{\phi}_{k'}^{(\tau)} (\theta_n)} \right),\label{eqn:prelim:rhoprior}
\end{align}
where we have abbreviated the $\tau$-qubit state as $\ket{\bm{\phi}_{k'}^{(\tau)} (\theta_n)} := \ket{\psi_{k'} (\theta_n)}^{\otimes \tau}$.
To bring Equation \eqref{eqn:prelim:rhoprior} into a nice form, we can take advantage of the following consideration:
A possible basis for a single qubit is given by the set $\Mg{ \ket{0_z}, \ket{1_z} }$ of states, which are eigenstates of the Pauli-$\sigma_z$
operator. If we treat a tensor product of $\tau$ single qubits, the basis vectors are consequently given by the set
\begin{align}\label{eqn:privkey:basisi}
 \Mg{\ket{i}: i \in \Z_{2^\tau}},
\end{align}
where $\ket{i}$ represents an eigenstate of the $\tau$-folded tensor product of Pauli-$\sigma_z$ operators $\sigma_z^{\otimes \tau}$. 
Eve knows from the construction of the protocol, that these $\tau$ copies of the qubit (which individually represent the public key) are equal. 
The state $\ket{\bm{\phi}_{k'}^{(\tau)} (\theta_n)}$ is therefore invariant under a permutation of the copies. Thus, a decomposition of that
state in the basis defined by Equation \eqref{eqn:privkey:basisi} will have equal parameters in those terms, where the number of ones in the
string $i$ coincides. The state $\ket{\bm{\phi}_{k'}^{(\tau)} (\theta_n)}$ can then be represented by the basis
\begin{align}
 \ket{l} = \left( \sum_{i=1}^{2^\tau} \delta_{l,H(i)} \ket{i} \right) /\sqrt{\binom{\tau}{l}},
\end{align}
where the \emph{Hamming weight}\index{Hamming weight} $H(i)$ is the number of ones in the string $i$. These
$\tau +1$ different states $\ket{l}$ (with $l \in \MgN{\tau}$) refer to subspaces of the $2^{\tau}$ dimensional Hilbert space.\par
We can use this basis in order to reformulate the state of Equation \eqref{eqn:prelim:rhoprior} with
\begin{align}
 \ket{\bm{\phi}_{k'}^{(\tau)} (\theta_n)} = \sum_{l=0}^{\tau} \sqrt{\binom{\tau}{l}} f_{\tau,l} (k' \theta_n) \ket{l},
\end{align}
where
\begin{align}
 f_{\tau,l} ( k' \theta_n) = \left( \cos \left(\frac{k' \theta_n}{2}\right) \right)^{\tau-l} \left( \sin \left(\frac{k' \theta_n}{2}\right) \right)^{l}.
\end{align}
Introducing the coefficients $C_{l,l'}$ with
\begin{align}
 C_{l,l'} = \frac{1}{2^n} \sqrt{\binom{\tau}{l}\binom{\tau}{l'}} \sum_{k'=0}^{2^n-1} f_{\tau,l}(k' \theta_n) f^*_{\tau,l'}(k' \theta_n),
\end{align}
leads finally to
\begin{align}\label{eqn:privkey:rhopriorll}
 \rho_{\mathrm{prior}}^{(\tau)} = \sum_{l,l'=0}^{\tau} C_{l,l'} \ketbra{l}{l'}.
\end{align}
We can use this form of the \emph{a priori} state Eve may hold, in order to obtain an upper bound on the amount of information about the private key,
Eve may extract from the state $\rho_{\mathrm{prior}}^{(\tau)}$. Therefore, we need the \emph{Holevo quantity}\index{Holevo quantity} $\chi$ which limits
the average amount of information $I_{\mathrm{av}}$ that can be extracted from an unknown state $\rho$. It is given by
\begin{align}
 \chi = S(\rho) - \sum_k p_k S(\rho_k) \geq I_{\mathrm{av}},
\end{align}
where $S(\rho)$ denotes the usual \emph{von-Neumann entropy} of a state $\rho$, where the system is in a mixed state with states $\rho_k$ and corresponding
probabilities $p_k$.\footnote{Discussions, derivations, and proofs on quantities of quantum information theory can be found in the book of
Nielsen and Chuang \cite{NC00}.} As seen in Equation~\eqref{eqn:privkey:rhopriorll},
the $\tau$-qubit state $\rho_{\mathrm{prior}}^{(\tau)}$ can be expressed in terms of a $\tau+1$ dimensional Hilbert space, therefore, its von-Neumann
entropy is again upper bounded as\footnote{Simulations suggest that the eigenvalues of Equation \eqref{eqn:privkey:rhopriorll} are given by $\lambda_i = 2^{-\tau} \binom{\tau}{i}$.
Then $\rho_{\mathrm{prior}}^{(\tau)}$ would be the entropy of a binomial distribution with mean $\tau/2$ and variance $\tau/4$ and upper bounded by the
Gaussian distribution; this would lead to a lower upper bound of $I_{\mathrm{av}}$,
namely, $I_{\mathrm{av}} \leq S(\rho_{\mathrm{prior}}^{(\tau)}) \leq \frac{1}{2} \log_2 \tau + \frac{1}{2}\log_2 (\pi \EZ /2)$.}
\begin{align}
 I_{\mathrm{av}} \leq \chi \leq S(\rho_{\mathrm{prior}}^{(\tau)}) \leq \log_2 (\tau+1).
\end{align}
By the construction of the private key and the resulting public key as given in Section~\ref{sec:georgios}, the state of the public key is in one
of $2^n$ different states, which means that this key is secure, as long as the information Eve may gain from her measurements is much smaller
than the number of bits which describe the information, namely, as long as
\begin{align}
 n \gg \log_2 (\tau+1)
\end{align}
holds. As in principle $T'$ copies of the public key exist and at least one is consumed to encrypt a message, this limit reads in terms of $T'=\tau+1$ as
\begin{align}\label{eqn:privkey:nlimit}
 n \gg \log_2 (T').
\end{align}
With the help of this limitation, the security of the private key can be guaranteed. Nevertheless, direct attacks on the message
are possible and will be discussed in Section~\ref{sec:message}.

\section{Security of a message}\label{sec:message}
The security of the sent message in the protocol of Nikolopoulos relies essentially on the
security of the private key that can be guaranteed according to the
limitations given in Equation \eqref{eqn:privkey:nlimit}. Thus, if the private key can be reconstructed
perfectly, a potential eavesdropper \emph{Eve} may use this information in order to obtain messages which are
sent with this key. But as the mentioned limitation should be achieved by construction, the optimal
strategy for Eve may be different. Instead of recovering the private key perfectly, she may recover the
key approximately. As the private key is in one-to-one correspondence to the public key, we can easily
discuss this issue in terms of the public key. Let us assume first that
Eve manages to receive $T'-1$ copies of the public key as well as an encrypted message which uses
the same key. Since the integers of the corresponding private key are chosen individually
and uniformly distributed, we are free to consider a single-bit message.\par
The strategy of Eve is
then to estimate the state of the public-key qubits, by taking into account that each qubit is prepared in
a state which lies on the $z$-$x$ plane of the Bloch sphere (cf. Appendix \ref{app:pauli:bloch}). The security parameter $s$, which was
introduced in Equation \eqref{eqn:georgios:s}, requires an $s$-bit codeword, which encodes
a bit of the message and is subsequently encrypted into $s$ qubits of the public key. At first, we will
consider Eve's attack on a single qubit, which encodes only a single bit of the codeword $\word$ and finally
generalize this approach on the complete ciphertext of the codeword. As already seen in Part \ref{part:mubs}
of this work, an optimal choice to estimate an unknown state with a minimal number of single-qubit measurements
is realized by using mutually unbiased bases\index{Mutually unbiased bases}. Taking the limitation to the $z$-$x$ plane of the
Bloch sphere into account, Eve's best measurement strategy is to use the eigenbases of the Pauli-$\sigma_z$ and
Pauli-$\sigma_x$ operators equally often. To simplify matters, we assume that the number of distributed
keys is given by
\begin{align}
 T' = 2 T + 1.
\end{align}
This results in an estimation of the public key and defines an estimation of the basis that was used to
encrypt the message. If Eve uses this basis estimation to measure the ciphertext, the fidelity of the
basis which results from the private key and the estimation of the bases gives the success probability of
Eve's attack.\par
Eve can apply the two different experiments (measuring in the eigenbasis of the Pauli-$\sigma_z$ or the
Pauli-$\sigma_x$ operator, respectively) and will receive either the measurement outcome zero or one.
This outcome depends in average on the expectation value which is given by testing $\ket{0_z}$ on the single-qubit state
$\ket{\psi_{k} (\theta_n)}$, that is introduced in Equation~\eqref{eqn:georgios:PKsingle}, namely
\begin{align}
 \braket{\psi_{k} (\theta_n)}{0_z}\braket{0_z}{\psi_{k} (\theta_n)},
\end{align}
which yields the probability
\begin{align}\label{eqn:message:p0z}
 p_0^{(z)}(k) = \cos^2 \left(\frac{k \theta_n}{2}\right),
\end{align}
to measure a zero in the eigenbasis of the $\sigma_z$ operator. To measure a zero in the eigenbasis of the
$\sigma_x$ operator, we find accordingly
\begin{align}\label{eqn:message:p0x}
 p_0^{(x)}(k) = \cos^2 \left(\frac{\pi}{4}-\frac{k \theta_n}{2}\right).
\end{align}
In both cases, the probabilities to measure a one are obviously given by $p_1^{(z)}(k) = 1 -p_0^{(z)}(k)$ and
$p_1^{(x)}(k) = 1 -p_0^{(x)}(k)$, respectively. If we denote the number of outcomes of a zero in the eigenbasis
of the $\sigma_z$ operator, by $T_0^{(z)}$ and the number of outcomes of a zero in the eigenbasis of the
$\sigma_x$ operator, by $T_0^{(x)}$, Eve's attack leads to one out of $(T+1)^2$ ordered pairs $\left(T_0^{(z)}, T_0^{(x)}\right)$.

\subsection{Information about the public key}
In a first step, this special attack yields information about the public key (and therefore about the private key).
The \emph{a posteriori} probability for Eve to guess a certain bit value $k'$ with the help of the
ordered pair $\left(T_0^{(z)}, T_0^{(x)}\right)$ is given by \emph{Bayes law} as
\begin{align}\label{eqn:message:aposteriori}
 p(k' \vert T_0^{(z)}, T_0^{(x)}) = \frac{p(T_0^{(z)}, T_0^{(x)}\vert k')}{2^n p(T_0^{(z)}, T_0^{(x)})}.
\end{align}
To calculate the \emph{a priori} probability of the right hand side of this expression, we can use the quantities
defined in Equation \eqref{eqn:message:p0z} and \eqref{eqn:message:p0x}, keeping in mind all possible permutations
of the measurement outcomes, which leads to
\begin{align}\label{eqn:message:apriori}
 p(T_0^{(z)}, T_0^{(x)}\vert k') = \prod_{b \in \Mg{z,x}} \binom{T}{T_0^{(b)}} \left( p_0^{(b)}(k') \right)^{T_0^{(b)}} \left( p_1^{(b)}(k') \right)^{T-T_0^{(b)}}.
\end{align}
The probability for Eve to obtain a certain ordered pair $\left(T_0^{(z)}, T_0^{(x)}\right)$ independently of the private-key state
is given by the sum of the \emph{a priory} probability over all possible values of $k'$, namely
\begin{align}
 p(T_0^{(z)}, T_0^{(x)}) = \frac{1}{2^n} \sum_{k'=0}^{2^n-1} p(T_0^{(z)}, T_0^{(x)}\vert k').
\end{align}
The information which Eve can accumulate in average by evaluating the discussed measurement outcome is given by
the difference of the Shannon entropies before and after the measurements, thus,
\begin{align}
 I_{\mathrm{av}} =& H_{\mathrm{prior}} - \langle H_{\mathrm{post}} \rangle\\
  =& n + \sum_{T_0^{(z)}, T_0^{(x)}} p(T_0^{(z)}, T_0^{(x)}) \sum_{k'=0}^{2^n-1} p(k' \vert T_0^{(z)}, T_0^{(x)}) \log_2 p(k' \vert T_0^{(z)}, T_0^{(x)}).
\end{align}
As the \emph{a priori} state is uniformly distributed and can be described by $n$ bits, the corresponding \emph{a priori}
entropy is given by $n$.

\subsection{Information about the message}
Alternatively, the measurement outcome $\left(T_0^{(z)}, T_0^{(x)}\right)$ can be used by Alice to estimate the public-key state.
With the help of the \emph{a posteriori} probability, that a certain value $k$ may have been the private key, as it
was introduced in Equation \eqref{eqn:message:aposteriori}, we can formulate the \emph{a posteriori} state which Eve has
after the measurement, taking all possible values for $k$ into account. This $\tau$-qubit state is given by
\begin{align}
 \rho_{\mathrm{post}}^{(\tau)}(T_0^{(z)}, T_0^{(x)}) = \sum_{k'=0}^{2^n-1} p(k' \vert T_0^{(z)}, T_0^{(x)}) \ketbra{\bm{\phi}_{k'}^{(\tau)} (\theta_n)}{\bm{\phi}_{k'}^{(\tau)} (\theta_n)}.
\end{align}
Although Eve intercepted $\tau$ raw copies of the public key, she is only interested in a single-qubit state, which can be
obtained by tracing out $\tau-1$ states and results in
\begin{align}
 \rho_{\mathrm{post}}^{(1)}(T_0^{(z)}, T_0^{(x)}) = \sum_{k'=0}^{2^n-1} p(k' \vert T_0^{(z)}, T_0^{(x)}) \ketbra{\psi_{k'} (\theta_n)}{\psi_{k'} (\theta_n)}.
\end{align}
Analogously to Equation \eqref{eqn:georgios:bloch} we can represent the single-qubit \emph{a posteriori} state by the Bloch vector (cf. Equation \ref{eqn:pauli:blochvector})
\begin{align}
 \tilde {\bm R} = \sum_{k'=0}^{2^n-1} p(k' \vert T_0^{(z)}, T_0^{(x)}) (\cos (k' \theta_n) \sigma_z + \sin (k' \theta_n) \sigma_x),
\end{align}
where this vector represents in general a mixed state and has therefore a length which is given by $\vert \tilde{\bm R} \vert \leq 1$.\par
In principle, we may compare this estimated Bloch vector now with the original Bloch vector and calculate the probability that both
coincide. However, we have to call back into mind, that a single bit is encoded by an $s$-bit codeword~$\word$ and subsequently encrypted
into $s$ qubits defining the security parameter that was formulated in Equation \eqref{eqn:georgios:s}. Following this construction, the first
step is to calculate the probability that Eve correctly estimates the bit $w_j$ of the codeword $\word = \Mg{1,\ldots, s}$, which is given
by
\begin{align}
 P(\suc \vert w_j,k, T_0^{(z)}, T_0^{(x)}) = \cos^2 (\Omega_j/2),
\end{align}
where $\Omega_j$ denotes the angle between the two Bloch vectors $\tilde{\bm R_j}$ and ${\bm R_j}$ and the index $j$ refers to
the qubit that encrypts the bit $w_j$. Solving this approach leads to the expression
\begin{align}
 P(\suc \vert w_j,k, T_0^{(z)}, T_0^{(x)}) = \frac{1}{2} + \frac{\tilde{\bm R_j} \cdot \bm R_j}{2 \vert \tilde{\bm R_j} \vert},
\end{align}
with the scalar product
\begin{align}
 \tilde{\bm R_j} \cdot \bm R_j = \sum_{k'=0}^{2^n-1} p(k' \vert T_0^{(z)}, T_0^{(x)}) \cos ((k' -k) \theta_n).
\end{align}
The ordered pair $\left(T_0^{(z)}, T_0^{(x)}\right)$ is an output of the measurements by Eve which is not observed by Alice. A more
meaningful quantity is therefore given by
\begin{figure}[t]
 \centering
 \def\svgwidth{230pt} 
 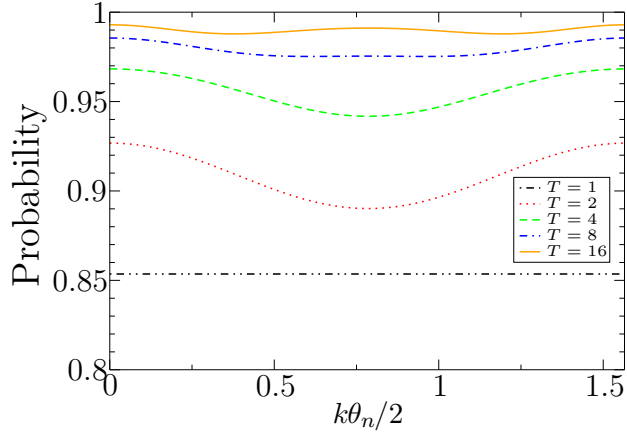
 \caption{Conditional probability $P(\suc\vert w_j,k)$ that Eve successfully attacks the public key for $n=10$ and different values of $T$.}
 \label{fig:message:psuc_wk}
\end{figure}
\begin{align}
 P(\suc \vert w_j, k) = \sum_{T_0^{(z)}, T_0^{(x)}} P(\suc, w_j,k, T_0^{(z)}, T_0^{(x)}) p(T_0^{(z)}, T_0^{(x)} \vert k)
\end{align}
and visualized in Figure \ref{fig:message:psuc_wk} for different values of $T$. Plotted against the key value $\key$, the
probability shows an oscillatory behavior around its mean value, with decreasing amplitude for a larger number of available copies.
The mean value can be calculated as
\begin{align}\label{eqn:message:psuc_wj}
 \bar P (\suc \vert w_j) = \frac{1}{2^n} \sum_{k=0}^{2^n-1} P(\suc \vert w_j, k)
\end{align}
and scales for $T>1$ like
\begin{align}\label{eqn:message:psuc_wj_bound}
 \bar P (\suc \vert w_j) \lesssim 1 - \frac{1}{6T},
\end{align}
which is given together with the optimal probability that is discussed below, in Figure \ref{fig:message:psuc}.
\begin{figure}[t!]
 \centering
 \def\svgwidth{230pt} 
 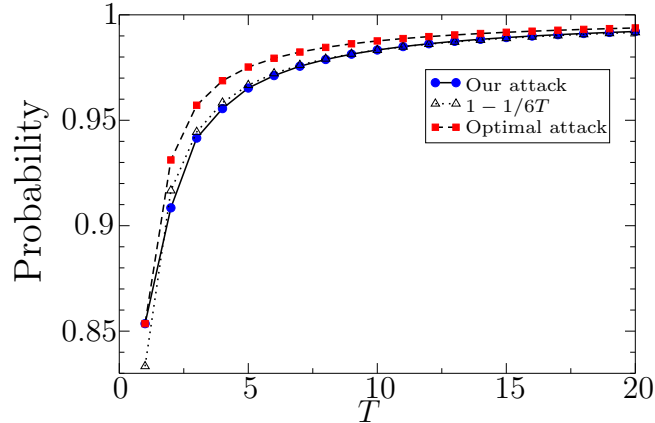
 \caption{Conditional probability $\bar P (\suc \vert w_j)$ that Eve successfully attacks the message for $n=10$ for the discussed attack,
 the asymptotic behavior of this attack and the optimal collective attack \cite{DBE98}, plotted against $T$.}
 \label{fig:message:psuc}
\end{figure}
The comparison with the \emph{optimal} probability of successfully estimating the prepared state which takes advantage of
a collective measurement method (and is therefore harder to get implemented) shows the good performance of this approach,
where this optimal probability is given by \cite{DBE98}
\begin{align}
 \bar P_{\mathrm{opt}} (\suc \vert w_j) = \frac{1}{2}+ \frac{1}{2^{2T+1}} \sum_{i=0}^{2T-1}\sqrt{\binom{2T}{i} \binom{2T}{i+1}},
\end{align}
and scales as
\begin{align}
 \bar P_{\mathrm{opt}} (\suc \vert w_j) \sim 1- \frac{1}{8T}.
\end{align}
Even this scaling can be reached by measurements on individual qubits, as it was shown by Bagan \etal{} \cite{Bagan89} with
an attack which is similar to the discussed one. Contrary to the attack given in \cite{NI09}, the probability of success does
not depend on the value $w_j$ and can be implemented using only measurement devices instead of the more complicated quantum
operations which are used in the mentioned approach.\par
Hence, we are well prepared to discuss the final step of this attack, namely considering the probability of correctly guessing
the message bit $m$. As seen in Section~\ref{sec:georgios}, the bit $m$ is encoded into the parity of an $s$-bit codeword $\bm w$,
where all appropriate codewords are equally probable. As seen in Equation \eqref{eqn:message:psuc_wj}, Eve has a certain probability
for each bit of the codeword to succeed. But even though, namely in cases where Eve fails an even number of times to predict the
correct value of a bit $w_j$ of the codeword, the parity she gets results in the same and therefore correct message bit $m$. Additionally,
as we have already seen in Figure \ref{fig:message:psuc_wk}, the amplitude of the oscillation of the success probability becomes
negligible for large values of $T$, at least it is an order of magnitude smaller than the mean. Therefore, we
will not gain a much better notion by observing the mean value of the success probability, as it should be comparable to the results
within a single-run experiment. Following the structure of the codeword, Eve has a mean probability of success, given a certain codeword
and a certain message bit,
\begin{align}\label{eqn:message:pssuc_mw}
 \bar P_s (\suc \vert m, \word) = \sum_{\substack{\alpha = 0,\\\mathrm{even}}}^s \binom{s}{\alpha} (1-\bar P(\suc \vert w_j))^{\alpha} (\bar P(\suc \vert w_j))^{s-\alpha},
\end{align}
where $\alpha$ denotes the number of incorrectly estimated bits of the codeword. Since this value does not depend on the message
bit and all codewords have the same probability, we finally find
\begin{align}
 \bar P_s (\suc) = \bar P_s (\suc \vert m, \word).
\end{align}
Within Figure \ref{fig:message:pssuc} we show the behavior of this probability for different values of $T$ against the security parameter~$s$.
\begin{figure}[t!]
 \centering
 \def\svgwidth{250pt} 
 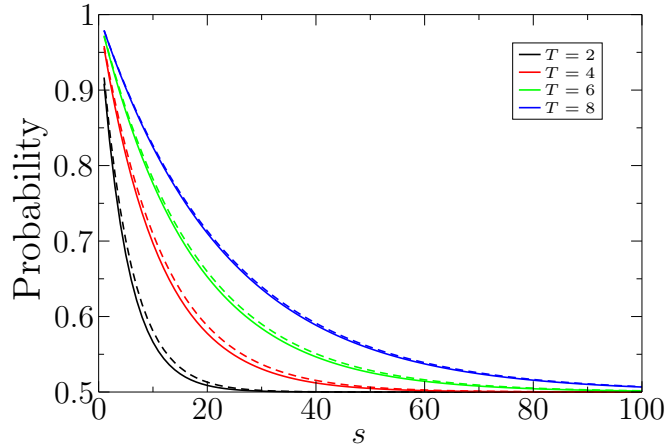
 \caption{Conditional probability $\bar P_s (\suc \vert m, \word)$ that Eve successfully attacks the message for $n=10$ and different values of $T$,
 plotted against the security parameter $s$. The solid lines refer to the exact results of Equation \eqref{eqn:message:pssuc_mw} and the dashed lines to the
 upper bound given by Equation \eqref{eqn:message:pssuc_mw_bound}.}
 \label{fig:message:pssuc}
\end{figure}
As expected, this parameter produces a monotonic decrease in the success probability of Eve and can be tightly upper bounded with the help
of the following lemma that can be used in order to generalize Equation \eqref{eqn:message:psuc_wj_bound}.\pagebreak
\begin{lem}[Success probabilities with increasing security parameter]\label{lem:message:genbound}\hfill\\
 If the probability (or an upper bound) of successfully eavesdropping a single-bit message encoded by a codeword of length $s=1$ is given by
 \begin{align}\label{eqn:message:q1}
  Q^{(1)} = \frac{1}{2} + \frac{\lambda}{2},
 \end{align}
 with $\lambda \in \C$, the corresponding quantity for a single-bit message, encoded by an $s$-bit codeword, such that
 the parity of the codeword equals the message bit and all possible codewords are equally probable, is given by
 \begin{align}\label{eqn:message:qs}
  Q^{(s)} = \frac{1}{2} + \frac{\lambda^s}{2}.
 \end{align}
\end{lem}
\begin{proof}
 As an even number of wrong results in the attack leads to a correct estimate of the parity of the codeword which represents the sent message, a
 generalization of Equation \eqref{eqn:message:q1} to a codeword of length $s$ is given by
 \begin{align}
  Q^{(s)} = \sum_{\substack{\alpha = 0,\\\mathrm{even}}}^s  \binom{s}{\alpha} ( 1- Q^{(1)})^{\alpha} ( Q^{(1)})^{s-\alpha}.
 \end{align}
 Alternatively, this generalization can also be written iteratively as
 \begin{align}
  Q^{(s)} =  Q^{(1)}Q^{(s-1)} + (1-Q^{(1)})(1-Q^{(s-1)}).
 \end{align}
 If for $Q^{(1)}$ Equation \eqref{eqn:message:q1} holds and we assume that also Equation \eqref{eqn:message:qs} is given, we find
 for $Q^{(s+1)}$, using the iterative construction,
 \begin{align}
  Q^{(s+1)} =& \left( \frac{1}{2} + \frac{1}{2} \right) \left( \frac{1}{2} + \frac{\lambda^s}{2} \right) + \left( \frac{1}{2} - \frac{1}{2} \right) \left( \frac{1}{2} - \frac{\lambda^s}{2} \right)\\
        =& \left( \frac{1}{2} + \frac{\lambda^{s+1}}{2} \right),
 \end{align}
 which completes the induction and provides the expected result.
\end{proof}
If we set $Q^{(1)}$ to the value of Equation \eqref{eqn:message:psuc_wj_bound}, we can use Lemma \ref{lem:message:genbound}
with $\lambda = 1- 1/(3T)$ and get finally a generalized upper bound for the correct guessing of the message\footnote{The message is encoded to the 
parity of an $s$-bit codeword as introduced in Section~\ref{sec:georgios}.} by
\begin{align}\label{eqn:message:pssuc_mw_bound}
 Q^{(s)} := \tilde P_s (\suc) = \frac{1}{2} + \frac{1}{2} \left( 1- \frac{1}{3T} \right)^s.
\end{align}

Finally, to achieve a security against the introduced attack, we have to show that for any small positive
value $\varepsilon \in (0,1/2]$, we find a finite value of the security parameter $s$, such that the probability
of successfully estimating goes below the sum of the probability of randomly guessing with that value.
This notion of security is also known as \emph{asymptotic security}%
\index{Asymptotic security}. Thus, to find a security parameter which fulfills $\bar P_s (\suc) < 1/2 + \varepsilon$,
we find, using Equation \eqref{eqn:message:pssuc_mw_bound},
\begin{align}
 s > \frac{1+\log_2 \varepsilon}{\log_2 \left(1- \frac{1}{3T} \right)},
\end{align}
which is also fulfilled if there holds
\begin{align}\label{eqn:message:slimitbound}
 s > 3T (1+\log_2 \varepsilon).
\end{align}
This easily implementable attack is stronger than the \emph{forward-search} attack\index{Forward-search attack} which was already discussed in the
original work of the protocol \cite{Niko08, NI09} and uses complicated quantum operations as \emph{Fourier transformations}
and permutations on large sets of qubits. Similarly to Equation \eqref{eqn:message:pssuc_mw_bound}, an upper bound
for its probability of successfully eavesdropping (which does not vary from run to run) is given by \cite{NI09}
\begin{align}\label{eqn:message:sFSlimit}
 \tilde P_s^{\mathrm{FS}} (\suc) = \frac{1}{2} + \frac{1}{2} \left( 1- \frac{1}{2T} \right)^s.
\end{align}
The success probability of that attack leads to a lower bound for the security parameter in the notion of Equation \eqref{eqn:message:slimitbound} as
\begin{align}\label{eqn:message:sFSlimitbound}
 s^{\mathrm{FS}} > 2T (1+\log_2 \varepsilon),
\end{align}
and differs only by a factor of $2/3$ with the new investigated attack.\footnote{The differences of Equations \eqref{eqn:message:sFSlimit} and
\eqref{eqn:message:sFSlimitbound} with their counterparts in \cite{NI09} are due to the fact that within the attack discussed here, $T$ is defined to refer
to the half of the number of keys which are used in order to attack the encrypted message. The missing factor of two in \cite[Equation (30)]{SNA12}, comparing
to Equation \eqref{eqn:message:sFSlimitbound}, is based on a misprint.}

\section{Noisy preprocessing}\label{sec:noisyprepro}
In order to increase the robustness of the protocol discussed in Section \ref{sec:georgios} against
various attacks like those investigated in Sections \ref{sec:privkey} and \ref{sec:message} and
potentially others, we show in the following exemplarily the behavior of a modified protocol against 
two simple attacks. Conversely to the original key-generation Step \ref{enum:georgios:genPK:public} of the protocol,
we assume that the states are not taken from the $z$-$x$-plane of the Bloch sphere. One may think of
an \emph{affine plane} which would imply states on a certain degree of latitude, areas on the surface which are
limited by two degrees of latitude, or even the complete surface of the sphere. Obviously, as we like
to leave the measurement process untouched, \emph{post-processing} (namely \emph{error correction}) steps will
be obligatory. Even within the original protocol they have to be added in order to get a practical implementation,
as errors in quantum states will always occur. Possibly, the introduced steps may help in order to model
the influence of errors in the original protocol. Nevertheless, as this consideration is based on two simple attacks,
the results can only be seen as limited, but indicate the positive behavior and motivate a subsequent
analysis.

\subsection{Single-key test attack}\label{sec:singletest}
The first attack we use in order to compare a notion of the security of the original protocol with the
security of the modified protocol is in principle a single-copy version of the attack, which was discussed
in Section \ref{sec:message} and published in the same work~\cite{SNA12}. If Eve intercepts a ciphertext and
a corresponding copy of the public key, she can attack the message bitwise by a simple experiment. For each
bit of the codeword $\word$ as introduced in Section \ref{sec:georgios}, she measures both corresponding
qubits in the same, but in principle arbitrary basis within the $z$-$x$-plane, which leads to a success probability similar to that of
Equation \eqref{eqn:message:p0z}. As rotations of the measurement axis around the $\sigma_y$ axis lead only to a constant term in the
argument of the $\cos^2$ and the possible states are arranged symmetrically with even probabilities, the success
probability is independent of the chosen axis for a random chosen public key bit, so we will take the $\sigma_z$ axis
and have
\begin{align}\label{eqn:prepro:p0}
 p_0(k) = \cos^2 \left(\frac{ k \theta_n}{2} \right),
\end{align}
to measure a value of zero on a qubit. As a codeword bit of one generates a flip on the Pauli-$\sigma_y$ axis, we have to
calculate the probability of measuring this flip--without knowing the correct basis--which is given by
\begin{multline}
 P^{(1)}(\suc \vert k, w_j)\\
 = \begin{cases}
  p^{(\mathrm{key})}_1(k) p^{(\mathrm{ct})}_0(k+2^{n-1}) + p^{(\mathrm{key})}_0(k) p^{(\mathrm{ct})}_1(k+2^{n-1})&\mathrm{for}\; w_j=1,\\
  p^{(\mathrm{key})}_0(k) p^{(\mathrm{ct})}_0(k) + p^{(\mathrm{key})}_1(k) p^{(\mathrm{ct})}_1(k)&\mathrm{for}\; w_j=0,
 \end{cases}
\end{multline}
where $p^{(\mathrm{key})}$ belongs to the measurement on the key and $p^{(\mathrm{ct})}$ to the measurement on the ciphertext.
Using Equation \eqref{eqn:prepro:p0}, we find
\begin{align}\label{eqn:prepro:psuc_k}
 P^{(1)}(\suc \vert k) = \sin^4\left(\frac{ k \theta_n}{2} \right) + \cos^4\left(\frac{ k \theta_n}{2} \right).
\end{align}
With the help of Lemma \ref{lem:message:genbound} we can generalize this expression to the parity of the codeword $\word$ and finally get
\begin{align}\label{eqn:prepro:pssuc_k}
 P^{(1)}_s(\suc \vert k) =& \frac{1}{2} + \frac{(1 + \cos(2 k \theta_n))^s}{2^{s+1}}\\
           =& \frac{1+\cos^{2s} (k \theta_n)}{2}.
\end{align}
For a mean probability by averaging over all possible private key states, we find with Equation \eqref{eqn:prepro:psuc_k} a probability of
\begin{align}\label{eqn:prepro:psuc1}
 \bar P^{(1)}(\suc) = \frac{1}{2^n} \sum_{k=1}^{2^n} P(\suc \vert k) = \frac{3}{4},
\end{align}
for Eve to measure the correct value of a single bit of the codeword $\word$ and with Equation \eqref{eqn:prepro:pssuc_k}
a probability of
\begin{align}
 \bar P_s^{(1)}(\suc) = \frac{1}{2} + \frac{(1/2)^s}{2},
\end{align}
to measure correctly the parity of the codeword.

\subsection{Double-key test attack}\label{sec:doubletest}
The second attack we discuss takes measurement results on two public keys and the corresponding ciphertext into account, so
we assume Eve holds those three qubit-strings. For each qubit triple that belongs to the bit of a codeword, she measures the first
key in the $\sigma_z$ basis as for the single-key test attack, but measures the second key in the orthogonal $\sigma_x$ basis.
This defines one out of two bases which is taken as the measurement basis for the ciphertext qubit.\par
In detail: If the measurements of $\sigma_z$ and $\sigma_x$ give the same output, that basis is chosen, which bisects the inner angle;
in the converse situation, a basis is chosen which bisects the outer angle. This attack is another formulation of the attack
discussed in Section \ref{sec:message} for $T=1$. For the probabilities of measuring the value zero or one from a copy of the key, we
can take the expressions of Equations \eqref{eqn:message:p0z} and \eqref{eqn:message:p0x}. The probability of getting the correct value of the
codeword bit by measuring the message within the resulting basis is given by the expectation value of this basis with the corresponding state
of the ciphertext. Following these ideas we find
\begin{align}\nonumber
 P^{(2)}(\suc \vert k,\theta_n) = p_0^{(z)}(k) p_0^{(x)}(k) \vert \braket{0_+}{\psi_k(\theta_n)} \vert^2 + p_1^{(z)}(k) p_1^{(x)}(k) \vert \braket{1_+}{\psi_k(\theta_n)} \vert^2\\
      + p_0^{(z)}(k) p_1^{(x)}(k) \vert \braket{0_-}{\psi_k(\theta_n)} \vert^2 + p_1^{(z)}(k) p_0^{(x)}(k) \vert \braket{1_-}{\psi_k(\theta_n)} \vert^2,
\label{eqn:prepr:p2suc}
\end{align}
where $\ket{0_+}$ and $\ket{1_+}$ refer to outcomes $0$ and $1$ for the basis which bisects the inner angle between $\ket{0_x}$ and $\ket{0_z}$ and
$\ket{0_-}$ and $\ket{1_-}$ to outcomes $0$ and $1$ for the basis which bisects the inner angle between $\ket{0_z}$ and $\ket{1_x}$. Figure \ref{fig:prepro:doublebases} illustrates this
structure.
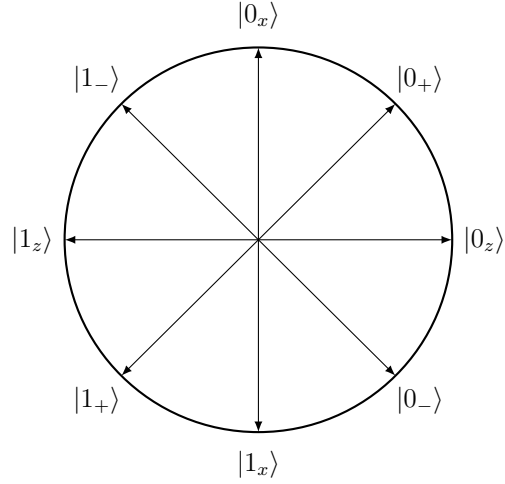
\begin{figure}[t]
  \centerline{
    \begin{tikzpicture}[scale=0.85, transform shape] 
      \draw [thick] (0,0) circle (3cm);
      \node at (0,3.5) {\ket{0_x}};
      \node at (0,-3.5) {\ket{1_x}};
      \node at (3.5,0) {\ket{0_z}};
      \node at (-3.5,0) {\ket{1_z}};
      \node at (2.5,2.5) {\ket{0_+}};
      \node at (-2.5,-2.5) {\ket{1_+}};
      \node at (2.5,-2.5) {\ket{0_-}};
      \node at (-2.5,2.5) {\ket{1_-}};
      \foreach \angle in {45, 90, ..., 360}
 	\draw[black, -latex] (0,0) -- (xyz polar cs:angle=\angle, radius=3);
    \end{tikzpicture}
  }
  \caption{The four bases which are defined in order to implement the double-key test attack.}
  \label{fig:prepro:doublebases}
\end{figure}
With
$\vert \braket{0_+}{\psi_k(\theta_n)} \vert^2 = \cos^2(\pi/8-k \theta_n/2)$ and $\vert \braket{0_-}{\psi_k(\theta_n)} \vert^2 = \cos^2(-\pi/8-k \theta_n/2)$,
we find
\begin{align}
 P^{(2)}(\suc) = p(\suc \vert k, \theta_n) = \frac{1}{4} (2+\sqrt{2}) \approx 0.85,
\end{align}
which does neither depend on $n$, nor on $k$, as the plot of Figure \ref{fig:message:psuc_wk} already indicated. For correctly estimating the
message bit, Eve needs to guess correctly the parity of the codeword $\word$, therefore we can take again advantage of Lemma \ref{lem:message:genbound}
and find for this attack
\begin{align}
 P_s^{(2)}(\suc) = \frac{1}{2} + \frac{(\sqrt{2}/2)^s}{2}.
\end{align}

\subsection{Protocol and security analysis}\label{sec:noisysecurity}
We extend the protocol which was discussed in Section \ref{sec:georgios} by a (private) displacement parameter $a\in [-1,1]$ which moves the plane
that is shown in Figure \ref{fig:georgios:xzplane} into the positive or negative $\sigma_y$ direction. The parameter is defined
to be one if the plane touches the surface only at $\ket{1_y}$ and equals minus one if it touches the opposite pole. It is linear on the $\sigma_z$ axis
in the representation of the Bloch sphere. Obviously, there is another representation of that parameter as an angle $\alpha \in [-\pi/2,\pi/2]$
with $a=\sin(\alpha)$, which is used in the following. The implementation into the protocol is achieved by generalizing the state $\ket{\psi_k(\theta_n)}$
which is defined in the key generation part in Step \ref{enum:georgios:genPK:public} as
\begin{align}
 \ket{\psi_k(\theta_n,\alpha)} =& \mathcal{R}_y(k\theta_n) \mathcal{R}_x(\alpha) \ket{0_z}\\
   =& \left(\cos \left(\frac{k \theta_n}{2}\right) \cos \frac{\alpha}{2} + \iE \sin \left(\frac{k \theta_n}{2}\right) \sin \frac{\alpha}{2}\right) \ket{0_z}\nonumber\\
   &+ \left(\sin \left(\frac{k \theta_n}{2}\right) \cos \frac{\alpha}{2} - \iE \cos \left(\frac{k \theta_n}{2}\right) \sin \frac{\alpha}{2}\right) \ket{1_z}.
\end{align}
The probability of Eve to correctly estimate the transmitted bit using the single-key test attack (which was discussed in Section \ref{sec:singletest}) is given by
\begin{align}
  \bar P^{(1)}_{\mathrm{E}}(\suc\vert\alpha) = \sum_{k=0}^{2^n-1}(\vert \braket{\xi}{\psi_k(\theta_n,\alpha)} \vert^2)^2 + (1-\vert \braket{\xi}{\psi_k(\theta_n,\alpha)} \vert^2)^2 = \frac{1}{8} (5+ \cos (2 \alpha))
\end{align}
for an arbitrary measurement basis in the $z$-$x$-plane, averaged over all $2^n$ possible states. For $\alpha=0$ this result coincides as expected with the
value derived in Equation~\eqref{eqn:prepro:psuc1}. As Alice does not know the transmitted message bit and the possible states are not within a single
plane, we assume for simplicity that she measures on an axis which is in the plane, where the states are defined for $\alpha=0$. This leads to a probability for Alice
to get the correct bit value, given by
\begin{align}
  \bar P_{\mathrm{A}}(\suc\vert\alpha) = \vert \bra{0_z} \mathcal{R}_x(\alpha) \ket{0_z} \vert^2 = \cos^2 \left(\frac{\alpha}{2}\right).
\end{align}
If we further assume, that an \emph{error correction}\index{Error correction} protocol exists which corrects as many bits as are distributed correctly
in average, the probability for Eve to find the correct bit value in the relevant cases is given by the relative success probability
\begin{align}\label{eqn:prepro:p1}
 \bar P^{(1)}(\suc\vert\alpha) = \frac{\bar P^{(1)}_{\mathrm{E}}(\suc\vert\alpha)}{\bar P_{\mathrm{A}}(\suc\vert\alpha)} = \frac{5+ \cos (2 \alpha)}{8 \cos^2 (\alpha/2)}.
\end{align}
Thus, to minimize the information Eve may gain by the single-key test attack for this modified protocol, we find
\begin{align}\label{eqn:prepro:a1min}
 \alpha^{(1)}_{\mathrm{min}} = \pm 2 \arccos \sqrt[\leftroot{-1}\uproot{2}\scriptstyle 4]{\frac{3}{4}},
\end{align}
which leads to
\begin{align}
 \bar P^{(1)}(\suc\vert\alpha^{(1)}_{\mathrm{min}}) = \sqrt{3}-1 \approx 0.732,
\end{align}
and is below the value $\bar P^{(1)}(\suc) = 3/4$ for $\alpha=0$.\par
Alternatively, we can use the double-key test attack (which was discussed in Section~\ref{sec:doubletest}) to check if the discussed tendency becomes stronger
or weaker in the case that Eve is able to catch more copies of the key. The probability for Eve to estimate the correct bit value can be calculated using
Equation \eqref{eqn:prepr:p2suc}, but replacing $\ket{\psi_k(\theta_n)}$ by $\ket{\psi_k(\theta_n,\alpha)}$ and deriving $p_v^{(b)}$ as
\begin{align}
 p_v^{(b)}(k) = \vert \braket{v_b}{\psi_k(\theta_n,\alpha)} \vert^2,
\end{align}
with $v \in \Mg{0,1}$ and $b \in \Mg{z,x}$, which leads to
\begin{align}
 \bar P^{(2)}_{\mathrm{E}}(\suc\vert\alpha) = \frac{1}{2} + \frac{\sqrt{2}}{4} \cos^2 \alpha.
\end{align}
As the probability for Alice to get the correct bit value does not depend on Eve's attack and therefore not the number of keys she uses for the attack, the
probability for Eve to get the correct bit value in the relevant cases is given by
\begin{align}\label{eqn:prepro:p2}
 \bar P^{(2)}(\suc\vert\alpha) = \frac{\bar P^{(2)}_{\mathrm{E}}(\suc\vert\alpha)}{\bar P_{\mathrm{A}}(\suc\vert\alpha)} = \frac{1 + (\cos^2 \alpha)/\sqrt{2}}{2 \cos^2 \frac{\alpha}{2}}.
\end{align}
Analogously to Equation \eqref{eqn:prepro:a1min} we can find the minimum of this probability in order to minimize the information gain for Alice by
\begin{align}
 \alpha^{(2)}_{\mathrm{min}} = \pm 2 \arccos \left( (\sqrt{2}-1)^{3/4} + (\sqrt{2}-1)^{3/4}/\sqrt{2} \right),
\end{align}
which leads to
\begin{align}
 \bar P^{(2)}(\suc\vert\alpha^{(2)}_{\mathrm{min}}) = \sqrt{2} \left(\sqrt{1 +\sqrt{2}} - 1\right) \approx 0.783,
\end{align}
and is considerably below Eve's success probability for $\alpha=0$ which is given by
\begin{align}\label{eqn:prepro:p2zero}
 \bar P^{(2)}(\suc\vert\alpha=0) \approx 0.854.
\end{align}
For a security parameter $s>1$, we can again use Lemma \ref{lem:message:genbound} in order to derive the success probability for Eve and
find 
\begin{align}
 \bar P^{(2)}_s(\suc\vert\alpha^{(2)}_{\mathrm{min}}) = \frac{1}{2} + \frac{1}{2} \left( \frac{5- 4 \cos^2 \frac{\alpha}{2} + \cos^2 \alpha}{4 \cos^2 \frac{\alpha}{2}} \right)^s,
\end{align}
which decreases faster than for $\alpha=0$ as can be seen in Figure \ref{fig:prepro:alphas}.%
\begin{figure}[t]
 \centering
 \def\svgwidth{220pt} 
 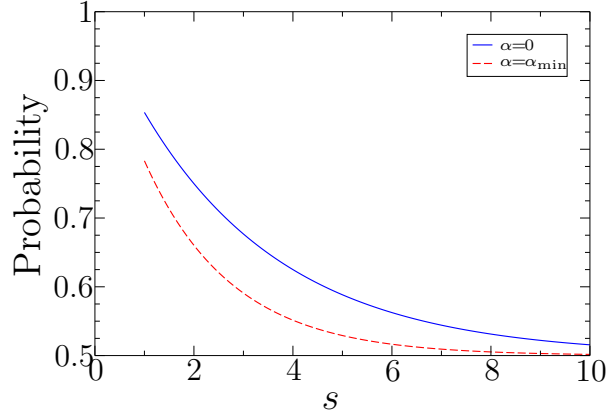
 \caption{Conditional probability $\bar P^{(2)}_s(\suc\vert\alpha)$ against the security parameter~$s$ that Eve successfully attacks the public key for $\alpha=0$ and $\alpha=\alpha_{\mathrm{min}}$,
  using the double-key test attack.\protect\footnotemark}
 \label{fig:prepro:alphas}
\end{figure}
Finally, we may imagine cases where $\alpha$ is not fixed, for example to test the robustness of the protocol against certain kinds of errors or simply to
adapt the protocol. We shortly calculate two situations: At first, $\alpha$ is a random parameter in the interval $[-\pi/2, \pi/2]$. In the second case, we
average over those values of $\alpha$ where Eve has a success probability that is not larger than for $\alpha=0$. In the first case, we find for the single-key test attack,
using Equation \eqref{eqn:prepro:p1},
\begin{align}
 \bar P^{(1)}(\suc) = \frac{1}{\pi} \int\limits_{-\pi/2}^{\pi/2} \frac{5+ \cos (2 \alpha)}{8 \cos^2 (\alpha/2)} \mathrm{d} \alpha \approx 0.773,
\end{align}
which is above the value of $3/4$ for $\alpha=0$, but for the double-key test attack, using Equation \eqref{eqn:prepro:p2}, we have
\begin{align}
 \bar P^{(2)}(\suc) = \frac{1}{\pi} \int\limits_{-\pi/2}^{\pi/2} \frac{1 + (\cos^2 \alpha)/\sqrt{2}}{2 \cos^2 \frac{\alpha}{2}} \mathrm{d} \alpha \approx 0.830,
\end{align}
which is slightly below the value for $\alpha=0$ which was given in Equation \eqref{eqn:prepro:p2zero}. If this tendency is monotonic, the protocol can in
principle be enhanced by robustness and security, if $\alpha$ is taken to be arbitrary.\footnotetext{Please keep in mind that $s$ is a discrete parameter,
the presentation is taken to get a better notion of the behavior of the probabilities.} 
\footnote{This statement holds as long as the number of public key copies is larger than two.}\par
For the second case we find
\begin{align}
 \bar P^{(1)}(\suc\vert \bar \alpha_{\mathrm{min}}^{(1)}) = \frac{3}{2\pi} \int\limits_{-\pi/3}^{\pi/3} \frac{5+ \cos (2 \alpha)}{8 \cos^2 (\alpha/2)} \mathrm{d} \alpha \approx 0.740,
\end{align}
which is below $3/4$ and
\begin{align}
 \bar P^{(2)}(\suc\vert \bar \alpha_{\mathrm{min}}^{(2)}) = \frac{1}{2 \alpha_{\mathrm{lim}}^{(2)}} \int\limits_{-\alpha_{\mathrm{lim}}^{(2)}}^{\alpha_{\mathrm{lim}}^{(2)}} \frac{1 + (\cos^2 \alpha)/\sqrt{2}}{2 \cos^2 \frac{\alpha}{2}} \mathrm{d} \alpha \approx 0.816,
\end{align}
which is below the value for $\alpha=0$ for the double-key test attack, with\linebreak $\alpha_{\mathrm{lim}}^{(2)} = 2 \arccos\left(\frac{\sqrt{1+\sqrt{2}}}{2}\right)$.\par
Concluding, we have seen in these two simple test attack scenarios, that a noisy-preprocessing in the introduced way is able to enhance the security of the
original protocol. Even if this additional parameter $\alpha$ is taken in a larger interval, the security against these test attacks increases.\footnote{For the
case of two copies of the public key, the test attack coincides by construction with the quite general attack discussed in Section \ref{sec:message}.} Finally,
the new parameter can be used as a first approach in order to model a protocol with errors in the preparation process of Alice, as Alice does not need the parameter~$\alpha$
for the considered measurement of the message.

\chapter{Conclusions and outlook}\label{chap:sumpk}
Quantum cryptographic primitives expand the field of classical cryptography and provide perspectives
which have the potential to influence existing schemes fundamentally. Whereas in classical cryptography
the security is based on \emph{computational complexity}, quantum cryptographic statements are based on
physically unsolvable problems. This leads to complete \emph{provably secure} cryptographic schemes.
The emerging field of quantum public-key encryption offers new ways in order to benefit from the
fundamental characteristics of quantum physics for cryptographic schemes.\par
Within the second part of this work, the recently invented
\emph{single-qubit-rotation protocol} is introduced and the \emph{security of the private key}
as well as the \emph{security of a message} which is transmitted with the help of the protocol are
discussed. A powerful and practical attack is used to consider the security of messages. Even though
it attacks individual qubits, its asymptotic behavior is comparable with collective attacks.
Finally, a possible extension of the protocol that implements a kind of \emph{noisy preprocessing} is
discussed and its advantage against simple test attacks is considered. The results indicate a positive effect
of this method also for more general attacks and its ability to model noise.\par
As a next step, a practical version of the protocol could be considered which contains \emph{error correction}
methods, an exact formulation of noise effects as well as preprocessing steps. Further analytical expressions
of important quantities would be useful to generalize Eve's success probabilities.

\appendix
\addcontentsline{toc}{part}{Appendix}

\chapter{Algebra and quantum information}\label{app:algebra}

\section{Fundamentals}
Within this section we arrange some known results from algebra that are relevant
for this work and were done for \cite{KRS10} and \cite{SR12}.

\begin{lem}[Determinant decomposition]\label{lem:algfund:detdecomp}\hfill\\
 Let $R$ be a ring. Then there holds
 \begin{align}
  \det \begin{pmatrix}A & B\\C & D\end{pmatrix} = \det(D) \det ( A- B D^{-1} C )
 \end{align}
 if $D$ is invertible, for $A, B, C, D \in M_n(R)$.
\end{lem}
\begin{proof}
 If $D$ is invertible, we are free to decompose the matrix as
 \begin{align}
  \begin{pmatrix}A & B\\C & D\end{pmatrix} = \begin{pmatrix}\Eins_n & B\\0_n & D\end{pmatrix} \begin{pmatrix}A-BD^{-1} C & 0_n\\D^{-1} C & \Eins_n\end{pmatrix}.
 \end{align}
 Since the determinant of block-triangular matrices is the product of the determinants of the blocks on the diagonal, the statement follows.
\end{proof}

\begin{lem}[Block matrix invertibility]\label{lem:algfund:detblock}\hfill\\
 Let $R$ be a commutative ring and $A, B, C, D \in R$ be commuting elements. Then the block matrix
 $\begin{pmatrix}A&B\\C&D\end{pmatrix}$ is invertible, if and only if $AD-BC$ is invertible.
\end{lem}
A proof is given by Silvester \cite{Silvester00}.
\pagebreak

\section{Finite fields}
A field with a finite number of elements is called a \emph{finite field} or \emph{Galois field}. It turns out that a finite
field exists only if the number of elements equals the power of a prime, which is called the \emph{order} of the field.
As for a given order, the corresponding finite field is unique up to isomorphisms, it can be denoted as
\begin{align}
 \F_q = \F_{p^m},
\end{align}
where $\F_{p^m}$ is a field with $p^m$ elements, thus with order $p^m$. The prime number~$p$
is called the \emph{characteristic} of the field and $m$ is a natural number. For $m=1$, the field is called
\emph{prime field}. An important linear transformation from $\F_{q^n}$ to its corresponding prime field, namely
$\F_{q}$, with $q=p^m$ and $n,m \in \N^*$, is the field trace $\tr_{L/K} (\alpha)$ which maps an element of a field $L$ to a field $K$, where
$K$ has not to be a prime field, but the order of $K$ has to be a divisor of $L$.

\begin{defi}[Field trace]\label{defi:app:algfield:trace}\hfill\\
 The trace of an element $\alpha \in L:=\F_{q^n}$ to a certain ground field $K:=\F_{q}$ with $q=p^m$, $p$
 prime and $m,n \in \N^*$, is given by
 \begin{align}
  \tr_{L/K} (\alpha) = \alpha + \alpha^{q} + \alpha^{q^2} + \ldots + \alpha^{q^{n-1}}.
 \end{align}
\end{defi}

\section{Hadamard matrix}\label{app:hadamard}
Hadamard matrices appear often in considerations of the quantum information theory. They also have applications
in coding theory.
\begin{defi}[Hadamard matrix]\label{defi:app:hadamard:hadamard}\hfill\\
 Matrices with entries $\pm 1$ and whose rows are pairwise orthogonal, are called \emph{Hadamard matrices}\index{Hadamard matrix}.
\end{defi}
It then follows obviously by this definition, that
the product of a Hadamard matrix with its transpose yields $n$ times the unity matrix, where $n$ is the
dimension. In quantum information theory, the most often used Hadamard matrix is defined for two-level systems
and given by
\begin{align}\label{eqn:app:hadamard:barH}
 \bar H := \begin{pmatrix} 1 & 1\\ 1 & -1\end{pmatrix}.
\end{align}
In many cases, a normalized version, which is then also the representation of a unitary operator, reads as
\begin{align}\label{eqn:app:hadamard:H}
 H := \frac{1}{\sqrt{2}}\begin{pmatrix} 1 & 1\\ 1 & -1\end{pmatrix},
\end{align}
thus, it is an orthogonal matrix.
Whereas different construction schemes for Hadamard matrices are known and their existence is only conjectured
for dimensions which are multiples of four\footnote{The existence of Hadamard matrices for the dimension one and two
seems to be an exception.}, we will only use the \emph{Sylvester construction} which is given as follows. Starting
with $\bar H_1 = (1)$ and $\bar H_2 = \bar H$, any matrix $\bar H_{2^m}$ with $m>1$ and $m \in \N^*$ is given by
\begin{align}\label{eqn:app:hadamard:sylvester}
 \bar H_{2^m} = \begin{pmatrix} \bar H_{2^{m-1}} & \bar H_{2^{m-1}}\\\bar H_{2^{m-1}} &-\bar H_{2^{m-1}}\end{pmatrix},
\end{align}
and equals the tensor product $\bar H_{2^m} = \bar H_2 \otimes \bar H_{2^{m-1}}$. The matrix $\bar H_{2^m}$ can then be
derived directly with the help of the $m$-folded tensor product of $\bar H$ as
\begin{align}\label{eqn:app:hadamard:mfolded}
 \bar H_{2^m} = \bar H^{\otimes m}.
\end{align}
Another representation of this matrix considers its entries and it holds
\begin{align}\label{eqn:app:hadamard:-1rep}
 (\bar H_{2^m})_{i,j} = (-1)^{i\cdot j},
\end{align}
with $i,j \in \F_{2^m}$, which can be read off directly from the construction of Sylvester given in Equation \eqref{eqn:app:hadamard:sylvester}.
The same arguments hold analogously for the normalized version, where an additional factor of $2^{-m/2}$ appears for the matrix $H_{2^m}$.

\section{Pauli operators}\label{app:pauli}
For a complex Hilbert space of dimension two, the three Pauli operators $\sigma_x, \sigma_y$, and $\sigma_z$ (or $X, Y$, and $Z$) are
defined as
\begin{align}\label{eqn:pauli:matrices}
 \sigma_x = \begin{pmatrix}0&1\\1&0\end{pmatrix},\quad \sigma_y = \begin{pmatrix}0&-\iE\\\iE&0\end{pmatrix},\quad\mathrm{and}\quad \sigma_z = \begin{pmatrix}1&0\\0&-1\end{pmatrix}.
\end{align}
For $x \rightarrow 1$, $y \rightarrow 2$, and $z \rightarrow 3$ it holds
\begin{align}
 \sigma_i \sigma_j = \delta_{ij} \Eins_2 + \iE \sum_{k=1}^3 \epsilon_{ijk} \sigma_k,
\end{align}
for $i,j \in \Mg{1,2,3}$ and with $\epsilon_{ijk}$ denoting the fully antisymmetric operator, also
known as \emph{Levi-Civita} symbol. It can be checked easily, that the eigenbases of the Pauli
operators are mutually unbiased.\par
For a $d$-dimensional complex Hilbert space with $d \in \N^*$, the generators of the \emph{generalized Pauli operators}\index{Generalized Pauli operators}
are defined by
\begin{align}
 Z \ket{i} = \omega^i \ket{i} \quad\mathrm{and}\quad X \ket{i} = \ket{i \oplus_d 1},
\end{align}
with $\omega = \exp{(2 \pi \iE /d)}$ being the first $d$-th root of unity. The set of Pauli operators which is generated by $Z$ and $X$ can be used
as a basis in order to represent a linear operator of the complex Hilbert space $\cH = \C^d$.
\begin{lem}[Orthogonality of generalized Pauli operators]\label{lem:pauli:orthogonal}\hfill\\
 The set of generalized Pauli operators is an orthogonal operator basis in the sense of Hilbert-Schmidt.
\end{lem}
\begin{proof}
 Using the Hilbert-Schmidt inner product, 
 \begin{align}
  \braket{Z^{\alpha}X^{\beta}}{Z^{\gamma}X^{\delta}}_{\mathrm{HS}} = \tr(Z^{\alpha-\gamma}X^{\beta-\delta}) = 
  \begin{cases}
   1\quad \mathrm{for}\;\alpha-\gamma=\beta-\delta=0,\\
   0\quad \mathrm{else},
  \end{cases}
 \end{align}
 holds for $\alpha,\beta, \gamma, \delta \in \MgE{d}$. Thus, the $d^2$ operators $Z^{\alpha}X^{\beta}$ form an orthogonal operator basis.
\end{proof}

\section{Bloch sphere representation}\label{app:pauli:bloch}
Each state of a two-dimensional quantum system can be represented as a \emph{Bloch vector}\index{Bloch vector} by three real parameters in the \emph{Bloch sphere},
where the axes are defined by the three Pauli matrices of Equation \eqref{eqn:pauli:matrices}. A pure quantum state which is written in the eigenbasis
of the Pauli-$\sigma_z$ operator, e.\,g.
\begin{align}
 \ket{\Psi} = \cos \left(\frac{\alpha}{2}\right)\ket{0_z} + \EZ^{\iE \varphi} \sin \left(\frac{\alpha}{2}\right) \ket{1_z}
\end{align}
leads to a density operator
\begin{align}
 \rho =& \ketbra{\Psi}{\Psi} = \frac{1}{2} \begin{pmatrix} 1+\cos \alpha &  \cos \varphi \sin \alpha - \iE \sin \varphi \sin \alpha \\  \cos \varphi \sin \alpha + \iE \sin \varphi \sin \alpha & 1-\cos \alpha \end{pmatrix}\\
     =& \frac{1}{2} \left( \Eins_2 + (\cos \varphi \sin \alpha) \sigma_x + (\sin \varphi \sin \alpha) \sigma_y + (\cos \alpha) \sigma_z \right), 
\end{align}
where basic trigonometric identities were used. The Bloch vector is given by
\begin{align}\label{eqn:pauli:blochvector}
 \bm R = (\cos \varphi \sin \alpha) \sigma_x + (\sin \varphi \sin \alpha) \sigma_y + (\cos \alpha) \sigma_z,
\end{align}
where the length of this vector equals one. Thus, all pure states are placed on the surface of the sphere.
As mixed states are given by $\sum_i p_i \rho_i$ with \mbox{$\sum_i p_i = 1$}, classical probabilities $p_i \in \R$ and
pure states $\rho_i$ with $i \in \N$, they lead to a sum of vectors in the Bloch sphere. Obviously, the resulting
vector will have a length smaller than one, therefore being placed inside the sphere. The zero-vector represents
then the \emph{completely mixed state}. Figure \ref{fig:pauli:bloch} illustrates the Bloch sphere, where the
extremal points on the axes are usually labeled by the eigenstates of the corresponding basis vectors. A similar
representation in higher dimensions is in principle not possible.
\begin{figure}[t]
 \centering
 \includegraphics[scale=0.20]{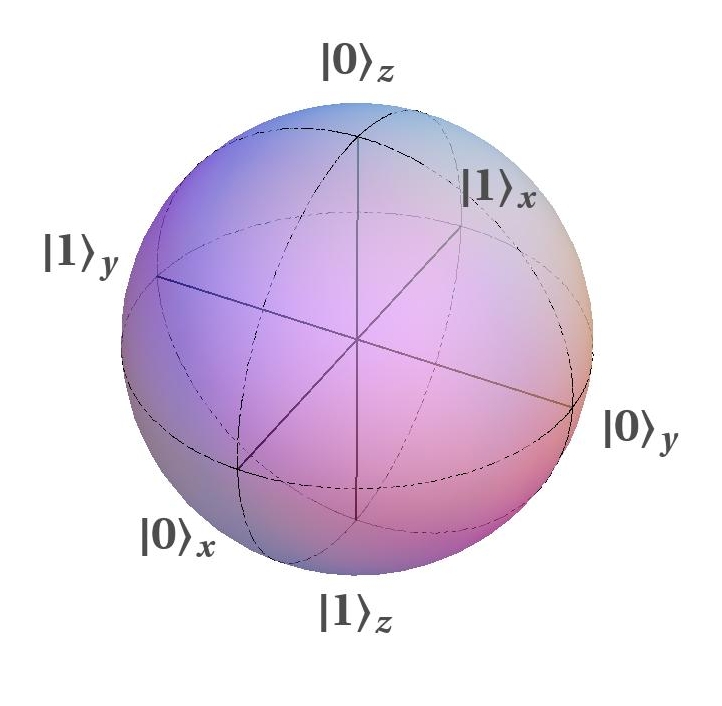}
 \caption{Illustration of a Bloch sphere, where a vector, the Bloch vector, starting at the origin of the sphere and pointing inside
 or on the surface, is capable to represent any state of a two-dimensional quantum system (qubit). The axes refer to the three Pauli
 operators, where their points on the surface are given by the corresponding eigenstates.}
 \label{fig:pauli:bloch}
\end{figure}

\section{Clifford group}\label{app:clifford}
Unitary matrices which map the set of Pauli operators onto the set of Pauli operators are called
\emph{Clifford unitary operators}\index{Clifford unitary operator} and can be represented as a symplectic matrix $A \in M_{2m}(\F_p)$
with $m \in \N^*$ describing the number of qudits, where the dimension of the Hilbert space is given by
$d=p^m$ with $p$ prime; the positive integer $p$ defines the Hilbert space dimension of a single qudit.
Within this section we give some properties of the symplectic vector space.

\begin{defi}[Symplectic matrix]\label{def:app:clifford:symplecticmatrix}\hfill\\
 A matrix $A \in M_{2m} ( \F_p )$ is called \emph{symplectic} \index{Symplectic}
 if there holds
 \begin{align}
   A^t S A = S \quad \mathrm{for} \quad S = \begin{pmatrix} 0_m & \Eins_m\\ - \Eins_m & 0_m\end{pmatrix}.
 \end{align}
\end{defi}

\begin{cor}[Symplectic matrix properties]\label{cor:app:clifford:symplecticmatrixprop}\hfill\\
 The matrix $A \in M_{2m} ( \F_p )$ is symplectic, if and only if there holds 
 \begin{align}
  s^t u = u^t s,\quad t^t v = v^t t,\quad\mathrm{and}\quad s^t v - u^t t =& \Eins_m\quad \mathrm{for}\quad A=
 \begin{pmatrix}
  s&t\\u&v
 \end{pmatrix},
 \end{align}
 for $s,t,u,v \in M_m(\F_p)$.
\end{cor}
\begin{proof}
 The three properties follow by simply applying Definition \ref{def:app:clifford:symplecticmatrix}
 to the matrix $A$.
\end{proof}

\begin{defi}[Symplectic product]\label{def:app:clifford:symplecticproduct}\hfill\\
  For two vectors $\vec a, \vec b \in \F_p^{2m}$, 
  the \emph{symplectic product}\index{Symplectic product} is defined as
  \begin{align}\label{eqn:app:clifford:symprod}
   (\vec a, \vec b)_{\mathrm{sp}} := \sum_{k=1}^m a_k^z b_k^x - a_k^x b_k^z \quad \mathrm{mod}\; p,
  \end{align}
  where $a^z$ refers to the first half of the vector entries and $a^x$ to the second half.
\end{defi}

\begin{cor}[Bilinearity of symplectic product]\label{cor:app:clifford:addsymp}\hfill\\
  For the symplectic product holds
  \begin{align}
    (\vec a + \vec b,  \vec c + \vec d)_{\mathrm{sp}} = (\vec a, \vec c)_{\mathrm{sp}} + (\vec a, \vec d)_{\mathrm{sp}} + (\vec b, \vec c)_{\mathrm{sp}} + (\vec b, \vec d)_{\mathrm{sp}}
  \end{align}
  with $\vec a, \vec b, \vec c$ and $\vec d \in K$.
\end{cor}
\begin{proof}
  The proof follows simply by using the definition of the symplectic product that is given by Equation \eqref{eqn:app:clifford:symprod}.
\end{proof}

\begin{defi}[Symplectic standard basis]\label{def:app:clifford:symplecticbasis}\hfill\\
  The vectors $\vec z_k, \vec x_k \in \F_p^{2m}$ with $k \in \MgE{m}$ form a \emph{symplectic basis}\index{Symplectic basis} of $\F_p^{2m}$,
  if there holds $(\vec z_k,\vec z_l)_{\mathrm{sp}} = (\vec x_k,\vec x_l)_{\mathrm{sp}} = 0$ 
  and $(\vec z_k,\vec x_l)_{\mathrm{sp}} = \delta_{kl}$ for $k,l \in \MgE{m}$.
  For the standard basis the vectors $\vec z_k$ are defined to have a one at position $k$ and zeros else, the vectors
  $\vec x_k$ have accordingly a one at position $m+k$ and zeros else.
\end{defi}

\begin{lem}[Symplectic basis]\label{lem:app:clifford:symplecticbasis}\hfill\\
  The images $\vec z'_k = C \vec z_k$ and $\vec x'_k = C \vec x_k$ of the symplectic standard basis with $k \in \MgE{m}$ and the matrix
  $C \in M_{2m}(\F_p)$ form again a symplectic basis, if and only if $C$ is a symplectic matrix in the sense of Definition \ref{def:app:clifford:symplecticmatrix}.
\end{lem}
\begin{proof}
  For a symplectic standard basis there holds $(\vec z_k,\vec z_l)_{\mathrm{sp}} = 0$ for all\linebreak $k,l \in \MgE{m}$. The transformation of these properties
  with a matrix
\begin{align}
 C=\begin{pmatrix} s&t\\u&v \end{pmatrix}
\end{align}
  may cause a new symplectic basis, if there holds $(C \vec z_k,C \vec z_l)_{\mathrm{sp}} = 0$.
  With $\vec \nu_k$ being the $k$-th row vector of $\nu \in \Mg{s,t,u,v}$ and $\vec z^z$ being the $z$-part of the vector $\vec z$ and $\vec x^x$ being
  the $x$-part of the vector $\vec x$, we find
  \begin{align}
   ( C \vec z_k, C \vec z_l )_\mathrm{sp} =& \sum_{i=1}^m (\vec s_i \cdot \vec z^z_k) (\vec u_i \cdot \vec z^z_l) - (\vec u_i \cdot \vec z^z_k) (\vec s_i \cdot \vec z^z_l)\\
                                          =& (\vec z^z_k)^t \cdot s^t \cdot u \cdot \vec z^z_l - (\vec z^z_k)^t \cdot u^t \cdot s \cdot \vec z^z_l.
  \end{align}
  Since this equation has to be zero for all $k,l \in \MgE{m}$, the expression \mbox{$s^t u = u^t s$} has to be fulfilled. Similarly, from 
  $(\vec x_k,\vec x_l)_{\mathrm{sp}} = 0$ follows that $t^t v = v^t t$. From $(\vec z_k,\vec x_l)_{\mathrm{sp}} = \delta_{kl}$ we get
  \begin{align}
   ( C \vec z_k, C \vec x_l )_\mathrm{sp} =& \sum_{i=1}^m (\vec s_i \cdot \vec z^z_k) (\vec v_i \cdot \vec x^x_l) - (\vec u_i \cdot \vec z^z_k) (\vec t_i \cdot \vec x^x_l)\\
                                          =& (\vec z^z_k)^t \cdot s^t \cdot v \cdot \vec x^x_l - (\vec z^z_k)^t \cdot u^t \cdot t \cdot \vec x^x_l.
  \end{align}
  In order to set this equation equal to $\delta_{ij}$, the expression $s^t v - u^t t = \Eins_m$ needs to be true. But by Corollary \ref{cor:app:clifford:symplecticmatrixprop}
  these three conditions hold if and only if $C$ is a symplectic matrix.
\end{proof}

\chapter{Wiedemann's conjecture proof approach}\label{app:wiedemann_proofs}
Within this chapter, an approach is given that seems to be a promising candidate for proving Wiedemann's conjecture of \cite{Wiedemann88}%
, where the proven analogy of Theorem~\ref{thm:ferset:wiedemann} is the origin of. The idea is based on the
construction of the unitary operator $U$ as it was given in Section \ref{sec:Uconstruction}.


\section{Unitary operator approach}
By Theorem \ref{thm:ferset:wiedemann} it was proven, that Wiedemann's conjecture is true, if and only if the recursive construction
of Equation \eqref{eqn:ferset:recursion} leads to a complete set of cyclic MUBs, which is given when the multiplicative
order of the stabilizer matrix $C_{2^k}$ of Equation~\eqref{eqn:ferset:C2k} equals $2^{2^k}+1$ for all $k \in \N$. As
the recursive construction of the unitary generator $U_{2^k} = V_{2^k}/(-\tr V_{2^k})$ that is given by Equation
\eqref{eqn:unitop:recursion}, is a representation of the stabilizer matrix (cf. Equation \eqref{eqn:unitop:Vm}), Wiedemann's
conjecture is analogously proven if it is shown that the multiplicative order of $U_{2^k}$ equals $2^{2^k}+1$. We
start with the following corollary that follows from Conjecture \ref{con:unitor:spectrum}:

\begin{cor}[Characteristic polynomial of Fermat sets]\label{cor:wiedproof:charpolfermat}\hfill\\ 
 If the characteristic polynomial of the unitary operator $U_{2^k} = V_{2^k}/(-\tr V_{2^k})$ with $V_m$ as given
 in Equation \eqref{eqn:unitop:Vmij}, equals $\sum_{l=0}^d x^l$ with $d=2^{2^k}$ for all $k \in \N$, $U_{2^k}$
 generates a complete set of cyclic MUBs and Wiedemann's conjecture is true. This would also prove Conjecture
 \ref{con:unitor:spectrum} in the case of Fermat sets.
\end{cor}
\begin{proof}
 If the multiplicative order of $U_{2^k}$ is $d+1$, its eigenvalues have to be roots of unity of order $d+1$,
 at least one of the eigenvalues has to be a principal root. The converse of Conjecture \ref{con:unitor:spectrum}
 is always true, namely if $U_{2^k}$ has $d$ different eigenvalues, it has order $d+1$. Thus, if we can show that
 $U_{2^k}$ has $d$ different eigenvalues of order $d+1$, it generates a complete set of cyclic MUBs. If we regard
 the characteristic polynomial $\chi_{U_{2^k}}$ of the unitary generator matrix with $d$ different eigenvalues of order $d+1$
 and exclude, as allowed, the eigenvalue $1$ from the set, we get
\begin{align}\label{eqn:wiedproof:prod}
 \chi_{U_{2^k}}(x) = \prod_{l=1}^d ( x - \EZ^{2 \pi \iE l / (d+1)} ).
\end{align}
 If we multiply Equation \eqref{eqn:wiedproof:prod} by $x-1$, the result is equal to $x^{d+1} - 1$, which, divided
 again by $x-1$, equals $\sum_{l=0}^d x^l$ as expected. 
\end{proof}
So we can prove Wiedemann's conjecture by calculating the characteristic polynomial of $U_{2^k}$, which is given by
\begin{align}
 \chi_{U_{2^k}}(x) :=& \det(x\Eins_d - U_{2^k})\\
		    =& x^d - a_1 x^{d-1} + \ldots + (-1)^{d-1} x^1 x_{1} + (-1)^d x^0 a_d;
\end{align}
see e.\,g. Jacobson \cite[p. 196]{Jacobson96}.
The value $a_1$ is defined as the sum of all diagonal elements, thus the trace of $U_{2^k}$. The value $a_2$ is the
sum of the two-rowed diagonal minors and so on; $a_d$ is then the determinant.
With $U_{2^k} = V_{2^k} / (-\tr V_{2^k})$ and $\tr V_{2^k} = -\iE 2^{m/2}$ as shown in Equation \eqref{eqn:unitop:trVm},
we find
\begin{align}
 \chi_{U_{2^k}}(x)  =& \det(x\Eins_d - \alpha V_{2^k})\\
		    =& x^d - a'_1 \alpha x^{d-1} + \ldots + (-1)^{d-1} \alpha^{d-1} x^1 a'_{d-1} + (-1)^d \alpha^d x^0 a'_d,
\end{align}
with $\alpha:=-\iE 2^{-m/2}$. In order to fulfill the requirements of Corollary \ref{cor:wiedproof:charpolfermat}, all the
coefficients $c_l$ of the characteristic polynomial $\chi_{U_{2^k}}$ of $U_{2^k}$ have to be one with
$\chi_{U_{2^k}}(x) = \sum_{l=0}^d c_l x^{d-l}$. Therefore, we should find for the minor sums~$a'_l$ of $V_{2^k}$, that
$a'_l = (-\iE 2^{m/2})^l$ holds for $l \in \MgE{d}$. For the trace $a'_1$ this was already shown in Equation \eqref{eqn:unitop:trVm}.
We will refer to $a'_l$, which is the sum of the $l$-rowed diagonal minors, as $\tr_l$ in the following. Since the proof
is not finished, we can only discuss the calculation of the two-rowed diagonal minors, but we expect that the higher orders may
be derived in a similar manner.

\subsection{Two-rowed diagonal minors}
The sum of the two-rowed diagonal minors of the matrix $V_{2^k}$ is given by
\begin{align}\label{eqn:wiedproof:trvVm}
 \tr_2 (V_{2^k}) := \sum_{j=2}^d \sum_{i=1}^{j-1} \left( (V_{2^k})_{i,i} \cdot (V_{2^k})_{j,j} - (V_{2^k})_{i,j} \cdot (V_{2^k})_{j,i} \right).
\end{align}
To calculate this sum, we will separate the sum into the diagonal and the anti-diagonal term. We like to recall that the elements
of $V_{2^k}$ are by construction limited as $(V_{2^k})_{i,j} \in \Mg{\pm 1, \pm \iE}$.\par
In those cases, where we pick a real and an imaginary element from the diagonal, their product would be imaginary. But we have seen in the
calculation of the trace in Equation \eqref{eqn:unitop:trVm}, that the sum of all real elements on the diagonal is zero. So the sum of these cases adds a zero to $\tr_2 (V_{2^k})$.\par
Whenever we take two real elements from the diagonal, they equal either $-1$ or $+1$, but we know their sum is zero and $d/2$ elements are real.
All their products give
\begin{align}\label{eqn:wiedproof:tr2diagrr}
 \frac{d^2-4d}{32} \cdot 1\cdot 1 + \frac{d^2-4d}{32} \cdot (-1)\cdot (-1) + \frac{d^2}{16} \cdot (-1) \cdot 1 = - \frac{d}{4}.
\end{align}
The last case holds when two imaginary elements from the diagonal axis are multiplied. Since the sum of the $2^{m-1}$ imaginary
elements gives $-\iE 2^{m/2}$ (cf. Equation \eqref{eqn:unitop:trVm}), there occur $2^{m-2} - 2^{m/2-1}$ elements on the
diagonal which are given by $+\iE$ and $2^{m-2} + 2^{m/2-1}$ elements on the diagonal which are given by $-\iE$. In analogy,
we can calculate the sum of the pairwise products of all imaginary elements on the diagonal as
\begin{align}\label{eqn:wiedproof:tr2diagii}
  \frac{(c^-)^2 - c^-}{2}\cdot \iE^2 + \frac{(c^+)^2 - c^+}{2}\cdot (-\iE)^2 + (c^- \cdot c^+)\cdot \iE\cdot(-\iE) = -\frac{d}{4},
\end{align}
with abbreviations $c^-:=(2^{m-2} - 2^{m/2 -1})$ and $c^+:=(2^{m-2} + 2^{m/2 -1})$.\par
In order to calculate the contributions of the anti-diagonal term of Equation~\eqref{eqn:wiedproof:trvVm}, some preliminary considerations
are needed. Both factors of the term, namely $(V_{2^k})_{i,j}$ and $(V_{2^k})_{j,i}$ have a factor of $(-1)^{i\cdot j}$ or $(-1)^{j\cdot i}$,
respectively, according to Equation~\eqref{eqn:unitop:Vmij}. Since their product is zero, also the product of $(-1)^{i\cdot 0}$ and
$(-1)^{j\cdot 0}$ is zero. Thus, we have
\begin{align}\label{eqn:wiedproof:Vi0j0}
 -(V_{2^k})_{i,j} \cdot (V_{2^k})_{j,i} = -(V_{2^k})_{i,0} \cdot (V_{2^k})_{j,0}.
\end{align}
By the recursion relation given in Equation \eqref{eqn:unitop:recursion} it is clear, that the set of elements of the first column of $V_{2^k}$ 
equals the set of elements of $V_{2^{k-1}}$. By construction, for any $V_{2^k}$, every second column has imaginary elements and in each column,
except for the first, the number of positive and negative elements is equal.\footnote{The Hadamard matrix $\bar H$ has an equal number of $+1$
and $-1$ in each column, except for the first, where only $+1$ appears (cf. Appendix \ref{app:hadamard}).}\par
The product of two elements from the first column, with one real and one imaginary element then vanishes, as the sum of imaginary elements in the first
column is zero.\par
If we take two imaginary elements from the first column of $V_{2^k}$, we have an equivalent calculation as for the real elements on the diagonal
to sum up the products of all pairs. With the same number of positive and negative imaginary elements the sum gives, involving the minus sign of Equation
\eqref{eqn:wiedproof:Vi0j0} and analogously to Equation \eqref{eqn:wiedproof:tr2diagrr},
\begin{align}
 -\left( \frac{d^2-4d}{32} \cdot \iE \cdot \iE + \frac{d^2-4d}{32} \cdot (-\iE)\cdot (-\iE) + \frac{d^2}{16} \cdot (-\iE) \cdot \iE\right) = - \frac{d}{4}.
\end{align}\par
For the negative sum of the products of two real elements of the first column of $V_{2^k}$ we remember, that the set of elements of this column
equals the set of all elements of $V_{2^{k-1}}$. This matrix is given by the multiplication of a phase vector with Sylvester's Hadamard matrix.
The rows with even index (starting with $0$), all have only real entries. Therefore, the number of $+1$ and $-1$ is equal in all columns, but not
in the first, where only $+1$ appears. Thus, there are $2^{m-2}+2^{m/2-1}$ entries with value $+1$ and $2^{m-2}-2^{m/2-1}$ entries with value $-1$.
The sum of their products can be calculated analogously to Equation \eqref{eqn:wiedproof:tr2diagii} as
\begin{align}
  - \left(\frac{(c^-)^2 - c^-}{2}\cdot 1^2 + \frac{(c^+)^2 - c^+}{2} (-1)^2 + (c^- \cdot c^+)\cdot 1\cdot (-1)\right) = -\frac{d}{4}.
\end{align}
Summing up four times $-d/4$ in the above discussion leads to
\begin{align}
 \tr_2 (V_{2^k}) = -d = -2^m,
\end{align}
which equals $\tr (V_{2^k})^2$ and is the expected result (cf. Equation \eqref{eqn:unitop:trVm}).\par
Continuing this process, maybe by using a complete induction which shows that $\tr_{l+1} (V_{2^k})$ is given by
$\tr_1 (V_{2^k}) \cdot \tr_{l} (V_{2^k})$, could lead to a proof of Wiedemann's conjecture.

\chapter{Results}\label{app:results}

\section{Solutions for cyclic MUBs}\label{app:MUBsolutions}

\subsection{Homogeneous sets with group structure}\label{app:homosets:group}
Homogeneous sets with an additive group structure in the $G_j^z$ components of the class generators for $j \in \MgE{d}$
were introduced in Section \ref{sec:homogeneous:group}. Results for dimension $d=2^3$ of the Hilbert space were
given. Here we list sets which appear by testing all stabilizer matrices $C$ in the form of Equation~\eqref{eqn:homosets:group:C}
for $d=2^4$. For each set we give the reduced stabilizer matrix $B$ as well as the
matrix~$R$, both with the vector $\vec n$ as their index which indicates the entanglement properties of the bases
of the specific set, as defined in Section \ref{sec:entangleprop}. Additionally, the corresponding quantum circuits
are derived by the methods of Section \ref{sec:gatedecomp:homogeneous} and listed below the matrices.

\begin{align}
 B_{(2,0,0,4,11)} = 
\begin{pmatrix}
 1 & 1 & 1 & 1\\
 1 & 1 & 1 & 0\\
 0 & 0 & 1 & 1\\
 1 & 0 & 0 & 0
\end{pmatrix}\quad\mathrm{and}\quad
 R_{(2,0,0,4,11)} = 
\begin{pmatrix}
 1 & 1 & 0 & 1\\
 1 & 0 & 1 & 0\\
 0 & 1 & 0 & 0\\
 1 & 0 & 0 & 0
\end{pmatrix}.
\end{align}

\centerline{
\begin{tikzpicture}[scale=0.8, transform shape]
\tikzstyle{operator} = [draw,fill=white,minimum size=1.5em]
\tikzstyle{phase} = [fill,shape=circle,minimum size=5pt,inner sep=0pt]
\tikzstyle{surround} = [fill=black!10,thick,draw=black,rounded corners=2mm]
\node at (1,0) (q1) {};
\node at (1,-1) (q2) {};
\node at (1,-2) (q3) {};
\node at (1,-3) (q4) {};
\node[phase] (p0) at (2,-3) {} edge [-] (q4);
\draw (2,0) node[draw,circle, minimum size=1.1em] {};
\draw[-] (p0) -- (2,0cm+.551em);
\node[phase] (p1) at (3,0) {} edge [-] (q1);
\draw (3,-3) node[draw,circle, minimum size=1.1em] {};
\draw[-] (p1) -- (3,-3cm-.551em);
\node[phase] (p2) at (4,-1) {} edge [-] (q2);
\draw (4,0) node[draw,circle, minimum size=1.1em] {};
\draw[-] (p2) -- (4,0cm+.551em);
\node[phase] (p3) at (4,-3) {} edge [-] (p0);
\draw (4,-2) node[draw,circle, minimum size=1.1em] {};
\draw[-] (p3) -- (4,-2cm+.551em);
\node[phase] (p4) at (5,-2) {} edge [-] (q3);
\draw (5,-1) node[draw,circle, minimum size=1.1em] {};
\draw[-] (p4) -- (5,-1cm+.551em);
\node[phase] (p5) at (6,-1) {} edge [-] (p2);
\draw (6,-2) node[draw,circle, minimum size=1.1em] {};
\draw[-] (p5) -- (6,-2cm-.551em);
\node[phase] (p6) at (7,-2) {} edge [-] (p4);
\draw (7,-1) node[draw,circle, minimum size=1.1em] {};
\draw[-] (p6) -- (7,-1cm+.551em);
\node[operator] (p7) at (8,0) {$\iE$} edge [-] (p1);
\node[operator] (p8) at (8,-2) {$\iE$} edge [-] (p6);
\node[phase] (p9) at (9,-1) {} edge [-] (p5);
\node[operator] (p10) at (9,-2) {$-1$} edge [-] (p8);
\draw[-] (p9) -- (p10);
\node[phase] (p11) at (10,-2) {} edge [-] (p10);
\node[operator] (p12) at (10,-3) {$-1$} edge [-] (p3);
\draw[-] (p11) -- (p12);
\node[operator] (p13) at (11,0) {$H$} edge [-] (p7);
\node[operator] (p14) at (11,-1) {$H$} edge [-] (p9);
\node[operator] (p15) at (11,-2) {$H$} edge [-] (p11);
\node[operator] (p16) at (11,-3) {$H$} edge [-] (p12);
\draw[-] (p13) -- (12,0);
\draw[-] (p14) -- (12,-1);
\draw[-] (p15) -- (12,-2);
\draw[-] (p16) -- (12,-3);
\draw[decorate,decoration={brace},thick] (12.3,0.2) to node[midway,right] (bracketR) {$\;\;=C_{(2,0,0,4,11)}$} (12.3,-3.2);
\begin{pgfonlayer}{background}
\node[surround] (background) [fit = (bracketL) (p13) (p16) (bracketR)] {};
\end{pgfonlayer}
\end{tikzpicture}
}
\pagebreak

\begin{align}
 B_{(2,0,1,4,10)} = 
\begin{pmatrix}
 1 & 1 & 1 & 0\\
 0 & 1 & 1 & 1\\
 0 & 0 & 1 & 1\\
 1 & 0 & 0 & 0
\end{pmatrix}\quad\mathrm{and}\quad
 R_{(2,0,1,4,10)} = 
\begin{pmatrix}
 1 & 0 & 0 & 1\\
 0 & 0 & 1 & 0\\
 0 & 1 & 0 & 0\\
 1 & 0 & 0 & 0
\end{pmatrix}.
\end{align}

\centerline{
\begin{tikzpicture}[scale=0.8, transform shape]
\tikzstyle{operator} = [draw,fill=white,minimum size=1.5em]
\tikzstyle{phase} = [fill,shape=circle,minimum size=5pt,inner sep=0pt]
\tikzstyle{surround} = [fill=black!10,thick,draw=black,rounded corners=2mm]
\node at (1,0) (q1) {};
\node at (1,-1) (q2) {};
\node at (1,-2) (q3) {};
\node at (1,-3) (q4) {};
\node[phase] (p0) at (2,-3) {} edge [-] (q4);
\draw (2,0) node[draw,circle, minimum size=1.1em] {};
\draw[-] (p0) -- (2,0cm+.551em);
\node[phase] (p1) at (3,0) {} edge [-] (q1);
\draw (3,-3) node[draw,circle, minimum size=1.1em] {};
\draw[-] (p1) -- (3,-3cm-.551em);
\node[phase] (p2) at (4,-2) {} edge [-] (q3);
\draw (4,-1) node[draw,circle, minimum size=1.1em] {};
\draw[-] (p2) -- (4,-1cm+.551em);
\node[phase] (p3) at (5,-1) {} edge [-] (q2);
\draw (5,-2) node[draw,circle, minimum size=1.1em] {};
\draw[-] (p3) -- (5,-2cm-.551em);
\node[phase] (p4) at (6,-2) {} edge [-] (p2);
\draw (6,-1) node[draw,circle, minimum size=1.1em] {};
\draw[-] (p4) -- (6,-1cm+.551em);
\node[operator] (p5) at (7,0) {$\iE$} edge [-] (p1);
\node[operator] (p6) at (7,-2) {$\iE$} edge [-] (p4);
\node[phase] (p7) at (8,-1) {} edge [-] (p3);
\node[operator] (p8) at (8,-2) {$-1$} edge [-] (p6);
\draw[-] (p7) -- (p8);
\node[phase] (p9) at (9,-1) {} edge [-] (p7);
\node[operator] (p10) at (9,-3) {$-1$} edge [-] (p0);
\draw[-] (p9) -- (p10);
\node[phase] (p11) at (10,-2) {} edge [-] (p8);
\node[operator] (p12) at (10,-3) {$-1$} edge [-] (p10);
\draw[-] (p11) -- (p12);
\node[operator] (p13) at (11,0) {$H$} edge [-] (p5);
\node[operator] (p14) at (11,-1) {$H$} edge [-] (p9);
\node[operator] (p15) at (11,-2) {$H$} edge [-] (p11);
\node[operator] (p16) at (11,-3) {$H$} edge [-] (p12);
\draw[-] (p13) -- (12,0);
\draw[-] (p14) -- (12,-1);
\draw[-] (p15) -- (12,-2);
\draw[-] (p16) -- (12,-3);
\draw[decorate,decoration={brace},thick] (12.3,0.2) to node[midway,right] (bracketR) {$\;\;=C_{(2,0,1,4,10)}$} (12.3,-3.2);
\begin{pgfonlayer}{background}
\node[surround] (background) [fit = (bracketL) (p13) (p16) (bracketR)] {};
\end{pgfonlayer}
\end{tikzpicture}
}

%
%

\begin{align}
 B_{(2,0,2,4,9)} = 
\begin{pmatrix}
 1 & 1 & 1 & 1\\
 0 & 0 & 1 & 0\\
 0 & 0 & 0 & 1\\
 1 & 0 & 0 & 0
\end{pmatrix}\quad\mathrm{and}\quad
 R_{(2,0,2,4,9)} = 
\begin{pmatrix}
 1 & 0 & 0 & 1\\
 0 & 1 & 1 & 0\\
 0 & 1 & 0 & 0\\
 1 & 0 & 0 & 0
\end{pmatrix}.
\end{align}

\centerline{
\begin{tikzpicture}[scale=0.8, transform shape]
\tikzstyle{operator} = [draw,fill=white,minimum size=1.5em]
\tikzstyle{phase} = [fill,shape=circle,minimum size=5pt,inner sep=0pt]
\tikzstyle{surround} = [fill=black!10,thick,draw=black,rounded corners=2mm]
\node at (1,0) (q1) {};
\node at (1,-1) (q2) {};
\node at (1,-2) (q3) {};
\node at (1,-3) (q4) {};
\node[phase] (p0) at (2,-3) {} edge [-] (q4);
\draw (2,0) node[draw,circle, minimum size=1.1em] {};
\draw[-] (p0) -- (2,0cm+.551em);
\node[phase] (p1) at (3,0) {} edge [-] (q1);
\draw (3,-3) node[draw,circle, minimum size=1.1em] {};
\draw[-] (p1) -- (3,-3cm-.551em);
\node[phase] (p2) at (4,-2) {} edge [-] (q3);
\draw (4,-1) node[draw,circle, minimum size=1.1em] {};
\draw[-] (p2) -- (4,-1cm+.551em);
\node[phase] (p3) at (5,-1) {} edge [-] (q2);
\draw (5,-2) node[draw,circle, minimum size=1.1em] {};
\draw[-] (p3) -- (5,-2cm-.551em);
\node[operator] (p4) at (6,0) {$\iE$} edge [-] (p1);
\node[operator] (p5) at (6,-2) {$\iE$} edge [-] (p2);
\node[operator] (p6) at (6,-3) {$\iE$} edge [-] (p0);
\node[phase] (p7) at (7,-1) {} edge [-] (p3);
\node[operator] (p8) at (7,-3) {$-1$} edge [-] (p6);
\draw[-] (p7) -- (p8);
\node[phase] (p9) at (8,-2) {} edge [-] (p5);
\node[operator] (p10) at (8,-3) {$-1$} edge [-] (p8);
\draw[-] (p9) -- (p10);
\node[operator] (p11) at (9,0) {$H$} edge [-] (p4);
\node[operator] (p12) at (9,-1) {$H$} edge [-] (p7);
\node[operator] (p13) at (9,-2) {$H$} edge [-] (p9);
\node[operator] (p14) at (9,-3) {$H$} edge [-] (p10);
\draw[-] (p11) -- (10,0);
\draw[-] (p12) -- (10,-1);
\draw[-] (p13) -- (10,-2);
\draw[-] (p14) -- (10,-3);
\draw[decorate,decoration={brace},thick] (10.3,0.2) to node[midway,right] (bracketR) {$\;\;=C_{(2,0,2,4,9)}$} (10.3,-3.2);
\begin{pgfonlayer}{background}
\node[surround] (background) [fit = (bracketL) (p11) (p14) (bracketR)] {};
\end{pgfonlayer}
\end{tikzpicture}
}

\begin{align}
 B_{(2,1,0,2,12)} = 
\begin{pmatrix}
 1 & 1 & 1 & 1\\
 1 & 1 & 0 & 0\\
 0 & 1 & 1 & 0\\
 1 & 0 & 0 & 0
\end{pmatrix}\quad\mathrm{and}\quad
 R_{(2,1,0,2,12)} = 
\begin{pmatrix}
 1 & 1 & 0 & 1\\
 1 & 0 & 1 & 0\\
 0 & 1 & 0 & 0\\
 1 & 0 & 0 & 0
\end{pmatrix}.
\end{align}

\centerline{
\begin{tikzpicture}[scale=0.8, transform shape]
\tikzstyle{operator} = [draw,fill=white,minimum size=1.5em]
\tikzstyle{phase} = [fill,shape=circle,minimum size=5pt,inner sep=0pt]
\tikzstyle{surround} = [fill=black!10,thick,draw=black,rounded corners=2mm]
\node at (1,0) (q1) {};
\node at (1,-1) (q2) {};
\node at (1,-2) (q3) {};
\node at (1,-3) (q4) {};
\node[phase] (p0) at (2,-3) {} edge [-] (q4);
\draw (2,0) node[draw,circle, minimum size=1.1em] {};
\draw[-] (p0) -- (2,0cm+.551em);
\node[phase] (p1) at (3,0) {} edge [-] (q1);
\draw (3,-3) node[draw,circle, minimum size=1.1em] {};
\draw[-] (p1) -- (3,-3cm-.551em);
\node[phase] (p2) at (4,-1) {} edge [-] (q2);
\draw (4,0) node[draw,circle, minimum size=1.1em] {};
\draw[-] (p2) -- (4,0cm+.551em);
\node[phase] (p3) at (4,-3) {} edge [-] (p0);
\draw (4,-2) node[draw,circle, minimum size=1.1em] {};
\draw[-] (p3) -- (4,-2cm+.551em);
\node[phase] (p4) at (5,-2) {} edge [-] (q3);
\draw (5,-1) node[draw,circle, minimum size=1.1em] {};
\draw[-] (p4) -- (5,-1cm+.551em);
\node[phase] (p5) at (6,-1) {} edge [-] (p2);
\draw (6,-2) node[draw,circle, minimum size=1.1em] {};
\draw[-] (p5) -- (6,-2cm-.551em);
\node[phase] (p6) at (7,-2) {} edge [-] (p4);
\draw (7,-1) node[draw,circle, minimum size=1.1em] {};
\draw[-] (p6) -- (7,-1cm+.551em);
\node[operator] (p7) at (8,0) {$\iE$} edge [-] (p1);
\node[operator] (p8) at (8,-1) {$\iE$} edge [-] (p5);
\node[operator] (p9) at (8,-3) {$\iE$} edge [-] (p3);
\node[phase] (p10) at (9,-1) {} edge [-] (p8);
\node[operator] (p11) at (9,-2) {$-1$} edge [-] (p6);
\draw[-] (p10) -- (p11);
\node[operator] (p12) at (10,0) {$H$} edge [-] (p7);
\node[operator] (p13) at (10,-1) {$H$} edge [-] (p10);
\node[operator] (p14) at (10,-2) {$H$} edge [-] (p11);
\node[operator] (p15) at (10,-3) {$H$} edge [-] (p9);
\draw[-] (p12) -- (11,0);
\draw[-] (p13) -- (11,-1);
\draw[-] (p14) -- (11,-2);
\draw[-] (p15) -- (11,-3);
\draw[decorate,decoration={brace},thick] (11.3,0.2) to node[midway,right] (bracketR) {$\;\;=C_{(2,1,0,2,12)}$} (11.3,-3.2);
\begin{pgfonlayer}{background}
\node[surround] (background) [fit = (bracketL) (p12) (p15) (bracketR)] {};
\end{pgfonlayer}
\end{tikzpicture}
}
\pagebreak

\begin{align}
 B_{(2,1,1,2,11)} = 
\begin{pmatrix}
 1 & 0 & 1 & 0\\
 0 & 0 & 1 & 1\\
 0 & 1 & 0 & 0\\
 1 & 0 & 0 & 0
\end{pmatrix}\quad\mathrm{and}\quad
 R_{(2,1,1,2,11)} = 
\begin{pmatrix}
 1 & 0 & 0 & 1\\
 0 & 0 & 1 & 0\\
 0 & 1 & 0 & 0\\
 1 & 0 & 0 & 0
\end{pmatrix}.
\end{align}

\centerline{
\begin{tikzpicture}[scale=0.8, transform shape]
\tikzstyle{operator} = [draw,fill=white,minimum size=1.5em]
\tikzstyle{phase} = [fill,shape=circle,minimum size=5pt,inner sep=0pt]
\tikzstyle{surround} = [fill=black!10,thick,draw=black,rounded corners=2mm]
\node at (1,0) (q1) {};
\node at (1,-1) (q2) {};
\node at (1,-2) (q3) {};
\node at (1,-3) (q4) {};
\node[phase] (p0) at (2,-3) {} edge [-] (q4);
\draw (2,0) node[draw,circle, minimum size=1.1em] {};
\draw[-] (p0) -- (2,0cm+.551em);
\node[phase] (p1) at (3,0) {} edge [-] (q1);
\draw (3,-3) node[draw,circle, minimum size=1.1em] {};
\draw[-] (p1) -- (3,-3cm-.551em);
\node[phase] (p2) at (4,-2) {} edge [-] (q3);
\draw (4,-1) node[draw,circle, minimum size=1.1em] {};
\draw[-] (p2) -- (4,-1cm+.551em);
\node[phase] (p3) at (5,-1) {} edge [-] (q2);
\draw (5,-2) node[draw,circle, minimum size=1.1em] {};
\draw[-] (p3) -- (5,-2cm-.551em);
\node[phase] (p4) at (6,-2) {} edge [-] (p2);
\draw (6,-1) node[draw,circle, minimum size=1.1em] {};
\draw[-] (p4) -- (6,-1cm+.551em);
\node[operator] (p5) at (7,0) {$\iE$} edge [-] (p1);
\node[operator] (p6) at (7,-1) {$\iE$} edge [-] (p3);
\node[operator] (p7) at (7,-2) {$\iE$} edge [-] (p4);
\node[phase] (p8) at (8,-2) {} edge [-] (p7);
\node[operator] (p9) at (8,-3) {$-1$} edge [-] (p0);
\draw[-] (p8) -- (p9);
\node[operator] (p10) at (9,0) {$H$} edge [-] (p5);
\node[operator] (p11) at (9,-1) {$H$} edge [-] (p6);
\node[operator] (p12) at (9,-2) {$H$} edge [-] (p8);
\node[operator] (p13) at (9,-3) {$H$} edge [-] (p9);
\draw[-] (p10) -- (10,0);
\draw[-] (p11) -- (10,-1);
\draw[-] (p12) -- (10,-2);
\draw[-] (p13) -- (10,-3);
\draw[decorate,decoration={brace},thick] (10.3,0.2) to node[midway,right] (bracketR) {$\;\;=C_{(2,1,1,2,11)}$} (10.3,-3.2);
\begin{pgfonlayer}{background}
\node[surround] (background) [fit = (bracketL) (p10) (p13) (bracketR)] {};
\end{pgfonlayer}
\end{tikzpicture}
}

\begin{align}
 B_{(2,1,2,2,10)} = 
\begin{pmatrix}
 1 & 0 & 1 & 0\\
 0 & 1 & 0 & 1\\
 0 & 1 & 1 & 0\\
 1 & 0 & 0 & 0
\end{pmatrix}\quad\mathrm{and}\quad
 R_{(2,1,2,2,10)} = 
\begin{pmatrix}
 1 & 0 & 0 & 1\\
 0 & 0 & 1 & 0\\
 0 & 1 & 0 & 0\\
 1 & 0 & 0 & 0
\end{pmatrix}.
\end{align}

\centerline{
\begin{tikzpicture}[scale=0.8, transform shape]
\tikzstyle{operator} = [draw,fill=white,minimum size=1.5em]
\tikzstyle{phase} = [fill,shape=circle,minimum size=5pt,inner sep=0pt]
\tikzstyle{surround} = [fill=black!10,thick,draw=black,rounded corners=2mm]
\node at (1,0) (q1) {};
\node at (1,-1) (q2) {};
\node at (1,-2) (q3) {};
\node at (1,-3) (q4) {};
\node[phase] (p0) at (2,-3) {} edge [-] (q4);
\draw (2,0) node[draw,circle, minimum size=1.1em] {};
\draw[-] (p0) -- (2,0cm+.551em);
\node[phase] (p1) at (3,0) {} edge [-] (q1);
\draw (3,-3) node[draw,circle, minimum size=1.1em] {};
\draw[-] (p1) -- (3,-3cm-.551em);
\node[phase] (p2) at (4,-2) {} edge [-] (q3);
\draw (4,-1) node[draw,circle, minimum size=1.1em] {};
\draw[-] (p2) -- (4,-1cm+.551em);
\node[phase] (p3) at (5,-1) {} edge [-] (q2);
\draw (5,-2) node[draw,circle, minimum size=1.1em] {};
\draw[-] (p3) -- (5,-2cm-.551em);
\node[phase] (p4) at (6,-2) {} edge [-] (p2);
\draw (6,-1) node[draw,circle, minimum size=1.1em] {};
\draw[-] (p4) -- (6,-1cm+.551em);
\node[operator] (p5) at (7,0) {$\iE$} edge [-] (p1);
\node[operator] (p6) at (7,-1) {$\iE$} edge [-] (p3);
\node[phase] (p7) at (8,-1) {} edge [-] (p6);
\node[operator] (p8) at (8,-2) {$-1$} edge [-] (p4);
\draw[-] (p7) -- (p8);
\node[phase] (p9) at (9,-2) {} edge [-] (p8);
\node[operator] (p10) at (9,-3) {$-1$} edge [-] (p0);
\draw[-] (p9) -- (p10);
\node[operator] (p11) at (10,0) {$H$} edge [-] (p5);
\node[operator] (p12) at (10,-1) {$H$} edge [-] (p7);
\node[operator] (p13) at (10,-2) {$H$} edge [-] (p9);
\node[operator] (p14) at (10,-3) {$H$} edge [-] (p10);
\draw[-] (p11) -- (11,0);
\draw[-] (p12) -- (11,-1);
\draw[-] (p13) -- (11,-2);
\draw[-] (p14) -- (11,-3);
\draw[decorate,decoration={brace},thick] (11.3,0.2) to node[midway,right] (bracketR) {$\;\;=C_{(2,1,2,2,10)}$} (11.3,-3.2);
\begin{pgfonlayer}{background}
\node[surround] (background) [fit = (bracketL) (p11) (p14) (bracketR)] {};
\end{pgfonlayer}
\end{tikzpicture}
}

\begin{align}
 B_{(2,2,2,0,11)} = 
\begin{pmatrix}
 1 & 0 & 0 & 1\\
 0 & 1 & 1 & 0\\
 1 & 1 & 0 & 0\\
 0 & 1 & 0 & 0
\end{pmatrix}\quad\mathrm{and}\quad
 R_{(2,2,2,0,11)} = 
\begin{pmatrix}
 1 & 1 & 0 & 1\\
 1 & 0 & 1 & 0\\
 0 & 1 & 0 & 0\\
 1 & 0 & 0 & 0
\end{pmatrix}.
\end{align}

\centerline{
\begin{tikzpicture}[scale=0.8, transform shape]
\tikzstyle{operator} = [draw,fill=white,minimum size=1.5em]
\tikzstyle{phase} = [fill,shape=circle,minimum size=5pt,inner sep=0pt]
\tikzstyle{surround} = [fill=black!10,thick,draw=black,rounded corners=2mm]
\node at (1,0) (q1) {};
\node at (1,-1) (q2) {};
\node at (1,-2) (q3) {};
\node at (1,-3) (q4) {};
\node[phase] (p0) at (2,-3) {} edge [-] (q4);
\draw (2,0) node[draw,circle, minimum size=1.1em] {};
\draw[-] (p0) -- (2,0cm+.551em);
\node[phase] (p1) at (3,0) {} edge [-] (q1);
\draw (3,-3) node[draw,circle, minimum size=1.1em] {};
\draw[-] (p1) -- (3,-3cm-.551em);
\node[phase] (p2) at (4,-1) {} edge [-] (q2);
\draw (4,0) node[draw,circle, minimum size=1.1em] {};
\draw[-] (p2) -- (4,0cm+.551em);
\node[phase] (p3) at (4,-3) {} edge [-] (p0);
\draw (4,-2) node[draw,circle, minimum size=1.1em] {};
\draw[-] (p3) -- (4,-2cm+.551em);
\node[phase] (p4) at (5,-2) {} edge [-] (q3);
\draw (5,-1) node[draw,circle, minimum size=1.1em] {};
\draw[-] (p4) -- (5,-1cm+.551em);
\node[phase] (p5) at (6,-1) {} edge [-] (p2);
\draw (6,-2) node[draw,circle, minimum size=1.1em] {};
\draw[-] (p5) -- (6,-2cm-.551em);
\node[phase] (p6) at (7,-2) {} edge [-] (p4);
\draw (7,-1) node[draw,circle, minimum size=1.1em] {};
\draw[-] (p6) -- (7,-1cm+.551em);
\node[operator] (p7) at (8,-1) {$\iE$} edge [-] (p5);
\node[operator] (p8) at (8,-2) {$\iE$} edge [-] (p6);
\node[operator] (p9) at (8,-3) {$\iE$} edge [-] (p3);
\node[phase] (p10) at (9,0) {} edge [-] (p1);
\node[operator] (p11) at (9,-1) {$-1$} edge [-] (p7);
\draw[-] (p10) -- (p11);
\node[operator] (p12) at (10,0) {$H$} edge [-] (p10);
\node[operator] (p13) at (10,-1) {$H$} edge [-] (p11);
\node[operator] (p14) at (10,-2) {$H$} edge [-] (p8);
\node[operator] (p15) at (10,-3) {$H$} edge [-] (p9);
\draw[-] (p12) -- (11,0);
\draw[-] (p13) -- (11,-1);
\draw[-] (p14) -- (11,-2);
\draw[-] (p15) -- (11,-3);
\draw[decorate,decoration={brace},thick] (11.3,0.2) to node[midway,right] (bracketR) {$\;\;=C_{(2,2,2,0,11)}$} (11.3,-3.2);
\begin{pgfonlayer}{background}
\node[surround] (background) [fit = (bracketL) (p12) (p15) (bracketR)] {};
\end{pgfonlayer}
\end{tikzpicture}
}
\pagebreak

\vspace{.5cm}

The listed sets are not unique. For each system we found the following number of different solutions
(which may give rise to possible symmetries):

\begin{align*}
 \vec n = (3,0,2,0,12) \Rightarrow&\quad 1\cdot4\cdot 15\cdot 4!\;\mathrm{solutions},\\
 \vec n = (2,0,0,4,11) \Rightarrow&\quad 4\cdot4\cdot 15\cdot 4!\;\mathrm{solutions},\\
 \vec n = (2,0,1,4,10) \Rightarrow&\quad 6\cdot4\cdot 15\cdot 4!\;\mathrm{solutions},\\
 \vec n = (2,0,2,4,09) \Rightarrow&\quad 1\cdot4\cdot 15\cdot 4!\;\mathrm{solutions},\\
 \vec n = (2,1,0,2,12) \Rightarrow&\quad 6\cdot4\cdot 15\cdot 4!\;\mathrm{solutions},\\
 \vec n = (2,1,1,2,11) \Rightarrow&\quad 6\cdot4\cdot 15\cdot 4!\;\mathrm{solutions},\\
 \vec n = (2,1,2,2,10) \Rightarrow&\quad 2\cdot4\cdot 15\cdot 4!\;\mathrm{solutions},\\
 \vec n = (2,2,2,0,11) \Rightarrow&\quad 2\cdot4\cdot 15\cdot 4!\;\mathrm{solutions}.
\end{align*}
As discussed in Section \ref{sec:homogeneous}, the factor of $15=2^4-1$ can be explained by a freedom
in the choice of $R$ (cf. Equation \eqref{eqn:homosets:group:Cfactor}). It is not clear, if this is
uncorrelated to all permutations of the indexing of the four qubits, which would explain the occurrence
of the factor $4!$.

\subsection{Homogeneous sets with semigroup structure}\label{app:homosets:semigroup}
\input{semigroupset_results}

\subsection{Inhomogeneous sets}\label{app:inhomosets}
\input{inhomoset_results}

\subsection{Testing Wiedemann's conjecture}\label{app:fermat_based:wiedemann}
Wiedemann tested his conjecture with dimensions $d=2^{2^k}$ for $k \in \MgN{8}$ in \cite{Wiedemann88}.
As claimed in Section \ref{sec:fermatset}, we are able to raise this test up to $k=11$. Since the testing procedure
is too complicated to be done by hand, we offer the used \emph{Matlab} code to provide maximal insight into
the calculation. It should be mentioned that the method \emph{divisor\_list} is generated by a
\emph{Mathematica} code, since \emph{Matlab} cannot handle very long integers that easily. This code
takes the factors of a Fermat number, calculates all divisors, writes them in a binary representation
and reverses that string to have the least significant bit on the left hand side. Finally, these
divisors are provided within a list that can be used for the \emph{Matlab} code. 
The following code uses the parameter \emph{ksize} that equals the variable $k$.

\begingroup
\small
\begin{lstlisting}[language=Matlab]
function [ ] = TestWiedemann(ksize)

% testing generator for 2^ksize qubits
size=2^ksize;

% initialize matrices
Zero=zeros([size,size]);
C=zeros([2*size,2*size]);

% create unity matrices One=1_{size} and One2=1_{2*size}
One  = eye(size);
One2 = eye(2*size);

% create matrix C
maxrow=0;
row=1;
for i=1:2*size
  C(i,row)=1;
  C(row,i)=1;
  row=row+1;
  if (row>maxrow)
    row=1;
    if (maxrow==0)
      maxrow=1;
    else
      maxrow=maxrow*2;
    end
  end
end
   

% test if C^{2^m+1}=One2
CC=C;
for l=1:size
  CC=mod(CC*CC,2);
end
  
% if yes...
if ( mod( CC*C,2 ) == One2) 
  % load divisors of 2^m+1 in binary
  % representation whearas the
  % least significant bit is on the left
  divisor_list=divisors_bin(size);
    
  wrong=0;
    
  for i=1:length(divs)
	  
    % initialize testmatrix
    test=One2;

    % load factor
    factor=divisor_list{i};

    % calculate the power of C given by the loaded
    % factor and write the result to the testmatrix
    CC=C;
    for j=1:length(factor)
      if (factor(j)=='1')
        test=mod(test*CC,2);
      end
      CC=mod(CC*CC,2);
    end
      
    % check if off-diagonal blocks equal zero
    if (test(1:size,size+1:2*size) == Zero)
      if (test(size+1:2*size,1:size) == Zero)
        % if yes, Wiedemanns conjecture would be wrong
        wrong=1;
        break
      end
    end
  end
    
  % if Wiedemanns conjecture has never been testet wrong, it
  % fits for the testet dimension
  if (wrong == 0) 
    fprintf('Wiedemanns conjecture tested correct');
    fprintf('for 2^%d qubits!\n\n',ksize); 
  else
    fprintf('Wiedemanns conjecture is wrong!\n\n');
  end
end
end
\end{lstlisting}
\endgroup
\pagebreak

\subsection{Triangle solutions}\label{app:fibonacci_based:triangle}
For the method we introduced in Chapter \ref{subsec:fibset:numerical} we are able to find
reduced stabilizer matrices $B$ in the form of Equation \eqref{eqn:fibset:A}. Here we list
solutions for the submatrix $A$ for dimensions $d=2^m$ with $m=\Mg{2,\ldots, 600}$.\footnote{The
used algorithm keeps the matrix $A$ as small as possible and is able to generate solutions
for even higher dimensions.}
The runtime of that algorithm is much shorter than expected in Section \ref{subsec:fibset:numerical},
since the number of divisors of an integer $2^m+1$ with $m \in \N^*$ seems in most times to be far
away from the supposed limit of $\sqrt{2^m+1}$, which is the leading factor in
Equation~\eqref{eqn:fibset:mmlimit} for large values of $m$.
Table \ref{table:app:triangle} gives solutions for the matrices~$A$, where the number in the leftmost
column indicates the number of qubits $m$ for the solution of $A$ in the subsequent column. Moving a
column to the right increments the number of qubits by one.

\begin{table}[h!]
\begin{tabular}{|c||c|c|c|c|c|c|}
\hline
$2$ \rule{0cm}{.15cm}  &
$\begin{smallmatrix}0\end{smallmatrix}$ &
$\begin{smallmatrix}0\end{smallmatrix}$ &
$\begin{smallmatrix}1\end{smallmatrix}$ &
$\begin{smallmatrix}0\end{smallmatrix}$ &
$\begin{smallmatrix}0\end{smallmatrix}$ &
$\begin{smallmatrix}1\end{smallmatrix}$ \\
\hline
$8$ \rule{0cm}{.55cm}  &
$\begin{smallmatrix}0&1\\1&1\end{smallmatrix}$ &
$\begin{smallmatrix}0\end{smallmatrix}$ &
$\begin{smallmatrix}1&0\\0&0\end{smallmatrix}$ &
$\begin{smallmatrix}0\end{smallmatrix}$ &
$\begin{smallmatrix}0&0&1\\0&0&0\\1&0&0\end{smallmatrix}$ &
$\begin{smallmatrix}1\end{smallmatrix}$ \\[1.15ex]	
\hline
$14$ \rule{0cm}{.44cm}  &
$\begin{smallmatrix}0\end{smallmatrix}$ &
$\begin{smallmatrix}0&1\\1&1\end{smallmatrix}$ &
$\begin{smallmatrix}1\end{smallmatrix}$ &
$\begin{smallmatrix}1\end{smallmatrix}$ &
$\begin{smallmatrix}0\end{smallmatrix}$ &
$\begin{smallmatrix}1\end{smallmatrix}$ \\[.35ex]
\hline
$20$ \rule{0cm}{.55cm}  &
$\begin{smallmatrix}1&0&0\\0&0&0\\0&0&1\end{smallmatrix}$ &
$\begin{smallmatrix}0&0&1\\0&0&0\\1&0&0\end{smallmatrix}$ &
$\begin{smallmatrix}0&1\\1&0\end{smallmatrix}$ &
$\begin{smallmatrix}0\end{smallmatrix}$ &
$\begin{smallmatrix}1&0\\0&1\end{smallmatrix}$ &
$\begin{smallmatrix}0&0&1\\0&0&0\\1&0&0\end{smallmatrix}$ \\[1.15ex]
\hline
$26$ \rule{0cm}{.55cm}  &
$\begin{smallmatrix}0\end{smallmatrix}$ &
$\begin{smallmatrix}1&0\\0&0\end{smallmatrix}$ &
$\begin{smallmatrix}0&0&1\\0&0&1\\1&1&1\end{smallmatrix}$ &
$\begin{smallmatrix}0\end{smallmatrix}$ &
$\begin{smallmatrix}0\end{smallmatrix}$ &
$\begin{smallmatrix}1&0&1\\0&0&1\\1&1&1\end{smallmatrix}$ \\[1.15ex]
\hline
$32$ \rule{0cm}{.55cm}  &
$\begin{smallmatrix}0&1\\1&0\end{smallmatrix}$ &
$\begin{smallmatrix}1&0&1\\0&1&0\\1&0&1\end{smallmatrix}$ &
$\begin{smallmatrix}0&1&1\\1&0&1\\1&1&1\end{smallmatrix}$ &
$\begin{smallmatrix}0\end{smallmatrix}$ &
$\begin{smallmatrix}0&0&1\\0&0&1\\1&1&1\end{smallmatrix}$ &
$\begin{smallmatrix}0&1&0\\1&1&0\\0&0&1\end{smallmatrix}$ \\[1.15ex]
\hline
$38$ \rule{0cm}{.65cm}  &
$\begin{smallmatrix}0&1\\1&1\end{smallmatrix}$ &
$\begin{smallmatrix}1&0&1\\0&0&0\\1&0&1\end{smallmatrix}$ &
$\begin{smallmatrix}1&0\\0&1\end{smallmatrix}$ &
$\begin{smallmatrix}0\end{smallmatrix}$ &
$\begin{smallmatrix}0&0&0&1\\0&0&0&0\\0&0&0&1\\1&0&1&0\end{smallmatrix}$ &
$\begin{smallmatrix}0&0&1\\0&0&1\\1&1&1\end{smallmatrix}$ \\[1.85ex]	
\hline
$44$ \rule{0cm}{.65cm}  &
$\begin{smallmatrix}1&0&0\\0&1&0\\0&0&0\end{smallmatrix}$ &
$\begin{smallmatrix}0&1&0\\1&0&0\\0&0&1\end{smallmatrix}$ &
$\begin{smallmatrix}1&0\\0&0\end{smallmatrix}$ &
$\begin{smallmatrix}1\end{smallmatrix}$ &
$\begin{smallmatrix}0&0&0&1\\0&1&0&1\\0&0&1&0\\1&1&0&0\end{smallmatrix}$ &
$\begin{smallmatrix}1\end{smallmatrix}$ \\[1.85ex]
\hline
$50$ \rule{0cm}{.65cm}  &
$\begin{smallmatrix}0\end{smallmatrix}$ &
$\begin{smallmatrix}0\end{smallmatrix}$ &
$\begin{smallmatrix}0&0&0&1\\0&0&1&1\\0&1&1&0\\1&1&0&0\end{smallmatrix}$ &
$\begin{smallmatrix}0&1&1\\1&1&1\\1&1&1\end{smallmatrix}$ &
$\begin{smallmatrix}1&0&1\\0&0&0\\1&0&1\end{smallmatrix}$ &
$\begin{smallmatrix}0&0&1\\0&1&0\\1&0&0\end{smallmatrix}$ \\[1.85ex]
\hline
$56$ \rule{0cm}{.65cm}  &
$\begin{smallmatrix}1&0&0\\0&0&0\\0&0&1\end{smallmatrix}$ &
$\begin{smallmatrix}0&0&1\\0&0&0\\1&0&0\end{smallmatrix}$ &
$\begin{smallmatrix}0&0&0&1\\0&0&1&0\\0&1&0&1\\1&0&1&0\end{smallmatrix}$ &
$\begin{smallmatrix}1&1&0\\1&0&1\\0&1&0\end{smallmatrix}$ &
$\begin{smallmatrix}0&1&0\\1&0&0\\0&0&1\end{smallmatrix}$ &
$\begin{smallmatrix}0&1\\1&0\end{smallmatrix}$ \\[1.85ex]
\hline
$62$ \rule{0cm}{.65cm}  &
$\begin{smallmatrix}0&1\\1&1\end{smallmatrix}$ &
$\begin{smallmatrix}1&0\\0&0\end{smallmatrix}$ &
$\begin{smallmatrix}0&0&1&0\\0&0&0&0\\1&0&1&1\\0&0&1&0\end{smallmatrix}$ &
$\begin{smallmatrix}0&0&1&0\\0&0&0&1\\1&0&0&1\\0&1&1&1\end{smallmatrix}$ &
$\begin{smallmatrix}1&0&1\\0&0&0\\1&0&1\end{smallmatrix}$ &
$\begin{smallmatrix}1&0&1\\0&0&0\\1&0&1\end{smallmatrix}$ \\[1.85ex]
\hline
$68$ \rule{0cm}{.65cm}  &
$\begin{smallmatrix}1&0&0\\0&0&0\\0&0&1\end{smallmatrix}$ &
$\begin{smallmatrix}0&0&0&1\\0&0&0&1\\0&0&1&0\\1&1&0&0\end{smallmatrix}$ &
$\begin{smallmatrix}0&0&1\\0&1&0\\1&0&0\end{smallmatrix}$ &
$\begin{smallmatrix}0&1&1\\1&0&1\\1&1&1\end{smallmatrix}$ &
$\begin{smallmatrix}0&1&1\\1&0&0\\1&0&1\end{smallmatrix}$ &
$\begin{smallmatrix}0&0&0&1\\0&0&1&0\\0&1&1&1\\1&0&1&1\end{smallmatrix}$ \\[1.85ex]
\hline
$74$ \rule{0cm}{.55cm}  &
$\begin{smallmatrix}1&0&0\\0&0&1\\0&1&0\end{smallmatrix}$ &
$\begin{smallmatrix}0&0&1\\0&0&0\\1&0&0\end{smallmatrix}$ &
$\begin{smallmatrix}1\end{smallmatrix}$ &
$\begin{smallmatrix}0&1\\1&0\end{smallmatrix}$ &
$\begin{smallmatrix}1&0\\0&0\end{smallmatrix}$ &
$\begin{smallmatrix}1\end{smallmatrix}$ \\[1.15ex]
\hline
$80$ \rule{0cm}{.65cm}  &
$\begin{smallmatrix}0&1\\1&1\end{smallmatrix}$ &
$\begin{smallmatrix}0&0&1\\0&0&1\\1&1&1\end{smallmatrix}$ &
$\begin{smallmatrix}1&0\\0&1\end{smallmatrix}$ &
$\begin{smallmatrix}1&0&1\\0&1&0\\1&0&0\end{smallmatrix}$ &
$\begin{smallmatrix}1&0&0\\0&0&0\\0&0&1\end{smallmatrix}$ &
$\begin{smallmatrix}0&0&1&0\\0&0&0&0\\1&0&0&1\\0&0&1&1\end{smallmatrix}$ \\[1.85ex]
\hline
$86$ \rule{0cm}{.65cm}  &
$\begin{smallmatrix}0\end{smallmatrix}$ &
$\begin{smallmatrix}0&0&0&1\\0&0&1&1\\0&1&0&1\\1&1&1&1\end{smallmatrix}$ &
$\begin{smallmatrix}0&0&0&1\\0&0&0&0\\0&0&0&0\\1&0&0&1\end{smallmatrix}$ &
$\begin{smallmatrix}0\end{smallmatrix}$ &
$\begin{smallmatrix}0\end{smallmatrix}$ &
$\begin{smallmatrix}1&0&1\\0&0&1\\1&1&1\end{smallmatrix}$ \\[1.85ex]
\hline
$92$ \rule{0cm}{.65cm}  &
$\begin{smallmatrix}0&0&0&1\\0&0&1&0\\0&1&0&1\\1&0&1&0\end{smallmatrix}$ &
$\begin{smallmatrix}0&0&0&1\\0&1&0&0\\0&0&0&1\\1&0&1&1\end{smallmatrix}$ &
$\begin{smallmatrix}0&0&1&1\\0&1&0&1\\1&0&1&1\\1&1&1&1\end{smallmatrix}$ &
$\begin{smallmatrix}0&1&0&0\\1&1&1&0\\0&1&0&1\\0&0&1&1\end{smallmatrix}$ &
$\begin{smallmatrix}0&1\\1&1\end{smallmatrix}$ &
$\begin{smallmatrix}0&0&0&1\\0&0&0&1\\0&0&1&0\\1&1&0&0\end{smallmatrix}$ \\[1.85ex]
\hline
\end{tabular}
\end{table}
\pagebreak

\begin{table}[h!]
\begin{tabular}{|c||c|c|c|c|c|c|}
\hline
$98$ \rule{0cm}{.65cm}  &
$\begin{smallmatrix}0&0&0&1\\0&0&1&1\\0&1&0&1\\1&1&1&1\end{smallmatrix}$ &
$\begin{smallmatrix}0&0&1&0\\0&0&1&1\\1&1&1&0\\0&1&0&1\end{smallmatrix}$ &
$\begin{smallmatrix}0&0&1&0\\0&1&0&0\\1&0&1&0\\0&0&0&0\end{smallmatrix}$ &
$\begin{smallmatrix}1\end{smallmatrix}$ &
$\begin{smallmatrix}0&0&0&1\\0&0&1&0\\0&1&1&1\\1&0&1&0\end{smallmatrix}$ &
$\begin{smallmatrix}0&0&0&1\\0&0&0&1\\0&0&1&0\\1&1&0&0\end{smallmatrix}$ \\[1.85ex]
\hline
$104$ \rule{0cm}{.65cm}  &
$\begin{smallmatrix}0&0&0&1\\0&1&1&0\\0&1&1&1\\1&0&1&0\end{smallmatrix}$ &
$\begin{smallmatrix}1&1&0\\1&1&1\\0&1&0\end{smallmatrix}$ &
$\begin{smallmatrix}0&0&0&1\\0&1&1&1\\0&1&1&1\\1&1&1&1\end{smallmatrix}$ &
$\begin{smallmatrix}0&0&0&1\\0&1&1&1\\0&1&1&1\\1&1&1&0\end{smallmatrix}$ &
$\begin{smallmatrix}0&0&1\\0&1&0\\1&0&0\end{smallmatrix}$ &
$\begin{smallmatrix}1\end{smallmatrix}$ \\[1.85ex]
\hline
$110$ \rule{0cm}{.65cm}  &
$\begin{smallmatrix}0&0&0&1\\0&0&1&0\\0&1&0&1\\1&0&1&0\end{smallmatrix}$ &
$\begin{smallmatrix}0&0&1&1\\0&1&1&1\\1&1&1&1\\1&1&1&1\end{smallmatrix}$ &
$\begin{smallmatrix}1&0&1\\0&0&0\\1&0&1\end{smallmatrix}$ &
$\begin{smallmatrix}0\end{smallmatrix}$ &
$\begin{smallmatrix}0&1&1\\1&1&1\\1&1&1\end{smallmatrix}$ &
$\begin{smallmatrix}0&0&1\\0&1&0\\1&0&1\end{smallmatrix}$ \\[1.85ex]
\hline
$116$ \rule{0cm}{.65cm}  &
$\begin{smallmatrix}1&0&0\\0&1&0\\0&0&1\end{smallmatrix}$ &
$\begin{smallmatrix}0&0&0&1\\0&1&1&0\\0&1&0&0\\1&0&0&1\end{smallmatrix}$ &
$\begin{smallmatrix}0&0&1&0\\0&0&1&1\\1&1&0&1\\0&1&1&1\end{smallmatrix}$ &
$\begin{smallmatrix}0&0&1&1\\0&1&1&0\\1&1&1&0\\1&0&0&1\end{smallmatrix}$ &
$\begin{smallmatrix}1&0&0\\0&1&0\\0&0&1\end{smallmatrix}$ &
$\begin{smallmatrix}0&0&1\\0&0&0\\1&0&0\end{smallmatrix}$ \\[1.85ex]
\hline
$122$ \rule{0cm}{.65cm}  &
$\begin{smallmatrix}0&0&0&1\\0&0&1&0\\0&1&1&1\\1&0&1&0\end{smallmatrix}$ &
$\begin{smallmatrix}0&0&0&1\\0&1&0&1\\0&0&0&1\\1&1&1&1\end{smallmatrix}$ &
$\begin{smallmatrix}0&1&1\\1&0&1\\1&1&1\end{smallmatrix}$ &
$\begin{smallmatrix}1&0&0&1\\0&0&0&0\\0&0&0&0\\1&0&0&1\end{smallmatrix}$ &
$\begin{smallmatrix}1&0&1\\0&0&0\\1&0&1\end{smallmatrix}$ &
$\begin{smallmatrix}0&0&0&1\\0&1&1&1\\0&1&0&1\\1&1&1&1\end{smallmatrix}$ \\[1.85ex]
\hline
$128$ \rule{0cm}{.65cm}  &
$\begin{smallmatrix}0&1&0\\1&1&0\\0&0&1\end{smallmatrix}$ &
$\begin{smallmatrix}0&0&1&0\\0&0&1&1\\1&1&1&0\\0&1&0&1\end{smallmatrix}$ &
$\begin{smallmatrix}1&1&0\\1&0&0\\0&0&0\end{smallmatrix}$ &
$\begin{smallmatrix}0\end{smallmatrix}$ &
$\begin{smallmatrix}0&0&0&1\\0&0&1&0\\0&1&0&1\\1&0&1&1\end{smallmatrix}$ &
$\begin{smallmatrix}1&0&1\\0&0&1\\1&1&0\end{smallmatrix}$ \\[1.85ex]
\hline
$134$ \rule{0cm}{.65cm}  &
$\begin{smallmatrix}0\end{smallmatrix}$ &
$\begin{smallmatrix}1&0&1&1\\0&1&1&0\\1&1&0&1\\1&0&1&1\end{smallmatrix}$ &
$\begin{smallmatrix}1&1&0&0\\1&1&1&1\\0&1&1&1\\0&1&1&1\end{smallmatrix}$ &
$\begin{smallmatrix}0&0&0&1\\0&1&0&1\\0&0&0&0\\1&1&0&0\end{smallmatrix}$ &
$\begin{smallmatrix}0&0&1&0\\0&1&0&0\\1&0&0&0\\0&0&0&1\end{smallmatrix}$ &
$\begin{smallmatrix}1&1&0\\1&0&1\\0&1&0\end{smallmatrix}$ \\[1.85ex]
\hline
$140$ \rule{0cm}{.65cm}  &
$\begin{smallmatrix}0&0&0&1\\0&1&0&0\\0&0&0&0\\1&0&0&1\end{smallmatrix}$ &
$\begin{smallmatrix}0&0&0&1\\0&1&1&0\\0&1&0&0\\1&0&0&1\end{smallmatrix}$ &
$\begin{smallmatrix}0&1&0&0\\1&1&0&0\\0&0&0&1\\0&0&1&1\end{smallmatrix}$ &
$\begin{smallmatrix}0&1\\1&1\end{smallmatrix}$ &
$\begin{smallmatrix}1&0&0\\0&0&0\\0&0&1\end{smallmatrix}$ &
$\begin{smallmatrix}0&0&0&1\\0&1&1&0\\0&1&0&1\\1&0&1&0\end{smallmatrix}$ \\[1.85ex]
\hline
$146$ \rule{0cm}{.65cm}  &
$\begin{smallmatrix}0&0&0&1\\0&0&1&0\\0&1&0&0\\1&0&0&0\end{smallmatrix}$ &
$\begin{smallmatrix}0&0&1&1\\0&1&0&1\\1&0&1&0\\1&1&0&1\end{smallmatrix}$ &
$\begin{smallmatrix}0&1&1\\1&0&1\\1&1&1\end{smallmatrix}$ &
$\begin{smallmatrix}0&1&1&1\\1&1&1&1\\1&1&0&0\\1&1&0&0\end{smallmatrix}$ &
$\begin{smallmatrix}1&0&0\\0&0&1\\0&1&0\end{smallmatrix}$ &
$\begin{smallmatrix}0&0&1\\0&0&1\\1&1&1\end{smallmatrix}$ \\[1.85ex]
\hline
$152$ \rule{0cm}{.65cm}  &
$\begin{smallmatrix}0&0&1&0\\0&1&1&0\\1&1&1&0\\0&0&0&0\end{smallmatrix}$ &
$\begin{smallmatrix}0&0&0&1\\0&0&0&0\\0&0&0&1\\1&0&1&0\end{smallmatrix}$ &
$\begin{smallmatrix}0&0&0&1\\0&0&1&1\\0&1&1&0\\1&1&0&0\end{smallmatrix}$ &
$\begin{smallmatrix}0\end{smallmatrix}$ &
$\begin{smallmatrix}1&0&0\\0&0&0\\0&0&1\end{smallmatrix}$ &
$\begin{smallmatrix}0&1&1&1\\1&1&1&1\\1&1&0&0\\1&1&0&0\end{smallmatrix}$ \\[1.85ex]
\hline
$158$ \rule{0cm}{.65cm}  &
$\begin{smallmatrix}0\end{smallmatrix}$ &
$\begin{smallmatrix}0&0&1\\0&1&0\\1&0&0\end{smallmatrix}$ &
$\begin{smallmatrix}0&0&0&1\\0&0&0&0\\0&0&0&0\\1&0&0&1\end{smallmatrix}$ &
$\begin{smallmatrix}0&0&1&0\\0&0&0&1\\1&0&0&1\\0&1&1&1\end{smallmatrix}$ &
$\begin{smallmatrix}0&0&0&1\\0&1&1&0\\0&1&0&1\\1&0&1&0\end{smallmatrix}$ &
$\begin{smallmatrix}0&0&0&1\\0&0&0&1\\0&0&1&0\\1&1&0&0\end{smallmatrix}$ \\[1.85ex]
\hline
$164$ \rule{0cm}{.65cm}  &
$\begin{smallmatrix}0&0&0&1\\0&0&1&1\\0&1&0&1\\1&1&1&1\end{smallmatrix}$ &
$\begin{smallmatrix}0&0&0&1\\0&0&0&1\\0&0&0&1\\1&1&1&0\end{smallmatrix}$ &
$\begin{smallmatrix}0&0&0&1\\0&0&1&0\\0&1&0&1\\1&0&1&1\end{smallmatrix}$ &
$\begin{smallmatrix}0&0&0&1\\0&1&1&0\\0&1&0&0\\1&0&0&1\end{smallmatrix}$ &
$\begin{smallmatrix}0&0&1&0\\0&0&0&0\\1&0&1&0\\0&0&0&0\end{smallmatrix}$ &
$\begin{smallmatrix}0&0&0&1\\0&0&0&1\\0&0&0&0\\1&1&0&0\end{smallmatrix}$ \\[1.85ex]
\hline
$170$ \rule{0cm}{.65cm}  &
$\begin{smallmatrix}0&0&0&1\\0&1&0&0\\0&0&0&0\\1&0&0&1\end{smallmatrix}$ &
$\begin{smallmatrix}0&0&0&1\\0&0&1&0\\0&1&1&1\\1&0&1&0\end{smallmatrix}$ &
$\begin{smallmatrix}0&0&0&1\\0&0&1&1\\0&1&1&0\\1&1&0&0\end{smallmatrix}$ &
$\begin{smallmatrix}0\end{smallmatrix}$ &
$\begin{smallmatrix}0\end{smallmatrix}$ &
$\begin{smallmatrix}0&1&0&0\\1&0&1&0\\0&1&0&1\\0&0&1&1\end{smallmatrix}$ \\[1.85ex]
\hline
$176$ \rule{0cm}{.65cm}  &
$\begin{smallmatrix}0&1\\1&1\end{smallmatrix}$ &
$\begin{smallmatrix}0&0&0&1\\0&0&1&0\\0&1&0&1\\1&0&1&1\end{smallmatrix}$ &
$\begin{smallmatrix}0&0&0&1\\0&0&0&1\\0&0&1&0\\1&1&0&0\end{smallmatrix}$ &
$\begin{smallmatrix}0&0&0&1\\0&1&0&0\\0&0&0&0\\1&0&0&1\end{smallmatrix}$ &
$\begin{smallmatrix}0&0&1\\0&0&1\\1&1&1\end{smallmatrix}$ &
$\begin{smallmatrix}0&1&1&1\\1&1&0&1\\1&0&1&0\\1&1&0&0\end{smallmatrix}$ \\[1.85ex]
\hline
$182$ \rule{0cm}{.65cm}  &
$\begin{smallmatrix}0&1&0&0\\1&0&0&0\\0&0&0&1\\0&0&1&1\end{smallmatrix}$ &
$\begin{smallmatrix}1&0&0&1\\0&1&0&0\\0&0&0&1\\1&0&1&1\end{smallmatrix}$ &
$\begin{smallmatrix}0&1&0&0\\1&0&1&1\\0&1&0&1\\0&1&1&0\end{smallmatrix}$ &
$\begin{smallmatrix}0&1\\1&1\end{smallmatrix}$ &
$\begin{smallmatrix}0\end{smallmatrix}$ &
$\begin{smallmatrix}0&0&0&1\\0&0&0&0\\0&0&1&0\\1&0&0&1\end{smallmatrix}$ \\[1.85ex]
\hline
$188$ \rule{0cm}{.65cm}  &
$\begin{smallmatrix}0&0&1&1\\0&1&1&1\\1&1&1&1\\1&1&1&1\end{smallmatrix}$ &
$\begin{smallmatrix}0&0&0&1\\0&1&1&0\\0&1&0&0\\1&0&0&1\end{smallmatrix}$ &
$\begin{smallmatrix}1&0\\0&1\end{smallmatrix}$ &
$\begin{smallmatrix}0&0&1&0\\0&1&1&1\\1&1&0&1\\0&1&1&0\end{smallmatrix}$ &
$\begin{smallmatrix}0&0&1&0\\0&0&0&0\\1&0&1&0\\0&0&0&0\end{smallmatrix}$ &
$\begin{smallmatrix}0&0&1&0\\0&0&0&1\\1&0&0&1\\0&1&1&1\end{smallmatrix}$ \\[1.85ex]
\hline
$194$ \rule{0cm}{.65cm}  &
$\begin{smallmatrix}0\end{smallmatrix}$ &
$\begin{smallmatrix}0&0&1&1\\0&1&1&1\\1&1&1&0\\1&1&0&0\end{smallmatrix}$ &
$\begin{smallmatrix}0&0&0&1\\0&1&1&1\\0&1&1&0\\1&1&0&1\end{smallmatrix}$ &
$\begin{smallmatrix}1&0&1&0\\0&1&1&1\\1&1&1&0\\0&1&0&0\end{smallmatrix}$ &
$\begin{smallmatrix}0&0&0&1\\0&1&1&1\\0&1&0&0\\1&1&0&1\end{smallmatrix}$ &
$\begin{smallmatrix}1&0&0&1\\0&1&0&0\\0&0&1&0\\1&0&0&1\end{smallmatrix}$ \\[1.85ex]
\hline
$200$ \rule{0cm}{.65cm}  &
$\begin{smallmatrix}0&0&0&1\\0&0&0&1\\0&0&0&1\\1&1&1&0\end{smallmatrix}$ &
$\begin{smallmatrix}0&0&0&1\\0&0&1&0\\0&1&1&1\\1&0&1&0\end{smallmatrix}$ &
$\begin{smallmatrix}1&0\\0&0\end{smallmatrix}$ &
$\begin{smallmatrix}0&0&1&0\\0&0&1&0\\1&1&1&1\\0&0&1&1\end{smallmatrix}$ &
$\begin{smallmatrix}0&0&0&1\\0&0&1&0\\0&1&1&1\\1&0&1&0\end{smallmatrix}$ &
$\begin{smallmatrix}0&0&0&1\\0&0&1&0\\0&1&1&1\\1&0&1&1\end{smallmatrix}$ \\[1.85ex]
\hline
$206$ \rule{0cm}{.65cm}  &
$\begin{smallmatrix}1&0&0&0\\0&0&0&0\\0&0&0&1\\0&0&1&1\end{smallmatrix}$ &
$\begin{smallmatrix}0&0&1&1\\0&1&1&0\\1&1&1&1\\1&0&1&1\end{smallmatrix}$ &
$\begin{smallmatrix}0&1&1\\1&0&1\\1&1&1\end{smallmatrix}$ &
$\begin{smallmatrix}0\end{smallmatrix}$ &
$\begin{smallmatrix}0\end{smallmatrix}$ &
$\begin{smallmatrix}0&0&1&1\\0&0&1&0\\1&1&1&1\\1&0&1&0\end{smallmatrix}$ \\[1.85ex]
\hline
\end{tabular}
\end{table}
\pagebreak

\begin{table}[h!]
\begin{tabular}{|c||c|c|c|c|c|c|}
\hline
$212$ \rule{0cm}{.65cm}  &
$\begin{smallmatrix}0&0&0&1\\0&1&0&0\\0&0&0&0\\1&0&0&1\end{smallmatrix}$ &
$\begin{smallmatrix}0&0&0&1\\0&0&0&1\\0&0&1&0\\1&1&0&1\end{smallmatrix}$ &
$\begin{smallmatrix}0&1&0&0\\1&1&0&1\\0&0&1&0\\0&1&0&1\end{smallmatrix}$ &
$\begin{smallmatrix}0&1&1&1\\1&0&0&1\\1&0&1&1\\1&1&1&1\end{smallmatrix}$ &
$\begin{smallmatrix}0&0&1\\0&1&0\\1&0&1\end{smallmatrix}$ &
$\begin{smallmatrix}1&1&0\\1&1&1\\0&1&0\end{smallmatrix}$ \\[1.85ex]
\hline
$218$ \rule{0cm}{.65cm}  &
$\begin{smallmatrix}0&0&0&1\\0&0&0&1\\0&0&0&1\\1&1&1&0\end{smallmatrix}$ &
$\begin{smallmatrix}0&0&0&1\\0&0&1&0\\0&1&1&1\\1&0&1&0\end{smallmatrix}$ &
$\begin{smallmatrix}0&1&0&0\\1&0&1&0\\0&1&0&1\\0&0&1&1\end{smallmatrix}$ &
$\begin{smallmatrix}0\end{smallmatrix}$ &
$\begin{smallmatrix}1&0&0\\0&0&1\\0&1&0\end{smallmatrix}$ &
$\begin{smallmatrix}0&1&0&0\\1&1&1&0\\0&1&1&0\\0&0&0&1\end{smallmatrix}$ \\[1.85ex]
\hline
$224$ \rule{0cm}{.65cm}  &
$\begin{smallmatrix}0&0&0&1\\0&1&1&1\\0&1&1&1\\1&1&1&0\end{smallmatrix}$ &
$\begin{smallmatrix}0&1&0&0\\1&1&1&1\\0&1&0&1\\0&1&1&0\end{smallmatrix}$ &
$\begin{smallmatrix}1&1&0\\1&0&0\\0&0&0\end{smallmatrix}$ &
$\begin{smallmatrix}0&0&0&1\\0&1&0&0\\0&0&0&0\\1&0&0&1\end{smallmatrix}$ &
$\begin{smallmatrix}1&0&1&0\\0&0&1&0\\1&1&1&1\\0&0&1&1\end{smallmatrix}$ &
$\begin{smallmatrix}0&0&1\\0&1&0\\1&0&0\end{smallmatrix}$ \\[1.85ex]
\hline
$230$ \rule{0cm}{.75cm}  &
$\begin{smallmatrix}0&0&0&1\\0&1&0&1\\0&0&0&0\\1&1&0&0\end{smallmatrix}$ &
$\begin{smallmatrix}0\end{smallmatrix}$ &
$\begin{smallmatrix}1&0\\0&1\end{smallmatrix}$ &
$\begin{smallmatrix}0&0&0&0&1\\0&0&0&0&0\\0&0&0&1&0\\0&0&1&0&1\\1&0&0&1&1\end{smallmatrix}$ &
$\begin{smallmatrix}1&1&1&0\\1&1&1&1\\1&1&1&1\\0&1&1&0\end{smallmatrix}$ &
$\begin{smallmatrix}0&0&1&1\\0&1&0&1\\1&0&1&1\\1&1&1&1\end{smallmatrix}$ \\[2.45ex]
\hline
$236$ \rule{0cm}{.65cm}  &
$\begin{smallmatrix}0&0&0&1\\0&1&1&1\\0&1&0&1\\1&1&1&0\end{smallmatrix}$ &
$\begin{smallmatrix}1&0&1\\0&1&0\\1&0&1\end{smallmatrix}$ &
$\begin{smallmatrix}0&0&1&1\\0&0&0&1\\1&0&0&1\\1&1&1&1\end{smallmatrix}$ &
$\begin{smallmatrix}1&0&0&0\\0&1&0&1\\0&0&0&1\\0&1&1&0\end{smallmatrix}$ &
$\begin{smallmatrix}0&1&1\\1&0&0\\1&0&1\end{smallmatrix}$ &
$\begin{smallmatrix}0&1&0&0\\1&1&0&0\\0&0&1&0\\0&0&0&0\end{smallmatrix}$ \\[1.85ex]
\hline
$242$ \rule{0cm}{.65cm}  &
$\begin{smallmatrix}0&0&0&1\\0&0&1&1\\0&1&0&1\\1&1&1&1\end{smallmatrix}$ &
$\begin{smallmatrix}0\end{smallmatrix}$ &
$\begin{smallmatrix}0&0&1\\0&1&0\\1&0&0\end{smallmatrix}$ &
$\begin{smallmatrix}0&1&0&0\\1&1&1&0\\0&1&1&0\\0&0&0&1\end{smallmatrix}$ &
$\begin{smallmatrix}0&0&0&1\\0&1&1&0\\0&1&1&1\\1&0&1&0\end{smallmatrix}$ &
$\begin{smallmatrix}0&0&1&0\\0&0&0&1\\1&0&1&0\\0&1&0&1\end{smallmatrix}$ \\[1.85ex]
\hline
$248$ \rule{0cm}{.65cm}  &
$\begin{smallmatrix}0&0&0&1\\0&1&1&1\\0&1&0&1\\1&1&1&0\end{smallmatrix}$ &
$\begin{smallmatrix}0&0&1&1\\0&1&0&1\\1&0&1&1\\1&1&1&0\end{smallmatrix}$ &
$\begin{smallmatrix}0&0&0&1\\0&0&0&0\\0&0&0&0\\1&0&0&1\end{smallmatrix}$ &
$\begin{smallmatrix}0\end{smallmatrix}$ &
$\begin{smallmatrix}0&0&1&1\\0&0&1&0\\1&1&1&1\\1&0&1&0\end{smallmatrix}$ &
$\begin{smallmatrix}0&0&0&1\\0&1&0&0\\0&0&0&1\\1&0&1&0\end{smallmatrix}$ \\[1.85ex]
\hline
$254$ \rule{0cm}{.65cm}  &
$\begin{smallmatrix}0&0&0&1\\0&0&1&0\\0&1&0&0\\1&0&0&0\end{smallmatrix}$ &
$\begin{smallmatrix}0&0&0&1\\0&1&1&0\\0&1&1&1\\1&0&1&0\end{smallmatrix}$ &
$\begin{smallmatrix}1&0&0&0\\0&0&0&0\\0&0&0&0\\0&0&0&0\end{smallmatrix}$ &
$\begin{smallmatrix}0&1&1&1\\1&1&0&1\\1&0&1&1\\1&1&1&1\end{smallmatrix}$ &
$\begin{smallmatrix}0&0&1&1\\0&0&0&1\\1&0&1&0\\1&1&0&0\end{smallmatrix}$ &
$\begin{smallmatrix}0&1&1&0\\1&1&0&0\\1&0&0&1\\0&0&1&0\end{smallmatrix}$ \\[1.85ex]
\hline
$260$ \rule{0cm}{.65cm}  &
$\begin{smallmatrix}1&0&0&0\\0&0&1&1\\0&1&1&1\\0&1&1&1\end{smallmatrix}$ &
$\begin{smallmatrix}0\end{smallmatrix}$ &
$\begin{smallmatrix}0&0&1\\0&0&0\\1&0&0\end{smallmatrix}$ &
$\begin{smallmatrix}1&0&0&1\\0&1&1&0\\0&1&0&1\\1&0&1&0\end{smallmatrix}$ &
$\begin{smallmatrix}1&0&0&1\\0&1&1&0\\0&1&0&1\\1&0&1&0\end{smallmatrix}$ &
$\begin{smallmatrix}1&0&1&1\\0&0&1&1\\1&1&1&0\\1&1&0&1\end{smallmatrix}$ \\[1.85ex]
\hline
$266$ \rule{0cm}{.65cm}  &
$\begin{smallmatrix}0&1\\1&0\end{smallmatrix}$ &
$\begin{smallmatrix}0&0&0&1\\0&0&0&1\\0&0&1&0\\1&1&0&1\end{smallmatrix}$ &
$\begin{smallmatrix}0&0&0&1\\0&1&1&1\\0&1&1&1\\1&1&1&1\end{smallmatrix}$ &
$\begin{smallmatrix}0&0&1&1\\0&0&0&1\\1&0&1&0\\1&1&0&1\end{smallmatrix}$ &
$\begin{smallmatrix}0\end{smallmatrix}$ &
$\begin{smallmatrix}0&0&1&0\\0&1&0&1\\1&0&1&0\\0&1&0&1\end{smallmatrix}$ \\[1.85ex]
\hline
$272$ \rule{0cm}{.65cm}  &
$\begin{smallmatrix}0&0&0&1\\0&1&0&0\\0&0&0&0\\1&0&0&1\end{smallmatrix}$ &
$\begin{smallmatrix}0&0&1&0\\0&1&0&0\\1&0&0&0\\0&0&0&0\end{smallmatrix}$ &
$\begin{smallmatrix}0&0&1&1\\0&0&1&1\\1&1&1&1\\1&1&1&1\end{smallmatrix}$ &
$\begin{smallmatrix}1&1&1&0\\1&1&0&1\\1&0&0&0\\0&1&0&1\end{smallmatrix}$ &
$\begin{smallmatrix}0&1&0&0\\1&1&1&1\\0&1&1&1\\0&1&1&0\end{smallmatrix}$ &
$\begin{smallmatrix}1&1&0\\1&1&1\\0&1&0\end{smallmatrix}$ \\[1.85ex]
\hline
$278$ \rule{0cm}{.65cm}  &
$\begin{smallmatrix}0\end{smallmatrix}$ &
$\begin{smallmatrix}1&0&0\\0&0&1\\0&1&0\end{smallmatrix}$ &
$\begin{smallmatrix}1&0&0&0\\0&0&0&0\\0&0&0&0\\0&0&0&0\end{smallmatrix}$ &
$\begin{smallmatrix}0\end{smallmatrix}$ &
$\begin{smallmatrix}0&1&0&0\\1&1&1&1\\0&1&1&1\\0&1&1&1\end{smallmatrix}$ &
$\begin{smallmatrix}1&0&1\\0&0&1\\1&1&1\end{smallmatrix}$ \\[1.85ex]
\hline
$284$ \rule{0cm}{.65cm}  &
$\begin{smallmatrix}0&1&1\\1&0&1\\1&1&1\end{smallmatrix}$ &
$\begin{smallmatrix}1&1&0&0\\1&0&1&1\\0&1&0&1\\0&1&1&0\end{smallmatrix}$ &
$\begin{smallmatrix}0&1&0&0\\1&0&1&0\\0&1&1&1\\0&0&1&0\end{smallmatrix}$ &
$\begin{smallmatrix}0&0&1&0\\0&0&0&0\\1&0&0&0\\0&0&0&1\end{smallmatrix}$ &
$\begin{smallmatrix}1&0&0\\0&0&0\\0&0&1\end{smallmatrix}$ &
$\begin{smallmatrix}1&0&0&1\\0&0&0&0\\0&0&1&0\\1&0&0&1\end{smallmatrix}$ \\[1.85ex]
\hline
$290$ \rule{0cm}{.75cm}  &
$\begin{smallmatrix}0&0&0&1\\0&0&0&1\\0&0&0&0\\1&1&0&0\end{smallmatrix}$ &
$\begin{smallmatrix}1&0&0&1\\0&0&0&0\\0&0&0&1\\1&0&1&0\end{smallmatrix}$ &
$\begin{smallmatrix}1&0\\0&1\end{smallmatrix}$ &
$\begin{smallmatrix}0\end{smallmatrix}$ &
$\begin{smallmatrix}0&0&0&1\\0&0&1&1\\0&1&0&0\\1&1&0&0\end{smallmatrix}$ &
$\begin{smallmatrix}0&0&0&0&1\\0&0&0&1&1\\0&0&0&0&0\\0&1&0&1&0\\1&1&0&0&0\end{smallmatrix}$ \\[2.45ex]
\hline
$296$ \rule{0cm}{.75cm}  &
$\begin{smallmatrix}0&0&1&1\\0&0&0&0\\1&0&1&0\\1&0&0&0\end{smallmatrix}$ &
$\begin{smallmatrix}0&1&0\\1&0&0\\0&0&1\end{smallmatrix}$ &
$\begin{smallmatrix}0&0&0&0&1\\0&0&0&1&1\\0&0&0&0&0\\0&1&0&1&0\\1&1&0&0&1\end{smallmatrix}$ &
$\begin{smallmatrix}0&0&0&0&1\\0&0&1&0&1\\0&1&0&1&1\\0&0&1&1&0\\1&1&1&0&0\end{smallmatrix}$ &
$\begin{smallmatrix}0&0&0&1\\0&1&0&0\\0&0&1&0\\1&0&0&0\end{smallmatrix}$ &
$\begin{smallmatrix}0&0&0&1\\0&0&1&0\\0&1&1&0\\1&0&0&0\end{smallmatrix}$ \\[2.45ex]
\hline
$302$ \rule{0cm}{.65cm}  &
$\begin{smallmatrix}0&1&0&0\\1&1&1&1\\0&1&1&1\\0&1&1&1\end{smallmatrix}$ &
$\begin{smallmatrix}0&1&1&0\\1&1&0&0\\1&0&1&0\\0&0&0&1\end{smallmatrix}$ &
$\begin{smallmatrix}0&0&1&1\\0&0&0&1\\1&0&1&0\\1&1&0&1\end{smallmatrix}$ &
$\begin{smallmatrix}0&0&0&1\\0&1&1&0\\0&1&0&0\\1&0&0&1\end{smallmatrix}$ &
$\begin{smallmatrix}0&1&0&1\\1&0&1&1\\0&1&0&0\\1&1&0&1\end{smallmatrix}$ &
$\begin{smallmatrix}0&0&1&1\\0&0&0&1\\1&0&1&0\\1&1&0&0\end{smallmatrix}$ \\[1.85ex]
\hline
$308$ \rule{0cm}{.65cm}  &
$\begin{smallmatrix}0&0&1&0\\0&0&0&0\\1&0&1&0\\0&0&0&0\end{smallmatrix}$ &
$\begin{smallmatrix}0\end{smallmatrix}$ &
$\begin{smallmatrix}0&0&0&1\\0&1&0&0\\0&0&0&0\\1&0&0&1\end{smallmatrix}$ &
$\begin{smallmatrix}0&0&1&1\\0&1&0&1\\1&0&1&1\\1&1&1&1\end{smallmatrix}$ &
$\begin{smallmatrix}1&0&0&1\\0&1&1&0\\0&1&0&0\\1&0&0&0\end{smallmatrix}$ &
$\begin{smallmatrix}0&0&1&0\\0&1&0&1\\1&0&1&1\\0&1&1&1\end{smallmatrix}$ \\[1.85ex]
\hline
$314$ \rule{0cm}{.75cm}  &
$\begin{smallmatrix}0&0&0&0&1\\0&0&0&0&1\\0&0&1&1&0\\0&0&1&0&0\\1&1&0&0&1\end{smallmatrix}$ &
$\begin{smallmatrix}0&0&1&0\\0&1&0&1\\1&0&1&1\\0&1&1&0\end{smallmatrix}$ &
$\begin{smallmatrix}1&1&0\\1&1&1\\0&1&0\end{smallmatrix}$ &
$\begin{smallmatrix}0&1&1\\1&1&1\\1&1&1\end{smallmatrix}$ &
$\begin{smallmatrix}0&1&0&1\\1&0&0&1\\0&0&1&0\\1&1&0&0\end{smallmatrix}$ &
$\begin{smallmatrix}0&0&1&1\\0&1&0&0\\1&0&1&0\\1&0&0&0\end{smallmatrix}$ \\[2.45ex]
\hline
\end{tabular}
\end{table}
\pagebreak

\begin{table}[h!]
\begin{tabular}{|c||c|c|c|c|c|c|}
\hline
$320$ \rule{0cm}{.75cm}  &
$\begin{smallmatrix}0&0&1&1\\0&0&0&0\\1&0&1&1\\1&0&1&0\end{smallmatrix}$ &
$\begin{smallmatrix}0&0&1&0\\0&0&1&0\\1&1&1&0\\0&0&0&1\end{smallmatrix}$ &
$\begin{smallmatrix}1&0&0&0\\0&0&0&1\\0&0&0&1\\0&1&1&1\end{smallmatrix}$ &
$\begin{smallmatrix}0\end{smallmatrix}$ &
$\begin{smallmatrix}0&0&0&0&1\\0&0&0&1&1\\0&0&0&1&1\\0&1&1&1&1\\1&1&1&1&1\end{smallmatrix}$ &
$\begin{smallmatrix}0&0&1&0\\0&0&1&1\\1&1&1&0\\0&1&0&1\end{smallmatrix}$ \\[2.45ex]
\hline
$326$ \rule{0cm}{.75cm}  &
$\begin{smallmatrix}0\end{smallmatrix}$ &
$\begin{smallmatrix}0&0&1&1\\0&0&0&0\\1&0&1&0\\1&0&0&0\end{smallmatrix}$ &
$\begin{smallmatrix}0&0&1&0\\0&1&0&1\\1&0&1&0\\0&1&0&1\end{smallmatrix}$ &
$\begin{smallmatrix}0\end{smallmatrix}$ &
$\begin{smallmatrix}1&0&0\\0&1&0\\0&0&1\end{smallmatrix}$ &
$\begin{smallmatrix}0&0&0&0&1\\0&0&0&1&1\\0&0&1&0&1\\0&1&0&1&1\\1&1&1&1&0\end{smallmatrix}$ \\[2.45ex]
\hline
$332$ \rule{0cm}{.65cm}  &
$\begin{smallmatrix}1&1&1\\1&1&0\\1&0&0\end{smallmatrix}$ &
$\begin{smallmatrix}1&0&0&1\\0&0&1&1\\0&1&0&0\\1&1&0&0\end{smallmatrix}$ &
$\begin{smallmatrix}1&0\\0&1\end{smallmatrix}$ &
$\begin{smallmatrix}1&0&1\\0&1&0\\1&0&0\end{smallmatrix}$ &
$\begin{smallmatrix}1&0&1&1\\0&0&1&0\\1&1&1&1\\1&0&1&1\end{smallmatrix}$ &
$\begin{smallmatrix}0&1&0&0\\1&1&0&0\\0&0&1&0\\0&0&0&0\end{smallmatrix}$ \\[1.85ex]
\hline
$338$ \rule{0cm}{.75cm}  &
$\begin{smallmatrix}0\end{smallmatrix}$ &
$\begin{smallmatrix}0&1&1&0\\1&1&1&1\\1&1&1&0\\0&1&0&0\end{smallmatrix}$ &
$\begin{smallmatrix}0&0&0&0&1\\0&0&0&0&0\\0&0&0&0&0\\0&0&0&0&0\\1&0&0&0&0\end{smallmatrix}$ &
$\begin{smallmatrix}0&0&1&0\\0&1&0&0\\1&0&1&0\\0&0&0&1\end{smallmatrix}$ &
$\begin{smallmatrix}0&0&1&1\\0&0&0&0\\1&0&0&1\\1&0&1&1\end{smallmatrix}$ &
$\begin{smallmatrix}1\end{smallmatrix}$ \\[2.45ex]
\hline
$344$ \rule{0cm}{.65cm}  &
$\begin{smallmatrix}1&0&0&1\\0&1&1&0\\0&1&0&1\\1&0&1&0\end{smallmatrix}$ &
$\begin{smallmatrix}0&1&1&1\\1&0&0&1\\1&0&1&1\\1&1&1&1\end{smallmatrix}$ &
$\begin{smallmatrix}0&0&1\\0&1&0\\1&0&0\end{smallmatrix}$ &
$\begin{smallmatrix}0&0&0&1\\0&1&0&1\\0&0&0&0\\1&1&0&0\end{smallmatrix}$ &
$\begin{smallmatrix}0&0&0&1\\0&0&0&0\\0&0&0&1\\1&0&1&0\end{smallmatrix}$ &
$\begin{smallmatrix}0&0&1&0\\0&1&0&0\\1&0&0&0\\0&0&0&0\end{smallmatrix}$ \\[1.85ex]
\hline
$350$ \rule{0cm}{.65cm}  &
$\begin{smallmatrix}0\end{smallmatrix}$ &
$\begin{smallmatrix}0&0&1\\0&1&0\\1&0&1\end{smallmatrix}$ &
$\begin{smallmatrix}0&0&0&1\\0&1&0&0\\0&0&0&1\\1&0&1&0\end{smallmatrix}$ &
$\begin{smallmatrix}1\end{smallmatrix}$ &
$\begin{smallmatrix}0\end{smallmatrix}$ &
$\begin{smallmatrix}1&1&0&0\\1&1&0&0\\0&0&0&1\\0&0&1&0\end{smallmatrix}$ \\[1.85ex]
\hline
$356$ \rule{0cm}{.75cm}  &
$\begin{smallmatrix}0&1\\1&0\end{smallmatrix}$ &
$\begin{smallmatrix}1&0&1&1\\0&0&0&1\\1&0&0&1\\1&1&1&0\end{smallmatrix}$ &
$\begin{smallmatrix}0&0&0&0&1\\0&0&0&0&0\\0&0&0&1&1\\0&0&1&1&0\\1&0&1&0&1\end{smallmatrix}$ &
$\begin{smallmatrix}0\end{smallmatrix}$ &
$\begin{smallmatrix}0&0&0&1\\0&0&0&0\\0&0&1&0\\1&0&0&1\end{smallmatrix}$ &
$\begin{smallmatrix}0&0&0&1\\0&0&1&0\\0&1&1&0\\1&0&0&0\end{smallmatrix}$ \\[2.45ex]
\hline
$362$ \rule{0cm}{.75cm}  &
$\begin{smallmatrix}0&0&0&0&1\\0&0&1&0&0\\0&1&1&0&1\\0&0&0&0&1\\1&0&1&1&0\end{smallmatrix}$ &
$\begin{smallmatrix}0&0&1&0\\0&1&0&0\\1&0&0&0\\0&0&0&0\end{smallmatrix}$ &
$\begin{smallmatrix}1&0&1\\0&0&1\\1&1&0\end{smallmatrix}$ &
$\begin{smallmatrix}1&0&0&1\\0&1&0&0\\0&0&1&0\\1&0&0&1\end{smallmatrix}$ &
$\begin{smallmatrix}0&0&1\\0&1&0\\1&0&0\end{smallmatrix}$ &
$\begin{smallmatrix}0&1&1&0\\1&1&1&0\\1&1&1&0\\0&0&0&1\end{smallmatrix}$ \\[2.45ex]
\hline
$368$ \rule{0cm}{.75cm}  &
$\begin{smallmatrix}0&1&1&0\\1&1&0&1\\1&0&1&1\\0&1&1&1\end{smallmatrix}$ &
$\begin{smallmatrix}0&0&0&0&1\\0&0&0&0&1\\0&0&1&0&1\\0&0&0&0&1\\1&1&1&1&1\end{smallmatrix}$ &
$\begin{smallmatrix}0&1&1&0\\1&1&0&0\\1&0&0&1\\0&0&1&0\end{smallmatrix}$ &
$\begin{smallmatrix}0&0&0&0&1\\0&0&0&0&1\\0&0&1&1&1\\0&0&1&1&0\\1&1&1&0&0\end{smallmatrix}$ &
$\begin{smallmatrix}0&0&1&0\\0&1&1&0\\1&1&0&1\\0&0&1&0\end{smallmatrix}$ &
$\begin{smallmatrix}0&0&1&1\\0&1&0&0\\1&0&1&0\\1&0&0&0\end{smallmatrix}$ \\[2.45ex]
\hline
$374$ \rule{0cm}{.75cm}  &
$\begin{smallmatrix}1&0&0&0\\0&1&0&0\\0&0&1&0\\0&0&0&1\end{smallmatrix}$ &
$\begin{smallmatrix}0\end{smallmatrix}$ &
$\begin{smallmatrix}0&0&0&0&1\\0&0&0&0&0\\0&0&0&0&1\\0&0&0&0&1\\1&0&1&1&0\end{smallmatrix}$ &
$\begin{smallmatrix}0&0&0&1\\0&1&0&0\\0&0&0&0\\1&0&0&1\end{smallmatrix}$ &
$\begin{smallmatrix}1&0&1&1\\0&1&0&1\\1&0&1&1\\1&1&1&0\end{smallmatrix}$ &
$\begin{smallmatrix}0&1&1&1\\1&1&0&1\\1&0&1&1\\1&1&1&1\end{smallmatrix}$ \\[2.45ex]
\hline
$380$ \rule{0cm}{.75cm}  &
$\begin{smallmatrix}1&1&1&1\\1&0&0&0\\1&0&0&1\\1&0&1&0\end{smallmatrix}$ &
$\begin{smallmatrix}1&0\\0&0\end{smallmatrix}$ &
$\begin{smallmatrix}0&0&1\\0&0&0\\1&0&0\end{smallmatrix}$ &
$\begin{smallmatrix}1&0&0&1\\0&1&0&0\\0&0&1&0\\1&0&0&0\end{smallmatrix}$ &
$\begin{smallmatrix}0&0&1&0\\0&0&1&1\\1&1&1&0\\0&1&0&1\end{smallmatrix}$ &
$\begin{smallmatrix}0&0&0&0&1\\0&0&0&1&1\\0&0&0&0&0\\0&1&0&1&0\\1&1&0&0&0\end{smallmatrix}$ \\[2.45ex]
\hline
$386$ \rule{0cm}{.65cm}  &
$\begin{smallmatrix}0&0&1&0\\0&0&1&0\\1&1&1&0\\0&0&0&1\end{smallmatrix}$ &
$\begin{smallmatrix}0&0&0&1\\0&0&0&0\\0&0&1&0\\1&0&0&1\end{smallmatrix}$ &
$\begin{smallmatrix}0&1&0&0\\1&1&0&0\\0&0&1&0\\0&0&0&1\end{smallmatrix}$ &
$\begin{smallmatrix}0&1&0\\1&1&0\\0&0&1\end{smallmatrix}$ &
$\begin{smallmatrix}1&0&1&1\\0&1&0&1\\1&0&0&0\\1&1&0&1\end{smallmatrix}$ &
$\begin{smallmatrix}0&0&1&0\\0&1&0&1\\1&0&1&0\\0&1&0&1\end{smallmatrix}$ \\[1.85ex]
\hline
$392$ \rule{0cm}{.65cm}  &
$\begin{smallmatrix}0&0&0&1\\0&1&0&1\\0&0&0&0\\1&1&0&0\end{smallmatrix}$ &
$\begin{smallmatrix}1&0&1&0\\0&1&0&1\\1&0&1&0\\0&1&0&0\end{smallmatrix}$ &
$\begin{smallmatrix}0&0&1&1\\0&1&0&1\\1&0&1&0\\1&1&0&1\end{smallmatrix}$ &
$\begin{smallmatrix}0&0&1&1\\0&0&1&1\\1&1&1&0\\1&1&0&0\end{smallmatrix}$ &
$\begin{smallmatrix}1&0&1&0\\0&0&1&1\\1&1&1&1\\0&1&1&1\end{smallmatrix}$ &
$\begin{smallmatrix}1&1&1\\1&1&0\\1&0&0\end{smallmatrix}$ \\[1.85ex]
\hline
$398$ \rule{0cm}{.65cm}  &
$\begin{smallmatrix}0\end{smallmatrix}$ &
$\begin{smallmatrix}0&0&1&0\\0&1&0&0\\1&0&0&0\\0&0&0&1\end{smallmatrix}$ &
$\begin{smallmatrix}1&0\\0&1\end{smallmatrix}$ &
$\begin{smallmatrix}0&0&1&1\\0&0&0&1\\1&0&1&1\\1&1&1&1\end{smallmatrix}$ &
$\begin{smallmatrix}0&0&1&1\\0&1&0&0\\1&0&0&1\\1&0&1&0\end{smallmatrix}$ &
$\begin{smallmatrix}1\end{smallmatrix}$ \\[1.85ex]
\hline
$404$ \rule{0cm}{.75cm}  &
$\begin{smallmatrix}0&0&1&1\\0&0&0&0\\1&0&1&0\\1&0&0&0\end{smallmatrix}$ &
$\begin{smallmatrix}1&0&0&1\\0&0&1&0\\0&1&0&1\\1&0&1&0\end{smallmatrix}$ &
$\begin{smallmatrix}0&0&0&1\\0&1&1&1\\0&1&1&1\\1&1&1&1\end{smallmatrix}$ &
$\begin{smallmatrix}0&0&0&0&1\\0&0&1&0&1\\0&1&1&1&1\\0&0&1&0&0\\1&1&1&0&0\end{smallmatrix}$ &
$\begin{smallmatrix}0&0&1&0\\0&1&0&1\\1&0&1&0\\0&1&0&1\end{smallmatrix}$ &
$\begin{smallmatrix}0&1&1&0\\1&1&1&0\\1&1&1&0\\0&0&0&1\end{smallmatrix}$ \\[2.45ex]
\hline
$410$ \rule{0cm}{.65cm}  &
$\begin{smallmatrix}0\end{smallmatrix}$ &
$\begin{smallmatrix}0\end{smallmatrix}$ &
$\begin{smallmatrix}0&1&0\\1&0&0\\0&0&1\end{smallmatrix}$ &
$\begin{smallmatrix}0&1&1&0\\1&1&0&1\\1&0&1&1\\0&1&1&1\end{smallmatrix}$ &
$\begin{smallmatrix}0\end{smallmatrix}$ &
$\begin{smallmatrix}0&0&1\\0&0&1\\1&1&1\end{smallmatrix}$ \\[1.85ex]
\hline
$416$ \rule{0cm}{.65cm}  &
$\begin{smallmatrix}1&0&1&0\\0&1&0&0\\1&0&1&0\\0&0&0&0\end{smallmatrix}$ &
$\begin{smallmatrix}1&0&0&1\\0&1&1&1\\0&1&0&0\\1&1&0&1\end{smallmatrix}$ &
$\begin{smallmatrix}1&0&0&0\\0&1&0&0\\0&0&0&0\\0&0&0&1\end{smallmatrix}$ &
$\begin{smallmatrix}0&0&1&1\\0&1&1&1\\1&1&1&0\\1&1&0&0\end{smallmatrix}$ &
$\begin{smallmatrix}0&0&0&1\\0&0&0&0\\0&0&1&0\\1&0&0&1\end{smallmatrix}$ &
$\begin{smallmatrix}1&0&1&0\\0&0&1&1\\1&1&0&0\\0&1&0&1\end{smallmatrix}$ \\[1.85ex]
\hline
\end{tabular}
\end{table}
\pagebreak

\begin{table}[h!]
\begin{tabular}{|c||c|c|c|c|c|c|}
\hline
$422$ \rule{0cm}{.75cm}  &
$\begin{smallmatrix}0&0&1&0\\0&1&0&1\\1&0&0&1\\0&1&1&0\end{smallmatrix}$ &
$\begin{smallmatrix}1&1&1&1\\1&0&1&1\\1&1&0&0\\1&1&0&1\end{smallmatrix}$ &
$\begin{smallmatrix}0&1&1\\1&0&1\\1&1&1\end{smallmatrix}$ &
$\begin{smallmatrix}0&0&0&1&0\\0&0&0&0&1\\0&0&1&1&0\\1&0&1&1&0\\0&1&0&0&1\end{smallmatrix}$ &
$\begin{smallmatrix}0\end{smallmatrix}$ &
$\begin{smallmatrix}0&0&0&1&0\\0&0&1&0&1\\0&1&1&0&0\\1&0&0&1&0\\0&1&0&0&1\end{smallmatrix}$ \\[2.45ex]
\hline
$428$ \rule{0cm}{.65cm}  &
$\begin{smallmatrix}0&1&1&1\\1&1&1&0\\1&1&0&0\\1&0&0&0\end{smallmatrix}$ &
$\begin{smallmatrix}0&1&1&0\\1&1&1&1\\1&1&1&0\\0&1&0&0\end{smallmatrix}$ &
$\begin{smallmatrix}0&0&1&0\\0&0&0&1\\1&0&0&0\\0&1&0&0\end{smallmatrix}$ &
$\begin{smallmatrix}0\end{smallmatrix}$ &
$\begin{smallmatrix}0&0&1&0\\0&1&0&0\\1&0&1&0\\0&0&0&1\end{smallmatrix}$ &
$\begin{smallmatrix}0&0&1&1\\0&0&0&0\\1&0&1&1\\1&0&1&1\end{smallmatrix}$ \\[1.85ex]
\hline
$434$ \rule{0cm}{.75cm}  &
$\begin{smallmatrix}0&0&0&0&1\\0&0&0&0&1\\0&0&0&1&1\\0&0&1&0&0\\1&1&1&0&0\end{smallmatrix}$ &
$\begin{smallmatrix}1&0&1&1\\0&1&0&1\\1&0&0&0\\1&1&0&1\end{smallmatrix}$ &
$\begin{smallmatrix}1&0&0&0\\0&1&0&1\\0&0&1&0\\0&1&0&0\end{smallmatrix}$ &
$\begin{smallmatrix}1&0&0&0\\0&1&0&1\\0&0&1&0\\0&1&0&1\end{smallmatrix}$ &
$\begin{smallmatrix}0\end{smallmatrix}$ &
$\begin{smallmatrix}0&0&0&0&1\\0&1&0&1&1\\0&0&1&0&0\\0&1&0&1&0\\1&1&0&0&0\end{smallmatrix}$ \\[2.45ex]
\hline
$440$ \rule{0cm}{.65cm}  &
$\begin{smallmatrix}0&1&1&0\\1&1&1&0\\1&1&1&0\\0&0&0&1\end{smallmatrix}$ &
$\begin{smallmatrix}0\end{smallmatrix}$ &
$\begin{smallmatrix}0&1&1&1\\1&1&0&0\\1&0&1&0\\1&0&0&1\end{smallmatrix}$ &
$\begin{smallmatrix}0\end{smallmatrix}$ &
$\begin{smallmatrix}0&0&1&1\\0&0&0&1\\1&0&1&0\\1&1&0&0\end{smallmatrix}$ &
$\begin{smallmatrix}0&0&1&0\\0&0&0&0\\1&0&1&1\\0&0&1&0\end{smallmatrix}$ \\[1.85ex]
\hline
$446$ \rule{0cm}{.75cm}  &
$\begin{smallmatrix}0&0&0&1\\0&0&0&1\\0&0&0&1\\1&1&1&1\end{smallmatrix}$ &
$\begin{smallmatrix}1&0&0&1\\0&1&0&1\\0&0&0&1\\1&1&1&1\end{smallmatrix}$ &
$\begin{smallmatrix}0&0&0&0&1\\0&0&0&0&1\\0&0&0&1&0\\0&0&1&0&0\\1&1&0&0&0\end{smallmatrix}$ &
$\begin{smallmatrix}0&0&1&0\\0&0&0&1\\1&0&0&1\\0&1&1&1\end{smallmatrix}$ &
$\begin{smallmatrix}0&0&1&0\\0&1&0&0\\1&0&0&0\\0&0&0&1\end{smallmatrix}$ &
$\begin{smallmatrix}1&0&0&1\\0&0&0&0\\0&0&0&0\\1&0&0&0\end{smallmatrix}$ \\[2.45ex]
\hline
$452$ \rule{0cm}{.65cm}  &
$\begin{smallmatrix}0&0&1&1\\0&0&1&1\\1&1&1&1\\1&1&1&1\end{smallmatrix}$ &
$\begin{smallmatrix}0\end{smallmatrix}$ &
$\begin{smallmatrix}0&0&1&0\\0&0&0&0\\1&0&0&1\\0&0&1&1\end{smallmatrix}$ &
$\begin{smallmatrix}1&0&1\\0&1&0\\1&0&0\end{smallmatrix}$ &
$\begin{smallmatrix}0&1&1&1\\1&1&0&0\\1&0&1&0\\1&0&0&1\end{smallmatrix}$ &
$\begin{smallmatrix}1\end{smallmatrix}$ \\[1.85ex]
\hline
$458$ \rule{0cm}{.65cm}  &
$\begin{smallmatrix}0&0&0&1\\0&1&0&0\\0&0&0&1\\1&0&1&0\end{smallmatrix}$ &
$\begin{smallmatrix}0&0&0&1\\0&0&1&1\\0&1&0&1\\1&1&1&1\end{smallmatrix}$ &
$\begin{smallmatrix}0&0&1&0\\0&1&0&1\\1&0&0&0\\0&1&0&0\end{smallmatrix}$ &
$\begin{smallmatrix}1&0&0&1\\0&1&1&1\\0&1&1&0\\1&1&0&1\end{smallmatrix}$ &
$\begin{smallmatrix}0&0&0&1\\0&0&0&0\\0&0&0&1\\1&0&1&0\end{smallmatrix}$ &
$\begin{smallmatrix}0&0&1&0\\0&0&0&1\\1&0&0&0\\0&1&0&0\end{smallmatrix}$ \\[1.85ex]
\hline
$464$ \rule{0cm}{.75cm}  &
$\begin{smallmatrix}0&1&1\\1&0&1\\1&1&1\end{smallmatrix}$ &
$\begin{smallmatrix}0&0&0&0&1\\0&0&1&1&1\\0&1&0&1&1\\0&1&1&0&0\\1&1&1&0&1\end{smallmatrix}$ &
$\begin{smallmatrix}1&1&0\\1&1&1\\0&1&0\end{smallmatrix}$ &
$\begin{smallmatrix}0&0&1&0\\0&1&0&1\\1&0&0&1\\0&1&1&0\end{smallmatrix}$ &
$\begin{smallmatrix}0&0&0&1\\0&0&1&1\\0&1&0&0\\1&1&0&0\end{smallmatrix}$ &
$\begin{smallmatrix}0&0&1&0\\0&1&0&0\\1&0&1&0\\0&0&0&0\end{smallmatrix}$ \\[2.45ex]
\hline
$470$ \rule{0cm}{.75cm}  &
$\begin{smallmatrix}0&0&0&0&1\\0&1&0&1&0\\0&0&0&1&1\\0&1&1&0&1\\1&0&1&1&1\end{smallmatrix}$ &
$\begin{smallmatrix}0&0&0&0&1\\0&0&1&0&1\\0&1&0&0&0\\0&0&0&1&0\\1&1&0&0&0\end{smallmatrix}$ &
$\begin{smallmatrix}1&1&0&0\\1&1&0&0\\0&0&0&0\\0&0&0&1\end{smallmatrix}$ &
$\begin{smallmatrix}0\end{smallmatrix}$ &
$\begin{smallmatrix}0&0&0&1&0\\0&0&0&0&0\\0&0&1&1&0\\1&0&1&1&1\\0&0&0&1&0\end{smallmatrix}$ &
$\begin{smallmatrix}0&0&0&0&1\\0&1&0&1&1\\0&0&0&1&0\\0&1&1&1&0\\1&1&0&0&1\end{smallmatrix}$ \\[2.45ex]
\hline
$476$ \rule{0cm}{.65cm}  &
$\begin{smallmatrix}1&0&1\\0&1&0\\1&0&0\end{smallmatrix}$ &
$\begin{smallmatrix}0&0&1&1\\0&1&0&0\\1&0&1&1\\1&0&1&0\end{smallmatrix}$ &
$\begin{smallmatrix}1&1&1\\1&1&0\\1&0&0\end{smallmatrix}$ &
$\begin{smallmatrix}0&0&1&1\\0&1&0&1\\1&0&1&1\\1&1&1&1\end{smallmatrix}$ &
$\begin{smallmatrix}1&0&1&0\\0&0&0&0\\1&0&1&1\\0&0&1&1\end{smallmatrix}$ &
$\begin{smallmatrix}1&0&1&0\\0&0&0&0\\1&0&0&1\\0&0&1&0\end{smallmatrix}$ \\[1.85ex]
\hline
$482$ \rule{0cm}{.75cm}  &
$\begin{smallmatrix}0&0&0&1\\0&1&1&0\\0&1&1&0\\1&0&0&1\end{smallmatrix}$ &
$\begin{smallmatrix}1&1&1&1\\1&0&0&0\\1&0&0&1\\1&0&1&1\end{smallmatrix}$ &
$\begin{smallmatrix}0&1&1&0\\1&0&0&0\\1&0&1&0\\0&0&0&1\end{smallmatrix}$ &
$\begin{smallmatrix}0&0&0&0&1\\0&1&0&1&1\\0&0&1&0&0\\0&1&0&1&0\\1&1&0&0&0\end{smallmatrix}$ &
$\begin{smallmatrix}0&1&1&0\\1&0&1&0\\1&1&1&0\\0&0&0&1\end{smallmatrix}$ &
$\begin{smallmatrix}0&1&0\\1&1&0\\0&0&1\end{smallmatrix}$ \\[2.45ex]
\hline
$488$ \rule{0cm}{.75cm}  &
$\begin{smallmatrix}0&0&1&0\\0&0&1&0\\1&1&1&0\\0&0&0&1\end{smallmatrix}$ &
$\begin{smallmatrix}0&0&1&0\\0&0&1&0\\1&1&0&1\\0&0&1&1\end{smallmatrix}$ &
$\begin{smallmatrix}1&0&0&0\\0&1&1&0\\0&1&1&1\\0&0&1&0\end{smallmatrix}$ &
$\begin{smallmatrix}0\end{smallmatrix}$ &
$\begin{smallmatrix}0&0&1&1\\0&1&0&0\\1&0&0&1\\1&0&1&0\end{smallmatrix}$ &
$\begin{smallmatrix}0&0&0&0&1\\0&0&0&0&1\\0&0&1&1&0\\0&0&1&1&1\\1&1&0&1&1\end{smallmatrix}$ \\[2.45ex]
\hline
$494$ \rule{0cm}{.75cm}  &
$\begin{smallmatrix}1&0&0&1\\0&1&0&1\\0&0&1&0\\1&1&0&0\end{smallmatrix}$ &
$\begin{smallmatrix}0\end{smallmatrix}$ &
$\begin{smallmatrix}0&0&0&1\\0&0&0&1\\0&0&1&0\\1&1&0&0\end{smallmatrix}$ &
$\begin{smallmatrix}0&0&0&0&1\\0&0&0&1&0\\0&0&1&1&0\\0&1&1&0&1\\1&0&0&1&1\end{smallmatrix}$ &
$\begin{smallmatrix}1&0&1&0\\0&0&1&0\\1&1&1&0\\0&0&0&1\end{smallmatrix}$ &
$\begin{smallmatrix}0&0&1&1\\0&1&0&1\\1&0&1&1\\1&1&1&1\end{smallmatrix}$ \\[2.45ex]
\hline
$500$ \rule{0cm}{.75cm}  &
$\begin{smallmatrix}0&0&1&1\\0&0&1&1\\1&1&1&0\\1&1&0&0\end{smallmatrix}$ &
$\begin{smallmatrix}0&0&0&1\\0&1&0&1\\0&0&0&0\\1&1&0&0\end{smallmatrix}$ &
$\begin{smallmatrix}0&0&0&0&1\\0&0&0&0&0\\0&0&0&1&1\\0&0&1&0&0\\1&0&1&0&0\end{smallmatrix}$ &
$\begin{smallmatrix}1&0&0&0\\0&1&1&1\\0&1&0&0\\0&1&0&0\end{smallmatrix}$ &
$\begin{smallmatrix}1&0&0\\0&0&1\\0&1&0\end{smallmatrix}$ &
$\begin{smallmatrix}1&0&1\\0&0&1\\1&1&1\end{smallmatrix}$ \\[2.45ex]
\hline
$506$ \rule{0cm}{.75cm}  &
$\begin{smallmatrix}0&0&0&1&0\\0&0&0&1&0\\0&0&1&0&0\\1&1&0&1&1\\0&0&0&1&0\end{smallmatrix}$ &
$\begin{smallmatrix}1&1&1&1\\1&0&0&0\\1&0&1&0\\1&0&0&1\end{smallmatrix}$ &
$\begin{smallmatrix}0&0&0&0&1\\0&0&0&1&0\\0&0&0&0&0\\0&1&0&0&0\\1&0&0&0&0\end{smallmatrix}$ &
$\begin{smallmatrix}1&0&1&1\\0&1&0&1\\1&0&0&1\\1&1&1&1\end{smallmatrix}$ &
$\begin{smallmatrix}0&0&0&0&1\\0&0&0&0&0\\0&0&1&0&0\\0&0&0&0&0\\1&0&0&0&1\end{smallmatrix}$ &
$\begin{smallmatrix}0&0&1&0\\0&0&1&0\\1&1&1&1\\0&0&1&1\end{smallmatrix}$ \\[2.45ex]
\hline
$512$ \rule{0cm}{.65cm}  &
$\begin{smallmatrix}1&0&0&0\\0&1&0&1\\0&0&1&0\\0&1&0&1\end{smallmatrix}$ &
$\begin{smallmatrix}1&0&0&1\\0&0&0&0\\0&0&0&0\\1&0&0&0\end{smallmatrix}$ &
$\begin{smallmatrix}0&1&0&0\\1&1&0&1\\0&0&0&1\\0&1&1&0\end{smallmatrix}$ &
$\begin{smallmatrix}0\end{smallmatrix}$ &
$\begin{smallmatrix}0&1&0&0\\1&0&1&0\\0&1&1&1\\0&0&1&0\end{smallmatrix}$ &
$\begin{smallmatrix}0&0&1&1\\0&1&0&0\\1&0&1&1\\1&0&1&0\end{smallmatrix}$ \\[1.85ex]
\hline
$518$ \rule{0cm}{.75cm}  &
$\begin{smallmatrix}0&0&0&0&1\\0&0&0&1&0\\0&0&0&1&0\\0&1&1&1&0\\1&0&0&0&1\end{smallmatrix}$ &
$\begin{smallmatrix}0\end{smallmatrix}$ &
$\begin{smallmatrix}1&0\\0&1\end{smallmatrix}$ &
$\begin{smallmatrix}0&0&0&0&1\\0&1&0&0&1\\0&0&0&0&0\\0&0&0&1&0\\1&1&0&0&1\end{smallmatrix}$ &
$\begin{smallmatrix}0&0&0&0&1\\0&1&0&1&0\\0&0&0&0&1\\0&1&0&0&0\\1&0&1&0&0\end{smallmatrix}$ &
$\begin{smallmatrix}1&0&0&0\\0&0&0&0\\0&0&0&0\\0&0&0&0\end{smallmatrix}$ \\[2.45ex]
\hline
\end{tabular}
\end{table}
\pagebreak

\begin{table}[h!]
\begin{tabular}{|c||c|c|c|c|c|c|}
\hline
$524$ \rule{0cm}{.65cm}  &
$\begin{smallmatrix}1&0&1\\0&1&0\\1&0&0\end{smallmatrix}$ &
$\begin{smallmatrix}0&1&0&1\\1&1&0&1\\0&0&1&0\\1&1&0&1\end{smallmatrix}$ &
$\begin{smallmatrix}0&0&1\\0&0&0\\1&0&0\end{smallmatrix}$ &
$\begin{smallmatrix}0&1&0&0\\1&1&0&1\\0&0&0&1\\0&1&1&0\end{smallmatrix}$ &
$\begin{smallmatrix}1&0&0&1\\0&0&0&1\\0&0&0&0\\1&1&0&1\end{smallmatrix}$ &
$\begin{smallmatrix}0&0&0&1\\0&0&1&1\\0&1&1&0\\1&1&0&0\end{smallmatrix}$ \\[1.85ex]
\hline
$530$ \rule{0cm}{.65cm}  &
$\begin{smallmatrix}0\end{smallmatrix}$ &
$\begin{smallmatrix}0\end{smallmatrix}$ &
$\begin{smallmatrix}1&1&0&0\\1&0&1&1\\0&1&1&1\\0&1&1&1\end{smallmatrix}$ &
$\begin{smallmatrix}0&0&1&1\\0&1&0&1\\1&0&1&1\\1&1&1&1\end{smallmatrix}$ &
$\begin{smallmatrix}1&0&0&1\\0&1&1&1\\0&1&0&1\\1&1&1&0\end{smallmatrix}$ &
$\begin{smallmatrix}1&0&0&1\\0&0&1&1\\0&1&0&1\\1&1&1&1\end{smallmatrix}$ \\[1.85ex]
\hline
$536$ \rule{0cm}{.65cm}  &
$\begin{smallmatrix}1&1&1\\1&1&0\\1&0&0\end{smallmatrix}$ &
$\begin{smallmatrix}1&0&0&1\\0&0&0&1\\0&0&0&1\\1&1&1&1\end{smallmatrix}$ &
$\begin{smallmatrix}0&0&0&1\\0&0&0&1\\0&0&1&0\\1&1&0&0\end{smallmatrix}$ &
$\begin{smallmatrix}0&1&0&0\\1&0&0&0\\0&0&1&0\\0&0&0&0\end{smallmatrix}$ &
$\begin{smallmatrix}0&0&0&1\\0&1&1&0\\0&1&1&1\\1&0&1&0\end{smallmatrix}$ &
$\begin{smallmatrix}0&1&1&1\\1&0&1&1\\1&1&0&0\\1&1&0&1\end{smallmatrix}$ \\[1.85ex]
\hline
$542$ \rule{0cm}{.75cm}  &
$\begin{smallmatrix}0&0&1&1\\0&1&1&0\\1&1&0&0\\1&0&0&0\end{smallmatrix}$ &
$\begin{smallmatrix}0\end{smallmatrix}$ &
$\begin{smallmatrix}0&0&0&0&1\\0&0&0&1&0\\0&0&1&1&0\\0&1&1&1&1\\1&0&0&1&0\end{smallmatrix}$ &
$\begin{smallmatrix}0\end{smallmatrix}$ &
$\begin{smallmatrix}0&0&0&0&1\\0&0&1&0&0\\0&1&0&1&1\\0&0&1&1&0\\1&0&1&0&0\end{smallmatrix}$ &
$\begin{smallmatrix}1&0&1&1\\0&0&0&1\\1&0&0&1\\1&1&1&1\end{smallmatrix}$ \\[2.45ex]
\hline
$548$ \rule{0cm}{.75cm}  &
$\begin{smallmatrix}0&1&0&1\\1&0&0&1\\0&0&1&0\\1&1&0&0\end{smallmatrix}$ &
$\begin{smallmatrix}1&0&0&1\\0&1&1&0\\0&1&0&0\\1&0&0&1\end{smallmatrix}$ &
$\begin{smallmatrix}0&0&1&0\\0&1&0&0\\1&0&1&0\\0&0&0&0\end{smallmatrix}$ &
$\begin{smallmatrix}0&0&0&1\\0&1&1&1\\0&1&1&1\\1&1&1&0\end{smallmatrix}$ &
$\begin{smallmatrix}1&0&0&1\\0&1&1&1\\0&1&0&0\\1&1&0&1\end{smallmatrix}$ &
$\begin{smallmatrix}0&0&0&0&1\\0&0&0&0&1\\0&0&1&1&1\\0&0&1&0&1\\1&1&1&1&0\end{smallmatrix}$ \\[2.45ex]
\hline
$554$ \rule{0cm}{.75cm}  &
$\begin{smallmatrix}0&0&0&1&0\\0&1&1&1&0\\0&1&0&1&0\\1&1&1&0&0\\0&0&0&0&1\end{smallmatrix}$ &
$\begin{smallmatrix}0&1&0&1\\1&0&1&1\\0&1&0&0\\1&1&0&1\end{smallmatrix}$ &
$\begin{smallmatrix}1&0&0&0\\0&0&0&1\\0&0&1&0\\0&1&0&0\end{smallmatrix}$ &
$\begin{smallmatrix}1&0&0&1\\0&0&1&0\\0&1&1&0\\1&0&0&0\end{smallmatrix}$ &
$\begin{smallmatrix}0\end{smallmatrix}$ &
$\begin{smallmatrix}0&1&0\\1&0&0\\0&0&1\end{smallmatrix}$ \\[2.45ex]
\hline
$560$ \rule{0cm}{.75cm}  &
$\begin{smallmatrix}0&0&0&1\\0&0&1&1\\0&1&0&1\\1&1&1&1\end{smallmatrix}$ &
$\begin{smallmatrix}1&0&1&0\\0&0&0&0\\1&0&0&0\\0&0&0&1\end{smallmatrix}$ &
$\begin{smallmatrix}0&1&1&0\\1&1&0&0\\1&0&0&1\\0&0&1&0\end{smallmatrix}$ &
$\begin{smallmatrix}0&0&0&1&0\\0&0&0&0&0\\0&0&1&1&1\\1&0&1&1&0\\0&0&1&0&0\end{smallmatrix}$ &
$\begin{smallmatrix}0&0&0&0&1\\0&0&1&1&1\\0&1&1&0&0\\0&1&0&1&1\\1&1&0&1&1\end{smallmatrix}$ &
$\begin{smallmatrix}0&1&1&0\\1&1&0&0\\1&0&1&1\\0&0&1&0\end{smallmatrix}$ \\[2.45ex]
\hline
$566$ \rule{0cm}{.75cm}  &
$\begin{smallmatrix}0&0&0&1\\0&0&1&0\\0&1&0&1\\1&0&1&0\end{smallmatrix}$ &
$\begin{smallmatrix}1&0&0&1\\0&1&1&0\\0&1&0&1\\1&0&1&0\end{smallmatrix}$ &
$\begin{smallmatrix}0&0&1&1\\0&0&0&1\\1&0&1&0\\1&1&0&0\end{smallmatrix}$ &
$\begin{smallmatrix}0&0&1&0\\0&1&0&1\\1&0&1&1\\0&1&1&0\end{smallmatrix}$ &
$\begin{smallmatrix}0&0&0&0&1\\0&0&0&1&0\\0&0&1&1&0\\0&1&1&0&1\\1&0&0&1&1\end{smallmatrix}$ &
$\begin{smallmatrix}0&1&0&0\\1&0&1&0\\0&1&1&0\\0&0&0&0\end{smallmatrix}$ \\[2.45ex]
\hline
$572$ \rule{0cm}{.75cm}  &
$\begin{smallmatrix}0&0&0&0&1\\0&0&0&1&0\\0&0&0&1&0\\0&1&1&1&0\\1&0&0&0&1\end{smallmatrix}$ &
$\begin{smallmatrix}1&0&1\\0&1&0\\1&0&0\end{smallmatrix}$ &
$\begin{smallmatrix}1&0&1\\0&0&1\\1&1&1\end{smallmatrix}$ &
$\begin{smallmatrix}1&0&0&0\\0&0&1&0\\0&1&0&0\\0&0&0&1\end{smallmatrix}$ &
$\begin{smallmatrix}1&0&0&1\\0&1&1&0\\0&1&0&0\\1&0&0&0\end{smallmatrix}$ &
$\begin{smallmatrix}0&0&0&0&1\\0&1&1&1&1\\0&1&1&0&0\\0&1&0&0&0\\1&1&0&0&1\end{smallmatrix}$ \\[2.45ex]
\hline
$578$ \rule{0cm}{.75cm}  &
$\begin{smallmatrix}1\end{smallmatrix}$ &
$\begin{smallmatrix}0&0&0&1\\0&1&0&0\\0&0&1&0\\1&0&0&1\end{smallmatrix}$ &
$\begin{smallmatrix}1&0&0&0\\0&0&1&0\\0&1&0&0\\0&0&0&1\end{smallmatrix}$ &
$\begin{smallmatrix}0&1&0&1\\1&0&0&1\\0&0&1&0\\1&1&0&1\end{smallmatrix}$ &
$\begin{smallmatrix}0&0&0&0&1\\0&0&1&0&0\\0&1&0&1&0\\0&0&1&0&1\\1&0&0&1&0\end{smallmatrix}$ &
$\begin{smallmatrix}1&1&0&0\\1&0&1&1\\0&1&0&1\\0&1&1&1\end{smallmatrix}$ \\[2.45ex]
\hline
$584$ \rule{0cm}{.75cm}  &
$\begin{smallmatrix}1&0&1&1\\0&1&0&0\\1&0&0&1\\1&0&1&0\end{smallmatrix}$ &
$\begin{smallmatrix}0&0&0&0&1\\0&0&1&0&1\\0&1&1&0&1\\0&0&0&1&0\\1&1&1&0&1\end{smallmatrix}$ &
$\begin{smallmatrix}0&0&0&0&1\\0&1&0&0&1\\0&0&1&0&1\\0&0&0&0&0\\1&1&1&0&0\end{smallmatrix}$ &
$\begin{smallmatrix}0&0&0&0&1\\0&0&0&0&0\\0&0&1&1&0\\0&0&1&0&1\\1&0&0&1&1\end{smallmatrix}$ &
$\begin{smallmatrix}1&0&0&1\\0&0&0&0\\0&0&0&1\\1&0&1&1\end{smallmatrix}$ &
$\begin{smallmatrix}0&0&0&0&1\\0&0&1&1&1\\0&1&0&0&1\\0&1&0&1&0\\1&1&1&0&1\end{smallmatrix}$ \\[2.45ex]
\hline
$590$ \rule{0cm}{.75cm}  &
$\begin{smallmatrix}1&0&1&1\\0&0&1&0\\1&1&1&0\\1&0&0&1\end{smallmatrix}$ &
$\begin{smallmatrix}0&0&0&0&1\\0&0&0&0&1\\0&0&1&0&0\\0&0&0&1&0\\1&1&0&0&0\end{smallmatrix}$ &
$\begin{smallmatrix}0&0&1&1\\0&1&0&0\\1&0&1&1\\1&0&1&0\end{smallmatrix}$ &
$\begin{smallmatrix}0\end{smallmatrix}$ &
$\begin{smallmatrix}0&1&0&0\\1&1&1&1\\0&1&1&1\\0&1&1&0\end{smallmatrix}$ &
$\begin{smallmatrix}0&0&0&0&1\\0&0&0&0&0\\0&0&1&1&1\\0&0&1&1&1\\1&0&1&1&0\end{smallmatrix}$ \\[2.45ex]
\hline
$596$ \rule{0cm}{.65cm}  &
$\begin{smallmatrix}0&0&1&0\\0&1&0&0\\1&0&1&0\\0&0&0&0\end{smallmatrix}$ &
$\begin{smallmatrix}0&0&1&0\\0&1&0&0\\1&0&0&0\\0&0&0&1\end{smallmatrix}$ &
$\begin{smallmatrix}0&1&1&1\\1&0&1&0\\1&1&0&0\\1&0&0&1\end{smallmatrix}$ &
$\begin{smallmatrix}1&0&0&0\\0&1&0&0\\0&0&1&0\\0&0&0&1\end{smallmatrix}$ &
$\begin{smallmatrix}0&1&1&1\\1&0&0&0\\1&0&1&0\\1&0&0&0\end{smallmatrix}$ & \\[1.85ex]
\hline
\end{tabular}
\caption{This table lists the matrices $A$ which correspond to the lower right corner of the matrices $B\in M_m(\mathbb{F}_2)$ for $m=\Mg{2,\ldots, 600}$
 in the form of Equation \eqref{eqn:fibset:A} satisfying Conditions (i) and (ii') as discussed in Chapter \ref{sec:fibset} and can be used to construct
 complete sets of Fibonacci-based cyclic MUBs for dimensions $2^m$.\label{table:app:triangle}}
\end{table}

\subsection{Companion matrix solutions}\label{app:fibonacci_based:companion}
In Chapter \ref{subsec:fibset:analytical} we conjecture the existence of a symmetric companion matrix
over $\GF{2}$. As long as we do not have the specific form, we propose found solutions for dimensions $d=2^m$
with $m=\Mg{2,\ldots, 36}$. Contrary to the solutions of Appendix \ref{app:fibonacci_based:triangle}, the
number of matrices we have to test before getting a correct solution seems to scale worse. According to
the representation of the solutions, as in Equation \eqref{eqn:fibset:antidiag}, we write the solutions
in Table \ref{table:app:antidiag} as a binary string, starting from left with $s_1$ and ending with the
last entry that equals one.

\begin{table}[h!]
\begin{tabular}{|c|l|}
\hline
$m$ & $s$\\
\hline
$2$ & 11\\
\hline
$4$ & 1101\\
\hline
$5$ & 01101\\
\hline
$6$ & 001101\\
\hline
$7$ & 1101001\\
\hline
$8$ & 11010001\\
\hline
$9$ & 011001001\\
\hline
$10$ & 0001100001\\
\hline
$11$ & 01100100001\\
\hline
$12$ & 001101000001\\
\hline
$13$ & 1100101000001\\
\hline
$14$ & 10100110000001\\
\hline
$15$ & 110100010000001\\
\hline
$16$ & 1101000100000001\\
\hline
$17$ & 10011000100000001\\
\hline
$18$ & 000000011000000001\\
\hline
$19$ & 1110100001000000001\\
\hline
$20$ & 00011000010000000001\\
\hline
$21$ & 111101000010000000001\\
\hline
$22$ & 0001111000100000000001\\
\hline
$23$ & 01100100000100000000001\\
\hline
$24$ & 001101000001000000000001\\
\hline
$25$ & 1000011000001000000000001\\
\hline
$26$ & 01100000100010000000000001\\
\hline
$27$ & 110010100000010000000000001\\
\hline
$28$ & 1010011000000100000000000001\\
\hline
$29$ & 01010001000000100000000000001\\
\hline
$30$ & 110110111000001000000000000001\\
\hline
$31$ & 0000110100000001000000000000001\\
\hline
$32$ & 11010001000000010000000000000001\\
\hline
$33$ & 010110001000000010000000000000001\\
\hline
$34$ & 0100100011000000100000000000000001\\
\hline
$35$ & 10110100100000000100000000000000001\\
\hline
$36$ & 100111001000000001000000000000000001\\
\hline
\end{tabular}
\caption{Solutions for the symmetric companion matrix\index{Symmetric companion matrix} which is discussed
in Section~\ref{subsec:fibset:analytical} and represents a reduced stabilizer matrix $B$ in order
to construct a complete set of Fibonacci-based cyclic MUBs (cf. Section \ref{sec:fibset}) for dimensions
$d=2^m$ with $m=2,\ldots,36$. The solutions are represented as a binary string, starting from left with
$s_1$ and ending with the last entry that equals one, corresponding to Equation \eqref{eqn:fibset:antidiag}.
\label{table:app:antidiag}}
\end{table}

\section{Fractal patterns}\label{app:fractals}
In Section \ref{sec:fibonacci_polynomials}, properties of the Fibonacci polynomials over $\GF{2}$ are discussed.
The structure of these polynomials is shown in Table \ref{table:app:Fibpattern} and related to the Sierpinski triangle.
Also the characteristic polynomials of matrices over $\GF{2}$, that have only ones in the upper left half and
zeros otherwise (thus, setting $A$ to zero in Equation \eqref{eqn:fibset:A}), show an equivalent pattern (see Table \ref{table:app:Apattern}).
Removing from Table \ref{table:app:Fibpattern} every second line and every second column leads to (an obviously smaller version of) Table \ref{table:app:Apattern}.
Both are related to Pascal's triangle over $\GF{2}$ (see Table \ref{table:app:Pascal}).
To have a good impression, we encode ones by ``\#'' and zeros by a space.

\begin{table}[ht!]
\begingroup
    \fontsize{8pt}{8pt}
\begin{verbatim}
#
 #
# #
   #
# # #
 #   #
#   # #
       #
#   # # #
 #   #   #
# # #   # #
   #       #
# #     # # #
 #       #   #
#       #   # #
               #
#       #   # # #
 #       #   #   #
# #     # # #   # #
   #       #       #
# # #   # #     # # #
 #   #   #       #   #
#   # # #       #   # #
       #               #
#   # #         #   # # #
 #   #           #   #   #
# # #           # # #   # #
   #               #       #
# #             # #     # # #
 #               #       #   #
#               #       #   # #
                               #
#               #       #   # # #
 #               #       #   #   #
# #             # #     # # #   # #
   #               #       #       #
# # #           # # #   # #     # # #
 #   #           #   #   #       #   #
#   # #         #   # # #       #   # #
       #               #               #
#   # # #       #   # #         #   # # #
 #   #   #       #   #           #   #   #
# # #   # #     # # #           # # #   # #
   #       #       #               #       #
# #     # # #   # #             # #     # # #
 #       #   #   #               #       #   #
#       #   # # #               #       #   # #
               #                               #
#       #   # #                 #       #   # # #
 #       #   #                   #       #   #   #
# #     # # #                   # #     # # #   # #
   #       #                       #       #       #
# # #   # #                     # # #   # #     # # #
 #   #   #                       #   #   #       #   #
#   # # #                       #   # # #       #   # #
       #                               #               #
#   # #                         #   # #         #   # # #
 #   #                           #   #           #   #   #
# # #                           # # #           # # #   # #
   #                               #               #       #
# #                             # #             # #     # # #
 #                               #               #       #   #
#                               #               #       #   # #
                                                               #
#                               #               #       #   # # #
 #                               #               #       #   #   #
# #                             # #             # #     # # #   # #
   #                               #               #       #       #
# # #                           # # #           # # #   # #     # # #
 #   #                           #   #           #   #   #       #   #
\end{verbatim}
\endgroup
\caption{Fibonacci polynomials over $\GF{2}$, where ``\#'' encodes a coefficient
that equals one in the polynomial, and a space encodes a zero. The rightmost coefficient is the
highest, the rows indicate the index $j$ of the Fibonacci polynomial $F_j(x)$, starting with
$1$ and ending with $70$.}\label{table:app:Fibpattern}
\end{table}

\begin{table}[ht!]
\begingroup
\fontsize{8pt}{8pt}
\begin{verbatim}
##
###
# ##
# ###
### ##
##  ###
#   # ##
#   # ###
##  ### ##
### ##  ###
# ###   # ##
# ##    # ###
###     ### ##
##      ##  ###
#       #   # ##
#       #   # ###
##      ##  ### ##
###     ### ##  ###
# ##    # ###   # ##
# ###   # ##    # ###
### ##  ###     ### ##
##  ### ##      ##  ###
#   # ###       #   # ##
#   # ##        #   # ###
##  ###         ##  ### ##
### ##          ### ##  ###
# ###           # ###   # ##
# ##            # ##    # ###
###             ###     ### ##
##              ##      ##  ###
#               #       #   # ##
#               #       #   # ###
##              ##      ##  ### ##
###             ###     ### ##  ###
# ##            # ##    # ###   # ##
# ###           # ###   # ##    # ###
### ##          ### ##  ###     ### ##
##  ###         ##  ### ##      ##  ###
#   # ##        #   # ###       #   # ##
#   # ###       #   # ##        #   # ###
##  ### ##      ##  ###         ##  ### ##
### ##  ###     ### ##          ### ##  ###
# ###   # ##    # ###           # ###   # ##
# ##    # ###   # ##            # ##    # ###
###     ### ##  ###             ###     ### ##
##      ##  ### ##              ##      ##  ###
#       #   # ###               #       #   # ##
#       #   # ##                #       #   # ###
##      ##  ###                 ##      ##  ### ##
###     ### ##                  ###     ### ##  ###
# ##    # ###                   # ##    # ###   # ##
# ###   # ##                    # ###   # ##    # ###
### ##  ###                     ### ##  ###     ### ##
##  ### ##                      ##  ### ##      ##  ###
#   # ###                       #   # ###       #   # ##
#   # ##                        #   # ##        #   # ###
##  ###                         ##  ###         ##  ### ##
### ##                          ### ##          ### ##  ###
# ###                           # ###           # ###   # ##
# ##                            # ##            # ##    # ###
###                             ###             ###     ### ##
##                              ##              ##      ##  ###
#                               #               #       #   # ##
#                               #               #       #   # ###
##                              ##              ##      ##  ### ##
###                             ###             ###     ### ##  ###
# ##                            # ##            # ##    # ###   # ##
# ###                           # ###           # ###   # ##    # ###
### ##                          ### ##          ### ##  ###     ### ##
##  ###                         ##  ###         ##  ### ##      ##  ###
\end{verbatim}
\endgroup
\caption{Characteristic polynomials of matrices over $\GF{2}$ with ones in the upper left half and
zeros otherwise ($A$ in Equation \eqref{eqn:fibset:A} is set to zero). ``\#'' encodes a coefficient
that equals one in the polynomial, whereas a space encodes a zero. The leftmost coefficient is the
highest, the rows indicate the dimension $m$ of the matrix $B$, starting with $1$ and ending with $70$.}\label{table:app:Apattern}
\end{table}

\begin{table}[ht!]
\begingroup
\fontsize{8pt}{8pt}
\begin{verbatim}
#
##
# #
####
#   #
##  ##
# # # #
########
#       #
##      ##
# #     # #
####    ####
#   #   #   #
##  ##  ##  ##
# # # # # # # #
################
#               #
##              ##
# #             # #
####            ####
#   #           #   #
##  ##          ##  ##
# # # #         # # # #
########        ########
#       #       #       #
##      ##      ##      ##
# #     # #     # #     # #
####    ####    ####    ####
#   #   #   #   #   #   #   #
##  ##  ##  ##  ##  ##  ##  ##
# # # # # # # # # # # # # # # #
################################
#                               #
##                              ##
# #                             # #
####                            ####
#   #                           #   #
##  ##                          ##  ##
# # # #                         # # # #
########                        ########
#       #                       #       #
##      ##                      ##      ##
# #     # #                     # #     # #
####    ####                    ####    ####
#   #   #   #                   #   #   #   #
##  ##  ##  ##                  ##  ##  ##  ##
# # # # # # # #                 # # # # # # # #
################                ################
#               #               #               #
##              ##              ##              ##
# #             # #             # #             # #
####            ####            ####            ####
#   #           #   #           #   #           #   #
##  ##          ##  ##          ##  ##          ##  ##
# # # #         # # # #         # # # #         # # # #
########        ########        ########        ########
#       #       #       #       #       #       #       #
##      ##      ##      ##      ##      ##      ##      ##
# #     # #     # #     # #     # #     # #     # #     # #
####    ####    ####    ####    ####    ####    ####    ####
#   #   #   #   #   #   #   #   #   #   #   #   #   #   #   #
##  ##  ##  ##  ##  ##  ##  ##  ##  ##  ##  ##  ##  ##  ##  ##
# # # # # # # # # # # # # # # # # # # # # # # # # # # # # # # #
################################################################
#                                                               #
##                                                              ##
# #                                                             # #
####                                                            ####
#   #                                                           #   #
##  ##                                                          ##  ##
# # # #                                                         # # # #
\end{verbatim}
\endgroup
\caption{Pascal's  triangle over $\GF{2}$, where ``\#'' encodes a binomial coefficient
that equals one and a space encodes a zero. For a binomial coefficient $\binom{n}{k}$ with $n,k \in \N$
and $k\leq n$, the columns indicate different values of $k$, where the ordering is symmetric and
the rows indicate the values for $n$, starting with $0$ and ending with $70$.}\label{table:app:Pascal}
\end{table}

\bibliographystyle{scriptum} 
\bibliography{references}

\printindex

\selectlanguage{german}
%
%
%

\end{document}